%% file: Hcolorings-23-12-15arxiv.tex
\documentclass[11pt]{article}

\usepackage{amsmath,amssymb,amsthm,hyperref,color}
\hypersetup{colorlinks=true,citecolor=blue,urlcolor=blue, linkcolor=blue}

\usepackage[margin = 1in]{geometry}
\usepackage{float}
\usepackage{caption}
\usepackage{subcaption}
\usepackage{graphics}
\usepackage{tikz}
\usepackage{pgfplots}
\newtheorem{theorem}{Theorem}

\newtheorem{lemma}[theorem]{Lemma}

\newtheorem{remark}[theorem]{Remark}

\newtheorem{definition}[theorem]{Definition}
 
 \newtheorem*{thmmain}{Theorem \ref{thm:main}}
 \newtheorem*{lemmain}{Lemma \ref{lem:main}}
 \newtheorem*{lemfixingHCOL}{Lemma \ref{lem:fixingHCOL}}
 \newtheorem*{lovaszone}{Lemma \ref{lem:lovaszone}}
 \newtheorem*{lovasztwo}{Lemma \ref{lem:lovasztwo}}

\def\ex{\mathrm{ex}}

\def\epsilon{\varepsilon}

\newcommand\COL[1]{\##1\mbox{-}\mathsf{Col}}

\newcommand\FixCOL[1]{\#\mathsf{Fixed}#1\mbox{-}\mathsf{Col}}
\newcommand\InjFixCOL[1]{\#\mathsf{InjFixed}#1\mbox{-}\mathsf{Col}}

\def\BIS{\#\mathrm{BIS}}
\def\sharpP{\#\mathsf{P}}
\def\FPRAS{\mathsf{FPRAS}}
\def\FPAS{\mathsf{FPAS}}
\def\PAS{\mathsf{PAS}}
\def\AP{\mathsf{AP}}
\def\NP{\mathsf{NP}}
\def\R{\Lambda}
\def\r{\lambda}
\def\j{\cap}
\def\Nu{N_{\cup}}

\def\Bip{\mathsf{Bip}}

\def\Cc{\mathcal{C}}

\def\Hc{\mathcal{H}}

\newcommand{\Lovasz}{Lov\'asz}
\newcommand{\Nesetril}{Ne\v{s}et\v{r}il}

\begin{document} 

\title{Approximately Counting $H$-Colourings is $\BIS$-Hard\thanks{These results 
were
presented in preliminary form in the proceedings of \textit{ICALP 2015} (Track A).}
}
\author{
Andreas Galanis\thanks{
  Department of Computer Science, University of Oxford, Wolfson Building, Parks Road, Oxford, OX1~3QD, UK.
  The research leading to these results has received funding from the European Research Council under
  the European Union's Seventh Framework Programme (FP7/2007-2013) ERC grant agreement no.\ 334828. The paper
  reflects only the authors' views and not the views of the ERC or the European Commission.
  The European Union is not liable for any use that may be made of the information contained therein.}
  \and
  Leslie Ann Goldberg$^\dagger$
  \and  Mark Jerrum\thanks{
  School of Mathematical Sciences,
  Queen Mary, University of London, Mile End Road, London E1 4NS, United Kingdom.}
 }

\date{} 
\maketitle

\begin{abstract} We consider the problem of counting $H$-colourings
from an input graph~$G$ to a target graph~$H$.
We show that if $H$ is any fixed graph without trivial components,
then the problem is as hard as  the well-known problem $\BIS$, which is
the problem of (approximately) counting independent sets in a bipartite graph.
$\BIS$ is a complete problem in an important complexity class for approximate counting,
and is believed not to have an FPRAS. If this is so, then our result
shows
that for every graph $H$ without trivial components, 
the $H$-colouring counting problem has no FPRAS.
This problem was studied a decade ago by Goldberg, Kelk and Paterson. 
They were able to show that approximately sampling $H$-colourings 
is $\BIS$-hard, but it was not known how to get the result for approximate counting.
Our solution builds on non-constructive ideas using the work of \Lovasz.
\end{abstract}

\section{Introduction}

The independent set and $k$-colouring models are well-known models from statistical physics which have also been thoroughly studied in computer science. A particularly 
interesting question from the  computer science perspective
is the complexity of counting and approximate counting in these 
models. Given an input graph~$G$,  the problem is
to approximate the number of independent sets of~$G$
or the number of   proper $k$-colourings of~$G$. 
Both of these problems can be viewed as special cases of
the more general problem of approximately counting   $H$-colourings.
This paper studies the complexity of the more general problem.

We  begin with few relevant definitions.
Let $H=(V(H),E(H))$ be a fixed graph which is allowed to have self-loops, but not parallel edges. An $H$-colouring of a graph $G=(V(G),E(G))$ is a homomorphism from $G$ to $H$, i.e., an assignment $h: V(G)\rightarrow V(H)$ such that for every edge $(u,v)$ in $G$, it holds that $(h(u),h(v))$ is also an edge of $H$. It is helpful to think of the vertices of $H$ as ``colours'' and the edges of $H$ as pairs of colours which are allowed to be assigned to adjacent vertices in $G$.  Given an input graph $G$,  we are interested in computing the number of $H$-colourings of $G$. We will refer to this problem as the $\COL{H}$ problem. Also, we  denote by $\COL{H}(G)$ the number of $H$-colourings of $G$. 
The examples mentioned earlier correspond to $H$-colourings in the following way.
Proper $k$-colourings of~$G$ correspond to $H$-colourings of~$G$ when
$H$ is a $k$-clique.  Also, independent sets of~$G$ correspond to $H$-colourings of $G$ when
$H$ is the connected $2$-vertex graph with exactly one self-loop.

A fully polynomial randomised approximation scheme ($\FPRAS$) for 
approximately counting $H$-colourings is a randomised algorithm that takes as input a graph~$G$ 
and an accuracy parameter $\epsilon>0$ and outputs a number 
which, with probability at least~$3/4$, is in the range $[\exp(-\epsilon) \COL{H}(G), \exp(\epsilon) \COL{H}(G)]$.
The  running time of the algorithm is bounded by a polynomial in~$|V(G)|$ and $\epsilon^{-1}$.

Our goal in this work is to quantify the computational complexity of approximately counting $H$-colourings. In particular, we seek to determine for which
graphs $H$  the problem  $\COL{H}$ admits  
an $\FPRAS$.
Dyer and Greenhill \cite{DyerGreenhill} have completely classified the computational complexity of \emph{exactly} counting $H$-colourings in terms of the parameter graph $H$.
We say that a connected graph is \emph{trivial} if it is either a clique with self-loops on every
vertex or a complete bipartite graph with no self-loops. Dyer and Greenhill showed
that the problem
 $\COL{H}$ is polynomial-time solvable when each connected component of $H$ is 
 trivial;
 otherwise it is $\sharpP$-complete.  
The complexity of the corresponding decision problem has also been characterised.
Hell and \Nesetril~\cite{HN1} showed that
  deciding whether an input graph $G$ 
 has 
  an $H$-colouring is  $\NP$-complete unless $H$ contains a self-loop or $H$ is bipartite (in which case it admits a trivial polynomial-time algorithm).

A \emph{polynomial approximate sampler} ($\PAS$) for sampling $H$-colourings is 
a randomised algorithm that takes as input a graph~$G$
and an accuracy parameter~$\epsilon\in(0,1]$
and gives an output such that the total variation distance between the output distribution of the
algorithm and the uniform distribution on $H$-colourings of~$G$ is at most~$\epsilon$.
The running time is bounded by a polynomial in $|V(G)|$ and $\epsilon^{-1}$.
The $\PAS$ is said to be a \emph{fully polynomial approximate sampler} ($\FPAS$) if the
running time is bounded by a polynomial in $|V(G)|$ and $\log(\epsilon^{-1})$.
The problem of approximately 
sampling
$H$-colourings 
has been shown in \cite{HCOL} to be $\BIS$-hard provided that $H$ contains no trivial components\footnote{\label{foot:trivia}Note that the presence of trivial components may lead to artificial approximation schemes.
For example \cite[Section 7]{HCOL} gives the example in which $H$ consists of
two components ---  a (trivial) size-$3$ clique with self-loops and a connected $2$-vertex graph
with exactly one self-loop. 
There is a $\PAS$ for this example.
If $\epsilon$ is not too small then  it
typically outputs an $H$-colouring using the trivial component. See \cite[Section 7]{HCOL} for details.}. 
More precisely, for any fixed graph~$H$ with no trivial components,
the existence of a $\PAS$ for sampling 
$H$-colourings would imply 
that there is
an $\FPRAS$ for $\BIS$, which is the problem of
counting the independent sets of a bipartite graph.
 $\BIS$ plays an important role in approximation complexity.
Despite many attempts, nobody has found an FPRAS for
$\BIS$ and it is conjectured 
that
none exists (even though it is unlikely that approximating $\BIS$ is NP-hard). 
Various natural algorithms have been ruled out
as candidate FPRASes for $\BIS$  \cite{MWW09, GS10, GJ12a}.
Moreover, Dyer et al.~\cite{DGGJ} 
showed that $\BIS$ is complete under approximation-preserving (AP) reductions
 in a logically defined class of
problems, called  \#RH$\Pi_1$, to which an increasing variety
of problems have been shown to belong.
Other typical complete problems in \#RH$\Pi_1$ include counting the
number of downsets in a partially ordered set \cite{DGGJ}
and computing the partition function of the ferromagnetic Ising model with local external fields \cite{GJ07}.

Perhaps surprisingly, the hardness result of \cite{HCOL} for sampling $H$-colourings does not imply hardness for approximately counting $H$-colourings. This might be  puzzling at first since,  for  the independent set and $k$-colouring models, approximate counting is well-known to be  equivalent to approximate sampling (this equivalence has been proved in \cite{JVV} for the so-called class of self-reducible problems in $\sharpP$). However, for general graphs $H$ it is only known \cite{DGJ} that an $\FPAS$ 
for sampling $H$-colourings implies an $\FPRAS$ for counting $H$-colourings (but not the reverse direction).  For a thorough discussion of this point we refer the reader to \cite{DGJ} (where also an example of a problem in $\sharpP$ is given which, under usual complexity theory assumptions, admits an $\FPRAS$ but not an $\FPAS$).  

In this paper, we address the following questions. 
\begin{itemize}
\item ``Is there a graph $H$  for which approximately counting $H$-colourings is substantially easier than approximately sampling $H$-colourings?''  
\item``Is there a graph   $H$ such that $\COL{H}$ lies between $\mathsf{P}$ and the class of $\BIS$-hard problems?''
\end{itemize}
 We present 
  the analogue of the hardness result of \cite{HCOL} in the counting setting, therefore providing evidence that the answers to the previous questions are negative. To formally state the result, recall the notion of an approximation preserving reduction $\leq_{\AP}$ (introduced in \cite{DGGJ}). For counting problems $\#\mathsf{A}$ and $\#\mathsf{B}$, $\#\mathsf{A}{\leq}_{\AP}\#\mathsf{B}$ implies that an $\FPRAS$ for $\#\mathsf{B}$ yields an $\FPRAS$ for $\#\mathsf{A}$. Our main result is the following.

\newcommand{\statethmmain}{Let $H$ be a graph (possibly with self-loops but without parallel edges), all of whose connected components are non-trivial. Then
$\BIS\leq_{\AP}\COL{H}$.}

\begin{theorem}\label{thm:main}
\statethmmain{}
\end{theorem} 

Interestingly, in the proof of Theorem~\ref{thm:main} we use a non-constructive approach, partly inspired by tools introduced by Lov\'{a}sz \cite{Lovasz} in the context of graph homomorphisms. Before delving into the technical details of the proof, in the next section we will describe, in a simplified setting, the main obstacles to proving Theorem~\ref{thm:main} and the most interesting elements of the proof.

\section{Reductions for sampling versus reductions for counting}\label{sec:previouswork}

We start by considering the closely related work \cite{HCOL}.
The assumptions on the graph~$H$ 
are the same
as in Theorem~\ref{thm:main} --- namely that $H$ does not have
trivial components. The proof in~\cite{HCOL}
shows  that approximately \emph{sampling} $H$-colourings is at least as hard as $\BIS$.
We will first overview the approach of this   paper
since we will use several ingredients of their proof.  We will also describe the new  ingredients  which will allow  us to leap  from the sampling setting to the counting setting.

Let $H$ be a graph for which we wish to show that $\BIS\leq_{\AP}\COL{H}$. To do this, it clearly suffices to find a subgraph $H'$ of $H$ such that $\BIS\leq_{\AP}\COL{H'}$ and  $\COL{H'}\leq_{\AP}\COL{H}$.

We use the following definitions. The \emph{degree} of a vertex is the number of edges incident to it,
so a self-loop contributes~$1$ to the degree of a vertex.
Given a subset $S$ of $V(H)$, let $H[S]$ denote the subgraph of $H$ induced by the set $S$. 
Further, let $\Nu(S)$ denote the neighbourhood of $S$ in $H$, i.e., the set of vertices in $H$ which are adjacent to a vertex in $S$. We will also denote the vertices of $H$ by $v_1,v_2,\hdots$.

\subsection{Restricting $H$-colourings to induced-subgraph-colourings}\label{sec:induced}

To motivate how to find such a subgraph $H'$, we give a high-level reduction scheme. This will reveal that the subgraphs induced by the neighbourhoods of 
maximum-degree
vertices  of~$H$  are  natural choices for $H'$.\footnote{Later, we will modify this reduction scheme in a non-trivial way, but still the aspects 
that we highlight of the simplified reduction scheme will carry 
over to the modified reduction scheme.} 
To see this, 
let's temporarily suppose that we already have a subgraph~$H'$ of~$H$ in mind, and let $G'$ be an input for $\COL{H'}$ such that $|V(G')|=n'$. We next construct an instance $G$ of $\COL{H}$ by adding to $G'$ a special vertex $w$ and a large independent set $I$ with $|I|\gg n'$. We also add all edges between the special vertex $w$ and the vertices of $G'$ and $I$. Let $h$ be a homomorphism from $G$ to $H$ with $h(w)=v_i$. 
\begin{figure}[h]
\begin{center}
\scalebox{0.6}[0.6]{\input{Hexample3.tex}}
\end{center}
\end{figure}
Observe that the edges between $w$ and the rest of $G$ enforce that the restriction of $h$ to 
$G'$ is an $N_\cup(v_i)$-colouring of~$G'$. Similarly, the restriction
of $h$ to~$I$ is   a $N_\cup(v_i)$-colouring of $I$. 
Furthermore,  there are $(\mathrm{deg}_H(v_i))^{|I|}$ choices for the restriction of $h$ on the independent set $I$. From this, it is not hard to formally argue that the effect of the large independent set $I$ is to enforce that in all but a negligible fraction of $H$-colourings of $G$, the special vertex $w$ must be assigned a colour among the vertices of $H$ with maximum degree and thus $G'$ is coloured using the subgraph of $H$ induced by the neighbourhood of such a vertex. 

In an ideal scenario, there is a unique vertex $v$ in $H$ with maximum degree and, further, the subgraph $H'=H[N_{\cup}(v)]$ of $H$ is non-trivial and different from $H$. If both of these hypotheses hold, the reduction above can be used to show that $\COL{H'}\leq_{\AP}\COL{H}$ (since the uniqueness of $v$ implies  that  $h(w)=v$ in ``almost all" $H$-colorings) and (say, by induction) we have $\BIS\leq_{\AP}\COL{H'}$. But what happens when these hypotheses do not hold?

A significant problem which arises  at this point is the existence of multiple relevant neighbourhoods.
That is, there may be several maximum-degree vertices in~$H$,
and the subgraphs induced by their neighbourhoods may not be isomorphic. 
It is much easier to deal with this problem in the sampling setting
than in the counting setting. 
In the next section, we describe the difference between  
these settings, and we describe our approach to this (initial) hurdle.
 
\subsection{Counting subgraph-induced-colourings: the case of multiple subgraphs}\label{sec:multiple}

To illustrate concretely  the part of the sampling argument in \cite{HCOL} which breaks down in the counting setting, let us consider the toy example of the graph $H$ depicted in Figure~\ref{fig:toy} and overview how the argument in \cite{HCOL} works. Since $H$ is regular,  the relevant (induced) neighbourhoods of the vertices in $H$ are given by the graphs $H_1$ and $H_2$ in Figure~\ref{fig:toy}. These  correspond to the neighbourhoods of the vertices $v_1$ and $v_2$, respectively (note that the remaining neighbourhoods are isomorphic to $H_2$). Before proceeding, we remark that $\BIS\leq_{\AP}\COL{H_1}$ and $\BIS\leq_{\AP}\COL{H_2}$ (see \cite{Kelk} where all graphs $H$ with up to four vertices are classified; in \cite[p. 314--316]{Kelk} $H_1$ corresponds to the graph n.51 while $H_2$ corresponds to the graph n.59).

\begin{figure}[H]
\centering
\scalebox{0.8}[0.8]{\input{Hexample2.tex}}
\caption{Example of a graph $H$ and the subgraphs induced by the neighbourhoods of its (maximum degree) vertices.}\label{fig:toy}
\end{figure}
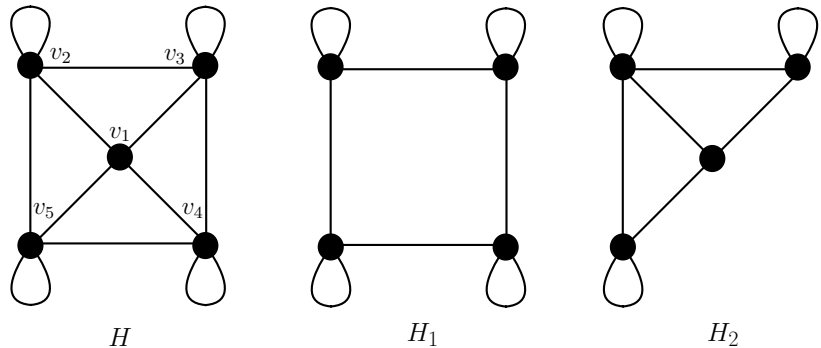

 So we already have sampling reductions from~$\BIS$ to~$\COL{H_1}$ and~$\COL{H_2}$ and
we want a sampling reduction from~$\BIS$ to $\COL{H}$, e.g., an algorithm for sampling 
bipartite independent sets using an oracle for sampling $H$-colourings.
Here's how it works. Let $G'$ be an input to $\BIS$. 
Using the sampling reductions to~$\COL{H_1}$ and~$\COL{H_2}$, construct from~$G'$
two graphs $G_1$ and $G_2$
such that
an (approximately) uniform  $H_1$-colouring of $G_1$ allows
us to construct an (approximately) uniform independent set of~$G'$
and similarly an (approximately) uniform $H_2$-colouring of $G_2$ also 
allows this. Then let $G$ be the graph 
obtained by 
taking the disjoint union of $G_1$ and $G_2$, adding a special vertex $w$, and adding all edges between $w$ and $G_1$ and $G_2$ (note, there is no need for the independent set $I$ used previously since $H$ is regular; see the graph $G$ in Figure~\ref{fig:sampler}).
\begin{figure}[h]
\begin{center}
\scalebox{0.6}[0.6]{\input{Hexample4.tex}}
\end{center}
\caption{``Gluing" the reductions (\cite{HCOL})}\label{fig:sampler}
\end{figure}
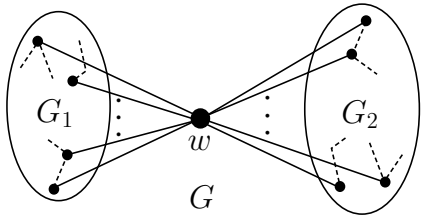

Observe that given a random $H$-colouring $h$ of $G$,  by just revealing the colour of the special vertex $w$, we can generate either a random $H_1$-colouring of $G_1$ or a random $H_2$-colouring of $G_2$.   In particular, if $h(w)$ is $v_1$, by considering the restriction of $h$ on $G_1$ we obtain a random $H_1$-colouring of $G_1$. Similarly, if  $h(w)$ is one of the vertices $v_2,v_3,v_4,v_5$, we obtain a random $H_2$-colouring of $G_2$. By our assumptions for $G_1$ and $G_2$, in each case we can then obtain a random independent set of $G'$. 

In contrast, the aforementioned reduction scheme fails in the  counting setting. Namely, considering cases for the colour $h(w)$, we obtain the following equality:
\begin{equation}\label{eq:highlevel1}
\COL{H}(G)=\COL{H_1}(G_1)\, \COL{H_1}(G_2)+4\COL{H_2}(G_1)\,\COL{H_2}(G_2)
\end{equation}
Given an approximation of $\COL{H}(G)$, say $Z$, observe that \eqref{eq:highlevel1} yields little information about whether $Z$ is a good approximation for  $\COL{H_1}(G_1)$ or $\COL{H_2}(G_2)$. This issue goes away in the sampling setting precisely because we can distinguish between the two cases by just looking at the colour of the special vertex $w$ in the random $H$-colouring of $G$. 

Thus, to proceed with the reduction in the counting setting we have to focus our attention on one of $H_1$ or $H_2$, say $H_1$, and somehow prove that $\COL{H_1}\leq_{\AP}\COL{H}$. The question which arises is how to choose between $H_1$ and $H_2$.  This becomes  more complicated for general graphs $H$  since it is not hard to imagine that instead of just two graphs $H_1,H_2$ we will typically have a collection of graphs $H_1,\hdots,H_t$ corresponding to the induced neighbourhoods of vertices $v_1,\hdots,v_t$ of $H$, for some $t$ which can be arbitrarily large (depending on the graph $H$). To make matters worse, apart from very basic information on  the $H_i$'s (such as connectedness or number of vertices/edges), we will not be able to control significantly their graph structure.

 At this point, we employ a non-constructive approach using a tool from \cite{Lovasz}: for arbitrary non-isomorphic graphs $H_1,H_2$ there exists a (fixed) graph $J$ depending only on $H_1$ and $H_2$ such that $\COL{H_{1}}(J)\neq\COL{H_{2}}(J)$. We extend this to an arbitrary collection of pairwise non-isomorphic graphs $H_1,\hdots,H_t$ as follows: we prove the existence of a graph $J$ so that for some 
 $i^*\in\{1,\ldots,t\}$
 it holds that $\COL{H_{i^*}}(J)>\COL{H_{i}}(J)$ for all $i\neq i^*$. Intuitively, the graph $J$ will be used to ``select" the subgraph $H_{i^*}$. Note that we will not require any further knowledge about what $H_{i^*}$ or $J$ is, freeing us from the cumbersome (and perhaps difficult) task of looking into the finer details of the structure of the graphs $H_1,\hdots,H_t$. With the graph $J$ in hand, we then  take sufficiently many disjoint copies of $J$ 
 and we connect them   to the special vertex $w$.
 This ensures  that, in a typical $H$-colouring, the vertex $w$ gets coloured with the vertex $v_i$. 

To utilise the above, we will further need to ensure that $\BIS\leq_{\AP}\COL{H_i}$ for every $i$. If we could ensure that the $H_i$'s are proper subgraphs of $H$ and non-trivial, then using the arguments above, we could complete the proof using induction. However, this is clearly not possible in general since for example, as we noted earlier, there may exist a vertex $v$ in $H$ such that $N_{\cup}(v)=V(H)$.  
Dealing with such cases 
is the bulk of the work in the sampling setting of \cite{HCOL} and
these cases cause even more problems for us. To deal with them, we need a further non-constructive argument --- one that turns out to be
more technical than, and substantially different from the ideas in Lov\'{a}sz \cite{Lovasz}. 

In the next section, we introduce
the concepts that we will use to formalise the intuitive arguments that we
have given above. These concepts will be critical in our argument.

\section{Proof Outline}
Since we are interested in instances of $\BIS$,
which are bipartite
graphs,   we will often need to consider $H$-colourings of bipartite graphs. We will assume that every bipartite graph $G$ comes with a (fixed) proper 2-colouring of its parts with colours $\{L,R\}$. We will use $L(G)$ and $R(G)$ to denote the vertices of $G$ coloured with $L$ and $R$ respectively. We will frequently refer to a bipartite graph as a 2-coloured graph to emphasise the proper 2-colouring of its vertices.
In Section~\ref{sec:fixingHCOL}
we will define the  bipartite double cover $\Bip(H)$
of a graph~$H$. 
This is a bipartite graph which can be viewed as the tensor 
product\footnote{The tensor product of graphs~$H_1$ and~$H_2$ 
has vertex set $V(H_1) \times V(H_2)$.
Vertices $(u_1,u_2)$ and $(v_1,v_2)$ are adjacent in the tensor product
if and only if $(u_1,v_1) \in E(H_1)$ and $(u_2,v_2) \in E(H_2)$.
} of $H$ with a single edge.
To simplify matters, for bipartite 
graphs~$H$, it will be convenient to restrict our attention to colour-preserving homomorphisms which are formally defined as follows,
where $h(S)$ denotes the image of a set $S\subseteq V(G)$ under
a homomorphism $h$ from~$G$ to~$H$, i.e.,
$h(S) = \{ h(v) \mid v \in S\}$.
\begin{definition}\label{def:color}
Let $G$ and $H$ be 2-coloured bipartite graphs. A homomorphism $h$ from $G$ to $H$ is colour-preserving if $h$ preserves the colouring of the parts of $G$, i.e., $h(L(G))\subseteq L(H)$ and $h(R(G))\subseteq R(H)$.
\end{definition}
As we will see in Section~\ref{sec:fixingHCOL},
the $H$-colourings of a connected bipartite graph~$G$ 
correspond to the colour-preserving homomorphisms
from~$G$ to~$\Bip(H)$. 
 By restricting our attention to colour-preserving homomorphisms, we obtain the following version of the $\COL{H}$ problem.\\

\noindent\textit{Parameter}. A bipartite graph $H$ with vertex partition $(L(H),R(H))$.\\
\textit{Name}. $\FixCOL{H}$\\
\textit{Instance}. A bipartite graph $G$ with vertex partition $(L(G),R(G))$.\\
\textit{Output}. 
The number of colour-preserving homomorphisms from $G$ to $H$,
which we denote $\FixCOL{H}(G)$.\\

Before justifying our interest in the problem~$\FixCOL{H}$, let us introduce one more piece of terminology. A bipartite graph $H$ will be called \textit{full} if there exist vertices $u\in L(H)$ and $v\in R(H)$ such that $u$ is adjacent to every vertex in $R(H)$ and $v$ is adjacent to every vertex in $L(H)$.  In this case, vertices $u$ and~$v$ are also called full. Note that if $H$ is full, then $H$ is also connected.

The following lemma allows us to restrict our attention to the $\FixCOL{H}$ problem. The proof of Lemma~\ref{lem:fixingHCOL} is along the lines   described in Sections~\ref{sec:induced} 
and~\ref{sec:multiple}, albeit with  some modifications to account for    technical details. We thus defer the proof to Section~\ref{sec:fixingHCOL}.

\newcommand{\statelemfixingHCOL}{
Let 
$H$ be a graph with no trivial components. Then, there exists a 2-coloured graph $H'$, which is 
full but not trivial,
such that
$$\FixCOL{H'}\leq_\AP \COL{H}.$$
}

\begin{lemma} [Analogue of Lemma 7 in \cite{HCOL}]
\label{lem:fixingHCOL}
\statelemfixingHCOL{}
\end{lemma}

Thus, our main result will be easy to conclude  from the following central lemma, which will be proved in Section~\ref{sec:lemmafull}. In the upcoming Section~\ref{sec:wer}, we give an overview of the ingredients in the proof.
\newcommand{\statelemmain}{Let $H$ be a 2-coloured graph which is 
full but not trivial.
Then $\BIS\leq_\AP\FixCOL{H}$.}
\begin{lemma}\label{lem:main}
\statelemmain{}
\end{lemma}

Assuming Lemmas~\ref{lem:fixingHCOL} and~\ref{lem:main}, the proof of Theorem~\ref{thm:main} is immediate.
 
 \begin{thmmain}
\statethmmain{}
\end{thmmain} 

\begin{proof} 
Just combine  Lemmas~\ref{lem:fixingHCOL} and~\ref{lem:main}.
\end{proof}

\section{Tools from graph homomorphisms}\label{sec:preliminaries}
In this section, we present some tools from graph-homomorphism theory which are central for our non-constructive arguments.

Recall that an isomorphism from a graph $H_1$ to a graph $H_2$ is an injective homomorphism $h$ from $H_1$ to $H_2$ whose inverse $h^{-1}$ is also a homomorphism from $H_2$ to $H_1$. 
Similarly,
a colour-preserving isomorphism from a 2-coloured bipartite graph~$H_1$ 
to a 2-coloured bipartite graph~$H_2$ is 
a colour-preserving homomorphism~$h$ 
(see Definition~\ref{def:color}), 
from~$H_1$ to~$H_2$
whose inverse   is also a colour-preserving homomorphism from $H_2$ to $H_1$. 
We will use $H_1\cong H_2$ to denote that $H_1$ and $H_2$ are isomorphic and $H_1\cong_c H_2$ to denote that $H_1$ and $H_2$ are isomorphic under a colour-preserving isomorphism. 
The following lemma has its origins in \cite[Theorem (3.6)]{Lovasz}. For completeness, 
in Section~\ref{sec:prelimproofs}, 
we 
present a  proof for the present setting which is a close adaptation of  \cite[Proof of Theorem 2.11]{HellNesetril}.

\newcommand{\statelovaszone}{Let $H_1$ and $H_2$ be 2-coloured bipartite graphs such that $H_1\ncong_c H_2$. Then,  there exists a 2-coloured bipartite graph $J$ such that $\FixCOL{H_1}(J)\neq\FixCOL{H_2}(J)$.
Moreover, $|V(J)|\leq\max\{|V(H_1)|,|V(H_2)|\}$.}
\begin{lemma}\label{lem:lovaszone}
\statelovaszone{}
\end{lemma}

Lemma~\ref{lem:lovaszone} can be viewed as a way of obtaining a gadget~$J$
to distinguish between the graphs~$H_1$ and~$H_2$.
In fact, we will need to be able to
pick out a particular graph~$H_{i^*}$ from a 
whole set of non-isomorphic graphs. 
For a positive integer $k$, let $[k]=\{1,\ldots,k\}$.
In Section~\ref{sec:prelimproofs} we 
give an inductive proof, powering up the effect of Lemma~\ref{lem:lovaszone}  to obtain the following. 
\newcommand{\statelovasztwo}{Let $H_1,\hdots,H_k$ be 2-coloured bipartite graphs such that $H_i\ncong_c H_j$ for all $\{i,j\}\in \binom{[k]}{2}$.
There exists a 2-coloured graph $J$ and an integer~$i^*\in [k]$ such that $\FixCOL{H_{i^*}}(J)>\FixCOL{H_i}(J)$ for all  $i\in[k]\backslash \{i^*\}$.}
\begin{lemma}\label{lem:lovasztwo}
\statelovasztwo{}
\end{lemma}

\section{Overview of proof of Lemma~\ref{lem:main}}\label{sec:wer}

Lemma~\ref{lem:main} is our central lemma. It will be proved in Section~\ref{sec:lemmafull}.
Here, we give an overview of the ingredients in the proof.
Thus, for the purpose of this section, we will assume that $H$ is a (2-coloured) bipartite graph which is full (thus, also connected) 
but not trivial.
Our goal will be to show $\BIS\leq_\AP\FixCOL{H}$.

Let $(V_L,V_R)$ denote the vertex partition of $H$ and let $F_L,F_R$ be the subsets of full vertices in $V_L,V_R$, respectively (i.e., every vertex in $F_L$ is connected to every vertex in $V_R$ and every vertex in $F_R$ is connected to every vertex in $V_L$). Since $H$ is full, we have $F_L,F_R\neq \emptyset$.  

Recall that, for a subset $S$ of $V(H)$, $H[S]$ is the subgraph of $H$ induced by the set $S$  
and that $\Nu(S)$ denotes the neighbourhood of $S$ in $H$, i.e., the set of vertices in $H$ which are adjacent to a vertex in $S$.  We will also use $N_{\j}(S)$ to denote the \emph{joint} neighbourhood of $S$ in $H$, i.e., the set of vertices in $H$ which are adjacent to \emph{every} vertex in $S$.

\subsection{An inductive approach using maximal bicliques of $H$}
The proof of  Lemma~\ref{lem:main} will be by induction on the number of vertices of $H$. Our goal will be to find a subgraph $H'$ of $H$ (which will also be 2-coloured, full 
but not trivial) with 
$|V(H')| <  |V(H)|$ such that $\FixCOL{H'}\leq_\AP\FixCOL{H}$.
If we find such an~$H'$, we will finish using the
inductive hypothesis that $\BIS\leq_\AP\FixCOL{H'}$. 
When we are not able to find such a subgraph~$H'$, we will
use an alternative method to  show that $\BIS\leq_\AP\FixCOL{H}$.

To select $H'$, we consider the set $\Cc$ of bicliques in $H$.
\begin{definition}
A \emph{biclique} in a graph $H$ with vertex partition $(V_L,V_R)$ is
a pair $(S_L,S_R)$ such that
$S_L \subseteq V_L$, $S_R \subseteq V_R$ and
$S_L\times S_R\subseteq E(H)$.
Given a fixed~$H$, we use $\Cc$ to denote the set of all bicliques in $H$.
A biclique $(S_L,S_R)\in \Cc$ is said to be \emph{maximal} 
if it is inclusion maximal, in the sense that 
there is no other biclique $(S'_L,S'_R)\in \Cc$ with $S_L \subseteq S'_L$ and $S_R \subseteq S'_R$.
\end{definition}
Note that $(F_L,V_R),(V_L,F_R)$ are maximal bicliques in $H$. For lack of better terminology, we will refer to these two special bicliques as the \emph{extremal} bicliques.

Our interest in maximal bicliques is justified by the following simple claim which will allow us to carry out our inductive step. 
\begin{lemma}\label{inductionstep}
Let $(S_L,S_R)$ be a maximal biclique which is not extremal, i.e, $S_L\neq F_L$, $S_R\neq F_R$. We have that $S_L\neq V_L$, $S_R\neq V_R$, $\Nu(S_L)=V_R$ and $\Nu(S_R)=V_L$.  Let $H_1=H[S_L\cup V_R]$ and $H_2=H[V_L\cup S_R]$. Then, for $i\in\{1,2\}$, we have that $H_i$ is  full 
but not trivial
and further satisfies $|V(H_i)|<|V(H)|$.
\end{lemma}
\begin{proof}
Since $(S_L,S_R)$ is a maximal biclique, we have $F_L\subseteq S_L$ and $F_R\subseteq S_R$. It follows that $\Nu(S_L)=V_R$ and $\Nu(S_R)=V_L$. It is easy to see that both $H_1,H_2$ are connected, bipartite and full. We also have that $|V(H_1)|, |V(H_2)|<|V(H)|$. If, for example, $|V(H_1)|=|V(H)|$ we must have that $S_L=V_L$. Since $(S_L,S_R)$ is a maximal biclique, this would imply that $S_R=F_R$. Similarly, we obtain that $H_1$, $H_2$ are not trivial. If, for example, $H_1$ is trivial, then we would have that $S_L=F_L$.
\end{proof}

\subsection{The basic gadget}\label{sec:basicgadget}
In this section, we discuss a gadget that is used in \cite{HCOL}. While the gadget in \cite{HCOL} will not work for us, it will nevertheless help to motivate our later selection of a more elaborate gadget for our needs. The gadget used in \cite{HCOL} is a complete bipartite graph $K_{a,b}$ with $a$ vertices on the left and $b$ vertices on the right. The  integers~$a$ and~$b$ should be thought of as sufficiently large numbers  which may depend on the size of the input  to $\FixCOL{H}$. Roughly speaking, we will be interested in the colours appearing on the left and right of $K_{a,b}$ in a typical colour-preserving homomorphism from $K_{a,b}$ to $H$. 

To make this precise, for a colour-preserving homomorphism $h:K_{a,b}\rightarrow H$, the \emph{phase} of $h$ is the pair $\big(h(L(K_{a,b})), h(R(K_{a,b}))\big)$, i.e., the subsets of $V_L$ and $V_R$ appearing on the left and right of $K_{a,b}$ under the homomorphism $h$, respectively. Since $K_{a,b}$ is a complete bipartite graph, we have that a phase is a biclique of $H$, i.e., an element of $\Cc$.  
Let $(S_L,S_R)\in \Cc$. For convenience, we will refer to the total number of colour-preserving homomorphisms whose phase equals $(S_L,S_R)$ as the contribution of the phase/biclique $(S_L,S_R)$ to the gadget. Our induction step crucially depends on analysing the \emph{dominant phases} of the gadget, i.e., the phases with the largest contribution.

 It is not hard to see that the contribution of a phase/biclique $(S_L,S_R)$ to the gadget $K_{a,b}$ 
is roughly equal to\footnote{The quantity $|S_L|^{a}$
approximates the number of 
surjective mappings from $L(K_{a,b})$ to $S_L$.
Similarly, $|S_R|^{b}$ approximates the number
of surjective mappings from $R(K_{a,b})$ to $S_R$. 
We will use fairly standard techniques to make sure that the relative error of the approximation (see the relevant Lemma~\ref{lem:surj}) is sufficiently small.}
\[|S_L|^{a}|S_R|^b.\]
 Thus, the dominant phases are determined by the ratio $a/b$. Rather than restricting ourselves to integers~$a$ and~$b$ it will be convenient to consider positive \emph{real} numbers $\alpha,\beta>0$ and the corresponding phases with dominant contribution.
\begin{definition}[The set of dominating bicliques $\Cc_{\alpha,\beta}$]\label{def:ccab}
Let $\alpha$ and $\beta$ be positive real numbers. Define $\Cc_{\alpha,\beta}$  to be the following subset of $\Cc$:
\[\Cc_{\alpha,\beta}:=\big\{(S_L,S_R)\in \Cc:\, (S_L,S_R)=\arg\max_{(S_{L}',S_{R}')\in \Cc} |S_{L}'|^{\alpha } |S_{R}'|^{\beta }\big\}.\]
Note that for positive $\alpha,\beta$ the bicliques in $\Cc_{\alpha,\beta}$ are in fact maximal. 
\end{definition}
For the purpose of the following discussion and to avoid delving into (at this point) unnecessary technical details, we will assume 
for now that $\alpha$ and $\beta$ are rationals. In particular, there exists an integer $Q$ so that $a=Q\alpha$  and $b=Q\beta$ are integers and $\alpha/\beta=a/b$.\footnote{For irrational $\alpha$ and $\beta$, we will later use Dirichlet's approximation theorem, cf. Lemma~\ref{lem:dirichlet}, to obtain integers $a$ and $b$ such that the ratio $\alpha/\beta$ is approximately equal to $a/b$ (to any desired polynomial precision).}

\subsection{A reduction scheme}\label{sec:qw}
The structure of our reduction scheme  expands on the work of~\cite{HCOL}. The following are implicit in \cite{HCOL}:
\begin{enumerate}
\item\label{it:extremal}
if  $|\Cc_{\alpha,\beta}|=2$ and $\Cc_{\alpha,\beta}$ consists only of the extremal bicliques 
then $\BIS$ reduces to $\FixCOL{H}$. 
\item\label{it:nonextremal} if $|\Cc_{\alpha,\beta}|=1$ and the unique dominating biclique in $\Cc_{\alpha,\beta}$ is \emph{not} an extremal biclique,  then $\FixCOL{H'}$ reduces to $\FixCOL{H}$ for some non-trivial 
proper
subgraph $H'$ of $H$.

\end{enumerate} 
Unfortunately, it is not hard to come up with examples of graphs $H$ such that, for every choice of $\alpha$ and $\beta$, we do not fall into one of the cases \ref{it:extremal} and \ref{it:nonextremal}, see Figure~\ref{fig:coexistence} for an explicit such example.
\begin{figure}[h]
\begin{center}
\scalebox{0.8}[0.8]{\input{Hexample7.tex}}
\end{center}
\caption{A graph $H$ where for any choice of the parameters $\alpha,\beta$, the set of dominating bicliques $\Cc_{\alpha,\beta}$ consists either of a single extremal biclique or of the two extremal and two non-extremal bicliques. Namely, when $\alpha>\beta$ (resp. $\alpha<\beta$) we have that the unique element of $\Cc_{\alpha,\beta}$ is the extremal biclique $(\{1,2,3,4\},\{1'\})$ (resp. $(\{1\},\{1',2',3',4'\}))$; when $\alpha=\beta$ we have that $\Cc_{\alpha,\beta}$ consists of the bicliques $(\{1\},\{1',2',3',4'\})$,  $(\{1,2,3,4\},\{1'\})$, $(\{1,2\},\{1',2'\})$ and $(\{1,3\},\{1',3'\})$.}\label{fig:coexistence}
\end{figure}
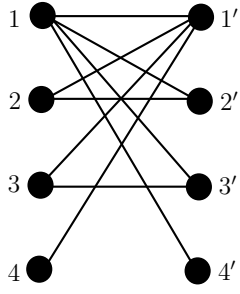
It will nevertheless be useful to see how the gadget $K_{a,b}$ is used, so we give a quick overview of the reductions which yield Items~\ref{it:extremal} and~\ref{it:nonextremal}  (since these are only implicit in \cite{HCOL}). 

For Item~\ref{it:extremal}, let $G'$ be a (2-coloured) bipartite graph which is an input to $\BIS$. To construct an instance of $\FixCOL{H}$, replace each vertex of $G'$ with a distinct copy of $K_{a,b}$. Further, for each edge $(u,v)$ in $G'$ with $u\in L(G')$ and $v\in R(G')$ add all edges between  the right part of $u$'s copy of $K_{a,b}$ and the left part of $v$'s copy of $K_{a,b}$. In the final graph, say  $G$, by scaling $a,b$ to be much larger than the size of $G'$ (while keeping fixed the ratio $a/b=\alpha/\beta$), the phases of the gadgets $K_{a,b}$ in ``almost all" colour-preserving homomorphisms from $G$ to $H$ are elements of $\Cc_{\alpha,\beta}$ and in particular are extremal bicliques. 
It then remains to observe that independent sets of $G'$ are encoded by those homomorphisms where 
\begin{itemize}
\item 
the phase of a gadget corresponding to a vertex in $L(G')$ is  
$(F_L,V_R)$ if the vertex is in the independent set and $(V_L,F_R)$ otherwise, and
\item the  phase of a gadget corresponding to a vertex in $R(G')$ is $(V_L,F_R)$
if the vertex is in the independent set and $(F_L,V_R)$ otherwise.
\end{itemize}

For Item~\ref{it:nonextremal}, the use of the gadget $K_{a,b}$  is  depicted in Figure~\ref{fig:gadgetuse}. Namely, for a 2-coloured bipartite graph $G'$, consider the graph $G$ obtained by adding all edges between the left part of $G'$ and the right part of $K_{a,b}$. We will typically denote the graph obtained by this construction as $K_{a,b}(G')$. For the following discussion, let us  set $G:=K_{a,b}(G')$. In the setting of Item~\ref{it:nonextremal}, we have that $\Cc_{\alpha,\beta}$ consists of a unique maximal biclique $(S_L,S_R)$ which is not extremal. Once again, by making $a,b$ large relative to the size of $G'$ (while maintaining the ratio $a/b=\alpha/\beta$), the phase of the gadget $K_{a,b}$ in ``almost all" homomorphisms $h$ of the graph $G$ will be the  dominating biclique $(S_L,S_R)$.  Let us consider such a homomorphism $h$ whose  restriction on $K_{a,b}$ has as its phase 
the maximal biclique $(S_L,S_R)$. The edges between $R(K_{a,b})$ and $L(G')$ enforce that $h(L(G'))\subseteq N_{\cap}(S_R)=S_L$, where in the latter equality we used that $(S_L,S_R)$ is a maximal biclique. It follows that $h(R(G'))\subseteq N_{\cup}(S_L)=V_R$ (see Lemma~\ref{inductionstep} for the latter equality). It follows that the restriction of $h$ on $G'$ is an $H_1$-colouring of the graph $G'$, where $H_1=H[S_L\cup V_R]$ is the same graph as in Lemma~\ref{inductionstep}. Viewing $G'$ as an instance of $\COL{H_1}$ and $G$ as an instance of $\COL{H}$, one obtains $\COL{H_1}\leq_{\AP}\COL{H}$. Since $H_1$ is full, not trivial and has fewer vertices than $H$, one can use the inductive hypothesis to conclude $\BIS\leq_{\AP}\COL{H_1}$.
\begin{figure}[h]
\begin{center}
\scalebox{0.6}[0.6]{\input{Hexample6.tex}}
\end{center}
\caption{The reduction for Item~\ref{it:nonextremal}}\label{fig:gadgetuse}
\end{figure}
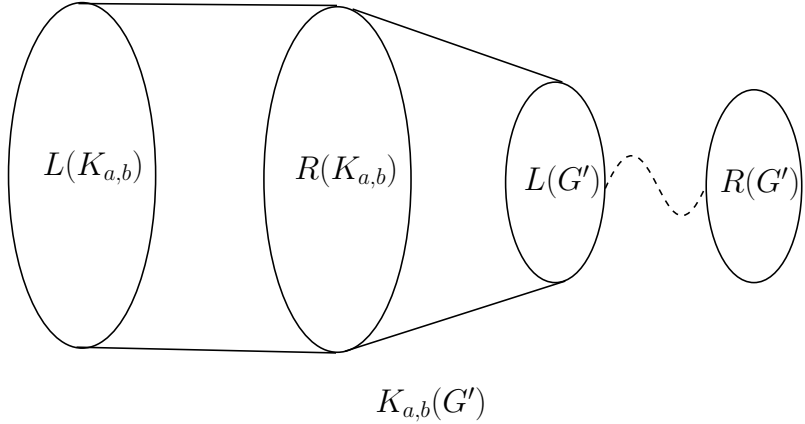

We remark here that using the non-constructive approach of Lemma~\ref{lem:lovasztwo} and along the lines we described in Section~\ref{sec:multiple}, we will be able to remove the restriction that $|\Cc_{\alpha,\beta}|=1$ as long as $\Cc_{\alpha,\beta}$ does not include an extremal biclique of $H$.

In view of Items~\ref{it:extremal} and~\ref{it:nonextremal}, the scheme pursued in \cite{HCOL} (and which we will also follow to a certain extent) is to fix $0<\alpha,\beta<1$  such that 
\begin{equation}\label{eq:defa0b0}
|F_L|^{\alpha}|V_R|^{\beta}=|V_L|^{\alpha}|F_R|^{\beta},
\end{equation}
so that the contribution of the extremal bicliques $(F_L,V_R)$ and $(V_L,F_R)$ to the gadget $K_{a,b}$ is equal. This has the beneficial effect that $\Cc_{\alpha,\beta}$ includes either none or both of the extremal bicliques. The only very bad scenario remaining is when $\Cc_{\alpha,\beta}$ includes both the extremal bicliques as well as (at least) one non-extremal biclique, since then not only $|\Cc_{\alpha,\beta}|>1$ (which is already a problem for the approach implicit in \cite{HCOL}) but also the coexistence of extremal and non-extremal bicliques in $\Cc_{\alpha,\beta}$ impedes the non-constructive approach of Lemma~\ref{lem:lovasztwo}. While for the sampling problem studied in \cite{HCOL} the coexistence of extremal and non-extremal bicliques was recoverable by  ``gluing" the reductions together (as we explained in a simplified setting in Section~\ref{sec:multiple}), this is no longer the case in the counting setting. More precisely, for the counting problem we will have to understand for which graphs $H$ the coexistence of extremal and non-extremal bicliques occurs and  consider more elaborate gadgets in the reduction to overcome this coexistence.  

\subsection{A non-constructive gadget}
The key idea is to introduce another  non-constructive  argument (in addition to the approach suggested by Lemma~\ref{lem:lovasztwo}) by viewing the construction of $K_{a,b}(G')$, cf. Figure~\ref{fig:gadgetuse}, as a gadget parameterised by the graph $G'$. To emphasize that $G'$ is no longer an input graph, let us switch notation from $G'$ to $\Gamma$, i.e., $\Gamma$ is a 2-coloured graph and $K_{a,b}(\Gamma)$ is the graph in Figure~\ref{fig:gadgetuse} where $G'$ is replaced by the graph $\Gamma$. We will choose $a,b$ sufficiently large so that the graph $\Gamma$ is ``small" relative to the graph $K_{a,b}$, so its effect on the dominant phases will be of second order. We stress here that we will never try to specify $\Gamma$ explicitly; all we need is the existence of a helpful $\Gamma$. In the following, we expand on this point and set up some relevant quantities for the proof.

As for the basic gadget, we define the phase of a colour-preserving homomorphism $h: K_{a,b}(\Gamma)\rightarrow H$ as the pair $\big(h(L(K_{a,b})), h(R(K_{a,b}))\big)$. Note that the phase of $h$ is determined by its restriction on $K_{a,b}$ (but not on $\Gamma$) and, thus, as before the phases are supported on bicliques of $H$. We once again set $a=Q\alpha,b=Q\beta$ and let $Q$ be a large integer relative to the size of $\Gamma$. 
(Again, we will later use Dirichilet's approximation theorem for the case
in which~$\alpha$ and~$\beta$ are not rational, but we do not worry about this here.)
With this setup, the phases with the dominant contribution in $K_{a,b}(\Gamma)$ are related to those in $K_{a,b}$ and in particular 
we will  make $a$ and $b$ 
sufficiently large to ensure that they 
are a subset of $\Cc_{\alpha,\beta}$. Note however that the graph $\Gamma$ has the effect of reweighting each phase contribution in $K_{a,b}(\Gamma)$ relative to the one in $K_{a,b}$. 

To understand the reweighted contribution, let us consider a homomorphism $h: K_{a,b}(\Gamma)\rightarrow H$ whose phase is a biclique $(S_L,S_R)\in \Cc$. For such $h$, the edges between $R(K_{a,b})$ and $L(\Gamma)$ enforce that $h(L(\Gamma))\subseteq N_{\j}(S_R)$ and thus\footnote{\label{foot:technical}Technically, to conclude that $h(R(\Gamma))\subseteq \Nu(N_{\j}(S_R))$ from $h(L(\Gamma))\subseteq N_{\j}(S_R)$ we need that every vertex in $R(\Gamma)$ has at least one neighbour in $L(\Gamma)$. In the upcoming Lemmas~\ref{lem:slsr} and~\ref{lem:flvr}, we address this technical point by requiring that $\Gamma$ has no isolated vertices in $R(\Gamma)$.} we obtain that $h(R(\Gamma))\subseteq \Nu(N_{\j}(S_R))$; it follows that the restriction of $h$ to $\Gamma$ is supported by vertices in $H[N_{\j}(S_R)\cup \Nu(N_{\j}(S_R))]$. It is useful to see what happens when the phase $(S_L,S_R)$ of the homomorphism is a maximal biclique (say in $\Cc_{\alpha,\beta}$): then, $N_{\j}(S_R)=S_L$ and $\Nu(S_L)=V_R$ (from Lemma~\ref{inductionstep}). Thus, in the case where the phase of $h$ corresponds to a maximal biclique, the restriction of $h$ to $\Gamma$ is supported by vertices in $H[S_L\cup V_R]$.

It will be useful to distill the following definitions from the above remarks.
\begin{definition}[The graph $H_{S_L,S_R}$]\label{def:specialgraph}
Let $(S_L,S_R)$ be a biclique in $H$, i.e., $(S_L,S_R)\in \Cc$. Define $H_{S_L,S_R}$ to be the (bipartite) graph \[H[N_{\j}(S_R)\cup \Nu(N_{\j}(S_R))],\] whose 2-colouring is naturally induced by the 2-colouring of $H$. Note that when $(S_L,S_R)$ is a maximal biclique, we have that $H_{S_L,S_R}=H[S_L\cup V_R]$.
\end{definition}

\begin{definition}[The parameter $\zeta(S_L,S_R,\Gamma)$]\label{def:zeta}
Let $\Gamma$ be a 2-coloured graph and let $(S_L,S_R)$ be a biclique in $H$, i.e., $(S_L,S_R)\in \Cc$. We will denote
\begin{equation}\label{def:zetaslsr}
\zeta(S_L,S_R,\Gamma):= \FixCOL{H_{S_L,S_R}}(\Gamma),
\end{equation}
where $H_{S_L,S_R}$ is as in Definition~\ref{def:specialgraph}.
\end{definition}

Utilising the above definitions and the remarks earlier, we obtain that the contribution of the biclique $(S_L,S_R)$ to the gadget $K_{a,b}(\Gamma)$ is roughly equal to
\[\zeta(S_L,S_R,\Gamma) |S_L|^a |S_R|^b.\]

 Similarly to the outline for the simple gadget, the guiding principle will be to choose  $a$, $b$,
 and~$\Gamma$ appropriately so that the dominant phases are supported either on (both of) the extremal bicliques or on the non-extremal bicliques (but not a combination of both). Roughly, the choice of $a$ and $b$ will restrict the dominant phases in $K_{a,b}(\Gamma)$ to be a subset of  $\Cc_{\alpha,\beta}$, while the graph $\Gamma$ will pick out either the extremal bicliques or a set of non-extremal bicliques. (In the latter case, we will further need to ensure that exactly one non-extremal biclique has significant contribution to the gadget. To do this, we will utilise Lemma~\ref{lem:lovasztwo}.) When there is no such graph $\Gamma$, we will 
 use an alternative method to find a (2-coloured) 
 proper
 subgraph $H'$ of $H$ which is also full 
 but not trivial
 such that $\FixCOL{H'}$ reduces to $\FixCOL{H}$. In other words, the non-existence of a ``helpful" gadget for any choice of $\Gamma$  will  establish a useful property for $\FixCOL{H}$ on an arbitrary input. 

To do this, we will again equalise the contribution of the extremal bicliques in the gadget $K_{a,b}(\Gamma)$, so it will be useful to explicitly write out the definition \eqref{def:zetaslsr} for the extremal bicliques. 
\begin{definition}[The parameters $\zeta^{\ex}_1(\Gamma)$, $\zeta^{\ex}_2(\Gamma)$]\label{def:zetaextremal}
Let $\Gamma$ be a 2-coloured graph. Let $\zeta^{\ex}_1(\Gamma)$, $\zeta^{\ex}_2(\Gamma)$ be the values of $\zeta(S_L,S_R,\Gamma)$ when $(S_L,S_R)$ is the extremal biclique $(F_L,V_R),(V_L,F_R)$ respectively.
\begin{equation}\label{eq:zetaextremal}
\zeta^{\ex}_1(\Gamma):=\zeta(F_L,V_R,\Gamma)=|F_L|^{|L(\Gamma)|}|V_R|^{|R(\Gamma)|} \mbox{ and } \zeta^{\ex}_2(\Gamma):=\zeta(V_L,F_R,\Gamma)= \FixCOL{H}(\Gamma).
\end{equation}
\end{definition}
(To see the second equality in the definition of $\zeta^{\ex}_1(\Gamma)$, note that $H[F_L\cup V_R]$ is a complete bipartite graph; for the second equality in the definition of $\zeta^{\ex}_2(\Gamma)$, note that $H[V_L\cup V_R]=H$. We remark here that the asymmetry 
in the definitions of $\zeta^{\ex}_1(\Gamma)$ and $\zeta^{\ex}_2(\Gamma)$ is caused by the choice of connecting the right part of $K_{a,b}$ to the left part of $\Gamma$.) 

To equalise the contribution of the extremal bicliques in the final gadget we will need to slightly perturb our selection of $a,b$. That is, instead of setting $a=Q\alpha$ and $b=Q\beta$, we will choose $\hat{a}=Q\alpha$ and $\hat{b}=Q\beta+\gamma$ for some appropriate choice of $\gamma$ (note that we only perturb the size of $b$). Now, for a phase $(S_L,S_R)$ the multiplicative correction to its contribution in $K_{\hat{a},\hat{b}}(\Gamma)$ relative to the one in $K_{a,b}$ is given by 
\[\zeta(S_L,S_R,\Gamma) |S_R|^{\gamma}.\]
Thus, to equalise the contribution of the extremal bicliques in $K_{\hat{a},\hat{b}}(\Gamma)$, we will need the following parameter $\gamma=\gamma(\Gamma)$ (we will drop the dependence of $\gamma$ on the graph $\Gamma$ when there is no danger of confusion), which is formally defined below.

\begin{definition}[The parameter $\gamma(\Gamma)$]\label{def:defgamma}
Let $\Gamma$ be a 2-coloured graph. Define $\gamma(\Gamma)$ to be the unique (real) solution to the following equation:
\begin{equation}\label{eq:defgamma}
\zeta^{\ex}_1(\Gamma)\, |V_R|^{\gamma(\Gamma)}=\zeta^{\ex}_2(\Gamma)\, |F_R|^{\gamma(\Gamma)}.
\end{equation}
\end{definition}
\begin{remark}
Note that $F_R\subset V_R$ so $\gamma(\Gamma)$ is well-defined for all $\Gamma$. Further, it is not hard to see that $\zeta^{\ex}_1(\Gamma)\leq \zeta^{\ex}_2(\Gamma)$ (since $H_{F_L,V_R}$ is a subgraph of $H_{V_L,F_R}$),  so $\gamma(\Gamma)\geq 0$ for all $\Gamma$.
\end{remark}

With these definitions, for a 2-coloured graph $\Gamma$, we define the following subset of $\Cc_{\alpha,\beta}$ (while the definition makes sense for general values of $\alpha,\beta$, we will typically assume that $\alpha,\beta$ are defined by \eqref{eq:defa0b0}):

\begin{definition}[The set of dominating bicliques $\Cc_{\alpha,\beta}^{\Gamma}$]\label{def:ccabgamma}
Let $0<\alpha,\beta<1$  and $\Gamma$ be a 2-coloured graph. Define $\Cc^{\Gamma}_{\alpha,\beta}$  to be the following subset of $\Cc_{\alpha,\beta}$:
\begin{equation}\label{eq:ccgammacompl}
\Cc^{\Gamma}_{\alpha,\beta}:=\left\{(S_L,S_R)\in \Cc_{\alpha,\beta}\mathrel{}\middle|\mathrel{} (S_L,S_R)=
   \arg\max_{(S'_L,S_R')\in \Cc_{\alpha,\beta}}\zeta(S_L',S_R',\Gamma)|S_R'|^{\gamma(\Gamma)}\right\}.
\end{equation}
Note that for all $\alpha,\beta>0$ and all 2-coloured graphs $\Gamma$, the elements of $\Cc^{\Gamma}_{\alpha,\beta}$ are maximal bicliques of $H$.
\end{definition}

It is useful to see how the above definitions degenerate in the case where $\Gamma$ is the empty graph. Then, we have that $\zeta(S_L,S_R,\Gamma)=1$ for all $(S_L,S_R)\in\Cc$, and in particular we have that $\zeta^{\ex}_1(\Gamma)=\zeta^{\ex}_2(\Gamma)=1$. It follows from \eqref{eq:defgamma} that $\gamma(\Gamma)=0$. It follows that when $\Gamma$ is the empty graph, $\Cc^{\Gamma}_{\alpha,\beta}$ is identical to $\Cc_{\alpha,\beta}$. 

In the next section, we give an outline for the cases which arise in the proof of Lemma~\ref{lem:main} and give examples for each case. 

\subsection{The cases in the proof of Lemma~\ref{lem:main} (overview with examples)}\label{sec:examples}
Consider the set of maximal bicliques $\Cc_{\alpha,\beta}$ where $0<\alpha,\beta<1$ are as in \eqref{eq:defa0b0}. For the discussion in this section we may assume that $\Cc_{\alpha,\beta}$ includes both extremal bicliques and at least one non-extremal biclique. Let 
\[\big(S^{(1)}_L,S^{(1)}_R\big),\hdots, \big(S^{(t)}_L,S^{(t)}_R\big)\]
be an enumeration of the non-extremal bicliques in $\Cc_{\alpha,\beta}$. Recall that all elements of $\Cc_{\alpha,\beta}$ are maximal bicliques of $H$. For convenience, in this section let $H_i$ denote the subgraph $H[S_L^{(i)}\cup V_R]$ (this corresponds to the graph $H_{S_L^{(i)}, S_R^{(i)}}$ in Definition~\ref{def:specialgraph}) and set $\zeta_i(\Gamma)=\zeta(S_L^{(i)},S_R^{(i)},\Gamma)=\FixCOL{H_i}(\Gamma)$ (cf. Definition~\ref{def:zeta}). For the extremal bicliques we will instead use the notation $H^{\ex}_1, H^{\ex}_2$ to denote the graphs $H_{F_L,V_R},H_{V_L,F_R}$ respectively. Note that $H^{\ex}_1$ is a complete bipartite graph with bipartition $\{F_L,V_R\}$ while $H^{\ex}_2$ is $H$ itself.

There are three complementary cases to consider for the proof of Lemma~\ref{lem:main}. Recall by construction that for every 2-coloured graph $\Gamma$, the extremal bicliques have equal contribution in the graph $K_{\hat{a},\hat{b}}(\Gamma)$ (where $\hat{a}=Qa$, $\hat{b}=Qb+\gamma(\Gamma)$ for some large integer $Q$).
The three cases are as follows. 
\begin{description} 
\item[Case 1:] There exists $i\in [t]$ and a 2-coloured graph $\Gamma$ such that the biclique $\big(S^{(i)}_L,S^{(i)}_R\big)$ dominates over the extremal bicliques in the gadget $K_{\hat{a},\hat{b}}(\Gamma)$. 
\item [Case 2:] There exists $i\in [t]$ such that for every 2-coloured graph $\Gamma$ the biclique $\big(S^{(i)}_L,S^{(i)}_R\big)$ has the same contribution as the extremal bicliques in the gadget $K_{\hat{a},\hat{b}}(\Gamma)$. 
\item [Case 3:]  For all $i\in [t]$ and every 2-coloured graph $\Gamma$, the contribution of the biclique $\big(S^{(i)}_L,S^{(i)}_R\big)$ is less 
than
or equal to the contribution of the extremal bicliques in the gadget $K_{\hat{a},\hat{b}}(\Gamma)$. Further, for all $i\in [t]$ there exists a 2-coloured graph $\Gamma_i$ such that the biclique $\big(S^{(i)}_L,S^{(i)}_R\big)$ is dominated by the extremal bicliques in the gadget $K_{\hat{a},\hat{b}}(\Gamma_i)$.
\end{description}
We next give an example for each case and overview how our proof works.
\vskip 0.2cm

{\bf Example of Case 1:}\quad An example of Case~1 is depicted in
Figure~\ref{fig:caseone}. 
\begin{figure}[h]
\begin{subfigure}{0.25\textwidth}
  \centering
  \scalebox{0.45}[0.45]{\input{Hexample9.tex}}
  \caption{The graph $H$.}
  \label{fig:caseone1}
\end{subfigure}
\begin{subfigure}{0.58\textwidth}
  \centering
  \scalebox{0.35}[0.35]{\input{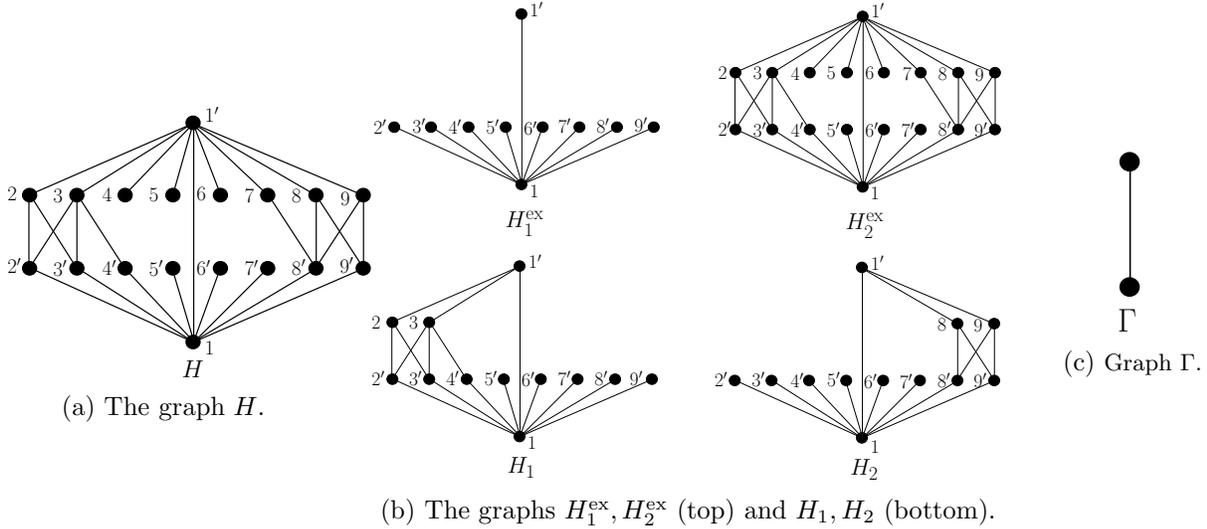}}
  \caption{The graphs $H^{\ex}_1,H^{\ex}_2$ (top) and $H_1,H_2$ (bottom).}
  \label{fig:caseone2}
\end{subfigure}
\begin{subfigure}{0.12\textwidth}
  \centering
  \scalebox{0.6}[0.6]{\input{Hexample9b.tex}}
  \caption{\footnotesize{Graph $\Gamma$.}}
  \label{fig:caseone3}
\end{subfigure}
\caption[.]{An example of a graph $H$ which falls into Case~1 of our analysis.}\label{fig:caseone}
\end{figure}  
Let us first see why the example is in Case~1.
The full vertices of $H$, as depicted in Figure~\ref{fig:caseone},
 are vertices~$1$ and $1'$ and the extremal bicliques 
 of~$H$ are $(\{1\},[9'])$ and $([9],\{1'\})$.
 Thus, the $\alpha,\beta$ 
 pairs which equalise the contribution of the extremal bicliques in $K_{a,b}$ satisfy $\alpha=\beta$.  The dominating bicliques $\Cc_{\alpha,\beta}$ for $\alpha=\beta$ are the extremal bicliques and the two bicliques $(S_L^{(1)},S_R^{(1)})=(\{1,2,3\},\{1',2',3'\})$ and $(S_L^{(2)},S_R^{(2)})=(\{1,8,9\},\{1',8',9'\})$.

 For the graph $\Gamma$ in Figure~\ref{fig:caseone3}, it holds that $\FixCOL{\widehat{H}}(\Gamma)=E(\widehat{H})$ where $\widehat{H}$ is any of the graphs $H^{\ex}_1,H^{\ex}_2, H_1,H_2$. Thus, 
\[\FixCOL{H^{\ex}_1}(\Gamma)=9, \ \FixCOL{H^{\ex}_2}(\Gamma)=27, \ 
\FixCOL{H_1}(\Gamma)= 16, 
\ \FixCOL{H_2}(\Gamma)= 15.
\]
It follows that $\gamma(\Gamma)= 1/2$ 
(cf. Definition~\ref{def:defgamma}) and 
$\zeta^{\ex}_1(\Gamma) |V_R|^{\gamma}=\zeta^{\ex}_2(\Gamma) |F_R|^{\gamma}=27.$
For the non-extremal elements of $\Cc_{\alpha,\beta}$ it holds that 
$\zeta_1(\Gamma) |S_R^{(1)}| ^{\gamma}=16 \times  \sqrt{3}
>27,
\quad \zeta_2(\Gamma) |S_R^{(2)}|^{\gamma}=15 \times  \sqrt{3}<
27.$
Thus the only dominating biclique in $\Cc^{\Gamma}_{\alpha,\beta}$ for the gadget $K_{\hat{a},\hat{b}}(\Gamma)$ is $(\{1,2,3\},\{1',2',3'\})$.

Now, in the general setting of Case~1, we have that, in the gadget $K_{\hat{a},\hat{b}}(\Gamma)$, the extremal bicliques are not dominating (since they are dominated by the biclique $(S_L^{(i)},S^{(i)}_R)$). Note however that there may still be more than one element in $\Cc^{\Gamma}_{\alpha,\beta}$, unlike the example in Figure~\ref{fig:caseone}. 
To pick out only one biclique from the set $\Cc^{\Gamma}_{\alpha,\beta}$ we further apply  Lemma~\ref{lem:lovasztwo} on the graphs $H_i$ corresponding to bicliques in $\Cc^{\Gamma}_{\alpha,\beta}$. This yields a graph $J$ which ``prefers" a particular graph, say, $H_j$. Then, using an argument analogous to the one in Section~\ref{sec:multiple} (i.e., by pasting  sufficiently many disjoint copies of $J$ in $K_{\hat{a},\hat{b}}(\Gamma)$), one can show that $\FixCOL{H_j}\leq_{\AP}\FixCOL{H}$. By induction, we have $\BIS\leq_{\AP}\FixCOL{H_j}$ and hence $\BIS\leq_{\AP}\FixCOL{H}$ as well.
\vskip 0.2cm
{\bf Example of Case  2:}\quad An example of Case~2 is depicted
in Figure~\ref{fig:casetwo}. \begin{figure}[h]
\centering
\begin{subfigure}{0.2\textwidth}
  \centering
  \scalebox{0.78}[0.78]{\input{Hexample7.tex}}
  \caption{The graph $H$.}
  \label{fig:casetwo1}
\end{subfigure}
\begin{subfigure}{0.7\textwidth}
  \centering
  \scalebox{0.68}[0.68]{\input{Hexample7a.tex}}
  \caption{The graphs $H^{\ex}_1,H^{\ex}_2,H_1,H_2$.}
  \label{fig:casetwo2}
\end{subfigure}
\caption[.]{An example of a graph $H$ which falls into Case~2 of our analysis.
}\label{fig:casetwo}
\end{figure}
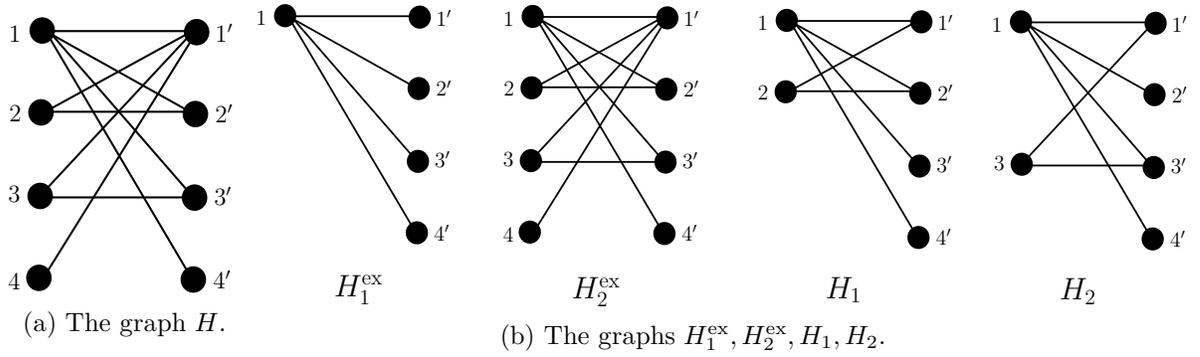
Note that this is the same graph that we encountered earlier in Figure~\ref{fig:coexistence}. The $\alpha,\beta$ 
 pairs which equalise the contribution of the extremal bicliques in $K_{a,b}$ satisfy $\alpha=\beta$, in which case \[\Cc_{\alpha,\beta}=\{(\{1\},[4']),([4],\{1'\}),(\{1,2\},\{1',2'\}),(\{1,3\},\{1',3'\})\}.\]

To see why $H$ falls into case~2 of our analysis, let us first prove the following equality which holds for every 2-coloured graph $\Gamma$:
\begin{equation}\label{eq:identity}
(\FixCOL{H_1}(\Gamma))^2=\FixCOL{H^{\ex}_1}(\Gamma)\, \FixCOL{H^{\ex}_2}(\Gamma).
\end{equation}
An analogous equality holds if we replace $H_1$ with $H_2$ (since $H_1\cong_c H_2$). The fastest way to derive \eqref{eq:identity} is to observe that the tensor product
 of $H_1$ with itself
 is isomorphic to the tensor product of $H^{\ex}_1$ and $H^{\ex}_2$. An alternative way to deduce the equality using tensor products of smaller graphs is as follows. Let $P_n$ denote the (2-coloured) path with $n$ vertices, whose 2-colouring is induced by colouring one of its degree-one vertices with the colour $R$. Then it is not hard to see that $H^{\ex}_1\cong_c P_3\times P_3$, $H^{\ex}_2\cong_c P_4\times P_4$, $H_1\cong_c P_3\times P_4$ (where $\times$ denotes graph tensor product). Thus, for every 2-coloured graph $\Gamma$ it holds that
\begin{gather*}
\FixCOL{H^{\ex}_1}(\Gamma)=(\FixCOL{P_3}(\Gamma))^2,\quad \FixCOL{H^{\ex}_2}(\Gamma)=(\FixCOL{P_4}(\Gamma))^2, \\
\FixCOL{H_1}(\Gamma)=\FixCOL{P_3}(\Gamma)\,\FixCOL{P_4}(\Gamma),
\end{gather*}
and \eqref{eq:identity} follows. (Note, by Lemma~\ref{lem:lovaszone}, \eqref{eq:identity} implies that $H_1\times H_1$ is isomorphic to $H^{\ex}_1\times H^{\ex}_2$.)

Using \eqref{eq:identity}, we next show that, for $\alpha=\beta$, for all 2-coloured graphs $\Gamma$, the set  $\Cc^{\Gamma}_{\alpha,\beta}$ consists of the extremal bicliques and the non-extremal bicliques $(S_L^{(1)},S_R^{(1)})=(\{1,2\},\{1',2'\})$,  $(S_L^{(2)},S_R^{(2)})=(\{1,3\},\{1',3'\})$. Note that $\FixCOL{H_1}(\Gamma)=\zeta_1(\Gamma)$ and 
\[\FixCOL{H^{\ex}_1}(\Gamma)\, \FixCOL{H^{\ex}_2}(\Gamma)=4^{|R(\Gamma)|}\FixCOL{H}(\Gamma)=4^{|R(\Gamma)|}\zeta_2^{\ex}(\Gamma).\]
Since $\gamma(\Gamma)$ by definition satisfies $\zeta_2^{\ex}(\Gamma)=\zeta_1^{\ex}(\Gamma)4^{\gamma}=4^{|R(\Gamma)|}4^{\gamma}$, we obtain from \eqref{eq:identity}  that 
$\zeta_1(\Gamma) =  4^{|R(\Gamma)|}2^{\gamma}$.
So, since $|S_R^{(1)}|=2$, we have
$\zeta_1(\Gamma)|S_R^{(1)}|^{\gamma}=\zeta_1(\Gamma)2^\gamma=4^{|R(\Gamma)|}4^{\gamma}=\zeta_2^{\ex}(\Gamma)$. 
Analogously, $\zeta_1(\Gamma)||S_R^{(2)}|^{\gamma}=\zeta_2^{\ex}(\Gamma)$. It follows that for all 2-coloured graphs $\Gamma$,  $\Cc^{\Gamma}_{\alpha,\beta}$ consists of the extremal bicliques and the non-extremal bicliques $(S_L^{(1)},S_R^{(1)})$, $(S_L^{(2)},S_R^{(2)})$. 
 
Now, in the general setting of Case~2, the non-existence of a gadget  $K_{\hat{a},\hat{b}}(\Gamma)$ that distinguishes between $(S_L^{(i)},S_R^{(i)})$ and the extremal bicliques for any choice of the graph $\Gamma$ allows us to quantitatively relate $\FixCOL{H}(\Gamma)$ and $\FixCOL{H_i}(\Gamma)$ on every ``input" $\Gamma$, i.e., obtain an equality analogous to \eqref{eq:identity}. It thus follows that  $\FixCOL{H_i}\leq_{\AP}\FixCOL{H}$ and we thus obtain that $\BIS\leq_{\AP}\FixCOL{H}(\Gamma)$ as in Case~1.
\vskip 0.2cm
{\bf Example of Case  3:}\quad
An example of Case~3 is depicted in Figure~\ref{fig:casethree}.
\begin{figure}[h]
\begin{subfigure}{0.25\textwidth}
  \centering
  \scalebox{0.45}[0.45]{\input{Hexample10.tex}}
  \caption{The graph $H$.}
  \label{fig:casethree1}
\end{subfigure}
\begin{subfigure}{0.57\textwidth}
  \centering
  \scalebox{0.35}[0.35]{\input{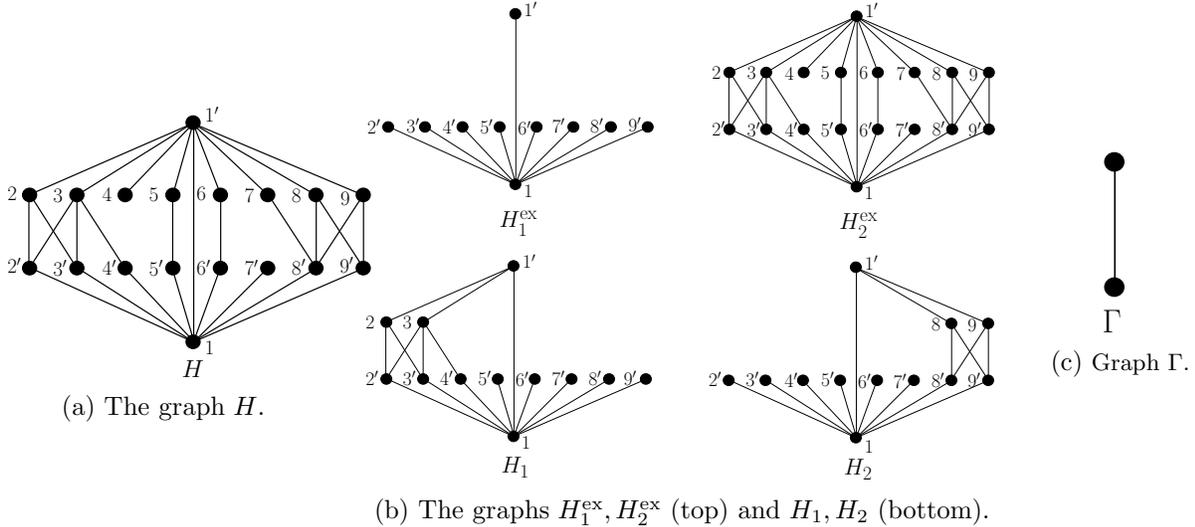}}
  \caption{The graphs $H^{\ex}_1,H^{\ex}_2$ (top) and $H_1,H_2$ (bottom).}
  \label{fig:casethree2}
\end{subfigure}
\begin{subfigure}{0.12\textwidth}
  \centering
  \scalebox{0.6}[0.6]{\input{Hexample9b.tex}	}
  \caption{\footnotesize{Graph $\Gamma$.}}
  \label{fig:casethree3}
\end{subfigure}
\caption[.]{An example of a graph $H$ which falls into Case 3 of our analysis. 
}\label{fig:casethree}
\end{figure}

The only difference 
between~$H$ here and
the graph $H$ in Figure~\ref{fig:caseone} 
is that here there are the extra edges $(5,5')$ and $(6,6')$. In particular, only the graph $H^{\ex}_2$ changes. The dominating bicliques $\Cc_{\alpha,\beta}$ for $\alpha=\beta$ are once again the extremal bicliques and the two bicliques $(S_L^{(1)},S_R^{(1)})=(\{1,2,3\},\{1',2',3'\})$, $(S_L^{(2)},S_R^{(2)})=(\{1,8,9\},\{1',8',9'\})$.

For the graph $\Gamma$ in Figure~\ref{fig:casethree3}, it holds that $\FixCOL{\widehat{H}}(\Gamma)=E(\widehat{H})$ where $\widehat{H}$ is any of the graphs $H^{\ex}_1,H^{\ex}_2, H_1,H_2$. Thus, 
\[\FixCOL{H^{\ex}_1}(\Gamma)=9, \ \FixCOL{H^{\ex}_2}(\Gamma)=29, \ \FixCOL{H_1}(\Gamma)=16, 
\ \FixCOL{H_2}(\Gamma)=15 
.\]
From $\zeta^{\ex}_1(\Gamma) |V_R|^{\gamma}=\zeta^{\ex}_2(\Gamma) |F_R|^{\gamma}=29$ (cf. Definition~\ref{def:defgamma}), we obtain that $\gamma(\Gamma)$ satisfies $9^{\gamma}=29/9$.
For the non-extremal elements of $\Cc_{\alpha,\beta}$ it holds that 
$\zeta_1(\Gamma) |S_R^{(1)}| ^{\gamma}=16
\sqrt{29}/3<29,\quad \zeta_2(\Gamma) |S_R^{(2)}|^{\gamma}=15
\sqrt{29}/3<29.$
Thus the dominating bicliques in $\Cc^{\Gamma}_{\alpha,\beta}$ for the gadget $K_{\hat{a},\hat{b}}(\Gamma)$ are the extremal bicliques.

Now consider the general setting of Case~3.
Recall from the statement of Case~3
that for all $i\in[t]$ there is a 2-coloured graph $\Gamma_i$ such that the biclique $\big(S^{(i)}_L,S^{(i)}_R\big)$ is dominated by the extremal bicliques in the gadget $K_{\hat{a},\hat{b}}(\Gamma_i)$.
Consider the graph $\Gamma$ which is the disjoint union of the $\Gamma_i$'s. It is not hard then  to show that in the gadget  $K_{\hat{a},\hat{b}}(\Gamma)$ the extremal bicliques are  dominating over all non-extremal bicliques. One can then use the reduction for Item~\ref{it:extremal} in Section~\ref{sec:qw} to show that $\BIS\leq_{\AP}\FixCOL{H}$ (instead of using the gadget $K_{a,b}$ as discussed there, 
we instead use  the gadget $K_{\hat{a},\hat{b}}(\Gamma)$).

\section{Technical lemmas}
We next state a few technical approximation results which will be used in the course of the proofs. The first will be used in the reductions to approximate real numbers with appropriate integers.
\begin{lemma}[Dirichlet's Approximation Theorem, cf. {\cite[p. 34]{Schmidt}}]\label{lem:dirichlet}
Let $\alpha_1,\hdots,\alpha_d>0$ be real numbers and $N$ be a natural number. Then, there exist positive integers $p_1,\hdots,p_d,q$  with $q\leq N$ such that $|q\alpha_i-p_i|\leq 1/N^{1/d}$ for every $i\in[d]$.
\end{lemma}

We will also need the following estimate for the number of \emph{surjective} mappings from a set of size $n$ to a set of size $k$. We will denote this number by $T(n,k)$.  In terms of the Stirling number of the second kind $S(n,k)$, it holds that 
\[T(n,k)=S(n,k) k!.\]
It is well known that for fixed $k$, the asymptotic order of $S(n,k)$ is $k^n/k!$ and thus $T(n,k)$ is  $(1+o(1))k^n$. We will need the following more precise estimate.
\begin{lemma}[Lemma 18 in \cite{DGGJ}]\label{lem:surj}
Let $n,k$ be positive integers. Denote by $T(n,k)$ be the number of \emph{surjective} mappings from a set of size $n$ to a set of size $k$, i.e., $T$ equals the cardinality of the set $\big\{h:[n]\rightarrow [k]\mid h([n])=[k]\big\}$. Then, for all $n\geq 2k\ln k$, 
\[(1-2k/n)k^n\leq T(n,k)\leq k^n.\]
\end{lemma}
\begin{proof}
It is shown in \cite{DGGJ} that  
\[\big(1-e^{-n/(2k)}\big)k^n\leq T(n,k)\leq k^n.\] 
The result follows from $1/x>\exp(-x)$ applied to 
$x=n/(2k)$. 
\end{proof}
Finally, we will use the following bound to bound expressions of the form $c^{\pm 1/n}$, where $c$ is a constant independent of $n$.
\begin{lemma}\label{lem:xzapprox}
Let $n,K$ be positive integers with $n\geq K$. Suppose that $x$ is a real number such that $1\leq x\leq K$ and let $z$ be a real number such that $|z|\leq 1/n$. Then
\[|x^z-1|\leq 2K/n.\]
\end{lemma}
\begin{proof}
Let $f(z)=x^z$. By the mean value theorem on $f$, there exists $\rho$ satisfying $|\rho|\leq |z|\leq 1/n$ such that $f(z)-f(0)=z f'(\rho)$. Using that $x\geq 1$, we thus obtain
\[|x^z-1|= |z|\, |f'(\rho)|=|z|\, |x^\rho\ln x|=|z|\, x^\rho\ln x\leq (2\ln x)/n\leq 2K/n.\]
where in the middle inequality we used that $|z|\leq 1/n$ and $x^\rho \leq 2$, while in the last inequality we used that $\ln x\leq \ln K<K$. We argue briefly for the validity of $x^\rho \leq 2$: since $f$ is increasing (from $x\geq 1$),  we have $x^\rho\leq x^{|\rho|}\leq x^{1/n}$. We  also have $2^n\geq n\geq K\geq x$, which yields $x^{1/n}\leq 2$.
\end{proof}

We conclude this section with a piece of notation. For non-negative real numbers $X,Z$ and $\epsilon>0$ we will write $X=(1\pm \epsilon)Z$ to denote that 
\[(1-\epsilon)Z\leq X\leq (1+\epsilon)Z.\]
It is not hard to see that if $X_1=(1\pm\epsilon_1) Z_1$ and $X_2=(1\pm \epsilon_2)Z_2$ for $0<\epsilon_1,\epsilon_2<1$, then $X_1X_2=(1\pm \epsilon)Z_1Z_2$ where $\epsilon=\frac{3}{2}(\epsilon_1+\epsilon_2)$.

\section{Proof of graph homomorphism lemma}\label{sec:prelimproofs}

In ths section, we prove the 
graph-homomorphism lemmas from Section~\ref{sec:preliminaries}.
First, the proof of the following 
lemma is a close adaptation of  \cite[Proof of Theorem 2.11]{HellNesetril}.

\begin{lovaszone}
\statelovaszone{}
\end{lovaszone}

\begin{proof}  For the sake of contradiction, we may assume that for all 2-coloured graphs 
$J$ with $|V(J)|\leq\max\{|V(H_1)|,|V(H_2)|\}$
it holds that 
\begin{equation}\label{eq:contra}
\FixCOL{H_1}(J)=\FixCOL{H_2}(J).
\end{equation} 

For 2-coloured graphs $H,J$ denote by $\InjFixCOL{H}(J)$ the number of \emph{injective} colour-preserving homomorphisms from $J$ to $H$. Assuming that \eqref{eq:contra} holds for all 2-coloured graphs
$J$ with $|V(J)|\leq\max\{|V(H_1)|,|V(H_2)|\}$,
we will further show that for all 
such
2-coloured graphs $J$ it holds that 
\begin{equation}\label{eq:contrainj}
\InjFixCOL{H_1}(J)=\InjFixCOL{H_2}(J).
\end{equation}
Plugging $J=H_1$ and $J=H_2$ into \eqref{eq:contrainj} yields the existence of a colour-preserving isomorphism between $H_1$ and $H_2$, i.e., $H_1\cong_c H_2$, contradicting the assumption in the statement of the lemma.

To prove \eqref{eq:contrainj}, we proceed by induction on $|V(J)|$. When $|V(J)|=1$, for every 2-coloured graph $H$ we (trivially) have that $\FixCOL{H}(J)=\InjFixCOL{H}(J)$ from where we obtain $\InjFixCOL{H_1}(J)=\InjFixCOL{H_2}(J)$. So assume that:
\begin{equation}\label{eq:inductioninj}
\begin{gathered}
\mbox{for all 2-coloured graphs $J'$ with $|V(J')|<|V(J)|$,}\\
\mbox{it holds that $\InjFixCOL{H_1}(J')=\InjFixCOL{H_2}(J')$.}
\end{gathered}
\end{equation}
We will show that $\InjFixCOL{H_1}(J)=\InjFixCOL{H_2}(J)$ as well. Let $\Theta_L, \Theta_R$ be partitions of $L(J), R(J)$, respectively. We denote by $J/(\Theta_L, \Theta_R)$ the 2-coloured  graph which is obtained by contracting every part in the partition $\Theta_L$ and every part in the partition $\Theta_R$. For 2-coloured graphs $H,J$ it is straightforward to check that
\begin{equation}\label{eq:basicsuminj}
\FixCOL{H}(J)=\sum_{\Theta_L}\sum_{\Theta_R}\InjFixCOL{H}\big(J/(\Theta_L,\Theta_R)\big),
\end{equation}
where in the latter sum $\Theta_L,\Theta_R$ range over all partitions of $L(J)$ and $R(J)$. Note that $J/(\Theta_L,\Theta_R)$ has fewer vertices than $J$, with the single exception of  the case that $\Theta_L,\Theta_R$ consist only of singleton sets where we have that $J/(\Theta_L,\Theta_R)=J$. Applying \eqref{eq:basicsuminj} for $H=H_1$ and $H=H_2$, from \eqref{eq:inductioninj} and the assumption $\FixCOL{H_1}(J)=\FixCOL{H_2}(J)$, we obtain $\InjFixCOL{H_1}(J)=\InjFixCOL{H_2}(J)$, as desired.
\end{proof}

We will also use the following generalisation of  Lemma~\ref{lem:lovaszone}.

\begin{lovasztwo}
\statelovasztwo{}
\end{lovasztwo}

\begin{proof}  
The proof is by induction on $k$. For $k=2$, the claim is immediate from Lemma~\ref{lem:lovaszone}. Assume that the claim holds for $k-1$, we next show it for $k$. By applying the inductive hypothesis to the graphs $H_2,\hdots,H_k$ (and renaming if necessary),  we may assume that $J_2$ is a bipartite graph such that 
\begin{equation}\label{eq:choiceJ2}
\FixCOL{H_2}(J_2)>\FixCOL{H_i}(J_2)\mbox{ for all $i>2$}.
\end{equation} 
 We may assume that
\begin{equation}\label{eq:H1H2}
\FixCOL{H_1}(J_2)= \FixCOL{H_2}(J_2)=:M.
\end{equation}
(Otherwise we may choose $J=J_2$ and $i^*=\arg\max_{i\in\{1,2\}}\FixCOL{H_i}(J_2)$.) By Lemma~\ref{lem:lovaszone}  (and renaming $H_1,H_2$ if necessary; this is justified by \eqref{eq:H1H2}), we may assume that $J'$ is a bipartite graph such that 
\begin{equation}\label{eq:choiceJp}
\FixCOL{H_1}(J')>\FixCOL{H_2}(J').
\end{equation}
Let $J$  be the disjoint union of $J'$ and $t$ copies of $J_2$, where $t:=\left\lceil C M\right\rceil$ where $M$ is as in \eqref{eq:H1H2} and  
\[C:=\max_{i\in[k]\backslash\{1,2\}}\frac{\FixCOL{H_i}(J')}{\FixCOL{H_1}(J')}.\] 
For every $i\in[k]$, we have that 
\[\FixCOL{H_i}(J)=\big(\FixCOL{H_i}(J_2)\big)^t\, \FixCOL{H_i}(J').\]
From \eqref{eq:H1H2} and \eqref{eq:choiceJp}, we have that $\FixCOL{H_1}(J)>\FixCOL{H_2}(J)$. Further, for $i\notin\{1,2\}$, we have
\begin{align*}
\frac{\FixCOL{H_1}(J)}{\FixCOL{H_i}(J)}&=\frac{\FixCOL{H_1}(J')}{\FixCOL{H_i}(J')}\Big(\frac{\FixCOL{H_1}(J_2)}{\FixCOL{H_i}(J_2)}\Big)^t\\
&\geq \frac{1}{C}\Big(1+\frac{1}{\FixCOL{H_i}(J_2)}\Big)^t\\
&> \frac{t}{C \FixCOL{H_i}(J_2)}\geq 1,
\end{align*}
where we used \eqref{eq:choiceJ2} to deduce $\FixCOL{H_1}(J')\geq \FixCOL{H_2}(J')+1$, the inequality $(1+x)^t\geq 1+t x>tx$ (for $t\geq 1$ and $x\geq 0$) as well as the choice of $t$ (cf. the definitions of $C,M$).
\end{proof}

\section{Proof of Lemma~\ref{lem:main}}\label{sec:lemmafull}
The following two lemmas will be important for the proof of Lemma~\ref{lem:main}.
\begin{lemma}\label{lem:slsr}
Let $H$ be a 2-coloured graph which is full  
but not trivial.
Let $0<\alpha,\beta<1$ be as in \eqref{eq:defa0b0}. Suppose that $\Gamma$ is a 2-coloured graph with no isolated vertices in $R(\Gamma)$. Assume further  that $(F_L,V_R),(V_L,F_R)\notin \Cc^{\Gamma}_{\alpha,\beta}$. 
Then a 2-coloured, full 
but not trivial
graph $H'$ can be specified  
such that $|V(H')|<|V(H)|$ and 
\begin{equation*}
\FixCOL{H'}\leq_\AP\FixCOL{H}.
\end{equation*}
\end{lemma}
\begin{proof}
For $(S_L,S_R)\in\Cc^{\Gamma}_{\alpha,\beta}$, recall that $(S_L,S_R)$ is a maximal biclique. From the assumptions, $(S_L,S_R)$ is not extremal,  so by Lemma~\ref{inductionstep} the graph $H_{S_L,S_R}$ is full 
but not trivial
and has fewer vertices  than $H$ (to see this, recall the remark in  Definition~\ref{def:specialgraph} that for a maximal biclique $(S_L,S_R)$ we have that $H_{S_L,S_R}=H[S_L\cup V_R]$; thus $H_{S_L,S_R}$ corresponds to the graph $H_1$ in Lemma~\ref{inductionstep}). We remark for later use that for a maximal biclique $(S_L,S_R)$ it holds that $R(H_{S_L,S_R})=V_R$.

We will find a biclique $(S_L,S_R)$ in $\Cc^{\Gamma}_{\alpha,\beta}$ such that $\FixCOL{H_{S_L,S_R}}\leq_{\AP}\FixCOL{H}$. To find such a biclique, we will use Lemma~\ref{lem:lovasztwo}. Specifically, let $\Hc$ be the set consisting of the (labelled) graphs $H_{S_L,S_R}$ when $(S_L,S_R)$ ranges over maximal bicliques in $\Cc_{\alpha,\beta}^{\Gamma}$. Further, let $\Hc_1,\hdots,\Hc_k$ be the equivalence classes of $\Hc$ under the (colour-preserving isomorphism) relation $\cong_c$ which was formally defined in Section~\ref{sec:preliminaries}. For $i\in[k]$, let $H_i$ be a representative 
of the class $\Hc_i$. By Lemma~\ref{lem:lovasztwo}, there exists a gadget $J$ and $i\in [k]$ such that for every 
$t\in[k]\setminus\{i\}$,
it holds that $\FixCOL{H_i}(J)>\FixCOL{H_t}(J)$. We will show that $H'=H_i$ satisfies the statement of the lemma. For future use, let $\Cc^{*}\subseteq \Cc^{\Gamma}_{\alpha,\beta}$ be the subset of bicliques $(S_L,S_R)$ whose corresponding graph $H_{S_L,S_R}$ is isomorphic to $H'$ under a colour-preserving isomorphism. Formally,
\begin{equation}\label{eq:defcstar}
\Cc^*:=\big\{(S_L,S_R)\in \Cc_{\alpha,\beta}^{\Gamma}\, |\, H_{S_L,S_R}\cong_c H'\big\}.
\end{equation}
For technical convenience we assume that the graph $J$ has no isolated vertices in $R(J)$. 
This is without loss of generality. To see this, suppose that $I\subseteq R(J)$ is the set of isolated vertices in $R(J)$. Let $J'$ be the 2-coloured graph which is obtained from $J$ by removing the vertices in $I$; clearly, $J'$ has no isolated vertices in $R(J)$. Observe that for each $t\in [k]$ it holds that $R(H_t)=V_R$ since by definition $H_t$ is the graph $H_{S_L,S_R}$ for some maximal biclique $(S_L,S_R)$. Thus, for every $t\in [k]$ it holds that    $\FixCOL{H_t}(J)=|V_R|^{|I|}\FixCOL{H_t}(J')$. Hence, by the choice of $J$, we have that $\FixCOL{H_i}(J')>\FixCOL{H_t}(J')$ for every $t\neq i$.

Having defined the graph $H'$, we next give the details of the reduction.  Let $0<\epsilon<1$ and $G'$ be an input to $\FixCOL{H'}$. Our goal is to approximate $\FixCOL{H'}(G')$ within a multiplicative factor of $(1\pm \epsilon)$ using an oracle for $\FixCOL{H}$. Using standard arguments,\footnote{If $G'_1,\hdots, G'_l$ are the connected components of $G'$ then it holds that $\FixCOL{H'}(G')=\prod^{l}_{t=1}\FixCOL{H'}(G'_l)$.} we may assume w.l.o.g. that $G'$ is connected.

Suppose that $n'=|V(G')|$ and let $n=\left\lceil  \kappa n'/\epsilon\right\rceil$ where $\kappa$ is a sufficiently large constant depending only on $H,\Gamma$. Note that by letting $\kappa$ be large, we also ensure that $n$ is sufficiently large both in absolute terms as well as relatively to $n'$ and $1/\epsilon$. For the time being, let us assume that $\kappa$ is such that $\alpha n\geq 1$, $\beta n \geq 1$, $n\geq 2\gamma$, $2\max\{|V_L|,|V_R|\}/n\leq \epsilon/10^4$. Later, we will refine appropriately the choice of $\kappa$.

By Lemma~\ref{lem:dirichlet} applied to the numbers $\alpha n^3$ and $n^2(\beta n+\gamma)$, there exists a positive integer $Q\leq n^2$ and positive integers $a,b$ such that 
\begin{equation}\label{eq:abdefinition}
|Q \alpha  n^3-a|\leq 1/n,\quad |Q (\beta n^3+\gamma n^2)-b|\leq 1/n.
\end{equation}
Note that since $1\leq Q\leq n^2$ and $\alpha,\beta,\gamma$ are constants, for all sufficiently large $n$, it holds that $n^2\leq a,b\leq 2n^{5}$. The choice of $a,b$ ensures that for $(S_L,S_R)\in \Cc$, it holds that 
\begin{align}
\zeta(S_L,S_R,\Gamma)^{Qn^2}|S_L|^a|S_R|^b&=\big(1\pm \epsilon/10^3\big)\big(\zeta(S_L,S_R,\Gamma)\, |S_L|^{\alpha n} |S_R|^{\beta n+\gamma}\big)^{Qn^2}\notag\\
&=\big(1\pm \epsilon/10^3\big)\big(|S_L|^{\alpha} |S_R|^{\beta}\big)^{Qn^3}\big(\zeta(S_L,S_R,\Gamma)\,  |S_R|^{\gamma}\big)^{Qn^2}\label{eq:approxslsr2},
\end{align}
where in the first equality we used Lemma~\ref{lem:xzapprox} for $x=|S_L|$ and $K=|V_L|$ to bound $|S_L|^{a-Q\alpha n^3}=1\pm 2|V_L|/n=1\pm \epsilon/10^{4}$ and similarly $|S_R|^{b-Q (\beta n^3+\gamma n^2)}=1\pm \epsilon/10^{4}$, 
while the second equality is a mere rearrangement of terms.

We next construct an instance $G$ of $\FixCOL{H}$. Recall that  $K_{a,b}$ is the complete bipartite graph with $a,b$ vertices on the left and right, respectively. Let $\overline{G}$ be the disjoint union of the graphs $G'$, $K_{a,b}$, $Qn^2$ disjoint copies of the graph $\Gamma$ and $\lceil n^{3/2}\rceil $ copies of the gadget $J$.  We denote by $\Gamma_r$ the $r$-th copy of $\Gamma$ and by $J_r$ the $r$-th copy of $J$.  To construct $G$, attach on $\overline{G}$ all edges in $R(K_{a,b})\times L(G')$, for $r=1,\hdots,Qn^2$ add all edges in $R(K_{a,b})\times L(\Gamma_r)$ and, finally, for $r=1,\hdots,\lceil n^{3/2}\rceil$ add all edges in $R(K_{a,b})\times L(J_r)$. Note that $G$ has size polynomial in $n,1/\epsilon$. 

For the remainder of the proof, for 2-coloured graphs $\hat{G},\hat{H}$,  we will use for convenience the shorthand $Z_{\hat{G}}(\hat{H})$ to denote $\FixCOL{\hat{H}}(\hat{G})$, i.e., $Z_{\hat{G}}(\hat{H})$ denotes the number of colour-preserving homomorphisms from $\hat{G}$ to $\hat{H}$. It will also be helpful to denote  
\begin{equation*}
M_K:=\max_{(S_L,S_R)\in \Cc}|S_L|^{\alpha}|S_R|^{\beta},\quad  M_\Gamma:=\max_{(S_L,S_R)\in \Cc_{\alpha,\beta}}\zeta(S_L,S_R,\Gamma)|S_R|^{\gamma},\quad
M_J:=Z_J(H').
\end{equation*}
Denote further by $c^*$ the cardinality of the set $\Cc^*$, i.e., $c^{*}:=|\Cc^*|$. We will prove the following: 
\begin{equation}\label{eq:connectionstar}
Z_{G}(H)= \big(1\pm \epsilon/10\big)c^{*}M^{Qn^3}_KM^{Qn^2}_\Gamma M_J^{\lceil n^{3/2}\rceil}Z_{G'}(H').
\end{equation}
Note that we can approximate $M_K,M_\Gamma$ to any desired accuracy (polynomial in $n$ and $\epsilon$) since $\Gamma$ is fixed graph and $\alpha,\beta,\gamma$ are constants which are efficiently approximable using \eqref{eq:defa0b0} and \eqref{eq:defgamma}. Since $J$ is a fixed graph, $M_J$ and $c^*$ can be computed exactly by brute force (which only takes constant time). Thus, if we can approximate $Z_{G}(H)$ within a multiplicative factor of $(1\pm \epsilon/10)$, equation \eqref{eq:connectionstar} yields an approximation of $Z_{G'}(H')$ within a multiplicative factor of $(1\pm \epsilon)$ (with room to spare).

It remains to argue for the validity of \eqref{eq:connectionstar}. Let $\Sigma$ be the set of all colour-preserving homomorphisms from $G$ to $H$. For subsets $S_L\subseteq V_L$ and $S_R\subseteq V_R$, let $\Sigma(S_L,S_R)$ be the following subset of $\Sigma$:
\begin{equation}\label{eq:defsigmaslslr}
\Sigma(S_L,S_R)=\{h\in \Sigma\mid h(L(K_{a,b}))=S_L,\, h(R(K_{a,b}))=S_R\}.
\end{equation}
Since $K_{a,b}$ is a complete bipartite graph, note that  $\Sigma(S_L,S_R)$ is non-empty iff $S_L\times S_R\subseteq E(H)$, i.e., $(S_L,S_R)$ is a biclique in $H$. It follows that the collection of sets $\Sigma(S_L,S_R)$ where $(S_L,S_R)$ ranges over all bicliques in $H$ is a partition of the set $\Sigma$, so 
\begin{equation}\label{eq:qqq2}
Z_{G}(H)=|\Sigma|=\sum_{(S_L,S_R)\in \Cc}|\Sigma(S_L,S_R)|.
\end{equation}
Fix $(S_L,S_R)\in \Cc$ and let $h\in \Sigma(S_L,S_R)$. Note that for every copy $\Gamma_r$ of $\Gamma$ we have that $h(L(\Gamma_r))\subseteq N_{\j}(S_R)$ since by construction $R(K_{a,b})\times L(\Gamma_r)\subseteq E(G)$. Using the assumption in the lemma that $\Gamma$ has no isolated vertices in $R(\Gamma)$, we thus obtain that for every copy $\Gamma_r$ of $\Gamma$ it holds that $h(R(\Gamma_r))\subseteq \Nu(N_{\j}(S_R))$. Analogously, since $G'$ is connected and $J$ has no isolated vertices in $R(J)$, for $G'$ and every copy $J_r$ of $J$, it holds that $h(L(G')),h(L(J_r))\subseteq N_{\j}(S_R)$ and $h(R(G')),h(R(J_r))\subseteq \Nu(N_{\j}(S_R))$. It follows that
\begin{equation}\label{eq:basbas2}
|\Sigma(S_L,S_R)|=T(a,|S_L|)\, T(b,|S_R|)\,\big(Z_{\Gamma}\big(H_{S_L,S_R}\big)\big)^{Qn^2}\big(Z_{J}(H_{S_L,S_R})\big)^{\lceil n^{3/2}\rceil}Z_{G'}\big(H_{S_L,S_R}\big),
\end{equation} 
where, recall, $T(m,q)$ is the number of surjective mappings from $[m]$ to $[q]$. Applying Lemma~\ref{lem:surj} (and using that $2|S_L|/a,2|S_R|/b\leq 2\max\{|V_L|,|V_R|\}/n\leq\epsilon/10^4$), we have that
\[T(a,|S_L|)=(1\pm \epsilon/10^4)|S_L|^a, \quad T(b,|S_R|)=(1\pm \epsilon/10^4)|S_R|^b,\] 
and thus
\[T(a,|S_L|)T(b,|S_R|)=(1\pm \epsilon/10^3)|S_L|^a|S_R|^b.\]
Observe also that $Z_{\Gamma}\big(H_{S_L,S_R}\big)=\zeta(S_L,S_R,\Gamma)$, cf. Definition~\ref{def:zeta}.  Thus, \eqref{eq:basbas2} yields 
\begin{align}
\begin{split}
|\Sigma(S_L,S_R)|={}&\big(1\pm \epsilon/10^3\big)|S_L|^{a} |S_R|^{b}\,\zeta(S_L,S_R,\Gamma)^{Qn^2}\\
&\qquad\qquad\qquad\qquad\qquad\big(Z_{J}(H_{S_L,S_R})\big)^{\lceil n^{3/2}\rceil}Z_{G'}\big(H_{S_L,S_R}\big)\notag
\end{split}\\
\begin{split}
={}&\big(1\pm \epsilon/10^2\big)\,\big(|S_L|^{\alpha } |S_R|^{\beta}\big)^{Qn^3}\big(\zeta(S_L,S_R,\Gamma)\, |S_R|^{\gamma}\big)^{Qn^2}\\
&\qquad\qquad\qquad\qquad\qquad\big(Z_{J}(H_{S_L,S_R})\big)^{\lceil n^{3/2}\rceil}Z_{G'}\big(H_{S_L,S_R}\big),
\end{split}\label{eq:basbas}
\end{align}
where for the second equality we  used the approximation in \eqref{eq:approxslsr2}. 

We next use \eqref{eq:basbas} to find the terms in \eqref{eq:qqq2} with significant contribution. Observe that for every biclique $(S_L,S_R)$ we have the  bounds
\[|S_L|^{\alpha } |S_R|^{\beta}\leq M_K,\quad \zeta(S_L,S_R,\Gamma)\,|S_R|^{\gamma}\leq Z_\Gamma(H)\, |V_R|^{\gamma},\quad Z_{J}(H_{S_L,S_R})\leq Z_{J}(H),\]
i.e., the bases of the powers in the r.h.s. of \eqref{eq:basbas} are bounded by absolute constants. Further, 
\[Z_{G'}\big(H_{S_L,S_R}\big)\leq Z_{G'}\big(H\big)\leq |V(H)|^{n'}\leq |V(H)|^{n}.\]
Using these observations, we will consider the factors in the r.h.s. of \eqref{eq:basbas} in decreasing order of magnitude (which can easily be read off from the exponents)  and  argue that only bicliques in $\Cc^*$ have significant contribution.

From the definition of $\Cc^*$ (see \eqref{eq:defcstar}), we have that  a biclique $(S_L,S_R)\in \Cc^*$ contributes $(1\pm \epsilon/10^2)W$ to the sum in \eqref{eq:qqq2}, where
\begin{align*}
W:=M^{Qn^3}_KM^{Qn^2}_\Gamma M_J^{\lceil n^{3/2}\rceil}Z_{G'}(H').
\end{align*} 
Thus, the total contribution from bicliques in $\Cc^*$ is $(1\pm \epsilon/10^2)c^* W$, that is,
\begin{equation}\label{eq:contcstar}
\sum_{(S_L,S_R)\in \Cc^{*}}|\Sigma(S_L,S_R)|=(1\pm \epsilon/10^2)c^* W.
\end{equation} 
From the definition of $\Cc_{\alpha,\beta}$, we have that   for $(S_L,S_R)\not\in \Cc_{\alpha,\beta}$, it holds that 
\[M_K>|S_L|^{\alpha}|S_R|^{\beta},\]
so that terms in \eqref{eq:qqq2} such that  $(S_L,S_R)\not\in \Cc_{\alpha,\beta}$ are smaller 
than~$W$ by a multiplicative factor of $\exp(-\Omega(n^3))$. 
From the definition of $\Cc_{\alpha,\beta}^{\Gamma}$, for $(S_L,S_R)\in \Cc_{\alpha,\beta}\backslash\Cc_{\alpha,\beta}^{\Gamma}$, it holds that
\[M_\Gamma>\zeta(S_L,S_R,\Gamma) |S_R|^{\gamma},\] 
so that terms in \eqref{eq:qqq2} such that  $(S_L,S_R)\in \Cc_{\alpha,\beta}\backslash\Cc_{\alpha,\beta}^{\Gamma}$ are smaller 
than~$W$
by a multiplicative factor of $\exp(-\Omega(n^2))$. Finally, from the definition of $\Cc^*$ (see \eqref{eq:defcstar}), for  $(S_L,S_R)\in \Cc_{\alpha,\beta}^{\Gamma}\backslash \Cc^*$ we have
\[M_J>Z_{J}(H_{S_L,S_R})\]
so that terms in \eqref{eq:qqq2} such that  $(S_L,S_R)\in \Cc_{\alpha,\beta}^{\Gamma}\backslash \Cc^*$ are smaller 
than~$W$ by a multiplicative factor of $\exp(-\Omega(n^{3/2}))$.  

Combining the above estimates, we have that  for every biclique  $(S_L,S_R)\in \Cc\backslash\Cc^*$, it holds that 
\[|\Sigma(S_L,S_R)|\leq \exp\big(-\Omega(n^{3/2})\big)W\leq \epsilon W/(10^2|\Cc|).\]
(Note that in the 
combined
inequality we assumed that the constant $\kappa$ in the definition of $n$ is sufficiently large.) It follows that the aggregate contribution from bicliques which are not in $\Cc^*$ to the sum \eqref{eq:qqq2} is at most $\epsilon W/10^2$, that is,
\begin{equation}\label{eq:contnotcstar}
\sum_{(S_L,S_R)\in \Cc\backslash\Cc^{*}}|\Sigma(S_L,S_R)|\leq \epsilon W/10^2.
\end{equation} 
Plugging the  bounds \eqref{eq:contcstar} and \eqref{eq:contnotcstar} in the sum \eqref{eq:qqq2} yields \eqref{eq:connectionstar}, as desired. This completes the  proof of Lemma~\ref{lem:slsr}.
\end{proof}

\begin{lemma}\label{lem:flvr}
Let $H$ be a 2-coloured graph which is full 
but not trivial.
Let $0<\alpha,\beta<1$ be as in \eqref{eq:defa0b0}. Suppose that $\Gamma$ is a 2-coloured graph such that $\Gamma$ has no isolated vertices in $R(\Gamma)$. Further, assume that 
\[\Cc^{\Gamma}_{\alpha,\beta}=\big\{(F_L,V_R),(V_L,F_R)\big\}.\] Then,  $\BIS\leq_\AP\FixCOL{H}$.
\end{lemma}
\begin{proof} 
Let $0<\epsilon<1$ and $G'$ be an input to $\BIS$. Our goal is to approximate $\BIS(G')$ within a multiplicative factor of $(1\pm \epsilon)$ using an oracle for $\FixCOL{H}$. 

Suppose that $n'=|V(G')|$ and let $n=\left\lceil  \kappa n'^2/\epsilon\right\rceil$ where $\kappa$ is a sufficiently large constant depending only on $H$ and $\Gamma$. Similarly to the proof of Lemma~\ref{lem:slsr}, we will use $\kappa$ to ensure that $n$ is sufficiently large both in absolute terms as well as relatively to $n'^2$ and $1/\epsilon$. We will assume that $\kappa$ is such that $\alpha n\geq 1$, $\beta n \geq 1$, $n\geq 2\gamma$, 
$2\max\{|V_L|,|V_R|\}/n\leq \epsilon/(10^4 n')$
and, later, we will further refine the choice of $\kappa$. 

Analogously to the proof of Lemma~\ref{lem:slsr}, we apply Lemma~\ref{lem:dirichlet} to the numbers $\alpha n^3$ and $n^2(\beta n+\gamma)$. This yields positive numbers $a,b,Q$ with $Q\leq n^2$ such that 
\[|Q \alpha n^3-a|\leq 1/n,\quad |Q n^2(\beta n+\gamma)-b|\leq 1/n,\]
and, once again, note that $n^2\leq a,b\leq 2n^{5}$. The choice of $a,b$ yields the following analogue of \eqref{eq:approxslsr2} for $(S_L,S_R)\in \Cc$:
\begin{equation}\label{eq:approxslsr3} 
\zeta(S_L,S_R,\Gamma)^{Qn^2}|S_L|^a|S_R|^b=\big(1\pm \epsilon/(10^3 n')\big)\big(|S_L|^{\alpha} |S_R|^{\beta}\big)^{Qn^3}\big(\zeta(S_L,S_R,\Gamma)\,  |S_R|^{\gamma}\big)^{Qn^2},
\end{equation} 
(Note that the r.h.s. in \eqref{eq:approxslsr3} differs from the respective one in \eqref{eq:approxslsr2} only in the error term, which can be accounted by the slightly different choice of $n$.)

To construct an instance $G$ of $\FixCOL{H}$, we first construct a bipartite graph $J$ as follows. Let $K_{a,b}$ be the complete bipartite graph with $a,b$ vertices on the left and right, respectively. Let $\overline{K}$ be the disjoint union of the graphs $K_{a,b}$ and $Qn^2$ disjoint copies of the graph $\Gamma$. We denote by $\Gamma_r$ the $r$-th copy of $\Gamma$. Finally, the graph $J$ is obtained by attaching on $\overline{K}$ for $r=1,\hdots,Qn^2$ all edges in $R(K_{a,b})\times L(\Gamma_r)$. 

Let $\Sigma_J$ be the set of all colour-preserving homomorphisms from $J$ to $H$. For subsets $S_L\subseteq V_L$ and $S_R\subseteq V_R$, let $\Sigma_J(S_L,S_R)$ be the following subset of $\Sigma_J$:
\[\Sigma_J(S_L,S_R)=\{h\in \Sigma_J\mid h(L(K_{a,b}))=S_L,\, h(R(K_{a,b}))=S_R\}.\]
Since $K_{a,b}$ is a complete bipartite graph, note that  $\Sigma_J(S_L,S_R)$ is non-empty iff $S_L\times S_R\subseteq E(H)$, i.e., $(S_L,S_R)$ is a biclique in $H$. It follows that the collection of sets $\Sigma_J(S_L,S_R)$ where $(S_L,S_R)$ ranges over all bicliques in $H$ is a partition of the set $\Sigma_J$.
Fix $(S_L,S_R)\in \Cc$ and let $h\in \Sigma(S_L,S_R)$. Note that for every copy $\Gamma_r$ of $\Gamma$ we have that $h(L(\Gamma_r))\subseteq N_{\j}(S_R)$ and thus $h(R(\Gamma_r))\subseteq \Nu(N_{\j}(S_R))$ (since by the assumptions in the lemma $\Gamma$ has no isolated vertices in $R(\Gamma)$).  
It follows that
\[|\Sigma_J(S_L,S_R)|=T(a,|S_L|)\, T(b,|S_R|)\,\big(\zeta(S_L,S_R,\Gamma)\big)^{Qn^2},\]
where, recall that for positive integers $m,q$, $T(m,q)$ is the number of surjective mappings from $[m]$ to $[q]$ and, from Definition~\ref{def:zeta}, $\zeta(S_L,S_R,\Gamma)=\FixCOL{H_{S_L,S_R}}(\Gamma)$. Applying Lemma~\ref{lem:surj} and the approximation in \eqref{eq:approxslsr3} yields that for every $(S_L,S_R)\in \Cc$, it holds that
\begin{equation}\label{eq:qwerty}
|\Sigma_J(S_L,S_R)|=\big(1\pm \epsilon/(10^2 n')\big)\big(|S_L|^{\alpha} |S_R|^{\beta}\big)^{Qn^3}\big(\zeta(S_L,S_R,\Gamma)\,  |S_R|^{\gamma}\big)^{Qn^2}.
\end{equation}
Let us also denote
\begin{equation*}
\begin{gathered}
M_K:=\max_{(S_L,S_R)\in \Cc}|S_L|^{\alpha}|S_R|^{\beta},\quad  M_\Gamma:=\max_{(S_L,S_R)\in \Cc_{\alpha,\beta}}\zeta(S_L,S_R,\Gamma)|S_R|^{\gamma}.
\end{gathered}
\end{equation*}
From \eqref{eq:qwerty} and the assumption that the extremal bicliques are elements of both $\Cc_{\alpha,\beta}$ and $\Cc^{\Gamma}_{\alpha,\beta}$, we obtain
\begin{equation}\label{eq:extext}
|\Sigma_J(F_L,V_R)|=(1\pm \epsilon/(10^2n')) 
M_K^{Qn^3}M_\Gamma^{Qn^2},\, |\Sigma_J(V_L,F_R)|= 
(1\pm \epsilon/(10^2n'))M_K^{Qn^3}M_\Gamma^{Qn^2}.
\end{equation}
Now consider a non-extremal biclique $(S_L,S_R)$, i.e.,  $(S_L,S_R)\neq (F_L,V_R),(V_L,F_R)$. From the definition of $\Cc_{\alpha,\beta}$ (cf. Definition~\ref{def:ccab}), if  $(S_L,S_R)\in \Cc\backslash\Cc_{\alpha,\beta}$, we have $M_K>|S_L|^{\alpha}|S_R|^{\beta}$ and thus
\[|\Sigma_J(S_L,S_R)|\leq \exp(-\Omega(n^3))M_K^{Qn^3}M_\Gamma^{Qn^2}.\] 
If $(S_L,S_R)\in \Cc_{\alpha,\beta}$ but $(S_L,S_R)\notin \Cc_{\alpha,\beta}^{\Gamma}$, from the definition of $\Cc^{\Gamma}_{\alpha,\beta}$ (cf. Definition~\ref{def:ccabgamma}) we  have $M_\Gamma>\zeta(S_L,S_R,\Gamma)|S_R|^{\gamma}$ so that
\[|\Sigma_J(S_L,S_R)|\leq \exp(-\Omega(n^2))M_K^{Qn^3}M_\Gamma^{Qn^2}.\] 
Combining the above bounds, it follows that for every non-extremal biclique $(S_L,S_R)\in \Cc$ it holds that $|\Sigma_J(S_L,S_R)|\leq \exp(-\Omega(n^2))M_K^{Qn^3}M_\Gamma^{Qn^2}$, i.e.,
\begin{equation}\label{eq:M3}
M_3\leq \exp(-\Omega(n^2))M_K^{Qn^3}M_\Gamma^{Qn^2},\mbox{ where }M_3=\max_{\substack{(S_L,S_R)\in \Cc;\\ (S_L,S_R)\neq (F_L,V_R),(V_L,F_R)}}|\Sigma_J(S_L,S_R)|.
\end{equation}

With the gadget $J$ at hand, we now give the reduction from $\BIS$ to $\FixCOL{H}$. First, we construct an instance $G$ of $\FixCOL{H}$ as follows. For every vertex $u$ in $G'$ take a disjoint copy of $J$, which we shall denote by $J^{u}$. We will further denote by $K^u_{a,b}$ the copy of $K_{a,b}$ in $J^u$.  For every edge $(u,v)\in E(G')$ with $u\in L(G')$ and $v\in R(G')$ add all edges between $R(K^{u}_{a,b})$ and $L(K^{v}_{a,b})$. Denote by $G$ the final graph. The construction ensures that $G$ is a bipartite graph (the 2-colouring of $G$ is naturally induced by the 2-colourings of the disjoint copies of $J$).

We will show that
\begin{equation}\label{eq:connection5}
\FixCOL{H}(G)= \big(1\pm \epsilon/10\big)\,\big(M_K^{Qn^3}M_\Gamma^{Qn^2 }\big)^{n'}\BIS(G').
\end{equation}
Note that $M_K,M_\Gamma$ can be computed to any desired polynomial precision (in $n,1/\epsilon$), so if we can approximate $\FixCOL{H}(G)$ within a multiplicative factor of $1\pm \epsilon/10$, equation \eqref{eq:connection5} yields an approximation of $\BIS(G')$ within a multiplicative factor of $1\pm \epsilon$. 

We next argue for the validity of \eqref{eq:connection5}. Let $h$ be a colour-preserving homomorphism from $G$ to $H$. For a vertex $u$ in $G'$, we call $(S_L,S_R)$ the \emph{phase} of $u$ under $h$ iff the restriction of $h$ to the graph $J^u$ is a homomorphism in $\Sigma_J(S_L,S_R)$.   We will denote by $h^u$ the   \emph{phase} of $u$ under $h$. Finally, we let $\mathcal{Y}(h)$ to be the vector $\{h^u\}_{u\in V(G')}$, i.e., the collection of phases under $h$.

Let $\Sigma$ be the set of colour-preserving homomorphisms from $G$ to $H$ and for a given phase vector $Y\in \Cc^{n'}$, let $\Sigma(Y)$ be the set of homomorphisms $h$  such that $\mathcal{Y}(h)=Y$, i.e.,
\[\Sigma(Y)=\{h\in \Sigma\mid \mathcal{Y}(h)=Y\}.\] 
Note that the $\Sigma(Y)$ form a partition of $\Sigma$, so we have that
\[|\Sigma|=\sum_{Y\in \Cc^{n'}}|\Sigma(Y)|.\]
Note also the following upper bound on $|\Sigma(Y)|$:
\begin{equation}\label{eq:upperbound}
|\Sigma(Y)|\leq \prod_{u\in V(G')} |\Sigma_J(Y_u)|.
\end{equation}
We call a phase vector $Y\in \mathcal{\Cc}^{n'}$ \emph{bad} if  there exists a vertex $u\in G'$ whose phase $Y_u$ is not an extremal biclique. Otherwise, $Y$ will be called \emph{good}. Using the bound \eqref{eq:upperbound} in combination with \eqref{eq:extext} and \eqref{eq:M3} (which hold for every copy of the gadget $J$), for a bad phase vector $Y$ we have that 
\begin{equation}\label{eq:badphasevector}
|\Sigma(Y)|\leq  
\left(1\pm \frac{\epsilon}{10^2n'}\right)^{n'-1}
\big(M_K^{Qn^3}M_\Gamma^{Qn^2 }\big)^{n'-1}M_3\leq 2\big(M_K^{Qn^3}M_\Gamma^{Qn^2 }\big)^{n'-1}M_3.
\end{equation}
Let $\Sigma'$ be the set of colour-preserving homomorphisms from $G$ to $H$ whose phase vector is bad. Since there are less than $|\Cc|^{n'}$ choices for the bad phase vector $Y$, using \eqref{eq:badphasevector}, we have the crude bound
\begin{equation}\label{eq:verycrude3}
|\Sigma'|\leq 2 |\Cc|^{n'}\big(M_K^{Qn^3}M_\Gamma^{Qn^2}\big)^{n'-1}M_3<\exp(-\Omega(n^2))\big(M_K^{Qn^3}M_\Gamma^{Qn^2}\big)^{n'}.
\end{equation}
 Note that in the latter inequality we used that $M_3$ is smaller from  $M_K^{Qn^3}M_\Gamma^{Qn^2}$ by a factor of $\exp(-\Omega(n^2))$, cf. \eqref{eq:M3}. By letting the constant $\kappa$ to be sufficiently large, \eqref{eq:verycrude3}  yields that 
\begin{equation}\label{eq:verycrude2}
|\Sigma'|\leq (\epsilon/20)\big(M_K^{Qn^3}M_\Gamma^{Qn^2}\big)^{n'}.
\end{equation} 

Let $Y\in \mathcal{\Cc}^{n'}$ be a good phase vector. We call $Y$ \emph{permissible} if for every edge $(u,v)\in E(G)$ we have either $Y_u\neq (F_L,V_R)$ or $Y_v\neq (V_L,F_R)$. Note that if $Y$ is not permissible we have that $|\Sigma(Y)|=0$ (since $H$ is not trivial). If $Y$ is permissible, we have from \eqref{eq:extext} applied to every copy of $J$ that 
\begin{equation}
|\Sigma(Y)|= 
\left(1\pm \frac{\epsilon}{10^2n'}\right)^{n'}
\big(M_K^{Qn^3}M_\Gamma^{Qn^2 }\big)^{n'}=\left(1\pm \frac{\epsilon}{20}\right)\big(M_K^{Qn^3}M_\Gamma^{Qn^2 }\big)^{n'}.
\end{equation}
Now observe that  permissible phase vectors $Y$ are in one-to-one correspondence with independent sets of $G'$: a vertex $u\in L(G')$ is in the independent set if $Y_u=(F_L,V_R)$; a vertex $v\in R(G')$ is in the independent set if $Y_v=(V_L,F_R)$. It follows that the total number of homomorphisms with permissible phase vectors is  
\begin{equation}\label{eq:permissible}
(1\pm \epsilon/20)\big(M_K^{Qn^3}M_\Gamma^{Qn^2 }\big)^{n'}\BIS(G').
\end{equation}
 Since $\BIS(G')\geq 1$, combining \eqref{eq:verycrude2} and \eqref{eq:permissible} yields \eqref{eq:connection5}, as wanted. This completes the proof of Lemma~\ref{lem:flvr}.
\end{proof}

We are now ready  to prove Lemma~\ref{lem:main}.  

\begin{lemmain}
\statelemmain{}
\end{lemmain}

\begin{proof} 
We will prove  the lemma using induction on the number $|V(H)|$ of vertices of $H$.
Since $H$ is full but not trivial, the base case is the case in which $H$ is a path on four vertices.
In this case, 
it is known that $\BIS\leq_{\AP}\FixCOL{H}$ (see for example \cite{DGGJ}).  

We now do the inductive case. Roughly, we will show that we can choose $\Gamma$ so that either the assumptions of Lemma~\ref{lem:slsr} or Lemma~\ref{lem:flvr} hold. 

Consider the set of maximal bicliques $\Cc_{\alpha,\beta}$ where $0<\alpha,\beta<1$ are as in \eqref{eq:defa0b0}. For the rest of the proof we may assume that $\Cc_{\alpha,\beta}$ includes both extremal bicliques and at least one non-extremal biclique. Otherwise, observe that if an extremal biclique is not in $\Cc_{\alpha,\beta}$, then $\Cc_{\alpha,\beta}$ does not in fact include either of the extremal bicliques (this is enforced by \eqref{eq:defa0b0}, i.e., the selection of $\alpha,\beta$). It then follows by Lemma~\ref{lem:slsr} (where $\Gamma$ is the empty graph) that there is a full 
but not trivial
graph~$H'$ such that $|V(H')|<|V(H)|$ and $\FixCOL{H'}\leq_{\AP}\FixCOL{H}$. By induction, we have that $\BIS\leq_{\AP}\FixCOL{H'}$, so we obtain that $\BIS\leq_{\AP}\FixCOL{H}$. Similarly, if $\Cc_{\alpha,\beta}$ includes both extremal bicliques but no other biclique, by Lemma~\ref{lem:flvr} (where $\Gamma$ is the empty graph),  we have that $\BIS\leq_{\AP}\FixCOL{H}$.

Let 
\[\big(S^{(1)}_L,S^{(1)}_R\big),\hdots, \big(S^{(t)}_L,S^{(t)}_R\big)\]
be an enumeration of the non-extremal bicliques in $\Cc_{\alpha,\beta}$. Recall that all elements of $\Cc_{\alpha,\beta}$ are maximal bicliques of $H$.    Recall the definition of 
$\zeta^{\ex}_1(\Gamma)$, and $\zeta^{\ex}_2(\Gamma)$
from Definition~\ref{def:zetaextremal}.
We consider three cases, 
which correspond (in order) to the 3 cases in the overview.

\noindent\textbf{Case I.} There exists $i\in[t]$ and a 2-coloured graph $\Gamma$ such that
\begin{equation}\label{eq:gammaabtwo}
\ln \frac{|V_R|}{|F_R|}\ln\frac{\zeta\big(S_L^{(i)},S_R^{(i)},\Gamma\big)}{\zeta^{\ex}_1(\Gamma)}>\ln\frac{|V_R|}{\big|S_R^{(i)}\big|}\ln\frac{\zeta^{\ex}_2(\Gamma)}{\zeta^{\ex}_1(\Gamma)}.
\end{equation}
We claim that the extremal bicliques are not in $\Cc^{\Gamma}_{\alpha,\beta}$. To see this, from \eqref{eq:defgamma}, we have that $\gamma(\Gamma)=\ln\frac{\zeta^{\ex}_2(\Gamma)}{\zeta^{\ex}_1(\Gamma)}/\ln\frac{|V_R|}{|F_R|}$, so \eqref{eq:gammaabtwo} yields that $\ln\frac{\zeta\big(S_L^{(i)},S_R^{(i)},\Gamma\big)}{\zeta^{\ex}_1(\Gamma)}/\ln\frac{|V_R|}{\big|S_R^{(i)}\big|}>\gamma(\Gamma)$ which can be massaged into
\begin{equation}\label{eq:slsri}
\zeta(S_L^{(i)},S_R^{(i)},\Gamma)\, |S_R^{(i)}|^{\gamma(\Gamma)}>\zeta^{\ex}_1(\Gamma) |V_R|^{\gamma(\Gamma)}=\zeta^{\ex}_2(\Gamma)\, |F_R|^{\gamma(\Gamma)},
\end{equation}
where in the latter equality we used the definition of $\gamma(\Gamma)$, see \eqref{eq:defgamma}. From \eqref{eq:slsri}, we obtain that the extremal bicliques are not in $\Cc^{\Gamma}_{\alpha,\beta}$, as wanted. 

To apply Lemma~\ref{lem:slsr} using the graph $\Gamma$, we need to ensure that $\Gamma$ has no isolated vertices in $R(\Gamma)$ (see also footnote~\ref{foot:technical}). Let $I$ be the set of isolated vertices in $R(\Gamma)$ and let $\Gamma'$ be the 2-coloured graph obtained by removing the set $I$ from $\Gamma$ (and which otherwise inherits the 2-colouring of $\Gamma$), so that $\Gamma'$ has no isolated vertices in $R(\Gamma')$. From Definition~\ref{def:zeta}, we have that $\zeta(S_L^{(i)},S_R^{(i)},\Gamma)=\FixCOL{H_{S_L^{(i)},S_R^{(i)}}}(\Gamma)$,  where $H_{S_L^{(i)},S_R^{(i)}}=H[S_L^{(i)}\cup V_R]$ since $(S_L^{(i)},S_R^{(i)})$ is a maximal biclique, see also Definition~\ref{def:specialgraph} for details. The observation which is important for us is that $R(H_{S_L^{(i)},S_R^{(i)}})=V_R$. We thus have
\[\zeta(S_L^{(i)},S_R^{(i)},\Gamma)=|V_R|^{|I|}\zeta(S_L^{(i)},S_R^{(i)},\Gamma'),\]
and analogously (observing that $R(H_{S_L,S_R})=V_R$ for every maximal biclique and hence also for the extremal ones)
\[\zeta^{\ex}_1(\Gamma)=|V_R|^{|I|}\,\zeta^{\ex}_1(\Gamma'),\quad \zeta^{\ex}_2(\Gamma)=|V_R|^{|I|}\,\zeta^{\ex}_2(\Gamma').\]
It follows that \eqref{eq:gammaabtwo} also holds if we replace $\Gamma$ with $\Gamma'$. 

From the argument above, we may assume w.l.o.g. that $\Gamma$ has no isolated vertices. From Lemma~\ref{lem:slsr}, it then follows that there is a full 
but not trivial
graph~$H'$ such that $|V(H')|<|V(H)|$ and $\FixCOL{H'}\leq_{\AP}\FixCOL{H}$. By induction, we have that $\BIS\leq_{\AP}\FixCOL{H'}$, so we obtain that $\BIS\leq_{\AP}\FixCOL{H}$. This completes the argument in the first case. \vskip 0.15cm

\noindent\textbf{Case II.} There exists $i\in [t]$ such that for all 2-coloured  graphs $\Gamma$ it holds that 
\begin{equation}\label{eq:gammaab2}
\ln \frac{|V_R|}{|F_R|}\ln\frac{\zeta\big(S_L^{(i)},S_R^{(i)},\Gamma\big)}{\zeta^{\ex}_1(\Gamma)}=\ln\frac{|V_R|}{\big|S_R^{(i)}\big|}\ln\frac{\zeta^{\ex}_2(\Gamma)}{\zeta^{\ex}_1(\Gamma)}.
\end{equation}
We will show that $\FixCOL{H'}\leq_{\AP}\FixCOL{H}$ where $H'=H_{S_L^{(i)},S_R^{(i)}}=H[S_L^{(i)}\cup V_R]$. To see this, set $c_1:=\ln \frac{|V_R|}{|F_R|}$, $c_2:=\ln\frac{|V_R|}{|S_R^{(i)}|}$ and  $c:=c_2/c_1$ (note that $0<c<1$). Then, rewriting \eqref{eq:gammaab2}, for every 2-coloured graph $\Gamma$ it holds that
\begin{equation}\label{eq:equalitycase}
\zeta(S_L^{(i)},S_R^{(i)},\Gamma)=\frac{\big(\zeta^{\ex}_2(\Gamma)\big)^c}{(\zeta^{\ex}_1(\Gamma))^{c-1}},\mbox{ so that }\FixCOL{H'}(\Gamma)=\frac{\big(\FixCOL{H}(\Gamma)\big)^c}{|F_L|^{(c-1)|L(\Gamma)|}|V_R|^{(c-1)|R(\Gamma)|}},
\end{equation}
cf. Definitions~\ref{def:specialgraph},~\ref{def:zeta} and~\ref{def:zetaextremal}. Equation \eqref{eq:equalitycase}  trivially implies that $\FixCOL{H'}\leq_{\AP}\FixCOL{H}$. By Lemma~\ref{inductionstep}, we have that $H'$ is  full 
but not trivial
and further satisfies $|V(H')|<|V(H)|$, so by induction we have that $\BIS\leq_{\AP}\FixCOL{H'}$. It follows that $\BIS\leq_{\AP}\FixCOL{H}$.\vskip 0.15cm 

\noindent\textbf{Case III.} For all $i\in [t]$, for all 2-coloured  graphs $\Gamma$, it holds that 
\begin{equation}\label{eq:gammaab}
\ln \frac{|V_R|}{|F_R|}\ln\frac{\zeta\big(S_L^{(i)},S_R^{(i)},\Gamma\big)}{\zeta^{\ex}_1(\Gamma)}\leq\ln\frac{|V_R|}{\big|S_R^{(i)}\big|}\ln\frac{\zeta^{\ex}_2(\Gamma)}{\zeta^{\ex}_1(\Gamma)}.
\end{equation}
and, for all $i\in[t]$, the inequality in \eqref{eq:gammaab} is strict for at least one graph $\Gamma_i$. Using an identical argument to the one in Case I, we may assume that $\Gamma_i$ has no isolated vertices in $R(\Gamma_i)$.

Let $\Gamma^*$ be the disjoint union of the graphs $\Gamma_i$ for $i=1,\hdots,t$ (the 2-colouring of $\Gamma^*$ is  induced in the natural way by the 2-colourings of its components).  For every $i\in [t]$, we have 
\begin{equation}\label{eq:derive1}
\zeta\big(S_L^{(i)},S_R^{(i)},\Gamma^*\big)=\prod^{t}_{j=1}\zeta\big(S_L^{(i)},S_R^{(i)},\Gamma_j\big), \quad \zeta^{\ex}_1(\Gamma^*)=\prod^t_{j=1} \zeta^{\ex}_1\big(\Gamma_j\big), \quad \zeta^{\ex}_2(\Gamma^*)=\prod^t_{j=1} \zeta^{\ex}_2\big(\Gamma_j\big).
\end{equation}
Further, for every $i\in[t]$, from \eqref{eq:gammaab} and the choice of  $\Gamma_i$, we have
\begin{equation}\label{eq:derive2}
\ln \frac{|V_R|}{|F_R|} \sum_{j=1}^t \ln \frac{\zeta\big(S_L^{(i)},S_R^{(i)},\Gamma_j\big)}{\zeta_1^{\ex}(\Gamma_j)}
  < \ln \frac{|V_R|}{|S_R^{(i)}|} \sum_{j=1}^t \ln \frac{\zeta_2^{\ex}(\Gamma_j)}{\zeta_1^{\ex}(\Gamma_j)}.
\end{equation} 
Combining \eqref{eq:derive1} and \eqref{eq:derive2} we obtain
\begin{equation}\label{eq:gammastar}
\ln \frac{|V_R|}{|F_R|}\ln\frac{\zeta\big(S_L^{(i)},S_R^{(i)},\Gamma^*\big)}{\zeta^{\ex}_1(\Gamma^*)}<\ln\frac{|V_R|}{\big|S_R^{(i)}\big|}\ln\frac{\zeta^{\ex}_2(\Gamma^*)}{\zeta^{\ex}_1(\Gamma^*)}.
\end{equation}
By essentially the same calculation as in case I, inequality \eqref{eq:gammastar} implies that $\Cc^{\Gamma^*}_{\alpha,\beta}$ contains only the extremal bicliques; to see this note that for every $i\in[t]$, we have that the inequality \eqref{eq:gammaabtwo} is reversed and thus so is the inequality in \eqref{eq:slsri}. 

Since $\Cc^{\Gamma^*}_{\alpha,\beta}$ contains only the extremal bicliques and $\Gamma^*$ has no isolated vertices in $R(\Gamma^*)$ (by construction), Lemma~\ref{lem:flvr} yields $\BIS\leq_\AP\FixCOL{H}$, which finishes the proof of Lemma~\ref{lem:main}. 
\end{proof}

\section{Proof of Lemma~\ref{lem:fixingHCOL}}\label{sec:fixingHCOL}
We follow as closely as possible \cite[Section 6]{HCOL}. The reduction is similar to the one in \cite[Proof of Lemma 6]{HCOL}, yet to account for the counting setting (rather than the sampling setting in \cite{HCOL}), we will have to deviate at a certain point.

Let $H$ be a graph with no trivial components (note that in this section $H$ is allowed to have self-loops but not multiple edges). Following \cite{HCOL}, let $\Delta_1$ be the maximum degree 
\footnote{Recall that a self-loop $(u,u)$ contributes one to the degree of the vertex $u$.}
 of $H$ and let $\Delta_2$ be the maximum degree amongst \emph{the neighbours of maximum degree vertices}. Formally,
\begin{equation}\label{eq:delta12}
\Delta_1=\max_{u\in V(H)}\mathrm{deg}_H(u), \quad \Delta_2=\max_{\substack{v\in V(H);\\ v\in \Nu(u),\, \mathrm{deg}(u)=\Delta_1}}\mathrm{deg}_H(v).
\end{equation}

As in \cite{HCOL}, we also consider the bipartite double cover of $H$, which will be denoted by $\Bip(H)$. Formally, $\Bip(H)$ is a  bipartite graph with vertex set $V(H)\times \{1,2\}$ (the 2-colouring of $\Bip(H)$ is given by the natural projection of the set $V(H)\times \{1,2\}$ onto its second component). For convenience, for $u\in V(H)$, we will denote by $u_1,u_2$ the two copies of the vertex $u$ in $V(\Bip(H))$. The edge set of $\Bip(H)$ is then given by
\[E(\Bip(H))=\big\{ (u_1, v_2\big)\, | \, (u,v)\in E(H)\} \cup \big\{ (v_1, u_2)\, | \, (u,v)\in E(H)\big\}.\]
In particular, note that for every edge $(u,v)$ in $H$ which is not a self-loop, there are two edges $(u_1, v_2\big)$ and $(v_1, u_2\big)$ in $\Bip(H)$, while for every self-loop $(u,u)$ in $H$, there is a (single) edge $(u_1,u_2)$ in $\Bip(H)$. Note that  $\Bip(H)$ is a bipartite graph without self-loops. Also, since $H$ has no trivial components, it is simple to see that $\Bip(H)$ has no trivial components as well.

For an edge $(u,v)\in E(H)$, let $H_{u,v}$ be the subgraph of $\Bip(H)$ induced by the vertex set $\Nu'(u_1)\cup \Nu'(v_2)$ where $\Nu'(u_1),\Nu'(v_2)$ are the neighbourhoods of $u_1,v_2$, respectively, in $\Bip(H)$. Note that the graph $H_{u,v}$ is bipartite with a 2-colouring which is naturally induced by the 2-colouring of $\Bip(H)$, i.e., $L({H_{u,v}})=\Nu'(v_2)$ and $R({H_{u,v}})=\Nu'(u_1)$. Further, $H_{u,v}$ is full since $u_1$ is connected to $\Nu'(u_1)$ and $v_2$ is connected to $\Nu'(v_2)$.

Let 
\[\R(H):=\{(u,v)\, | \, (u,v)\in E(H), \, \mathrm{deg}_H(u)=\Delta_1,\, \mathrm{deg}_H(v)=\Delta_2\}.\]
The following lemma is proved in \cite{HCOL}.
\begin{lemma}{\cite[Lemma 6]{HCOL}}
Let $H$ be a graph with no trivial components and let $(u,v)\in \R(H)$. Then,  $H_{u,v}$ is a full 
but not trivial
graph.
\end{lemma}

With this setup, we are now ready to prove Lemma~\ref{lem:fixingHCOL}. 

\begin{lemfixingHCOL}
\statelemfixingHCOL{}
\end{lemfixingHCOL}

\begin{proof} 
Let $\Hc$ be the multiset consisting of the (labelled) graphs $H_{u,v}$ when $(u,v)$ ranges over elements of $\R(H)$. (Note that there might be multiple elements of $\Lambda(H)$ which map to the same graph.) Further, let $\Hc_1,\hdots,\Hc_k$ be the equivalence classes of $\Hc$ under the (colour-preserving isomorphism) relation $\cong_c$ which was formally defined in Section~\ref{sec:preliminaries}. For $i\in[k]$, let $H_i$ be a representantive of the class $\Hc_i$. By Lemma~\ref{lem:lovasztwo}, there exists a gadget $J$ and $i\in [k]$ such that for every $t\neq i$, it holds that $\FixCOL{H_i}(J)>\FixCOL{H_t}(J)$. We will show that $H'=H_i$ satisfies the statement of the lemma. For future use, let $\R^{*}(H)\subseteq \R(H)$ be the subset of elements  $(u,v)$ whose corresponding graph $H_{u,v}$ is isomorphic to $H'$ under a colour-preserving isomorphism, that is,
\begin{equation*}
\R^{*}(H):=\big\{(u,v)\in\R(H)\, |\, H_{u,v}\cong_c H'\big\}.
\end{equation*}

Let $0<\epsilon<1$ and $G'$ be an input to $\FixCOL{H'}$. Our goal is to approximate $\FixCOL{H'}(G')$ within a multiplicative factor of $(1\pm \epsilon)$ using an oracle for $\COL{H}$.  Suppose that $n'=|V(G')|$ and let $n=\left\lceil  \kappa n'/\epsilon\right\rceil$ where $\kappa$ is a sufficiently large constant depending only on $H$. 

To construct an instance $G$ of $\COL{H}$, consider first the (independent) sets of vertices $A$ and $B$  with 
\[|A|=n^4, |B|=n^3.\]
Let $\overline{G}$ be the disjoint union of $G'$, $n^2$ copies of the graph $J$, the independent sets $A,B$ and two ``special" vertices $w_A$ and $w_B$. We denote by $J_r$ the $r$-th copy of $J$. To obtain $G$, attach on $\overline{G}$ the edge $(w_A,w_B)$, all edges in $w_A \times A$, $w_B\times B$, $w_A\times R(G')$, $w_B\times L(G')$ and for $r=1,\hdots,n^2$ all edges in $w_A\times R(J_r)$, $w_B\times L(J_r)$. See Figure~\ref{fig:reductionlemma3}.

\begin{figure}[h]
\centering
\scalebox{0.8}[0.8]{\input{Hexample5.tex}}
\caption{The graph $G$}\label{fig:reductionlemma3}
\end{figure}
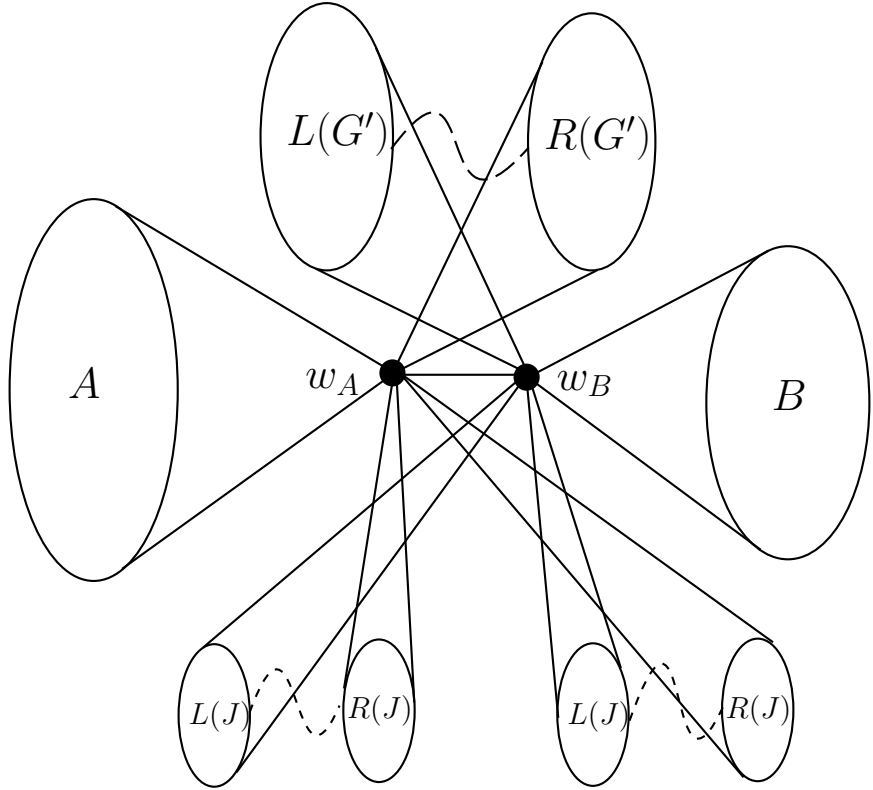

Let $\Sigma$ be the set of homomorphisms from $G$ to $H$. Note that $|\Sigma|=\COL{H}(G)$. We will show that
\begin{equation}\label{eq:approxapprox}
|\Sigma|=(1\pm \epsilon/10)\r^*\Delta_1^{n^4}\Delta_2^{n^3}(\FixCOL{H'}(J))^{n^2}\FixCOL{H'}(G'),
\end{equation}
where $\r^*:=|\R^{*}(H)|$ and, recall, $\Delta_1$, $\Delta_2$ are as in \eqref{eq:delta12}. Note that $\Delta_1,\Delta_2,\r^*$ are constants and $J$ is a fixed graph. It follows that an approximation of $|\Sigma|$ within a multiplicative factor of $1\pm \epsilon/2$ yields an approximation of $\FixCOL{H'}(G')$ within a multiplicative factor of $1\pm \epsilon$. 

Thus, it remains to show \eqref{eq:approxapprox}. Let $h$ be an arbitrary homomorphism in $\Sigma$. Since  $(w_A,w_B)\in E(G)$, we have that $(h(w_A),h(w_B))\in E(H)$. For vertices $u,v\in V(H)$, let $\Sigma(u,v)$ be the subset of $\Sigma$ such that $h(w_A)=u$ and $h(w_B)=v$. Formally,
\[\Sigma(u,v):=\{h\in \Sigma\, | \, h(w_A)=u, \, h(w_B)=v\}.\]
We have
\begin{equation}\label{eq:easysum}
|\Sigma|=\sum_{\substack{u,v\in V(H);\\ (u,v)\in E(H)}} |\Sigma(u,v)|.
\end{equation}
Let $h\in \Sigma(u,v)$. Then, the edges $w_A\times A, w_A\times R(G')$ enforce that $h(A), h(R(G'))\subseteq \Nu(u)$ (the neighbourhood of $u$ in $H$) and similarly $h(B), h(L(G'))\subseteq \Nu(v)$. Also, $h(R(J_r))\subseteq \Nu(u)$ and $h(L(J_r))\subseteq \Nu(v)$. Now observe that a homomorphism from  the bipartite graph $G$ to $H$ corresponds to a colour-preserving homomorphism from $G$ to $\Bip(H)$. Combining with the previous remarks, it is not hard to see  that 
\[|\Sigma(u,v)|=(\mathrm{deg}_H(u))^{|A|}(\mathrm{deg}_H(v))^{|B|} (\FixCOL{H_{u,v}}(J))^{n^2}\FixCOL{H_{u,v}}(G').\]
Note that for $(u,v)\in \R^*(H)$, we have 
\begin{equation}\label{eq:dominant}
|\Sigma(u,v)|=W,\mbox{ where } W:=\Delta_1^{n^4}\Delta_2^{n^3}(\FixCOL{H'}(J))^{n^2}\FixCOL{H'}(G').
\end{equation}
For the following estimates, observe that for every $(u,v)\in E(H)$, it holds that
\[\FixCOL{H_{u,v}}(G')\leq (2|V(H)|)^{n'}\leq (2|V(H)|)^{n}.\]
For $u\in H$ such that $\mathrm{deg}_H(u)<\Delta_1$, we have for all $v\in H$ that
\begin{equation}\label{eq:bound1}
|\Sigma(u,v)|\leq \exp(-\Omega(n^4)) W.
\end{equation}
For $u,v$ such that $\mathrm{deg}_H(u)=\Delta_1$ but $\mathrm{deg}_H(v)<\Delta_2$, we have that
\begin{equation}\label{eq:bound2}
|\Sigma(u,v)|\leq \exp(-\Omega(n^3)) W.
\end{equation} 
For $u,v$ such that $\mathrm{deg}_H(u)=\Delta_1$, $\mathrm{deg}_H(u)=\Delta_2$ but $(u,v)\notin \R^*(H)$, by the choice of $J$ we have $\FixCOL{H_{u,v}}(J)<\FixCOL{H'}(J)$ and hence
\begin{equation}\label{eq:bound3}
|\Sigma(u,v)|\leq \exp(-\Omega(n^2)) W.
\end{equation} 
Combining the estimates \eqref{eq:bound1}, \eqref{eq:bound2}, \eqref{eq:bound3} and letting $\kappa$ to be a sufficiently large constant, we have that for $(u,v)\notin \R^*(H)$ it holds that 
\begin{equation}\label{eq:negl}
|\Sigma(u,v)|\leq \epsilon W/(10|E(H)|).
\end{equation} 
Plugging \eqref{eq:dominant} and \eqref{eq:negl} back into \eqref{eq:easysum} yields \eqref{eq:approxapprox}, as wanted. This completes the proof of Lemma~\ref{lem:fixingHCOL}.
\end{proof}

\end{document}

%% file: Hexample3.tex
% C:\Users\agalanis\Dropbox\H-colorings\ICALP\H-colorings\Hexample3.jdr
% Created by Jpgfdraw version 0.5.6b
% 05-Feb-2015 15:30:14
\begin{pgfpicture}{0bp}{0bp}{293.200012bp}{210.999985bp}
\begin{pgfscope}
\pgfsetlinewidth{1.0bp}
\pgfsetrectcap 
\pgfsetmiterjoin \pgfsetmiterlimit{10.0}
\pgfpathmoveto{\pgfpoint{292.700012bp}{105.499988bp}}
\pgfpathcurveto{\pgfpoint{292.700012bp}{47.510089bp}}{\pgfpoint{270.403778bp}{0.499988bp}}{\pgfpoint{242.899994bp}{0.499988bp}}
\pgfpathcurveto{\pgfpoint{215.39624bp}{0.499988bp}}{\pgfpoint{193.100006bp}{47.510089bp}}{\pgfpoint{193.100006bp}{105.499988bp}}
\pgfpathcurveto{\pgfpoint{193.100006bp}{163.489886bp}}{\pgfpoint{215.39624bp}{210.499988bp}}{\pgfpoint{242.899994bp}{210.499988bp}}
\pgfpathcurveto{\pgfpoint{270.403778bp}{210.499988bp}}{\pgfpoint{292.700012bp}{163.489886bp}}{\pgfpoint{292.700012bp}{105.499988bp}}
\pgfclosepath
\color[rgb]{0.0,0.0,0.0}
\pgfusepath{stroke}
\end{pgfscope}
\begin{pgfscope}
\pgfsetlinewidth{1.0bp}
\pgfsetrectcap 
\pgfsetmiterjoin \pgfsetmiterlimit{10.0}
\pgfpathmoveto{\pgfpoint{143.692912bp}{98.763061bp}}
\pgfpathcurveto{\pgfpoint{143.692912bp}{95.507915bp}}{\pgfpoint{141.161805bp}{92.869096bp}}{\pgfpoint{138.039521bp}{92.869096bp}}
\pgfpathcurveto{\pgfpoint{134.917236bp}{92.869096bp}}{\pgfpoint{132.386123bp}{95.507915bp}}{\pgfpoint{132.386123bp}{98.763061bp}}
\pgfpathcurveto{\pgfpoint{132.386123bp}{102.018209bp}}{\pgfpoint{134.917236bp}{104.657027bp}}{\pgfpoint{138.039521bp}{104.657027bp}}
\pgfpathcurveto{\pgfpoint{141.161805bp}{104.657027bp}}{\pgfpoint{143.692912bp}{102.018209bp}}{\pgfpoint{143.692912bp}{98.763061bp}}
\pgfclosepath
\color[rgb]{0.0,0.0,0.0}\pgfseteorule\pgfusepath{fill}
\pgfpathmoveto{\pgfpoint{143.692912bp}{98.763061bp}}
\pgfpathcurveto{\pgfpoint{143.692912bp}{95.507915bp}}{\pgfpoint{141.161805bp}{92.869096bp}}{\pgfpoint{138.039521bp}{92.869096bp}}
\pgfpathcurveto{\pgfpoint{134.917236bp}{92.869096bp}}{\pgfpoint{132.386123bp}{95.507915bp}}{\pgfpoint{132.386123bp}{98.763061bp}}
\pgfpathcurveto{\pgfpoint{132.386123bp}{102.018209bp}}{\pgfpoint{134.917236bp}{104.657027bp}}{\pgfpoint{138.039521bp}{104.657027bp}}
\pgfpathcurveto{\pgfpoint{141.161805bp}{104.657027bp}}{\pgfpoint{143.692912bp}{102.018209bp}}{\pgfpoint{143.692912bp}{98.763061bp}}
\pgfclosepath
\color[rgb]{0.0,0.0,0.0}
\pgfusepath{stroke}
\end{pgfscope}
\begin{pgfscope}
\pgftransformcm{1.0}{0.0}{0.0}{1.0}{\pgfpoint{131.100006bp}{78.699988bp}}
\pgftext[left,base]{\sffamily\mdseries\upshape\huge
\color[rgb]{0.0,0.0,0.0}$w$}
\end{pgfscope}
\begin{pgfscope}
\pgfsetlinewidth{1.0bp}
\pgfsetrectcap 
\pgfsetmiterjoin \pgfsetmiterlimit{10.0}
\pgfpathmoveto{\pgfpoint{137.300006bp}{98.899988bp}}
\pgfpathlineto{\pgfpoint{23.300006bp}{138.499988bp}}
\color[rgb]{0.0,0.0,0.0}
\pgfusepath{stroke}
\end{pgfscope}
\begin{pgfscope}
\pgfsetlinewidth{1.0bp}
\pgfsetrectcap 
\pgfsetmiterjoin \pgfsetmiterlimit{10.0}
\pgfpathmoveto{\pgfpoint{135.500006bp}{97.699988bp}}
\pgfpathlineto{\pgfpoint{34.100006bp}{37.699988bp}}
\color[rgb]{0.0,0.0,0.0}
\pgfusepath{stroke}
\end{pgfscope}
\begin{pgfscope}
\pgfsetlinewidth{1.0bp}
\pgfsetrectcap 
\pgfsetmiterjoin \pgfsetmiterlimit{10.0}
\pgfpathmoveto{\pgfpoint{137.300006bp}{98.299988bp}}
\pgfpathlineto{\pgfpoint{43.700006bp}{58.699988bp}}
\color[rgb]{0.0,0.0,0.0}
\pgfusepath{stroke}
\end{pgfscope}
\begin{pgfscope}
\pgfsetlinewidth{1.0bp}
\pgfsetrectcap 
\pgfsetmiterjoin \pgfsetmiterlimit{10.0}
\pgfpathmoveto{\pgfpoint{24.789519bp}{138.76004bp}}
\pgfpathcurveto{\pgfpoint{24.789519bp}{137.132467bp}}{\pgfpoint{23.523965bp}{135.813058bp}}{\pgfpoint{21.962823bp}{135.813058bp}}
\pgfpathcurveto{\pgfpoint{20.401681bp}{135.813058bp}}{\pgfpoint{19.136124bp}{137.132467bp}}{\pgfpoint{19.136124bp}{138.76004bp}}
\pgfpathcurveto{\pgfpoint{19.136124bp}{140.387614bp}}{\pgfpoint{20.401681bp}{141.707023bp}}{\pgfpoint{21.962823bp}{141.707023bp}}
\pgfpathcurveto{\pgfpoint{23.523965bp}{141.707023bp}}{\pgfpoint{24.789519bp}{140.387614bp}}{\pgfpoint{24.789519bp}{138.76004bp}}
\pgfclosepath
\color[rgb]{0.0,0.0,0.0}\pgfseteorule\pgfusepath{fill}
\pgfpathmoveto{\pgfpoint{24.789519bp}{138.76004bp}}
\pgfpathcurveto{\pgfpoint{24.789519bp}{137.132467bp}}{\pgfpoint{23.523965bp}{135.813058bp}}{\pgfpoint{21.962823bp}{135.813058bp}}
\pgfpathcurveto{\pgfpoint{20.401681bp}{135.813058bp}}{\pgfpoint{19.136124bp}{137.132467bp}}{\pgfpoint{19.136124bp}{138.76004bp}}
\pgfpathcurveto{\pgfpoint{19.136124bp}{140.387614bp}}{\pgfpoint{20.401681bp}{141.707023bp}}{\pgfpoint{21.962823bp}{141.707023bp}}
\pgfpathcurveto{\pgfpoint{23.523965bp}{141.707023bp}}{\pgfpoint{24.789519bp}{140.387614bp}}{\pgfpoint{24.789519bp}{138.76004bp}}
\pgfclosepath
\color[rgb]{0.0,0.0,0.0}
\pgfusepath{stroke}
\end{pgfscope}
\begin{pgfscope}
\pgfsetlinewidth{1.0bp}
\pgfsetrectcap 
\pgfsetmiterjoin \pgfsetmiterlimit{10.0}
\pgfpathmoveto{\pgfpoint{45.789519bp}{58.36004bp}}
\pgfpathcurveto{\pgfpoint{45.789519bp}{56.732467bp}}{\pgfpoint{44.523965bp}{55.413058bp}}{\pgfpoint{42.962823bp}{55.413058bp}}
\pgfpathcurveto{\pgfpoint{41.401681bp}{55.413058bp}}{\pgfpoint{40.136124bp}{56.732467bp}}{\pgfpoint{40.136124bp}{58.36004bp}}
\pgfpathcurveto{\pgfpoint{40.136124bp}{59.987614bp}}{\pgfpoint{41.401681bp}{61.307023bp}}{\pgfpoint{42.962823bp}{61.307023bp}}
\pgfpathcurveto{\pgfpoint{44.523965bp}{61.307023bp}}{\pgfpoint{45.789519bp}{59.987614bp}}{\pgfpoint{45.789519bp}{58.36004bp}}
\pgfclosepath
\color[rgb]{0.0,0.0,0.0}\pgfseteorule\pgfusepath{fill}
\pgfpathmoveto{\pgfpoint{45.789519bp}{58.36004bp}}
\pgfpathcurveto{\pgfpoint{45.789519bp}{56.732467bp}}{\pgfpoint{44.523965bp}{55.413058bp}}{\pgfpoint{42.962823bp}{55.413058bp}}
\pgfpathcurveto{\pgfpoint{41.401681bp}{55.413058bp}}{\pgfpoint{40.136124bp}{56.732467bp}}{\pgfpoint{40.136124bp}{58.36004bp}}
\pgfpathcurveto{\pgfpoint{40.136124bp}{59.987614bp}}{\pgfpoint{41.401681bp}{61.307023bp}}{\pgfpoint{42.962823bp}{61.307023bp}}
\pgfpathcurveto{\pgfpoint{44.523965bp}{61.307023bp}}{\pgfpoint{45.789519bp}{59.987614bp}}{\pgfpoint{45.789519bp}{58.36004bp}}
\pgfclosepath
\color[rgb]{0.0,0.0,0.0}
\pgfusepath{stroke}
\end{pgfscope}
\begin{pgfscope}
\pgfsetlinewidth{1.0bp}
\pgfsetrectcap 
\pgfsetmiterjoin \pgfsetmiterlimit{10.0}
\pgfpathmoveto{\pgfpoint{36.789519bp}{37.36004bp}}
\pgfpathcurveto{\pgfpoint{36.789519bp}{35.732467bp}}{\pgfpoint{35.523965bp}{34.413058bp}}{\pgfpoint{33.962823bp}{34.413058bp}}
\pgfpathcurveto{\pgfpoint{32.401681bp}{34.413058bp}}{\pgfpoint{31.136124bp}{35.732467bp}}{\pgfpoint{31.136124bp}{37.36004bp}}
\pgfpathcurveto{\pgfpoint{31.136124bp}{38.987614bp}}{\pgfpoint{32.401681bp}{40.307023bp}}{\pgfpoint{33.962823bp}{40.307023bp}}
\pgfpathcurveto{\pgfpoint{35.523965bp}{40.307023bp}}{\pgfpoint{36.789519bp}{38.987614bp}}{\pgfpoint{36.789519bp}{37.36004bp}}
\pgfclosepath
\color[rgb]{0.0,0.0,0.0}\pgfseteorule\pgfusepath{fill}
\pgfpathmoveto{\pgfpoint{36.789519bp}{37.36004bp}}
\pgfpathcurveto{\pgfpoint{36.789519bp}{35.732467bp}}{\pgfpoint{35.523965bp}{34.413058bp}}{\pgfpoint{33.962823bp}{34.413058bp}}
\pgfpathcurveto{\pgfpoint{32.401681bp}{34.413058bp}}{\pgfpoint{31.136124bp}{35.732467bp}}{\pgfpoint{31.136124bp}{37.36004bp}}
\pgfpathcurveto{\pgfpoint{31.136124bp}{38.987614bp}}{\pgfpoint{32.401681bp}{40.307023bp}}{\pgfpoint{33.962823bp}{40.307023bp}}
\pgfpathcurveto{\pgfpoint{35.523965bp}{40.307023bp}}{\pgfpoint{36.789519bp}{38.987614bp}}{\pgfpoint{36.789519bp}{37.36004bp}}
\pgfclosepath
\color[rgb]{0.0,0.0,0.0}
\pgfusepath{stroke}
\end{pgfscope}
\begin{pgfscope}
\pgfsetlinewidth{1.0bp}
\pgfsetrectcap 
\pgfsetmiterjoin \pgfsetmiterlimit{10.0}
\pgfpathmoveto{\pgfpoint{137.900006bp}{100.699988bp}}
\pgfpathlineto{\pgfpoint{50.900006bp}{142.099988bp}}
\color[rgb]{0.0,0.0,0.0}
\pgfusepath{stroke}
\end{pgfscope}
\begin{pgfscope}
\pgfsetlinewidth{1.0bp}
\pgfsetrectcap 
\pgfsetmiterjoin \pgfsetmiterlimit{10.0}
\pgfpathmoveto{\pgfpoint{53.189519bp}{141.56004bp}}
\pgfpathcurveto{\pgfpoint{53.189519bp}{139.932467bp}}{\pgfpoint{51.923965bp}{138.613058bp}}{\pgfpoint{50.362823bp}{138.613058bp}}
\pgfpathcurveto{\pgfpoint{48.801681bp}{138.613058bp}}{\pgfpoint{47.536124bp}{139.932467bp}}{\pgfpoint{47.536124bp}{141.56004bp}}
\pgfpathcurveto{\pgfpoint{47.536124bp}{143.187614bp}}{\pgfpoint{48.801681bp}{144.507023bp}}{\pgfpoint{50.362823bp}{144.507023bp}}
\pgfpathcurveto{\pgfpoint{51.923965bp}{144.507023bp}}{\pgfpoint{53.189519bp}{143.187614bp}}{\pgfpoint{53.189519bp}{141.56004bp}}
\pgfclosepath
\color[rgb]{0.0,0.0,0.0}\pgfseteorule\pgfusepath{fill}
\pgfpathmoveto{\pgfpoint{53.189519bp}{141.56004bp}}
\pgfpathcurveto{\pgfpoint{53.189519bp}{139.932467bp}}{\pgfpoint{51.923965bp}{138.613058bp}}{\pgfpoint{50.362823bp}{138.613058bp}}
\pgfpathcurveto{\pgfpoint{48.801681bp}{138.613058bp}}{\pgfpoint{47.536124bp}{139.932467bp}}{\pgfpoint{47.536124bp}{141.56004bp}}
\pgfpathcurveto{\pgfpoint{47.536124bp}{143.187614bp}}{\pgfpoint{48.801681bp}{144.507023bp}}{\pgfpoint{50.362823bp}{144.507023bp}}
\pgfpathcurveto{\pgfpoint{51.923965bp}{144.507023bp}}{\pgfpoint{53.189519bp}{143.187614bp}}{\pgfpoint{53.189519bp}{141.56004bp}}
\pgfclosepath
\color[rgb]{0.0,0.0,0.0}
\pgfusepath{stroke}
\end{pgfscope}
\begin{pgfscope}
\pgftransformcm{0.0}{-1.0}{1.0}{0.0}{\pgfpoint{94.700006bp}{110.899988bp}}
\pgftext[left,base]{\sffamily\mdseries\upshape\fontsize{30.1125}{30.1125}\selectfont
\color[rgb]{0.0,0.0,0.0}$\hdots$}
\end{pgfscope}
\begin{pgfscope}
\pgfsetlinewidth{1.0bp}
\pgfsetrectcap 
\pgfsetmiterjoin \pgfsetmiterlimit{10.0}
\pgfpathmoveto{\pgfpoint{72.500005bp}{92.899985bp}}
\pgfpathcurveto{\pgfpoint{72.500005bp}{56.449194bp}}{\pgfpoint{56.382256bp}{26.9bp}}{\pgfpoint{36.500005bp}{26.9bp}}
\pgfpathcurveto{\pgfpoint{16.617755bp}{26.9bp}}{\pgfpoint{0.500006bp}{56.449194bp}}{\pgfpoint{0.500006bp}{92.899985bp}}
\pgfpathcurveto{\pgfpoint{0.500006bp}{129.350775bp}}{\pgfpoint{16.617755bp}{158.899985bp}}{\pgfpoint{36.500005bp}{158.899985bp}}
\pgfpathcurveto{\pgfpoint{56.382256bp}{158.899985bp}}{\pgfpoint{72.500005bp}{129.350775bp}}{\pgfpoint{72.500005bp}{92.899985bp}}
\pgfclosepath
\color[rgb]{0.0,0.0,0.0}
\pgfusepath{stroke}
\end{pgfscope}
\begin{pgfscope}
\pgfsetlinewidth{1.0bp}
\pgfsetrectcap 
\pgfsetmiterjoin \pgfsetmiterlimit{10.0}
\pgfsetdash{{2.0bp}{3.0bp}}{2.0bp}
\pgfpathmoveto{\pgfpoint{50.300006bp}{141.499988bp}}
\pgfpathlineto{\pgfpoint{55.700006bp}{117.499988bp}}
\pgfpathlineto{\pgfpoint{46.700006bp}{110.899988bp}}
\color[rgb]{0.0,0.0,0.0}
\pgfusepath{stroke}
\end{pgfscope}
\begin{pgfscope}
\pgfsetlinewidth{1.0bp}
\pgfsetrectcap 
\pgfsetmiterjoin \pgfsetmiterlimit{10.0}
\pgfsetdash{{2.0bp}{3.0bp}}{2.0bp}
\pgfpathmoveto{\pgfpoint{22.700006bp}{138.499988bp}}
\pgfpathlineto{\pgfpoint{10.100006bp}{119.899988bp}}
\color[rgb]{0.0,0.0,0.0}
\pgfusepath{stroke}
\end{pgfscope}
\begin{pgfscope}
\pgfsetlinewidth{1.0bp}
\pgfsetrectcap 
\pgfsetmiterjoin \pgfsetmiterlimit{10.0}
\pgfsetdash{{2.0bp}{3.0bp}}{2.0bp}
\pgfpathmoveto{\pgfpoint{22.700006bp}{137.899988bp}}
\pgfpathlineto{\pgfpoint{32.300006bp}{112.699988bp}}
\pgfpathlineto{\pgfpoint{32.300006bp}{112.699988bp}}
\color[rgb]{0.0,0.0,0.0}
\pgfusepath{stroke}
\end{pgfscope}
\begin{pgfscope}
\pgfsetlinewidth{1.0bp}
\pgfsetrectcap 
\pgfsetmiterjoin \pgfsetmiterlimit{10.0}
\pgfsetdash{{2.0bp}{3.0bp}}{2.0bp}
\pgfpathmoveto{\pgfpoint{43.100006bp}{59.299988bp}}
\pgfpathlineto{\pgfpoint{27.500006bp}{70.699988bp}}
\color[rgb]{0.0,0.0,0.0}
\pgfusepath{stroke}
\end{pgfscope}
\begin{pgfscope}
\pgfsetlinewidth{1.0bp}
\pgfsetrectcap 
\pgfsetmiterjoin \pgfsetmiterlimit{10.0}
\pgfsetdash{{2.0bp}{3.0bp}}{2.0bp}
\pgfpathmoveto{\pgfpoint{41.900006bp}{56.899988bp}}
\pgfpathlineto{\pgfpoint{33.500006bp}{37.099988bp}}
\color[rgb]{0.0,0.0,0.0}
\pgfusepath{stroke}
\end{pgfscope}
\begin{pgfscope}
\pgfsetlinewidth{1.0bp}
\pgfsetrectcap 
\pgfsetmiterjoin \pgfsetmiterlimit{10.0}
\pgfpathmoveto{\pgfpoint{137.900006bp}{99.499988bp}}
\pgfpathlineto{\pgfpoint{229.100006bp}{197.401752bp}}
\color[rgb]{0.0,0.0,0.0}
\pgfusepath{stroke}
\end{pgfscope}
\begin{pgfscope}
\pgfsetlinewidth{1.0bp}
\pgfsetrectcap 
\pgfsetmiterjoin \pgfsetmiterlimit{10.0}
\pgfpathmoveto{\pgfpoint{136.700006bp}{98.899988bp}}
\pgfpathlineto{\pgfpoint{248.900006bp}{191.899988bp}}
\color[rgb]{0.0,0.0,0.0}
\pgfusepath{stroke}
\end{pgfscope}
\begin{pgfscope}
\pgfsetlinewidth{1.0bp}
\pgfsetrectcap 
\pgfsetmiterjoin \pgfsetmiterlimit{10.0}
\pgfpathmoveto{\pgfpoint{138.500006bp}{99.499988bp}}
\pgfpathlineto{\pgfpoint{245.300006bp}{169.699988bp}}
\color[rgb]{0.0,0.0,0.0}
\pgfusepath{stroke}
\end{pgfscope}
\begin{pgfscope}
\pgfsetlinewidth{1.0bp}
\pgfsetrectcap 
\pgfsetmiterjoin \pgfsetmiterlimit{10.0}
\pgfpathmoveto{\pgfpoint{138.500006bp}{100.099988bp}}
\pgfpathlineto{\pgfpoint{228.500006bp}{15.499988bp}}
\color[rgb]{0.0,0.0,0.0}
\pgfusepath{stroke}
\end{pgfscope}
\begin{pgfscope}
\pgfsetlinewidth{1.0bp}
\pgfsetrectcap 
\pgfsetmiterjoin \pgfsetmiterlimit{10.0}
\pgfpathmoveto{\pgfpoint{139.100006bp}{99.499988bp}}
\pgfpathlineto{\pgfpoint{248.900006bp}{20.899988bp}}
\color[rgb]{0.0,0.0,0.0}
\pgfusepath{stroke}
\end{pgfscope}
\begin{pgfscope}
\pgfsetlinewidth{1.0bp}
\pgfsetrectcap 
\pgfsetmiterjoin \pgfsetmiterlimit{10.0}
\pgfpathmoveto{\pgfpoint{139.100006bp}{98.899988bp}}
\pgfpathlineto{\pgfpoint{245.300006bp}{41.299988bp}}
\color[rgb]{0.0,0.0,0.0}
\pgfusepath{stroke}
\end{pgfscope}
\begin{pgfscope}
\pgfsetlinewidth{1.0bp}
\pgfsetrectcap 
\pgfsetmiterjoin \pgfsetmiterlimit{10.0}
\pgfpathmoveto{\pgfpoint{231.989519bp}{197.96004bp}}
\pgfpathcurveto{\pgfpoint{231.989519bp}{196.332467bp}}{\pgfpoint{230.723965bp}{195.013058bp}}{\pgfpoint{229.162823bp}{195.013058bp}}
\pgfpathcurveto{\pgfpoint{227.601681bp}{195.013058bp}}{\pgfpoint{226.336124bp}{196.332467bp}}{\pgfpoint{226.336124bp}{197.96004bp}}
\pgfpathcurveto{\pgfpoint{226.336124bp}{199.587614bp}}{\pgfpoint{227.601681bp}{200.907023bp}}{\pgfpoint{229.162823bp}{200.907023bp}}
\pgfpathcurveto{\pgfpoint{230.723965bp}{200.907023bp}}{\pgfpoint{231.989519bp}{199.587614bp}}{\pgfpoint{231.989519bp}{197.96004bp}}
\pgfclosepath
\color[rgb]{0.0,0.0,0.0}\pgfseteorule\pgfusepath{fill}
\pgfpathmoveto{\pgfpoint{231.989519bp}{197.96004bp}}
\pgfpathcurveto{\pgfpoint{231.989519bp}{196.332467bp}}{\pgfpoint{230.723965bp}{195.013058bp}}{\pgfpoint{229.162823bp}{195.013058bp}}
\pgfpathcurveto{\pgfpoint{227.601681bp}{195.013058bp}}{\pgfpoint{226.336124bp}{196.332467bp}}{\pgfpoint{226.336124bp}{197.96004bp}}
\pgfpathcurveto{\pgfpoint{226.336124bp}{199.587614bp}}{\pgfpoint{227.601681bp}{200.907023bp}}{\pgfpoint{229.162823bp}{200.907023bp}}
\pgfpathcurveto{\pgfpoint{230.723965bp}{200.907023bp}}{\pgfpoint{231.989519bp}{199.587614bp}}{\pgfpoint{231.989519bp}{197.96004bp}}
\pgfclosepath
\color[rgb]{0.0,0.0,0.0}
\pgfusepath{stroke}
\end{pgfscope}
\begin{pgfscope}
\pgfsetlinewidth{1.0bp}
\pgfsetrectcap 
\pgfsetmiterjoin \pgfsetmiterlimit{10.0}
\pgfpathmoveto{\pgfpoint{253.589519bp}{193.76004bp}}
\pgfpathcurveto{\pgfpoint{253.589519bp}{192.132467bp}}{\pgfpoint{252.323965bp}{190.813058bp}}{\pgfpoint{250.762823bp}{190.813058bp}}
\pgfpathcurveto{\pgfpoint{249.201681bp}{190.813058bp}}{\pgfpoint{247.936124bp}{192.132467bp}}{\pgfpoint{247.936124bp}{193.76004bp}}
\pgfpathcurveto{\pgfpoint{247.936124bp}{195.387614bp}}{\pgfpoint{249.201681bp}{196.707023bp}}{\pgfpoint{250.762823bp}{196.707023bp}}
\pgfpathcurveto{\pgfpoint{252.323965bp}{196.707023bp}}{\pgfpoint{253.589519bp}{195.387614bp}}{\pgfpoint{253.589519bp}{193.76004bp}}
\pgfclosepath
\color[rgb]{0.0,0.0,0.0}\pgfseteorule\pgfusepath{fill}
\pgfpathmoveto{\pgfpoint{253.589519bp}{193.76004bp}}
\pgfpathcurveto{\pgfpoint{253.589519bp}{192.132467bp}}{\pgfpoint{252.323965bp}{190.813058bp}}{\pgfpoint{250.762823bp}{190.813058bp}}
\pgfpathcurveto{\pgfpoint{249.201681bp}{190.813058bp}}{\pgfpoint{247.936124bp}{192.132467bp}}{\pgfpoint{247.936124bp}{193.76004bp}}
\pgfpathcurveto{\pgfpoint{247.936124bp}{195.387614bp}}{\pgfpoint{249.201681bp}{196.707023bp}}{\pgfpoint{250.762823bp}{196.707023bp}}
\pgfpathcurveto{\pgfpoint{252.323965bp}{196.707023bp}}{\pgfpoint{253.589519bp}{195.387614bp}}{\pgfpoint{253.589519bp}{193.76004bp}}
\pgfclosepath
\color[rgb]{0.0,0.0,0.0}
\pgfusepath{stroke}
\end{pgfscope}
\begin{pgfscope}
\pgfsetlinewidth{1.0bp}
\pgfsetrectcap 
\pgfsetmiterjoin \pgfsetmiterlimit{10.0}
\pgfpathmoveto{\pgfpoint{248.789519bp}{170.36004bp}}
\pgfpathcurveto{\pgfpoint{248.789519bp}{168.732467bp}}{\pgfpoint{247.523965bp}{167.413058bp}}{\pgfpoint{245.962823bp}{167.413058bp}}
\pgfpathcurveto{\pgfpoint{244.401681bp}{167.413058bp}}{\pgfpoint{243.136124bp}{168.732467bp}}{\pgfpoint{243.136124bp}{170.36004bp}}
\pgfpathcurveto{\pgfpoint{243.136124bp}{171.987614bp}}{\pgfpoint{244.401681bp}{173.307023bp}}{\pgfpoint{245.962823bp}{173.307023bp}}
\pgfpathcurveto{\pgfpoint{247.523965bp}{173.307023bp}}{\pgfpoint{248.789519bp}{171.987614bp}}{\pgfpoint{248.789519bp}{170.36004bp}}
\pgfclosepath
\color[rgb]{0.0,0.0,0.0}\pgfseteorule\pgfusepath{fill}
\pgfpathmoveto{\pgfpoint{248.789519bp}{170.36004bp}}
\pgfpathcurveto{\pgfpoint{248.789519bp}{168.732467bp}}{\pgfpoint{247.523965bp}{167.413058bp}}{\pgfpoint{245.962823bp}{167.413058bp}}
\pgfpathcurveto{\pgfpoint{244.401681bp}{167.413058bp}}{\pgfpoint{243.136124bp}{168.732467bp}}{\pgfpoint{243.136124bp}{170.36004bp}}
\pgfpathcurveto{\pgfpoint{243.136124bp}{171.987614bp}}{\pgfpoint{244.401681bp}{173.307023bp}}{\pgfpoint{245.962823bp}{173.307023bp}}
\pgfpathcurveto{\pgfpoint{247.523965bp}{173.307023bp}}{\pgfpoint{248.789519bp}{171.987614bp}}{\pgfpoint{248.789519bp}{170.36004bp}}
\pgfclosepath
\color[rgb]{0.0,0.0,0.0}
\pgfusepath{stroke}
\end{pgfscope}
\begin{pgfscope}
\pgfsetlinewidth{1.0bp}
\pgfsetrectcap 
\pgfsetmiterjoin \pgfsetmiterlimit{10.0}
\pgfpathmoveto{\pgfpoint{250.589519bp}{40.16004bp}}
\pgfpathcurveto{\pgfpoint{250.589519bp}{38.532467bp}}{\pgfpoint{249.323965bp}{37.213058bp}}{\pgfpoint{247.762823bp}{37.213058bp}}
\pgfpathcurveto{\pgfpoint{246.201681bp}{37.213058bp}}{\pgfpoint{244.936124bp}{38.532467bp}}{\pgfpoint{244.936124bp}{40.16004bp}}
\pgfpathcurveto{\pgfpoint{244.936124bp}{41.787614bp}}{\pgfpoint{246.201681bp}{43.107023bp}}{\pgfpoint{247.762823bp}{43.107023bp}}
\pgfpathcurveto{\pgfpoint{249.323965bp}{43.107023bp}}{\pgfpoint{250.589519bp}{41.787614bp}}{\pgfpoint{250.589519bp}{40.16004bp}}
\pgfclosepath
\color[rgb]{0.0,0.0,0.0}\pgfseteorule\pgfusepath{fill}
\pgfpathmoveto{\pgfpoint{250.589519bp}{40.16004bp}}
\pgfpathcurveto{\pgfpoint{250.589519bp}{38.532467bp}}{\pgfpoint{249.323965bp}{37.213058bp}}{\pgfpoint{247.762823bp}{37.213058bp}}
\pgfpathcurveto{\pgfpoint{246.201681bp}{37.213058bp}}{\pgfpoint{244.936124bp}{38.532467bp}}{\pgfpoint{244.936124bp}{40.16004bp}}
\pgfpathcurveto{\pgfpoint{244.936124bp}{41.787614bp}}{\pgfpoint{246.201681bp}{43.107023bp}}{\pgfpoint{247.762823bp}{43.107023bp}}
\pgfpathcurveto{\pgfpoint{249.323965bp}{43.107023bp}}{\pgfpoint{250.589519bp}{41.787614bp}}{\pgfpoint{250.589519bp}{40.16004bp}}
\pgfclosepath
\color[rgb]{0.0,0.0,0.0}
\pgfusepath{stroke}
\end{pgfscope}
\begin{pgfscope}
\pgfsetlinewidth{1.0bp}
\pgfsetrectcap 
\pgfsetmiterjoin \pgfsetmiterlimit{10.0}
\pgfpathmoveto{\pgfpoint{252.989519bp}{20.96004bp}}
\pgfpathcurveto{\pgfpoint{252.989519bp}{19.332467bp}}{\pgfpoint{251.723965bp}{18.013058bp}}{\pgfpoint{250.162823bp}{18.013058bp}}
\pgfpathcurveto{\pgfpoint{248.601681bp}{18.013058bp}}{\pgfpoint{247.336124bp}{19.332467bp}}{\pgfpoint{247.336124bp}{20.96004bp}}
\pgfpathcurveto{\pgfpoint{247.336124bp}{22.587614bp}}{\pgfpoint{248.601681bp}{23.907023bp}}{\pgfpoint{250.162823bp}{23.907023bp}}
\pgfpathcurveto{\pgfpoint{251.723965bp}{23.907023bp}}{\pgfpoint{252.989519bp}{22.587614bp}}{\pgfpoint{252.989519bp}{20.96004bp}}
\pgfclosepath
\color[rgb]{0.0,0.0,0.0}\pgfseteorule\pgfusepath{fill}
\pgfpathmoveto{\pgfpoint{252.989519bp}{20.96004bp}}
\pgfpathcurveto{\pgfpoint{252.989519bp}{19.332467bp}}{\pgfpoint{251.723965bp}{18.013058bp}}{\pgfpoint{250.162823bp}{18.013058bp}}
\pgfpathcurveto{\pgfpoint{248.601681bp}{18.013058bp}}{\pgfpoint{247.336124bp}{19.332467bp}}{\pgfpoint{247.336124bp}{20.96004bp}}
\pgfpathcurveto{\pgfpoint{247.336124bp}{22.587614bp}}{\pgfpoint{248.601681bp}{23.907023bp}}{\pgfpoint{250.162823bp}{23.907023bp}}
\pgfpathcurveto{\pgfpoint{251.723965bp}{23.907023bp}}{\pgfpoint{252.989519bp}{22.587614bp}}{\pgfpoint{252.989519bp}{20.96004bp}}
\pgfclosepath
\color[rgb]{0.0,0.0,0.0}
\pgfusepath{stroke}
\end{pgfscope}
\begin{pgfscope}
\pgfsetlinewidth{1.0bp}
\pgfsetrectcap 
\pgfsetmiterjoin \pgfsetmiterlimit{10.0}
\pgfpathmoveto{\pgfpoint{233.189519bp}{13.76004bp}}
\pgfpathcurveto{\pgfpoint{233.189519bp}{12.132467bp}}{\pgfpoint{231.923965bp}{10.813058bp}}{\pgfpoint{230.362823bp}{10.813058bp}}
\pgfpathcurveto{\pgfpoint{228.801681bp}{10.813058bp}}{\pgfpoint{227.536124bp}{12.132467bp}}{\pgfpoint{227.536124bp}{13.76004bp}}
\pgfpathcurveto{\pgfpoint{227.536124bp}{15.387614bp}}{\pgfpoint{228.801681bp}{16.707023bp}}{\pgfpoint{230.362823bp}{16.707023bp}}
\pgfpathcurveto{\pgfpoint{231.923965bp}{16.707023bp}}{\pgfpoint{233.189519bp}{15.387614bp}}{\pgfpoint{233.189519bp}{13.76004bp}}
\pgfclosepath
\color[rgb]{0.0,0.0,0.0}\pgfseteorule\pgfusepath{fill}
\pgfpathmoveto{\pgfpoint{233.189519bp}{13.76004bp}}
\pgfpathcurveto{\pgfpoint{233.189519bp}{12.132467bp}}{\pgfpoint{231.923965bp}{10.813058bp}}{\pgfpoint{230.362823bp}{10.813058bp}}
\pgfpathcurveto{\pgfpoint{228.801681bp}{10.813058bp}}{\pgfpoint{227.536124bp}{12.132467bp}}{\pgfpoint{227.536124bp}{13.76004bp}}
\pgfpathcurveto{\pgfpoint{227.536124bp}{15.387614bp}}{\pgfpoint{228.801681bp}{16.707023bp}}{\pgfpoint{230.362823bp}{16.707023bp}}
\pgfpathcurveto{\pgfpoint{231.923965bp}{16.707023bp}}{\pgfpoint{233.189519bp}{15.387614bp}}{\pgfpoint{233.189519bp}{13.76004bp}}
\pgfclosepath
\color[rgb]{0.0,0.0,0.0}
\pgfusepath{stroke}
\end{pgfscope}
\begin{pgfscope}
\pgftransformcm{0.0}{-1.0}{1.0}{0.0}{\pgfpoint{171.100006bp}{114.899988bp}}
\pgftext[left,base]{\sffamily\mdseries\upshape\fontsize{30.1125}{30.1125}\selectfont
\color[rgb]{0.0,0.0,0.0}$\hdots$}
\end{pgfscope}
\begin{pgfscope}
\pgftransformcm{1.0}{0.0}{0.0}{1.0}{\pgfpoint{28.100006bp}{85.699988bp}}
\pgftext[left,base]{\sffamily\mdseries\upshape\huge
\color[rgb]{0.0,0.0,0.0}$G'$}
\end{pgfscope}
\begin{pgfscope}
\pgftransformcm{1.0}{0.0}{0.0}{1.0}{\pgfpoint{241.300006bp}{101.699988bp}}
\pgftext[left,base]{\sffamily\mdseries\upshape\huge
\color[rgb]{0.0,0.0,0.0}$I$}
\end{pgfscope}
\begin{pgfscope}
\pgftransformcm{1.0}{0.0}{0.0}{1.0}{\pgfpoint{130.700006bp}{14.899988bp}}
\pgftext[left,base]{\sffamily\mdseries\upshape\huge
\color[rgb]{0.0,0.0,0.0}$G$}
\end{pgfscope}
\end{pgfpicture}

%% file: Hexample2.tex
% C:\Users\agalanis\Dropbox\H-colorings\final revision\H-colorings\Hexample2.jdr
% Created by Jpgfdraw version 0.5.6b
% 23-Dec-2015 02:56:27
\begin{pgfpicture}{0bp}{0bp}{381.014679bp}{160.503556bp}
\begin{pgfscope}
\pgftransformcm{0.5}{-0.0}{0.0}{0.5}{\pgfpoint{46.578376bp}{0.0bp}}
\pgftext[left,bottom]{\sffamily\mdseries\upshape\fontsize{30.1125}{30.1125}\selectfont
\color[rgb]{0.0,0.0,0.0}$H$}
\end{pgfscope}
\begin{pgfscope}
\pgftransformcm{0.5}{-0.0}{0.0}{0.5}{\pgfpoint{186.978376bp}{0.0bp}}
\pgftext[left,bottom]{\sffamily\mdseries\upshape\fontsize{30.1125}{30.1125}\selectfont
\color[rgb]{0.0,0.0,0.0}$H_1$}
\end{pgfscope}
\begin{pgfscope}
\pgftransformcm{0.5}{-0.0}{0.0}{0.5}{\pgfpoint{328.478376bp}{0.0bp}}
\pgftext[left,bottom]{\sffamily\mdseries\upshape\fontsize{30.1125}{30.1125}\selectfont
\color[rgb]{0.0,0.0,0.0}$H_2$}
\end{pgfscope}
\begin{pgfscope}
\pgfsetlinewidth{1.0bp}
\pgfsetrectcap 
\pgfsetmiterjoin \pgfsetmiterlimit{10.0}
\pgfpathmoveto{\pgfpoint{151.484845bp}{18.615273bp}}
\pgfpathcurveto{\pgfpoint{164.024845bp}{19.770273bp}}{\pgfpoint{162.704845bp}{30.495273bp}}{\pgfpoint{151.319845bp}{46.987478bp}}
\pgfpathcurveto{\pgfpoint{139.934845bp}{30.495273bp}}{\pgfpoint{138.614845bp}{19.770273bp}}{\pgfpoint{151.154845bp}{18.615273bp}}
\color[rgb]{0.0,0.0,0.0}
\pgfusepath{stroke}
\end{pgfscope}
\begin{pgfscope}
\pgfsetlinewidth{1.0bp}
\pgfsetrectcap 
\pgfsetmiterjoin \pgfsetmiterlimit{10.0}
\pgfpathmoveto{\pgfpoint{233.874845bp}{18.560273bp}}
\pgfpathcurveto{\pgfpoint{246.414845bp}{19.715273bp}}{\pgfpoint{245.094845bp}{30.440273bp}}{\pgfpoint{233.709845bp}{46.932478bp}}
\pgfpathcurveto{\pgfpoint{222.324845bp}{30.440273bp}}{\pgfpoint{221.004845bp}{19.715273bp}}{\pgfpoint{233.544845bp}{18.560273bp}}
\color[rgb]{0.0,0.0,0.0}
\pgfusepath{stroke}
\end{pgfscope}
\begin{pgfscope}
\pgfsetlinewidth{1.0bp}
\pgfsetrectcap 
\pgfsetmiterjoin \pgfsetmiterlimit{10.0}
\pgfpathmoveto{\pgfpoint{289.094845bp}{18.615273bp}}
\pgfpathcurveto{\pgfpoint{301.634845bp}{19.770273bp}}{\pgfpoint{300.314845bp}{30.495273bp}}{\pgfpoint{288.929845bp}{46.987478bp}}
\pgfpathcurveto{\pgfpoint{277.544845bp}{30.495273bp}}{\pgfpoint{276.224845bp}{19.770273bp}}{\pgfpoint{288.764845bp}{18.615273bp}}
\color[rgb]{0.0,0.0,0.0}
\pgfusepath{stroke}
\end{pgfscope}
\begin{pgfscope}
\pgfsetlinewidth{1.0bp}
\pgfsetrectcap 
\pgfsetmiterjoin \pgfsetmiterlimit{10.0}
\pgfpathmoveto{\pgfpoint{92.029842bp}{159.849717bp}}
\pgfpathcurveto{\pgfpoint{79.489842bp}{158.694717bp}}{\pgfpoint{80.809842bp}{147.969717bp}}{\pgfpoint{92.194842bp}{131.477512bp}}
\pgfpathcurveto{\pgfpoint{103.579842bp}{147.969717bp}}{\pgfpoint{104.899842bp}{158.694717bp}}{\pgfpoint{92.359842bp}{159.849717bp}}
\color[rgb]{0.0,0.0,0.0}
\pgfusepath{stroke}
\end{pgfscope}
\begin{pgfscope}
\pgfsetlinewidth{1.0bp}
\pgfsetrectcap 
\pgfsetmiterjoin \pgfsetmiterlimit{10.0}
\pgfpathmoveto{\pgfpoint{97.726628bp}{132.175751bp}}
\pgfpathcurveto{\pgfpoint{97.726628bp}{128.920605bp}}{\pgfpoint{95.195521bp}{126.281786bp}}{\pgfpoint{92.073236bp}{126.281786bp}}
\pgfpathcurveto{\pgfpoint{88.950952bp}{126.281786bp}}{\pgfpoint{86.419839bp}{128.920605bp}}{\pgfpoint{86.419839bp}{132.175751bp}}
\pgfpathcurveto{\pgfpoint{86.419839bp}{135.430899bp}}{\pgfpoint{88.950952bp}{138.069716bp}}{\pgfpoint{92.073236bp}{138.069716bp}}
\pgfpathcurveto{\pgfpoint{95.195521bp}{138.069716bp}}{\pgfpoint{97.726628bp}{135.430899bp}}{\pgfpoint{97.726628bp}{132.175751bp}}
\pgfclosepath
\color[rgb]{0.0,0.0,0.0}\pgfseteorule\pgfusepath{fill}
\pgfpathmoveto{\pgfpoint{97.726628bp}{132.175751bp}}
\pgfpathcurveto{\pgfpoint{97.726628bp}{128.920605bp}}{\pgfpoint{95.195521bp}{126.281786bp}}{\pgfpoint{92.073236bp}{126.281786bp}}
\pgfpathcurveto{\pgfpoint{88.950952bp}{126.281786bp}}{\pgfpoint{86.419839bp}{128.920605bp}}{\pgfpoint{86.419839bp}{132.175751bp}}
\pgfpathcurveto{\pgfpoint{86.419839bp}{135.430899bp}}{\pgfpoint{88.950952bp}{138.069716bp}}{\pgfpoint{92.073236bp}{138.069716bp}}
\pgfpathcurveto{\pgfpoint{95.195521bp}{138.069716bp}}{\pgfpoint{97.726628bp}{135.430899bp}}{\pgfpoint{97.726628bp}{132.175751bp}}
\pgfclosepath
\color[rgb]{0.0,0.0,0.0}
\pgfusepath{stroke}
\end{pgfscope}
\begin{pgfscope}
\pgfsetlinewidth{1.0bp}
\pgfsetrectcap 
\pgfsetmiterjoin \pgfsetmiterlimit{10.0}
\pgfpathmoveto{\pgfpoint{10.024845bp}{19.330273bp}}
\pgfpathcurveto{\pgfpoint{22.564845bp}{20.485273bp}}{\pgfpoint{21.244845bp}{31.210273bp}}{\pgfpoint{9.859845bp}{47.702478bp}}
\pgfpathcurveto{\pgfpoint{-1.525155bp}{31.210273bp}}{\pgfpoint{-2.845155bp}{20.485273bp}}{\pgfpoint{9.694845bp}{19.330273bp}}
\color[rgb]{0.0,0.0,0.0}
\pgfusepath{stroke}
\end{pgfscope}
\begin{pgfscope}
\pgfsetlinewidth{1.0bp}
\pgfsetrectcap 
\pgfsetmiterjoin \pgfsetmiterlimit{10.0}
\pgfpathmoveto{\pgfpoint{15.391628bp}{47.365751bp}}
\pgfpathcurveto{\pgfpoint{15.391628bp}{44.110605bp}}{\pgfpoint{12.860521bp}{41.471786bp}}{\pgfpoint{9.738236bp}{41.471786bp}}
\pgfpathcurveto{\pgfpoint{6.615952bp}{41.471786bp}}{\pgfpoint{4.084839bp}{44.110605bp}}{\pgfpoint{4.084839bp}{47.365751bp}}
\pgfpathcurveto{\pgfpoint{4.084839bp}{50.620899bp}}{\pgfpoint{6.615952bp}{53.259716bp}}{\pgfpoint{9.738236bp}{53.259716bp}}
\pgfpathcurveto{\pgfpoint{12.860521bp}{53.259716bp}}{\pgfpoint{15.391628bp}{50.620899bp}}{\pgfpoint{15.391628bp}{47.365751bp}}
\pgfclosepath
\color[rgb]{0.0,0.0,0.0}\pgfseteorule\pgfusepath{fill}
\pgfpathmoveto{\pgfpoint{15.391628bp}{47.365751bp}}
\pgfpathcurveto{\pgfpoint{15.391628bp}{44.110605bp}}{\pgfpoint{12.860521bp}{41.471786bp}}{\pgfpoint{9.738236bp}{41.471786bp}}
\pgfpathcurveto{\pgfpoint{6.615952bp}{41.471786bp}}{\pgfpoint{4.084839bp}{44.110605bp}}{\pgfpoint{4.084839bp}{47.365751bp}}
\pgfpathcurveto{\pgfpoint{4.084839bp}{50.620899bp}}{\pgfpoint{6.615952bp}{53.259716bp}}{\pgfpoint{9.738236bp}{53.259716bp}}
\pgfpathcurveto{\pgfpoint{12.860521bp}{53.259716bp}}{\pgfpoint{15.391628bp}{50.620899bp}}{\pgfpoint{15.391628bp}{47.365751bp}}
\pgfclosepath
\color[rgb]{0.0,0.0,0.0}
\pgfusepath{stroke}
\end{pgfscope}
\begin{pgfscope}
\pgfsetlinewidth{1.0bp}
\pgfsetrectcap 
\pgfsetmiterjoin \pgfsetmiterlimit{10.0}
\pgfpathmoveto{\pgfpoint{9.584842bp}{159.959717bp}}
\pgfpathcurveto{\pgfpoint{-2.955158bp}{158.804717bp}}{\pgfpoint{-1.635158bp}{148.079717bp}}{\pgfpoint{9.749842bp}{131.587512bp}}
\pgfpathcurveto{\pgfpoint{21.134842bp}{148.079717bp}}{\pgfpoint{22.454842bp}{158.804717bp}}{\pgfpoint{9.914842bp}{159.959717bp}}
\color[rgb]{0.0,0.0,0.0}
\pgfusepath{stroke}
\end{pgfscope}
\begin{pgfscope}
\pgfsetlinewidth{1.0bp}
\pgfsetrectcap 
\pgfsetmiterjoin \pgfsetmiterlimit{10.0}
\pgfpathmoveto{\pgfpoint{15.281628bp}{132.285751bp}}
\pgfpathcurveto{\pgfpoint{15.281628bp}{129.030605bp}}{\pgfpoint{12.750521bp}{126.391786bp}}{\pgfpoint{9.628236bp}{126.391786bp}}
\pgfpathcurveto{\pgfpoint{6.505952bp}{126.391786bp}}{\pgfpoint{3.974839bp}{129.030605bp}}{\pgfpoint{3.974839bp}{132.285751bp}}
\pgfpathcurveto{\pgfpoint{3.974839bp}{135.540899bp}}{\pgfpoint{6.505952bp}{138.179716bp}}{\pgfpoint{9.628236bp}{138.179716bp}}
\pgfpathcurveto{\pgfpoint{12.750521bp}{138.179716bp}}{\pgfpoint{15.281628bp}{135.540899bp}}{\pgfpoint{15.281628bp}{132.285751bp}}
\pgfclosepath
\color[rgb]{0.0,0.0,0.0}\pgfseteorule\pgfusepath{fill}
\pgfpathmoveto{\pgfpoint{15.281628bp}{132.285751bp}}
\pgfpathcurveto{\pgfpoint{15.281628bp}{129.030605bp}}{\pgfpoint{12.750521bp}{126.391786bp}}{\pgfpoint{9.628236bp}{126.391786bp}}
\pgfpathcurveto{\pgfpoint{6.505952bp}{126.391786bp}}{\pgfpoint{3.974839bp}{129.030605bp}}{\pgfpoint{3.974839bp}{132.285751bp}}
\pgfpathcurveto{\pgfpoint{3.974839bp}{135.540899bp}}{\pgfpoint{6.505952bp}{138.179716bp}}{\pgfpoint{9.628236bp}{138.179716bp}}
\pgfpathcurveto{\pgfpoint{12.750521bp}{138.179716bp}}{\pgfpoint{15.281628bp}{135.540899bp}}{\pgfpoint{15.281628bp}{132.285751bp}}
\pgfclosepath
\color[rgb]{0.0,0.0,0.0}
\pgfusepath{stroke}
\end{pgfscope}
\begin{pgfscope}
\pgfsetlinewidth{1.0bp}
\pgfsetrectcap 
\pgfsetmiterjoin \pgfsetmiterlimit{10.0}
\pgfpathmoveto{\pgfpoint{92.414845bp}{19.275273bp}}
\pgfpathcurveto{\pgfpoint{104.954845bp}{20.430273bp}}{\pgfpoint{103.634845bp}{31.155273bp}}{\pgfpoint{92.249845bp}{47.647478bp}}
\pgfpathcurveto{\pgfpoint{80.864845bp}{31.155273bp}}{\pgfpoint{79.544845bp}{20.430273bp}}{\pgfpoint{92.084845bp}{19.275273bp}}
\color[rgb]{0.0,0.0,0.0}
\pgfusepath{stroke}
\end{pgfscope}
\begin{pgfscope}
\pgfsetlinewidth{1.0bp}
\pgfsetrectcap 
\pgfsetmiterjoin \pgfsetmiterlimit{10.0}
\pgfpathmoveto{\pgfpoint{97.781628bp}{47.310751bp}}
\pgfpathcurveto{\pgfpoint{97.781628bp}{44.055605bp}}{\pgfpoint{95.250521bp}{41.416786bp}}{\pgfpoint{92.128236bp}{41.416786bp}}
\pgfpathcurveto{\pgfpoint{89.005952bp}{41.416786bp}}{\pgfpoint{86.474839bp}{44.055605bp}}{\pgfpoint{86.474839bp}{47.310751bp}}
\pgfpathcurveto{\pgfpoint{86.474839bp}{50.565899bp}}{\pgfpoint{89.005952bp}{53.204716bp}}{\pgfpoint{92.128236bp}{53.204716bp}}
\pgfpathcurveto{\pgfpoint{95.250521bp}{53.204716bp}}{\pgfpoint{97.781628bp}{50.565899bp}}{\pgfpoint{97.781628bp}{47.310751bp}}
\pgfclosepath
\color[rgb]{0.0,0.0,0.0}\pgfseteorule\pgfusepath{fill}
\pgfpathmoveto{\pgfpoint{97.781628bp}{47.310751bp}}
\pgfpathcurveto{\pgfpoint{97.781628bp}{44.055605bp}}{\pgfpoint{95.250521bp}{41.416786bp}}{\pgfpoint{92.128236bp}{41.416786bp}}
\pgfpathcurveto{\pgfpoint{89.005952bp}{41.416786bp}}{\pgfpoint{86.474839bp}{44.055605bp}}{\pgfpoint{86.474839bp}{47.310751bp}}
\pgfpathcurveto{\pgfpoint{86.474839bp}{50.565899bp}}{\pgfpoint{89.005952bp}{53.204716bp}}{\pgfpoint{92.128236bp}{53.204716bp}}
\pgfpathcurveto{\pgfpoint{95.250521bp}{53.204716bp}}{\pgfpoint{97.781628bp}{50.565899bp}}{\pgfpoint{97.781628bp}{47.310751bp}}
\pgfclosepath
\color[rgb]{0.0,0.0,0.0}
\pgfusepath{stroke}
\end{pgfscope}
\begin{pgfscope}
\pgfsetlinewidth{1.0bp}
\pgfsetrectcap 
\pgfsetmiterjoin \pgfsetmiterlimit{10.0}
\pgfpathmoveto{\pgfpoint{92.359842bp}{131.139717bp}}
\pgfpathlineto{\pgfpoint{9.529842bp}{131.139717bp}}
\color[rgb]{0.0,0.0,0.0}
\pgfusepath{stroke}
\end{pgfscope}
\begin{pgfscope}
\pgfsetlinewidth{1.0bp}
\pgfsetrectcap 
\pgfsetmiterjoin \pgfsetmiterlimit{10.0}
\pgfpathmoveto{\pgfpoint{9.694842bp}{131.139717bp}}
\pgfpathlineto{\pgfpoint{9.694842bp}{47.979717bp}}
\color[rgb]{0.0,0.0,0.0}
\pgfusepath{stroke}
\end{pgfscope}
\begin{pgfscope}
\pgfsetlinewidth{1.0bp}
\pgfsetrectcap 
\pgfsetmiterjoin \pgfsetmiterlimit{10.0}
\pgfpathmoveto{\pgfpoint{9.694842bp}{48.309717bp}}
\pgfpathlineto{\pgfpoint{92.854842bp}{48.309717bp}}
\color[rgb]{0.0,0.0,0.0}
\pgfusepath{stroke}
\end{pgfscope}
\begin{pgfscope}
\pgfsetlinewidth{1.0bp}
\pgfsetrectcap 
\pgfsetmiterjoin \pgfsetmiterlimit{10.0}
\pgfpathmoveto{\pgfpoint{92.524842bp}{131.139717bp}}
\pgfpathlineto{\pgfpoint{92.524842bp}{48.144717bp}}
\color[rgb]{0.0,0.0,0.0}
\pgfusepath{stroke}
\end{pgfscope}
\begin{pgfscope}
\pgfsetlinewidth{1.0bp}
\pgfsetrectcap 
\pgfsetmiterjoin \pgfsetmiterlimit{10.0}
\pgfpathmoveto{\pgfpoint{10.083743bp}{130.750808bp}}
\pgfpathlineto{\pgfpoint{92.583743bp}{48.250808bp}}
\color[rgb]{0.0,0.0,0.0}
\pgfusepath{stroke}
\end{pgfscope}
\begin{pgfscope}
\pgfsetlinewidth{1.0bp}
\pgfsetrectcap 
\pgfsetmiterjoin \pgfsetmiterlimit{10.0}
\pgfpathmoveto{\pgfpoint{93.234894bp}{130.264673bp}}
\pgfpathlineto{\pgfpoint{10.734894bp}{47.599673bp}}
\color[rgb]{0.0,0.0,0.0}
\pgfusepath{stroke}
\end{pgfscope}
\begin{pgfscope}
\pgfsetlinewidth{1.0bp}
\pgfsetrectcap 
\pgfsetmiterjoin \pgfsetmiterlimit{10.0}
\pgfpathmoveto{\pgfpoint{57.521628bp}{89.055751bp}}
\pgfpathcurveto{\pgfpoint{57.521628bp}{85.800605bp}}{\pgfpoint{54.990521bp}{83.161786bp}}{\pgfpoint{51.868236bp}{83.161786bp}}
\pgfpathcurveto{\pgfpoint{48.745952bp}{83.161786bp}}{\pgfpoint{46.214839bp}{85.800605bp}}{\pgfpoint{46.214839bp}{89.055751bp}}
\pgfpathcurveto{\pgfpoint{46.214839bp}{92.310899bp}}{\pgfpoint{48.745952bp}{94.949716bp}}{\pgfpoint{51.868236bp}{94.949716bp}}
\pgfpathcurveto{\pgfpoint{54.990521bp}{94.949716bp}}{\pgfpoint{57.521628bp}{92.310899bp}}{\pgfpoint{57.521628bp}{89.055751bp}}
\pgfclosepath
\color[rgb]{0.0,0.0,0.0}\pgfseteorule\pgfusepath{fill}
\pgfpathmoveto{\pgfpoint{57.521628bp}{89.055751bp}}
\pgfpathcurveto{\pgfpoint{57.521628bp}{85.800605bp}}{\pgfpoint{54.990521bp}{83.161786bp}}{\pgfpoint{51.868236bp}{83.161786bp}}
\pgfpathcurveto{\pgfpoint{48.745952bp}{83.161786bp}}{\pgfpoint{46.214839bp}{85.800605bp}}{\pgfpoint{46.214839bp}{89.055751bp}}
\pgfpathcurveto{\pgfpoint{46.214839bp}{92.310899bp}}{\pgfpoint{48.745952bp}{94.949716bp}}{\pgfpoint{51.868236bp}{94.949716bp}}
\pgfpathcurveto{\pgfpoint{54.990521bp}{94.949716bp}}{\pgfpoint{57.521628bp}{92.310899bp}}{\pgfpoint{57.521628bp}{89.055751bp}}
\pgfclosepath
\color[rgb]{0.0,0.0,0.0}
\pgfusepath{stroke}
\end{pgfscope}
\begin{pgfscope}
\pgfsetlinewidth{1.0bp}
\pgfsetrectcap 
\pgfsetmiterjoin \pgfsetmiterlimit{10.0}
\pgfpathmoveto{\pgfpoint{233.489842bp}{159.134717bp}}
\pgfpathcurveto{\pgfpoint{220.949842bp}{157.979717bp}}{\pgfpoint{222.269842bp}{147.254717bp}}{\pgfpoint{233.654842bp}{130.762512bp}}
\pgfpathcurveto{\pgfpoint{245.039842bp}{147.254717bp}}{\pgfpoint{246.359842bp}{157.979717bp}}{\pgfpoint{233.819842bp}{159.134717bp}}
\color[rgb]{0.0,0.0,0.0}
\pgfusepath{stroke}
\end{pgfscope}
\begin{pgfscope}
\pgfsetlinewidth{1.0bp}
\pgfsetrectcap 
\pgfsetmiterjoin \pgfsetmiterlimit{10.0}
\pgfpathmoveto{\pgfpoint{239.186628bp}{131.460751bp}}
\pgfpathcurveto{\pgfpoint{239.186628bp}{128.205605bp}}{\pgfpoint{236.655521bp}{125.566786bp}}{\pgfpoint{233.533236bp}{125.566786bp}}
\pgfpathcurveto{\pgfpoint{230.410952bp}{125.566786bp}}{\pgfpoint{227.879839bp}{128.205605bp}}{\pgfpoint{227.879839bp}{131.460751bp}}
\pgfpathcurveto{\pgfpoint{227.879839bp}{134.715899bp}}{\pgfpoint{230.410952bp}{137.354716bp}}{\pgfpoint{233.533236bp}{137.354716bp}}
\pgfpathcurveto{\pgfpoint{236.655521bp}{137.354716bp}}{\pgfpoint{239.186628bp}{134.715899bp}}{\pgfpoint{239.186628bp}{131.460751bp}}
\pgfclosepath
\color[rgb]{0.0,0.0,0.0}\pgfseteorule\pgfusepath{fill}
\pgfpathmoveto{\pgfpoint{239.186628bp}{131.460751bp}}
\pgfpathcurveto{\pgfpoint{239.186628bp}{128.205605bp}}{\pgfpoint{236.655521bp}{125.566786bp}}{\pgfpoint{233.533236bp}{125.566786bp}}
\pgfpathcurveto{\pgfpoint{230.410952bp}{125.566786bp}}{\pgfpoint{227.879839bp}{128.205605bp}}{\pgfpoint{227.879839bp}{131.460751bp}}
\pgfpathcurveto{\pgfpoint{227.879839bp}{134.715899bp}}{\pgfpoint{230.410952bp}{137.354716bp}}{\pgfpoint{233.533236bp}{137.354716bp}}
\pgfpathcurveto{\pgfpoint{236.655521bp}{137.354716bp}}{\pgfpoint{239.186628bp}{134.715899bp}}{\pgfpoint{239.186628bp}{131.460751bp}}
\pgfclosepath
\color[rgb]{0.0,0.0,0.0}
\pgfusepath{stroke}
\end{pgfscope}
\begin{pgfscope}
\pgfsetlinewidth{1.0bp}
\pgfsetrectcap 
\pgfsetmiterjoin \pgfsetmiterlimit{10.0}
\pgfpathmoveto{\pgfpoint{156.851628bp}{46.650751bp}}
\pgfpathcurveto{\pgfpoint{156.851628bp}{43.395605bp}}{\pgfpoint{154.320521bp}{40.756786bp}}{\pgfpoint{151.198236bp}{40.756786bp}}
\pgfpathcurveto{\pgfpoint{148.075952bp}{40.756786bp}}{\pgfpoint{145.544839bp}{43.395605bp}}{\pgfpoint{145.544839bp}{46.650751bp}}
\pgfpathcurveto{\pgfpoint{145.544839bp}{49.905899bp}}{\pgfpoint{148.075952bp}{52.544716bp}}{\pgfpoint{151.198236bp}{52.544716bp}}
\pgfpathcurveto{\pgfpoint{154.320521bp}{52.544716bp}}{\pgfpoint{156.851628bp}{49.905899bp}}{\pgfpoint{156.851628bp}{46.650751bp}}
\pgfclosepath
\color[rgb]{0.0,0.0,0.0}\pgfseteorule\pgfusepath{fill}
\pgfpathmoveto{\pgfpoint{156.851628bp}{46.650751bp}}
\pgfpathcurveto{\pgfpoint{156.851628bp}{43.395605bp}}{\pgfpoint{154.320521bp}{40.756786bp}}{\pgfpoint{151.198236bp}{40.756786bp}}
\pgfpathcurveto{\pgfpoint{148.075952bp}{40.756786bp}}{\pgfpoint{145.544839bp}{43.395605bp}}{\pgfpoint{145.544839bp}{46.650751bp}}
\pgfpathcurveto{\pgfpoint{145.544839bp}{49.905899bp}}{\pgfpoint{148.075952bp}{52.544716bp}}{\pgfpoint{151.198236bp}{52.544716bp}}
\pgfpathcurveto{\pgfpoint{154.320521bp}{52.544716bp}}{\pgfpoint{156.851628bp}{49.905899bp}}{\pgfpoint{156.851628bp}{46.650751bp}}
\pgfclosepath
\color[rgb]{0.0,0.0,0.0}
\pgfusepath{stroke}
\end{pgfscope}
\begin{pgfscope}
\pgfsetlinewidth{1.0bp}
\pgfsetrectcap 
\pgfsetmiterjoin \pgfsetmiterlimit{10.0}
\pgfpathmoveto{\pgfpoint{151.044842bp}{159.244717bp}}
\pgfpathcurveto{\pgfpoint{138.504842bp}{158.089717bp}}{\pgfpoint{139.824842bp}{147.364717bp}}{\pgfpoint{151.209842bp}{130.872512bp}}
\pgfpathcurveto{\pgfpoint{162.594842bp}{147.364717bp}}{\pgfpoint{163.914842bp}{158.089717bp}}{\pgfpoint{151.374842bp}{159.244717bp}}
\color[rgb]{0.0,0.0,0.0}
\pgfusepath{stroke}
\end{pgfscope}
\begin{pgfscope}
\pgfsetlinewidth{1.0bp}
\pgfsetrectcap 
\pgfsetmiterjoin \pgfsetmiterlimit{10.0}
\pgfpathmoveto{\pgfpoint{156.741628bp}{131.570751bp}}
\pgfpathcurveto{\pgfpoint{156.741628bp}{128.315605bp}}{\pgfpoint{154.210521bp}{125.676786bp}}{\pgfpoint{151.088236bp}{125.676786bp}}
\pgfpathcurveto{\pgfpoint{147.965952bp}{125.676786bp}}{\pgfpoint{145.434839bp}{128.315605bp}}{\pgfpoint{145.434839bp}{131.570751bp}}
\pgfpathcurveto{\pgfpoint{145.434839bp}{134.825899bp}}{\pgfpoint{147.965952bp}{137.464716bp}}{\pgfpoint{151.088236bp}{137.464716bp}}
\pgfpathcurveto{\pgfpoint{154.210521bp}{137.464716bp}}{\pgfpoint{156.741628bp}{134.825899bp}}{\pgfpoint{156.741628bp}{131.570751bp}}
\pgfclosepath
\color[rgb]{0.0,0.0,0.0}\pgfseteorule\pgfusepath{fill}
\pgfpathmoveto{\pgfpoint{156.741628bp}{131.570751bp}}
\pgfpathcurveto{\pgfpoint{156.741628bp}{128.315605bp}}{\pgfpoint{154.210521bp}{125.676786bp}}{\pgfpoint{151.088236bp}{125.676786bp}}
\pgfpathcurveto{\pgfpoint{147.965952bp}{125.676786bp}}{\pgfpoint{145.434839bp}{128.315605bp}}{\pgfpoint{145.434839bp}{131.570751bp}}
\pgfpathcurveto{\pgfpoint{145.434839bp}{134.825899bp}}{\pgfpoint{147.965952bp}{137.464716bp}}{\pgfpoint{151.088236bp}{137.464716bp}}
\pgfpathcurveto{\pgfpoint{154.210521bp}{137.464716bp}}{\pgfpoint{156.741628bp}{134.825899bp}}{\pgfpoint{156.741628bp}{131.570751bp}}
\pgfclosepath
\color[rgb]{0.0,0.0,0.0}
\pgfusepath{stroke}
\end{pgfscope}
\begin{pgfscope}
\pgfsetlinewidth{1.0bp}
\pgfsetrectcap 
\pgfsetmiterjoin \pgfsetmiterlimit{10.0}
\pgfpathmoveto{\pgfpoint{239.241628bp}{46.595751bp}}
\pgfpathcurveto{\pgfpoint{239.241628bp}{43.340605bp}}{\pgfpoint{236.710521bp}{40.701786bp}}{\pgfpoint{233.588236bp}{40.701786bp}}
\pgfpathcurveto{\pgfpoint{230.465952bp}{40.701786bp}}{\pgfpoint{227.934839bp}{43.340605bp}}{\pgfpoint{227.934839bp}{46.595751bp}}
\pgfpathcurveto{\pgfpoint{227.934839bp}{49.850899bp}}{\pgfpoint{230.465952bp}{52.489716bp}}{\pgfpoint{233.588236bp}{52.489716bp}}
\pgfpathcurveto{\pgfpoint{236.710521bp}{52.489716bp}}{\pgfpoint{239.241628bp}{49.850899bp}}{\pgfpoint{239.241628bp}{46.595751bp}}
\pgfclosepath
\color[rgb]{0.0,0.0,0.0}\pgfseteorule\pgfusepath{fill}
\pgfpathmoveto{\pgfpoint{239.241628bp}{46.595751bp}}
\pgfpathcurveto{\pgfpoint{239.241628bp}{43.340605bp}}{\pgfpoint{236.710521bp}{40.701786bp}}{\pgfpoint{233.588236bp}{40.701786bp}}
\pgfpathcurveto{\pgfpoint{230.465952bp}{40.701786bp}}{\pgfpoint{227.934839bp}{43.340605bp}}{\pgfpoint{227.934839bp}{46.595751bp}}
\pgfpathcurveto{\pgfpoint{227.934839bp}{49.850899bp}}{\pgfpoint{230.465952bp}{52.489716bp}}{\pgfpoint{233.588236bp}{52.489716bp}}
\pgfpathcurveto{\pgfpoint{236.710521bp}{52.489716bp}}{\pgfpoint{239.241628bp}{49.850899bp}}{\pgfpoint{239.241628bp}{46.595751bp}}
\pgfclosepath
\color[rgb]{0.0,0.0,0.0}
\pgfusepath{stroke}
\end{pgfscope}
\begin{pgfscope}
\pgfsetlinewidth{1.0bp}
\pgfsetrectcap 
\pgfsetmiterjoin \pgfsetmiterlimit{10.0}
\pgfpathmoveto{\pgfpoint{233.819842bp}{130.424717bp}}
\pgfpathlineto{\pgfpoint{150.989842bp}{130.424717bp}}
\color[rgb]{0.0,0.0,0.0}
\pgfusepath{stroke}
\end{pgfscope}
\begin{pgfscope}
\pgfsetlinewidth{1.0bp}
\pgfsetrectcap 
\pgfsetmiterjoin \pgfsetmiterlimit{10.0}
\pgfpathmoveto{\pgfpoint{151.154842bp}{130.424717bp}}
\pgfpathlineto{\pgfpoint{151.154842bp}{47.264717bp}}
\color[rgb]{0.0,0.0,0.0}
\pgfusepath{stroke}
\end{pgfscope}
\begin{pgfscope}
\pgfsetlinewidth{1.0bp}
\pgfsetrectcap 
\pgfsetmiterjoin \pgfsetmiterlimit{10.0}
\pgfpathmoveto{\pgfpoint{151.154842bp}{47.594717bp}}
\pgfpathlineto{\pgfpoint{234.314842bp}{47.594717bp}}
\color[rgb]{0.0,0.0,0.0}
\pgfusepath{stroke}
\end{pgfscope}
\begin{pgfscope}
\pgfsetlinewidth{1.0bp}
\pgfsetrectcap 
\pgfsetmiterjoin \pgfsetmiterlimit{10.0}
\pgfpathmoveto{\pgfpoint{233.984842bp}{130.424717bp}}
\pgfpathlineto{\pgfpoint{233.984842bp}{47.429717bp}}
\color[rgb]{0.0,0.0,0.0}
\pgfusepath{stroke}
\end{pgfscope}
\begin{pgfscope}
\pgfsetlinewidth{1.0bp}
\pgfsetrectcap 
\pgfsetmiterjoin \pgfsetmiterlimit{10.0}
\pgfpathmoveto{\pgfpoint{371.099842bp}{159.134717bp}}
\pgfpathcurveto{\pgfpoint{358.559842bp}{157.979717bp}}{\pgfpoint{359.879842bp}{147.254717bp}}{\pgfpoint{371.264842bp}{130.762512bp}}
\pgfpathcurveto{\pgfpoint{382.649842bp}{147.254717bp}}{\pgfpoint{383.969842bp}{157.979717bp}}{\pgfpoint{371.429842bp}{159.134717bp}}
\color[rgb]{0.0,0.0,0.0}
\pgfusepath{stroke}
\end{pgfscope}
\begin{pgfscope}
\pgfsetlinewidth{1.0bp}
\pgfsetrectcap 
\pgfsetmiterjoin \pgfsetmiterlimit{10.0}
\pgfpathmoveto{\pgfpoint{376.796628bp}{131.460751bp}}
\pgfpathcurveto{\pgfpoint{376.796628bp}{128.205605bp}}{\pgfpoint{374.265521bp}{125.566786bp}}{\pgfpoint{371.143236bp}{125.566786bp}}
\pgfpathcurveto{\pgfpoint{368.020952bp}{125.566786bp}}{\pgfpoint{365.489839bp}{128.205605bp}}{\pgfpoint{365.489839bp}{131.460751bp}}
\pgfpathcurveto{\pgfpoint{365.489839bp}{134.715899bp}}{\pgfpoint{368.020952bp}{137.354716bp}}{\pgfpoint{371.143236bp}{137.354716bp}}
\pgfpathcurveto{\pgfpoint{374.265521bp}{137.354716bp}}{\pgfpoint{376.796628bp}{134.715899bp}}{\pgfpoint{376.796628bp}{131.460751bp}}
\pgfclosepath
\color[rgb]{0.0,0.0,0.0}\pgfseteorule\pgfusepath{fill}
\pgfpathmoveto{\pgfpoint{376.796628bp}{131.460751bp}}
\pgfpathcurveto{\pgfpoint{376.796628bp}{128.205605bp}}{\pgfpoint{374.265521bp}{125.566786bp}}{\pgfpoint{371.143236bp}{125.566786bp}}
\pgfpathcurveto{\pgfpoint{368.020952bp}{125.566786bp}}{\pgfpoint{365.489839bp}{128.205605bp}}{\pgfpoint{365.489839bp}{131.460751bp}}
\pgfpathcurveto{\pgfpoint{365.489839bp}{134.715899bp}}{\pgfpoint{368.020952bp}{137.354716bp}}{\pgfpoint{371.143236bp}{137.354716bp}}
\pgfpathcurveto{\pgfpoint{374.265521bp}{137.354716bp}}{\pgfpoint{376.796628bp}{134.715899bp}}{\pgfpoint{376.796628bp}{131.460751bp}}
\pgfclosepath
\color[rgb]{0.0,0.0,0.0}
\pgfusepath{stroke}
\end{pgfscope}
\begin{pgfscope}
\pgfsetlinewidth{1.0bp}
\pgfsetrectcap 
\pgfsetmiterjoin \pgfsetmiterlimit{10.0}
\pgfpathmoveto{\pgfpoint{294.461628bp}{46.650751bp}}
\pgfpathcurveto{\pgfpoint{294.461628bp}{43.395605bp}}{\pgfpoint{291.930521bp}{40.756786bp}}{\pgfpoint{288.808236bp}{40.756786bp}}
\pgfpathcurveto{\pgfpoint{285.685952bp}{40.756786bp}}{\pgfpoint{283.154839bp}{43.395605bp}}{\pgfpoint{283.154839bp}{46.650751bp}}
\pgfpathcurveto{\pgfpoint{283.154839bp}{49.905899bp}}{\pgfpoint{285.685952bp}{52.544716bp}}{\pgfpoint{288.808236bp}{52.544716bp}}
\pgfpathcurveto{\pgfpoint{291.930521bp}{52.544716bp}}{\pgfpoint{294.461628bp}{49.905899bp}}{\pgfpoint{294.461628bp}{46.650751bp}}
\pgfclosepath
\color[rgb]{0.0,0.0,0.0}\pgfseteorule\pgfusepath{fill}
\pgfpathmoveto{\pgfpoint{294.461628bp}{46.650751bp}}
\pgfpathcurveto{\pgfpoint{294.461628bp}{43.395605bp}}{\pgfpoint{291.930521bp}{40.756786bp}}{\pgfpoint{288.808236bp}{40.756786bp}}
\pgfpathcurveto{\pgfpoint{285.685952bp}{40.756786bp}}{\pgfpoint{283.154839bp}{43.395605bp}}{\pgfpoint{283.154839bp}{46.650751bp}}
\pgfpathcurveto{\pgfpoint{283.154839bp}{49.905899bp}}{\pgfpoint{285.685952bp}{52.544716bp}}{\pgfpoint{288.808236bp}{52.544716bp}}
\pgfpathcurveto{\pgfpoint{291.930521bp}{52.544716bp}}{\pgfpoint{294.461628bp}{49.905899bp}}{\pgfpoint{294.461628bp}{46.650751bp}}
\pgfclosepath
\color[rgb]{0.0,0.0,0.0}
\pgfusepath{stroke}
\end{pgfscope}
\begin{pgfscope}
\pgfsetlinewidth{1.0bp}
\pgfsetrectcap 
\pgfsetmiterjoin \pgfsetmiterlimit{10.0}
\pgfpathmoveto{\pgfpoint{288.654842bp}{159.244717bp}}
\pgfpathcurveto{\pgfpoint{276.114842bp}{158.089717bp}}{\pgfpoint{277.434842bp}{147.364717bp}}{\pgfpoint{288.819842bp}{130.872512bp}}
\pgfpathcurveto{\pgfpoint{300.204842bp}{147.364717bp}}{\pgfpoint{301.524842bp}{158.089717bp}}{\pgfpoint{288.984842bp}{159.244717bp}}
\color[rgb]{0.0,0.0,0.0}
\pgfusepath{stroke}
\end{pgfscope}
\begin{pgfscope}
\pgfsetlinewidth{1.0bp}
\pgfsetrectcap 
\pgfsetmiterjoin \pgfsetmiterlimit{10.0}
\pgfpathmoveto{\pgfpoint{294.351628bp}{131.570751bp}}
\pgfpathcurveto{\pgfpoint{294.351628bp}{128.315605bp}}{\pgfpoint{291.820521bp}{125.676786bp}}{\pgfpoint{288.698236bp}{125.676786bp}}
\pgfpathcurveto{\pgfpoint{285.575952bp}{125.676786bp}}{\pgfpoint{283.044839bp}{128.315605bp}}{\pgfpoint{283.044839bp}{131.570751bp}}
\pgfpathcurveto{\pgfpoint{283.044839bp}{134.825899bp}}{\pgfpoint{285.575952bp}{137.464716bp}}{\pgfpoint{288.698236bp}{137.464716bp}}
\pgfpathcurveto{\pgfpoint{291.820521bp}{137.464716bp}}{\pgfpoint{294.351628bp}{134.825899bp}}{\pgfpoint{294.351628bp}{131.570751bp}}
\pgfclosepath
\color[rgb]{0.0,0.0,0.0}\pgfseteorule\pgfusepath{fill}
\pgfpathmoveto{\pgfpoint{294.351628bp}{131.570751bp}}
\pgfpathcurveto{\pgfpoint{294.351628bp}{128.315605bp}}{\pgfpoint{291.820521bp}{125.676786bp}}{\pgfpoint{288.698236bp}{125.676786bp}}
\pgfpathcurveto{\pgfpoint{285.575952bp}{125.676786bp}}{\pgfpoint{283.044839bp}{128.315605bp}}{\pgfpoint{283.044839bp}{131.570751bp}}
\pgfpathcurveto{\pgfpoint{283.044839bp}{134.825899bp}}{\pgfpoint{285.575952bp}{137.464716bp}}{\pgfpoint{288.698236bp}{137.464716bp}}
\pgfpathcurveto{\pgfpoint{291.820521bp}{137.464716bp}}{\pgfpoint{294.351628bp}{134.825899bp}}{\pgfpoint{294.351628bp}{131.570751bp}}
\pgfclosepath
\color[rgb]{0.0,0.0,0.0}
\pgfusepath{stroke}
\end{pgfscope}
\begin{pgfscope}
\pgfsetlinewidth{1.0bp}
\pgfsetrectcap 
\pgfsetmiterjoin \pgfsetmiterlimit{10.0}
\pgfpathmoveto{\pgfpoint{371.429842bp}{130.424717bp}}
\pgfpathlineto{\pgfpoint{288.599842bp}{130.424717bp}}
\color[rgb]{0.0,0.0,0.0}
\pgfusepath{stroke}
\end{pgfscope}
\begin{pgfscope}
\pgfsetlinewidth{1.0bp}
\pgfsetrectcap 
\pgfsetmiterjoin \pgfsetmiterlimit{10.0}
\pgfpathmoveto{\pgfpoint{288.764842bp}{130.424717bp}}
\pgfpathlineto{\pgfpoint{288.764842bp}{47.264717bp}}
\color[rgb]{0.0,0.0,0.0}
\pgfusepath{stroke}
\end{pgfscope}
\begin{pgfscope}
\pgfsetlinewidth{1.0bp}
\pgfsetrectcap 
\pgfsetmiterjoin \pgfsetmiterlimit{10.0}
\pgfpathmoveto{\pgfpoint{288.764834bp}{130.424716bp}}
\pgfpathlineto{\pgfpoint{330.014834bp}{89.174716bp}}
\color[rgb]{0.0,0.0,0.0}
\pgfusepath{stroke}
\end{pgfscope}
\begin{pgfscope}
\pgfsetlinewidth{1.0bp}
\pgfsetrectcap 
\pgfsetmiterjoin \pgfsetmiterlimit{10.0}
\pgfpathmoveto{\pgfpoint{372.304894bp}{129.549673bp}}
\pgfpathlineto{\pgfpoint{289.804894bp}{46.884673bp}}
\color[rgb]{0.0,0.0,0.0}
\pgfusepath{stroke}
\end{pgfscope}
\begin{pgfscope}
\pgfsetlinewidth{1.0bp}
\pgfsetrectcap 
\pgfsetmiterjoin \pgfsetmiterlimit{10.0}
\pgfpathmoveto{\pgfpoint{336.591628bp}{88.340751bp}}
\pgfpathcurveto{\pgfpoint{336.591628bp}{85.085605bp}}{\pgfpoint{334.060521bp}{82.446786bp}}{\pgfpoint{330.938236bp}{82.446786bp}}
\pgfpathcurveto{\pgfpoint{327.815952bp}{82.446786bp}}{\pgfpoint{325.284839bp}{85.085605bp}}{\pgfpoint{325.284839bp}{88.340751bp}}
\pgfpathcurveto{\pgfpoint{325.284839bp}{91.595899bp}}{\pgfpoint{327.815952bp}{94.234716bp}}{\pgfpoint{330.938236bp}{94.234716bp}}
\pgfpathcurveto{\pgfpoint{334.060521bp}{94.234716bp}}{\pgfpoint{336.591628bp}{91.595899bp}}{\pgfpoint{336.591628bp}{88.340751bp}}
\pgfclosepath
\color[rgb]{0.0,0.0,0.0}\pgfseteorule\pgfusepath{fill}
\pgfpathmoveto{\pgfpoint{336.591628bp}{88.340751bp}}
\pgfpathcurveto{\pgfpoint{336.591628bp}{85.085605bp}}{\pgfpoint{334.060521bp}{82.446786bp}}{\pgfpoint{330.938236bp}{82.446786bp}}
\pgfpathcurveto{\pgfpoint{327.815952bp}{82.446786bp}}{\pgfpoint{325.284839bp}{85.085605bp}}{\pgfpoint{325.284839bp}{88.340751bp}}
\pgfpathcurveto{\pgfpoint{325.284839bp}{91.595899bp}}{\pgfpoint{327.815952bp}{94.234716bp}}{\pgfpoint{330.938236bp}{94.234716bp}}
\pgfpathcurveto{\pgfpoint{334.060521bp}{94.234716bp}}{\pgfpoint{336.591628bp}{91.595899bp}}{\pgfpoint{336.591628bp}{88.340751bp}}
\pgfclosepath
\color[rgb]{0.0,0.0,0.0}
\pgfusepath{stroke}
\end{pgfscope}
\begin{pgfscope}
\pgftransformcm{1.0}{0.0}{0.0}{1.0}{\pgfpoint{46.858722bp}{98.187677bp}}
\pgftext[left,base]{\sffamily\mdseries\upshape\large
\color[rgb]{0.0,0.0,0.0}$v_1$}
\end{pgfscope}
\begin{pgfscope}
\pgftransformcm{1.0}{0.0}{0.0}{1.0}{\pgfpoint{18.858722bp}{135.187677bp}}
\pgftext[left,base]{\sffamily\mdseries\upshape\large
\color[rgb]{0.0,0.0,0.0}$v_2$}
\end{pgfscope}
\begin{pgfscope}
\pgftransformcm{1.0}{0.0}{0.0}{1.0}{\pgfpoint{72.658722bp}{134.787677bp}}
\pgftext[left,base]{\sffamily\mdseries\upshape\large
\color[rgb]{0.0,0.0,0.0}$v_3$}
\end{pgfscope}
\begin{pgfscope}
\pgftransformcm{1.0}{0.0}{0.0}{1.0}{\pgfpoint{81.058722bp}{62.187677bp}}
\pgftext[left,base]{\sffamily\mdseries\upshape\large
\color[rgb]{0.0,0.0,0.0}$v_4$}
\end{pgfscope}
\begin{pgfscope}
\pgftransformcm{1.0}{0.0}{0.0}{1.0}{\pgfpoint{11.058722bp}{61.987677bp}}
\pgftext[left,base]{\sffamily\mdseries\upshape\large
\color[rgb]{0.0,0.0,0.0}$v_5$}
\end{pgfscope}
\end{pgfpicture}

%% file: Hexample4.tex
% C:\Users\agalanis\Dropbox\H-colorings\ICALP\H-colorings\Hexample4.jdr
% Created by Jpgfdraw version 0.5.6b
% 05-Feb-2015 15:30:58
\begin{pgfpicture}{0bp}{0bp}{260.042877bp}{132.999985bp}
\begin{pgfscope}
\pgftransformcm{1.0}{0.0}{0.0}{1.0}{\pgfpoint{114.149994bp}{41.1bp}}
\pgftext[left,base]{\sffamily\mdseries\upshape\huge
\color[rgb]{0.0,0.0,0.0}$w$}
\end{pgfscope}
\begin{pgfscope}
\pgftransformcm{1.0}{0.0}{0.0}{1.0}{\pgfpoint{18.949994bp}{58.3bp}}
\pgftext[left,base]{\sffamily\mdseries\upshape\huge
\color[rgb]{0.0,0.0,0.0}$G_1$}
\end{pgfscope}
\begin{pgfscope}
\pgftransformcm{1.0}{0.0}{0.0}{1.0}{\pgfpoint{114.949994bp}{4.3bp}}
\pgftext[left,base]{\sffamily\mdseries\upshape\huge
\color[rgb]{0.0,0.0,0.0}$G$}
\end{pgfscope}
\begin{pgfscope}
\pgfsetlinewidth{1.0bp}
\pgfsetrectcap 
\pgfsetmiterjoin \pgfsetmiterlimit{10.0}
\pgfpathmoveto{\pgfpoint{116.349984bp}{60.636926bp}}
\pgfpathcurveto{\pgfpoint{116.349984bp}{63.892073bp}}{\pgfpoint{118.881091bp}{66.530892bp}}{\pgfpoint{122.003375bp}{66.530892bp}}
\pgfpathcurveto{\pgfpoint{125.125659bp}{66.530892bp}}{\pgfpoint{127.656772bp}{63.892073bp}}{\pgfpoint{127.656772bp}{60.636926bp}}
\pgfpathcurveto{\pgfpoint{127.656772bp}{57.381779bp}}{\pgfpoint{125.125659bp}{54.742961bp}}{\pgfpoint{122.003375bp}{54.742961bp}}
\pgfpathcurveto{\pgfpoint{118.881091bp}{54.742961bp}}{\pgfpoint{116.349984bp}{57.381779bp}}{\pgfpoint{116.349984bp}{60.636926bp}}
\pgfclosepath
\color[rgb]{0.0,0.0,0.0}\pgfseteorule\pgfusepath{fill}
\pgfpathmoveto{\pgfpoint{116.349984bp}{60.636926bp}}
\pgfpathcurveto{\pgfpoint{116.349984bp}{63.892073bp}}{\pgfpoint{118.881091bp}{66.530892bp}}{\pgfpoint{122.003375bp}{66.530892bp}}
\pgfpathcurveto{\pgfpoint{125.125659bp}{66.530892bp}}{\pgfpoint{127.656772bp}{63.892073bp}}{\pgfpoint{127.656772bp}{60.636926bp}}
\pgfpathcurveto{\pgfpoint{127.656772bp}{57.381779bp}}{\pgfpoint{125.125659bp}{54.742961bp}}{\pgfpoint{122.003375bp}{54.742961bp}}
\pgfpathcurveto{\pgfpoint{118.881091bp}{54.742961bp}}{\pgfpoint{116.349984bp}{57.381779bp}}{\pgfpoint{116.349984bp}{60.636926bp}}
\pgfclosepath
\color[rgb]{0.0,0.0,0.0}
\pgfusepath{stroke}
\end{pgfscope}
\begin{pgfscope}
\pgfsetlinewidth{1.0bp}
\pgfsetrectcap 
\pgfsetmiterjoin \pgfsetmiterlimit{10.0}
\pgfpathmoveto{\pgfpoint{122.742889bp}{60.5bp}}
\pgfpathlineto{\pgfpoint{236.742889bp}{20.9bp}}
\color[rgb]{0.0,0.0,0.0}
\pgfusepath{stroke}
\end{pgfscope}
\begin{pgfscope}
\pgfsetlinewidth{1.0bp}
\pgfsetrectcap 
\pgfsetmiterjoin \pgfsetmiterlimit{10.0}
\pgfpathmoveto{\pgfpoint{124.542889bp}{61.7bp}}
\pgfpathlineto{\pgfpoint{225.942889bp}{121.7bp}}
\color[rgb]{0.0,0.0,0.0}
\pgfusepath{stroke}
\end{pgfscope}
\begin{pgfscope}
\pgfsetlinewidth{1.0bp}
\pgfsetrectcap 
\pgfsetmiterjoin \pgfsetmiterlimit{10.0}
\pgfpathmoveto{\pgfpoint{122.742889bp}{61.1bp}}
\pgfpathlineto{\pgfpoint{216.342889bp}{100.7bp}}
\color[rgb]{0.0,0.0,0.0}
\pgfusepath{stroke}
\end{pgfscope}
\begin{pgfscope}
\pgfsetlinewidth{1.0bp}
\pgfsetrectcap 
\pgfsetmiterjoin \pgfsetmiterlimit{10.0}
\pgfpathmoveto{\pgfpoint{235.253377bp}{20.639948bp}}
\pgfpathcurveto{\pgfpoint{235.253377bp}{22.267521bp}}{\pgfpoint{236.51893bp}{23.58693bp}}{\pgfpoint{238.080073bp}{23.58693bp}}
\pgfpathcurveto{\pgfpoint{239.641215bp}{23.58693bp}}{\pgfpoint{240.906771bp}{22.267521bp}}{\pgfpoint{240.906771bp}{20.639948bp}}
\pgfpathcurveto{\pgfpoint{240.906771bp}{19.012374bp}}{\pgfpoint{239.641215bp}{17.692965bp}}{\pgfpoint{238.080073bp}{17.692965bp}}
\pgfpathcurveto{\pgfpoint{236.51893bp}{17.692965bp}}{\pgfpoint{235.253377bp}{19.012374bp}}{\pgfpoint{235.253377bp}{20.639948bp}}
\pgfclosepath
\color[rgb]{0.0,0.0,0.0}\pgfseteorule\pgfusepath{fill}
\pgfpathmoveto{\pgfpoint{235.253377bp}{20.639948bp}}
\pgfpathcurveto{\pgfpoint{235.253377bp}{22.267521bp}}{\pgfpoint{236.51893bp}{23.58693bp}}{\pgfpoint{238.080073bp}{23.58693bp}}
\pgfpathcurveto{\pgfpoint{239.641215bp}{23.58693bp}}{\pgfpoint{240.906771bp}{22.267521bp}}{\pgfpoint{240.906771bp}{20.639948bp}}
\pgfpathcurveto{\pgfpoint{240.906771bp}{19.012374bp}}{\pgfpoint{239.641215bp}{17.692965bp}}{\pgfpoint{238.080073bp}{17.692965bp}}
\pgfpathcurveto{\pgfpoint{236.51893bp}{17.692965bp}}{\pgfpoint{235.253377bp}{19.012374bp}}{\pgfpoint{235.253377bp}{20.639948bp}}
\pgfclosepath
\color[rgb]{0.0,0.0,0.0}
\pgfusepath{stroke}
\end{pgfscope}
\begin{pgfscope}
\pgfsetlinewidth{1.0bp}
\pgfsetrectcap 
\pgfsetmiterjoin \pgfsetmiterlimit{10.0}
\pgfpathmoveto{\pgfpoint{214.253377bp}{101.039948bp}}
\pgfpathcurveto{\pgfpoint{214.253377bp}{102.667521bp}}{\pgfpoint{215.51893bp}{103.98693bp}}{\pgfpoint{217.080073bp}{103.98693bp}}
\pgfpathcurveto{\pgfpoint{218.641215bp}{103.98693bp}}{\pgfpoint{219.906771bp}{102.667521bp}}{\pgfpoint{219.906771bp}{101.039948bp}}
\pgfpathcurveto{\pgfpoint{219.906771bp}{99.412374bp}}{\pgfpoint{218.641215bp}{98.092965bp}}{\pgfpoint{217.080073bp}{98.092965bp}}
\pgfpathcurveto{\pgfpoint{215.51893bp}{98.092965bp}}{\pgfpoint{214.253377bp}{99.412374bp}}{\pgfpoint{214.253377bp}{101.039948bp}}
\pgfclosepath
\color[rgb]{0.0,0.0,0.0}\pgfseteorule\pgfusepath{fill}
\pgfpathmoveto{\pgfpoint{214.253377bp}{101.039948bp}}
\pgfpathcurveto{\pgfpoint{214.253377bp}{102.667521bp}}{\pgfpoint{215.51893bp}{103.98693bp}}{\pgfpoint{217.080073bp}{103.98693bp}}
\pgfpathcurveto{\pgfpoint{218.641215bp}{103.98693bp}}{\pgfpoint{219.906771bp}{102.667521bp}}{\pgfpoint{219.906771bp}{101.039948bp}}
\pgfpathcurveto{\pgfpoint{219.906771bp}{99.412374bp}}{\pgfpoint{218.641215bp}{98.092965bp}}{\pgfpoint{217.080073bp}{98.092965bp}}
\pgfpathcurveto{\pgfpoint{215.51893bp}{98.092965bp}}{\pgfpoint{214.253377bp}{99.412374bp}}{\pgfpoint{214.253377bp}{101.039948bp}}
\pgfclosepath
\color[rgb]{0.0,0.0,0.0}
\pgfusepath{stroke}
\end{pgfscope}
\begin{pgfscope}
\pgfsetlinewidth{1.0bp}
\pgfsetrectcap 
\pgfsetmiterjoin \pgfsetmiterlimit{10.0}
\pgfpathmoveto{\pgfpoint{223.253377bp}{122.039948bp}}
\pgfpathcurveto{\pgfpoint{223.253377bp}{123.667521bp}}{\pgfpoint{224.51893bp}{124.98693bp}}{\pgfpoint{226.080073bp}{124.98693bp}}
\pgfpathcurveto{\pgfpoint{227.641215bp}{124.98693bp}}{\pgfpoint{228.906771bp}{123.667521bp}}{\pgfpoint{228.906771bp}{122.039948bp}}
\pgfpathcurveto{\pgfpoint{228.906771bp}{120.412374bp}}{\pgfpoint{227.641215bp}{119.092965bp}}{\pgfpoint{226.080073bp}{119.092965bp}}
\pgfpathcurveto{\pgfpoint{224.51893bp}{119.092965bp}}{\pgfpoint{223.253377bp}{120.412374bp}}{\pgfpoint{223.253377bp}{122.039948bp}}
\pgfclosepath
\color[rgb]{0.0,0.0,0.0}\pgfseteorule\pgfusepath{fill}
\pgfpathmoveto{\pgfpoint{223.253377bp}{122.039948bp}}
\pgfpathcurveto{\pgfpoint{223.253377bp}{123.667521bp}}{\pgfpoint{224.51893bp}{124.98693bp}}{\pgfpoint{226.080073bp}{124.98693bp}}
\pgfpathcurveto{\pgfpoint{227.641215bp}{124.98693bp}}{\pgfpoint{228.906771bp}{123.667521bp}}{\pgfpoint{228.906771bp}{122.039948bp}}
\pgfpathcurveto{\pgfpoint{228.906771bp}{120.412374bp}}{\pgfpoint{227.641215bp}{119.092965bp}}{\pgfpoint{226.080073bp}{119.092965bp}}
\pgfpathcurveto{\pgfpoint{224.51893bp}{119.092965bp}}{\pgfpoint{223.253377bp}{120.412374bp}}{\pgfpoint{223.253377bp}{122.039948bp}}
\pgfclosepath
\color[rgb]{0.0,0.0,0.0}
\pgfusepath{stroke}
\end{pgfscope}
\begin{pgfscope}
\pgfsetlinewidth{1.0bp}
\pgfsetrectcap 
\pgfsetmiterjoin \pgfsetmiterlimit{10.0}
\pgfpathmoveto{\pgfpoint{122.142889bp}{58.7bp}}
\pgfpathlineto{\pgfpoint{209.142889bp}{17.3bp}}
\color[rgb]{0.0,0.0,0.0}
\pgfusepath{stroke}
\end{pgfscope}
\begin{pgfscope}
\pgfsetlinewidth{1.0bp}
\pgfsetrectcap 
\pgfsetmiterjoin \pgfsetmiterlimit{10.0}
\pgfpathmoveto{\pgfpoint{206.853377bp}{17.839948bp}}
\pgfpathcurveto{\pgfpoint{206.853377bp}{19.467521bp}}{\pgfpoint{208.11893bp}{20.78693bp}}{\pgfpoint{209.680073bp}{20.78693bp}}
\pgfpathcurveto{\pgfpoint{211.241215bp}{20.78693bp}}{\pgfpoint{212.506771bp}{19.467521bp}}{\pgfpoint{212.506771bp}{17.839948bp}}
\pgfpathcurveto{\pgfpoint{212.506771bp}{16.212374bp}}{\pgfpoint{211.241215bp}{14.892965bp}}{\pgfpoint{209.680073bp}{14.892965bp}}
\pgfpathcurveto{\pgfpoint{208.11893bp}{14.892965bp}}{\pgfpoint{206.853377bp}{16.212374bp}}{\pgfpoint{206.853377bp}{17.839948bp}}
\pgfclosepath
\color[rgb]{0.0,0.0,0.0}\pgfseteorule\pgfusepath{fill}
\pgfpathmoveto{\pgfpoint{206.853377bp}{17.839948bp}}
\pgfpathcurveto{\pgfpoint{206.853377bp}{19.467521bp}}{\pgfpoint{208.11893bp}{20.78693bp}}{\pgfpoint{209.680073bp}{20.78693bp}}
\pgfpathcurveto{\pgfpoint{211.241215bp}{20.78693bp}}{\pgfpoint{212.506771bp}{19.467521bp}}{\pgfpoint{212.506771bp}{17.839948bp}}
\pgfpathcurveto{\pgfpoint{212.506771bp}{16.212374bp}}{\pgfpoint{211.241215bp}{14.892965bp}}{\pgfpoint{209.680073bp}{14.892965bp}}
\pgfpathcurveto{\pgfpoint{208.11893bp}{14.892965bp}}{\pgfpoint{206.853377bp}{16.212374bp}}{\pgfpoint{206.853377bp}{17.839948bp}}
\pgfclosepath
\color[rgb]{0.0,0.0,0.0}
\pgfusepath{stroke}
\end{pgfscope}
\begin{pgfscope}
\pgftransformcm{-0.0}{1.0}{-1.0}{-0.0}{\pgfpoint{165.342889bp}{48.5bp}}
\pgftext[left,base]{\sffamily\mdseries\upshape\fontsize{30.1125}{30.1125}\selectfont
\color[rgb]{0.0,0.0,0.0}$\hdots$}
\end{pgfscope}
\begin{pgfscope}
\pgfsetlinewidth{1.0bp}
\pgfsetrectcap 
\pgfsetmiterjoin \pgfsetmiterlimit{10.0}
\pgfpathmoveto{\pgfpoint{187.542891bp}{66.500003bp}}
\pgfpathcurveto{\pgfpoint{187.542891bp}{102.950793bp}}{\pgfpoint{203.660639bp}{132.499988bp}}{\pgfpoint{223.542891bp}{132.499988bp}}
\pgfpathcurveto{\pgfpoint{243.425141bp}{132.499988bp}}{\pgfpoint{259.542889bp}{102.950793bp}}{\pgfpoint{259.542889bp}{66.500003bp}}
\pgfpathcurveto{\pgfpoint{259.542889bp}{30.049213bp}}{\pgfpoint{243.425141bp}{0.500003bp}}{\pgfpoint{223.542891bp}{0.500003bp}}
\pgfpathcurveto{\pgfpoint{203.660639bp}{0.500003bp}}{\pgfpoint{187.542891bp}{30.049213bp}}{\pgfpoint{187.542891bp}{66.500003bp}}
\pgfclosepath
\color[rgb]{0.0,0.0,0.0}
\pgfusepath{stroke}
\end{pgfscope}
\begin{pgfscope}
\pgfsetlinewidth{1.0bp}
\pgfsetrectcap 
\pgfsetmiterjoin \pgfsetmiterlimit{10.0}
\pgfsetdash{{2.0bp}{3.0bp}}{2.0bp}
\pgfpathmoveto{\pgfpoint{209.742889bp}{17.9bp}}
\pgfpathlineto{\pgfpoint{204.342889bp}{41.9bp}}
\pgfpathlineto{\pgfpoint{213.342889bp}{48.5bp}}
\color[rgb]{0.0,0.0,0.0}
\pgfusepath{stroke}
\end{pgfscope}
\begin{pgfscope}
\pgfsetlinewidth{1.0bp}
\pgfsetrectcap 
\pgfsetmiterjoin \pgfsetmiterlimit{10.0}
\pgfsetdash{{2.0bp}{3.0bp}}{2.0bp}
\pgfpathmoveto{\pgfpoint{237.342889bp}{20.9bp}}
\pgfpathlineto{\pgfpoint{249.942889bp}{39.5bp}}
\color[rgb]{0.0,0.0,0.0}
\pgfusepath{stroke}
\end{pgfscope}
\begin{pgfscope}
\pgfsetlinewidth{1.0bp}
\pgfsetrectcap 
\pgfsetmiterjoin \pgfsetmiterlimit{10.0}
\pgfsetdash{{2.0bp}{3.0bp}}{2.0bp}
\pgfpathmoveto{\pgfpoint{237.342889bp}{21.5bp}}
\pgfpathlineto{\pgfpoint{227.742889bp}{46.7bp}}
\pgfpathlineto{\pgfpoint{227.742889bp}{46.7bp}}
\color[rgb]{0.0,0.0,0.0}
\pgfusepath{stroke}
\end{pgfscope}
\begin{pgfscope}
\pgfsetlinewidth{1.0bp}
\pgfsetrectcap 
\pgfsetmiterjoin \pgfsetmiterlimit{10.0}
\pgfsetdash{{2.0bp}{3.0bp}}{2.0bp}
\pgfpathmoveto{\pgfpoint{216.942889bp}{100.1bp}}
\pgfpathlineto{\pgfpoint{232.542889bp}{88.7bp}}
\color[rgb]{0.0,0.0,0.0}
\pgfusepath{stroke}
\end{pgfscope}
\begin{pgfscope}
\pgfsetlinewidth{1.0bp}
\pgfsetrectcap 
\pgfsetmiterjoin \pgfsetmiterlimit{10.0}
\pgfsetdash{{2.0bp}{3.0bp}}{2.0bp}
\pgfpathmoveto{\pgfpoint{218.142889bp}{102.5bp}}
\pgfpathlineto{\pgfpoint{226.542889bp}{122.3bp}}
\color[rgb]{0.0,0.0,0.0}
\pgfusepath{stroke}
\end{pgfscope}
\begin{pgfscope}
\pgfsetlinewidth{1.0bp}
\pgfsetrectcap 
\pgfsetmiterjoin \pgfsetmiterlimit{10.0}
\pgfpathmoveto{\pgfpoint{65.299993bp}{68.349997bp}}
\pgfpathcurveto{\pgfpoint{65.299993bp}{35.544286bp}}{\pgfpoint{50.794019bp}{8.950011bp}}{\pgfpoint{32.899993bp}{8.950011bp}}
\pgfpathcurveto{\pgfpoint{15.005968bp}{8.950011bp}}{\pgfpoint{0.499994bp}{35.544286bp}}{\pgfpoint{0.499994bp}{68.349997bp}}
\pgfpathcurveto{\pgfpoint{0.499994bp}{101.155708bp}}{\pgfpoint{15.005968bp}{127.749997bp}}{\pgfpoint{32.899993bp}{127.749997bp}}
\pgfpathcurveto{\pgfpoint{50.794019bp}{127.749997bp}}{\pgfpoint{65.299993bp}{101.155708bp}}{\pgfpoint{65.299993bp}{68.349997bp}}
\pgfclosepath
\color[rgb]{0.0,0.0,0.0}
\pgfusepath{stroke}
\end{pgfscope}
\begin{pgfscope}
\pgfsetlinewidth{1.0bp}
\pgfsetrectcap 
\pgfsetmiterjoin \pgfsetmiterlimit{10.0}
\pgfsetdash{{2.0bp}{3.0bp}}{2.0bp}
\pgfpathmoveto{\pgfpoint{45.319994bp}{112.09bp}}
\pgfpathlineto{\pgfpoint{50.179994bp}{90.49bp}}
\pgfpathlineto{\pgfpoint{42.079994bp}{84.55bp}}
\color[rgb]{0.0,0.0,0.0}
\pgfusepath{stroke}
\end{pgfscope}
\begin{pgfscope}
\pgfsetlinewidth{1.0bp}
\pgfsetrectcap 
\pgfsetmiterjoin \pgfsetmiterlimit{10.0}
\pgfsetdash{{2.0bp}{3.0bp}}{2.0bp}
\pgfpathmoveto{\pgfpoint{20.479994bp}{109.39bp}}
\pgfpathlineto{\pgfpoint{9.139994bp}{92.65bp}}
\color[rgb]{0.0,0.0,0.0}
\pgfusepath{stroke}
\end{pgfscope}
\begin{pgfscope}
\pgfsetlinewidth{1.0bp}
\pgfsetrectcap 
\pgfsetmiterjoin \pgfsetmiterlimit{10.0}
\pgfsetdash{{2.0bp}{3.0bp}}{2.0bp}
\pgfpathmoveto{\pgfpoint{20.479994bp}{108.85bp}}
\pgfpathlineto{\pgfpoint{29.119994bp}{86.17bp}}
\pgfpathlineto{\pgfpoint{29.119994bp}{86.17bp}}
\color[rgb]{0.0,0.0,0.0}
\pgfusepath{stroke}
\end{pgfscope}
\begin{pgfscope}
\pgfsetlinewidth{1.0bp}
\pgfsetrectcap 
\pgfsetmiterjoin \pgfsetmiterlimit{10.0}
\pgfsetdash{{2.0bp}{3.0bp}}{2.0bp}
\pgfpathmoveto{\pgfpoint{38.839994bp}{38.11bp}}
\pgfpathlineto{\pgfpoint{24.799994bp}{48.37bp}}
\color[rgb]{0.0,0.0,0.0}
\pgfusepath{stroke}
\end{pgfscope}
\begin{pgfscope}
\pgfsetlinewidth{1.0bp}
\pgfsetrectcap 
\pgfsetmiterjoin \pgfsetmiterlimit{10.0}
\pgfsetdash{{2.0bp}{3.0bp}}{2.0bp}
\pgfpathmoveto{\pgfpoint{37.759994bp}{35.95bp}}
\pgfpathlineto{\pgfpoint{30.199994bp}{18.13bp}}
\color[rgb]{0.0,0.0,0.0}
\pgfusepath{stroke}
\end{pgfscope}
\begin{pgfscope}
\pgfsetlinewidth{1.0bp}
\pgfsetrectcap 
\pgfsetmiterjoin \pgfsetmiterlimit{10.0}
\pgfpathmoveto{\pgfpoint{44.639506bp}{84.160052bp}}
\pgfpathcurveto{\pgfpoint{44.639506bp}{82.532479bp}}{\pgfpoint{43.373953bp}{81.21307bp}}{\pgfpoint{41.812811bp}{81.21307bp}}
\pgfpathcurveto{\pgfpoint{40.251669bp}{81.21307bp}}{\pgfpoint{38.986112bp}{82.532479bp}}{\pgfpoint{38.986112bp}{84.160052bp}}
\pgfpathcurveto{\pgfpoint{38.986112bp}{85.787626bp}}{\pgfpoint{40.251669bp}{87.107035bp}}{\pgfpoint{41.812811bp}{87.107035bp}}
\pgfpathcurveto{\pgfpoint{43.373953bp}{87.107035bp}}{\pgfpoint{44.639506bp}{85.787626bp}}{\pgfpoint{44.639506bp}{84.160052bp}}
\pgfclosepath
\color[rgb]{0.0,0.0,0.0}\pgfseteorule\pgfusepath{fill}
\pgfpathmoveto{\pgfpoint{44.639506bp}{84.160052bp}}
\pgfpathcurveto{\pgfpoint{44.639506bp}{82.532479bp}}{\pgfpoint{43.373953bp}{81.21307bp}}{\pgfpoint{41.812811bp}{81.21307bp}}
\pgfpathcurveto{\pgfpoint{40.251669bp}{81.21307bp}}{\pgfpoint{38.986112bp}{82.532479bp}}{\pgfpoint{38.986112bp}{84.160052bp}}
\pgfpathcurveto{\pgfpoint{38.986112bp}{85.787626bp}}{\pgfpoint{40.251669bp}{87.107035bp}}{\pgfpoint{41.812811bp}{87.107035bp}}
\pgfpathcurveto{\pgfpoint{43.373953bp}{87.107035bp}}{\pgfpoint{44.639506bp}{85.787626bp}}{\pgfpoint{44.639506bp}{84.160052bp}}
\pgfclosepath
\color[rgb]{0.0,0.0,0.0}
\pgfusepath{stroke}
\end{pgfscope}
\begin{pgfscope}
\pgfsetlinewidth{1.0bp}
\pgfsetrectcap 
\pgfsetmiterjoin \pgfsetmiterlimit{10.0}
\pgfpathmoveto{\pgfpoint{22.639506bp}{109.160052bp}}
\pgfpathcurveto{\pgfpoint{22.639506bp}{107.532479bp}}{\pgfpoint{21.373953bp}{106.21307bp}}{\pgfpoint{19.812811bp}{106.21307bp}}
\pgfpathcurveto{\pgfpoint{18.251669bp}{106.21307bp}}{\pgfpoint{16.986112bp}{107.532479bp}}{\pgfpoint{16.986112bp}{109.160052bp}}
\pgfpathcurveto{\pgfpoint{16.986112bp}{110.787626bp}}{\pgfpoint{18.251669bp}{112.107035bp}}{\pgfpoint{19.812811bp}{112.107035bp}}
\pgfpathcurveto{\pgfpoint{21.373953bp}{112.107035bp}}{\pgfpoint{22.639506bp}{110.787626bp}}{\pgfpoint{22.639506bp}{109.160052bp}}
\pgfclosepath
\color[rgb]{0.0,0.0,0.0}\pgfseteorule\pgfusepath{fill}
\pgfpathmoveto{\pgfpoint{22.639506bp}{109.160052bp}}
\pgfpathcurveto{\pgfpoint{22.639506bp}{107.532479bp}}{\pgfpoint{21.373953bp}{106.21307bp}}{\pgfpoint{19.812811bp}{106.21307bp}}
\pgfpathcurveto{\pgfpoint{18.251669bp}{106.21307bp}}{\pgfpoint{16.986112bp}{107.532479bp}}{\pgfpoint{16.986112bp}{109.160052bp}}
\pgfpathcurveto{\pgfpoint{16.986112bp}{110.787626bp}}{\pgfpoint{18.251669bp}{112.107035bp}}{\pgfpoint{19.812811bp}{112.107035bp}}
\pgfpathcurveto{\pgfpoint{21.373953bp}{112.107035bp}}{\pgfpoint{22.639506bp}{110.787626bp}}{\pgfpoint{22.639506bp}{109.160052bp}}
\pgfclosepath
\color[rgb]{0.0,0.0,0.0}
\pgfusepath{stroke}
\end{pgfscope}
\begin{pgfscope}
\pgfsetlinewidth{1.0bp}
\pgfsetrectcap 
\pgfsetmiterjoin \pgfsetmiterlimit{10.0}
\pgfpathmoveto{\pgfpoint{41.239506bp}{37.760052bp}}
\pgfpathcurveto{\pgfpoint{41.239506bp}{36.132479bp}}{\pgfpoint{39.973953bp}{34.81307bp}}{\pgfpoint{38.412811bp}{34.81307bp}}
\pgfpathcurveto{\pgfpoint{36.851669bp}{34.81307bp}}{\pgfpoint{35.586112bp}{36.132479bp}}{\pgfpoint{35.586112bp}{37.760052bp}}
\pgfpathcurveto{\pgfpoint{35.586112bp}{39.387626bp}}{\pgfpoint{36.851669bp}{40.707035bp}}{\pgfpoint{38.412811bp}{40.707035bp}}
\pgfpathcurveto{\pgfpoint{39.973953bp}{40.707035bp}}{\pgfpoint{41.239506bp}{39.387626bp}}{\pgfpoint{41.239506bp}{37.760052bp}}
\pgfclosepath
\color[rgb]{0.0,0.0,0.0}\pgfseteorule\pgfusepath{fill}
\pgfpathmoveto{\pgfpoint{41.239506bp}{37.760052bp}}
\pgfpathcurveto{\pgfpoint{41.239506bp}{36.132479bp}}{\pgfpoint{39.973953bp}{34.81307bp}}{\pgfpoint{38.412811bp}{34.81307bp}}
\pgfpathcurveto{\pgfpoint{36.851669bp}{34.81307bp}}{\pgfpoint{35.586112bp}{36.132479bp}}{\pgfpoint{35.586112bp}{37.760052bp}}
\pgfpathcurveto{\pgfpoint{35.586112bp}{39.387626bp}}{\pgfpoint{36.851669bp}{40.707035bp}}{\pgfpoint{38.412811bp}{40.707035bp}}
\pgfpathcurveto{\pgfpoint{39.973953bp}{40.707035bp}}{\pgfpoint{41.239506bp}{39.387626bp}}{\pgfpoint{41.239506bp}{37.760052bp}}
\pgfclosepath
\color[rgb]{0.0,0.0,0.0}
\pgfusepath{stroke}
\end{pgfscope}
\begin{pgfscope}
\pgfsetlinewidth{1.0bp}
\pgfsetrectcap 
\pgfsetmiterjoin \pgfsetmiterlimit{10.0}
\pgfpathmoveto{\pgfpoint{32.839506bp}{16.160052bp}}
\pgfpathcurveto{\pgfpoint{32.839506bp}{14.532479bp}}{\pgfpoint{31.573953bp}{13.21307bp}}{\pgfpoint{30.012811bp}{13.21307bp}}
\pgfpathcurveto{\pgfpoint{28.451669bp}{13.21307bp}}{\pgfpoint{27.186112bp}{14.532479bp}}{\pgfpoint{27.186112bp}{16.160052bp}}
\pgfpathcurveto{\pgfpoint{27.186112bp}{17.787626bp}}{\pgfpoint{28.451669bp}{19.107035bp}}{\pgfpoint{30.012811bp}{19.107035bp}}
\pgfpathcurveto{\pgfpoint{31.573953bp}{19.107035bp}}{\pgfpoint{32.839506bp}{17.787626bp}}{\pgfpoint{32.839506bp}{16.160052bp}}
\pgfclosepath
\color[rgb]{0.0,0.0,0.0}\pgfseteorule\pgfusepath{fill}
\pgfpathmoveto{\pgfpoint{32.839506bp}{16.160052bp}}
\pgfpathcurveto{\pgfpoint{32.839506bp}{14.532479bp}}{\pgfpoint{31.573953bp}{13.21307bp}}{\pgfpoint{30.012811bp}{13.21307bp}}
\pgfpathcurveto{\pgfpoint{28.451669bp}{13.21307bp}}{\pgfpoint{27.186112bp}{14.532479bp}}{\pgfpoint{27.186112bp}{16.160052bp}}
\pgfpathcurveto{\pgfpoint{27.186112bp}{17.787626bp}}{\pgfpoint{28.451669bp}{19.107035bp}}{\pgfpoint{30.012811bp}{19.107035bp}}
\pgfpathcurveto{\pgfpoint{31.573953bp}{19.107035bp}}{\pgfpoint{32.839506bp}{17.787626bp}}{\pgfpoint{32.839506bp}{16.160052bp}}
\pgfclosepath
\color[rgb]{0.0,0.0,0.0}
\pgfusepath{stroke}
\end{pgfscope}
\begin{pgfscope}
\pgfsetlinewidth{1.0bp}
\pgfsetrectcap 
\pgfsetmiterjoin \pgfsetmiterlimit{10.0}
\pgfpathmoveto{\pgfpoint{122.149994bp}{61.3bp}}
\pgfpathlineto{\pgfpoint{20.149994bp}{109.3bp}}
\pgfpathlineto{\pgfpoint{20.149994bp}{109.3bp}}
\color[rgb]{0.0,0.0,0.0}
\pgfusepath{stroke}
\end{pgfscope}
\begin{pgfscope}
\pgfsetlinewidth{1.0bp}
\pgfsetrectcap 
\pgfsetmiterjoin \pgfsetmiterlimit{10.0}
\pgfpathmoveto{\pgfpoint{120.349994bp}{59.5bp}}
\pgfpathlineto{\pgfpoint{41.749994bp}{84.1bp}}
\color[rgb]{0.0,0.0,0.0}
\pgfusepath{stroke}
\end{pgfscope}
\begin{pgfscope}
\pgfsetlinewidth{1.0bp}
\pgfsetrectcap 
\pgfsetmiterjoin \pgfsetmiterlimit{10.0}
\pgfpathmoveto{\pgfpoint{122.149994bp}{60.1bp}}
\pgfpathlineto{\pgfpoint{38.749994bp}{37.9bp}}
\color[rgb]{0.0,0.0,0.0}
\pgfusepath{stroke}
\end{pgfscope}
\begin{pgfscope}
\pgfsetlinewidth{1.0bp}
\pgfsetrectcap 
\pgfsetmiterjoin \pgfsetmiterlimit{10.0}
\pgfpathmoveto{\pgfpoint{122.149994bp}{60.7bp}}
\pgfpathlineto{\pgfpoint{29.749994bp}{16.9bp}}
\color[rgb]{0.0,0.0,0.0}
\pgfusepath{stroke}
\end{pgfscope}
\begin{pgfscope}
\pgftransformcm{0.0}{-1.0}{1.0}{0.0}{\pgfpoint{69.549994bp}{75.5bp}}
\pgftext[left,base]{\sffamily\mdseries\upshape\fontsize{30.1125}{30.1125}\selectfont
\color[rgb]{0.0,0.0,0.0}$\hdots$}
\end{pgfscope}
\begin{pgfscope}
\pgftransformcm{1.0}{0.0}{0.0}{1.0}{\pgfpoint{210.549994bp}{57.5bp}}
\pgftext[left,base]{\sffamily\mdseries\upshape\huge
\color[rgb]{0.0,0.0,0.0}$G_2$}
\end{pgfscope}
\end{pgfpicture}

%% file: Hexample7.tex
% C:\Users\agalanis\Dropbox\H-colorings\Hexample7.jdr
% Created by Jpgfdraw version 0.5.6b
% 03-Feb-2015 17:50:58
\begin{pgfpicture}{0bp}{0bp}{108.175bp}{133.187927bp}
\begin{pgfscope}
\pgfsetlinewidth{1.0bp}
\pgfsetrectcap 
\pgfsetmiterjoin \pgfsetmiterlimit{10.0}
\pgfpathmoveto{\pgfpoint{21.872906bp}{126.793966bp}}
\pgfpathcurveto{\pgfpoint{21.872906bp}{123.53882bp}}{\pgfpoint{19.341799bp}{120.900001bp}}{\pgfpoint{16.219514bp}{120.900001bp}}
\pgfpathcurveto{\pgfpoint{13.09723bp}{120.900001bp}}{\pgfpoint{10.566117bp}{123.53882bp}}{\pgfpoint{10.566117bp}{126.793966bp}}
\pgfpathcurveto{\pgfpoint{10.566117bp}{130.049114bp}}{\pgfpoint{13.09723bp}{132.687931bp}}{\pgfpoint{16.219514bp}{132.687931bp}}
\pgfpathcurveto{\pgfpoint{19.341799bp}{132.687931bp}}{\pgfpoint{21.872906bp}{130.049114bp}}{\pgfpoint{21.872906bp}{126.793966bp}}
\pgfclosepath
\color[rgb]{0.0,0.0,0.0}\pgfseteorule\pgfusepath{fill}
\pgfpathmoveto{\pgfpoint{21.872906bp}{126.793966bp}}
\pgfpathcurveto{\pgfpoint{21.872906bp}{123.53882bp}}{\pgfpoint{19.341799bp}{120.900001bp}}{\pgfpoint{16.219514bp}{120.900001bp}}
\pgfpathcurveto{\pgfpoint{13.09723bp}{120.900001bp}}{\pgfpoint{10.566117bp}{123.53882bp}}{\pgfpoint{10.566117bp}{126.793966bp}}
\pgfpathcurveto{\pgfpoint{10.566117bp}{130.049114bp}}{\pgfpoint{13.09723bp}{132.687931bp}}{\pgfpoint{16.219514bp}{132.687931bp}}
\pgfpathcurveto{\pgfpoint{19.341799bp}{132.687931bp}}{\pgfpoint{21.872906bp}{130.049114bp}}{\pgfpoint{21.872906bp}{126.793966bp}}
\pgfclosepath
\color[rgb]{0.0,0.0,0.0}
\pgfusepath{stroke}
\end{pgfscope}
\begin{pgfscope}
\pgfsetlinewidth{1.0bp}
\pgfsetrectcap 
\pgfsetmiterjoin \pgfsetmiterlimit{10.0}
\pgfpathmoveto{\pgfpoint{21.272906bp}{87.193966bp}}
\pgfpathcurveto{\pgfpoint{21.272906bp}{83.93882bp}}{\pgfpoint{18.741799bp}{81.300001bp}}{\pgfpoint{15.619514bp}{81.300001bp}}
\pgfpathcurveto{\pgfpoint{12.49723bp}{81.300001bp}}{\pgfpoint{9.966117bp}{83.93882bp}}{\pgfpoint{9.966117bp}{87.193966bp}}
\pgfpathcurveto{\pgfpoint{9.966117bp}{90.449114bp}}{\pgfpoint{12.49723bp}{93.087931bp}}{\pgfpoint{15.619514bp}{93.087931bp}}
\pgfpathcurveto{\pgfpoint{18.741799bp}{93.087931bp}}{\pgfpoint{21.272906bp}{90.449114bp}}{\pgfpoint{21.272906bp}{87.193966bp}}
\pgfclosepath
\color[rgb]{0.0,0.0,0.0}\pgfseteorule\pgfusepath{fill}
\pgfpathmoveto{\pgfpoint{21.272906bp}{87.193966bp}}
\pgfpathcurveto{\pgfpoint{21.272906bp}{83.93882bp}}{\pgfpoint{18.741799bp}{81.300001bp}}{\pgfpoint{15.619514bp}{81.300001bp}}
\pgfpathcurveto{\pgfpoint{12.49723bp}{81.300001bp}}{\pgfpoint{9.966117bp}{83.93882bp}}{\pgfpoint{9.966117bp}{87.193966bp}}
\pgfpathcurveto{\pgfpoint{9.966117bp}{90.449114bp}}{\pgfpoint{12.49723bp}{93.087931bp}}{\pgfpoint{15.619514bp}{93.087931bp}}
\pgfpathcurveto{\pgfpoint{18.741799bp}{93.087931bp}}{\pgfpoint{21.272906bp}{90.449114bp}}{\pgfpoint{21.272906bp}{87.193966bp}}
\pgfclosepath
\color[rgb]{0.0,0.0,0.0}
\pgfusepath{stroke}
\end{pgfscope}
\begin{pgfscope}
\pgfsetlinewidth{1.0bp}
\pgfsetrectcap 
\pgfsetmiterjoin \pgfsetmiterlimit{10.0}
\pgfpathmoveto{\pgfpoint{96.472906bp}{125.993966bp}}
\pgfpathcurveto{\pgfpoint{96.472906bp}{122.73882bp}}{\pgfpoint{93.941799bp}{120.100001bp}}{\pgfpoint{90.819514bp}{120.100001bp}}
\pgfpathcurveto{\pgfpoint{87.69723bp}{120.100001bp}}{\pgfpoint{85.166117bp}{122.73882bp}}{\pgfpoint{85.166117bp}{125.993966bp}}
\pgfpathcurveto{\pgfpoint{85.166117bp}{129.249114bp}}{\pgfpoint{87.69723bp}{131.887931bp}}{\pgfpoint{90.819514bp}{131.887931bp}}
\pgfpathcurveto{\pgfpoint{93.941799bp}{131.887931bp}}{\pgfpoint{96.472906bp}{129.249114bp}}{\pgfpoint{96.472906bp}{125.993966bp}}
\pgfclosepath
\color[rgb]{0.0,0.0,0.0}\pgfseteorule\pgfusepath{fill}
\pgfpathmoveto{\pgfpoint{96.472906bp}{125.993966bp}}
\pgfpathcurveto{\pgfpoint{96.472906bp}{122.73882bp}}{\pgfpoint{93.941799bp}{120.100001bp}}{\pgfpoint{90.819514bp}{120.100001bp}}
\pgfpathcurveto{\pgfpoint{87.69723bp}{120.100001bp}}{\pgfpoint{85.166117bp}{122.73882bp}}{\pgfpoint{85.166117bp}{125.993966bp}}
\pgfpathcurveto{\pgfpoint{85.166117bp}{129.249114bp}}{\pgfpoint{87.69723bp}{131.887931bp}}{\pgfpoint{90.819514bp}{131.887931bp}}
\pgfpathcurveto{\pgfpoint{93.941799bp}{131.887931bp}}{\pgfpoint{96.472906bp}{129.249114bp}}{\pgfpoint{96.472906bp}{125.993966bp}}
\pgfclosepath
\color[rgb]{0.0,0.0,0.0}
\pgfusepath{stroke}
\end{pgfscope}
\begin{pgfscope}
\pgfsetlinewidth{1.0bp}
\pgfsetrectcap 
\pgfsetmiterjoin \pgfsetmiterlimit{10.0}
\pgfpathmoveto{\pgfpoint{95.872906bp}{86.393966bp}}
\pgfpathcurveto{\pgfpoint{95.872906bp}{83.13882bp}}{\pgfpoint{93.341799bp}{80.500001bp}}{\pgfpoint{90.219514bp}{80.500001bp}}
\pgfpathcurveto{\pgfpoint{87.09723bp}{80.500001bp}}{\pgfpoint{84.566117bp}{83.13882bp}}{\pgfpoint{84.566117bp}{86.393966bp}}
\pgfpathcurveto{\pgfpoint{84.566117bp}{89.649114bp}}{\pgfpoint{87.09723bp}{92.287931bp}}{\pgfpoint{90.219514bp}{92.287931bp}}
\pgfpathcurveto{\pgfpoint{93.341799bp}{92.287931bp}}{\pgfpoint{95.872906bp}{89.649114bp}}{\pgfpoint{95.872906bp}{86.393966bp}}
\pgfclosepath
\color[rgb]{0.0,0.0,0.0}\pgfseteorule\pgfusepath{fill}
\pgfpathmoveto{\pgfpoint{95.872906bp}{86.393966bp}}
\pgfpathcurveto{\pgfpoint{95.872906bp}{83.13882bp}}{\pgfpoint{93.341799bp}{80.500001bp}}{\pgfpoint{90.219514bp}{80.500001bp}}
\pgfpathcurveto{\pgfpoint{87.09723bp}{80.500001bp}}{\pgfpoint{84.566117bp}{83.13882bp}}{\pgfpoint{84.566117bp}{86.393966bp}}
\pgfpathcurveto{\pgfpoint{84.566117bp}{89.649114bp}}{\pgfpoint{87.09723bp}{92.287931bp}}{\pgfpoint{90.219514bp}{92.287931bp}}
\pgfpathcurveto{\pgfpoint{93.341799bp}{92.287931bp}}{\pgfpoint{95.872906bp}{89.649114bp}}{\pgfpoint{95.872906bp}{86.393966bp}}
\pgfclosepath
\color[rgb]{0.0,0.0,0.0}
\pgfusepath{stroke}
\end{pgfscope}
\begin{pgfscope}
\pgfsetlinewidth{1.0bp}
\pgfsetrectcap 
\pgfsetmiterjoin \pgfsetmiterlimit{10.0}
\pgfpathmoveto{\pgfpoint{20.872906bp}{46.793966bp}}
\pgfpathcurveto{\pgfpoint{20.872906bp}{43.53882bp}}{\pgfpoint{18.341799bp}{40.900001bp}}{\pgfpoint{15.219514bp}{40.900001bp}}
\pgfpathcurveto{\pgfpoint{12.09723bp}{40.900001bp}}{\pgfpoint{9.566117bp}{43.53882bp}}{\pgfpoint{9.566117bp}{46.793966bp}}
\pgfpathcurveto{\pgfpoint{9.566117bp}{50.049114bp}}{\pgfpoint{12.09723bp}{52.687931bp}}{\pgfpoint{15.219514bp}{52.687931bp}}
\pgfpathcurveto{\pgfpoint{18.341799bp}{52.687931bp}}{\pgfpoint{20.872906bp}{50.049114bp}}{\pgfpoint{20.872906bp}{46.793966bp}}
\pgfclosepath
\color[rgb]{0.0,0.0,0.0}\pgfseteorule\pgfusepath{fill}
\pgfpathmoveto{\pgfpoint{20.872906bp}{46.793966bp}}
\pgfpathcurveto{\pgfpoint{20.872906bp}{43.53882bp}}{\pgfpoint{18.341799bp}{40.900001bp}}{\pgfpoint{15.219514bp}{40.900001bp}}
\pgfpathcurveto{\pgfpoint{12.09723bp}{40.900001bp}}{\pgfpoint{9.566117bp}{43.53882bp}}{\pgfpoint{9.566117bp}{46.793966bp}}
\pgfpathcurveto{\pgfpoint{9.566117bp}{50.049114bp}}{\pgfpoint{12.09723bp}{52.687931bp}}{\pgfpoint{15.219514bp}{52.687931bp}}
\pgfpathcurveto{\pgfpoint{18.341799bp}{52.687931bp}}{\pgfpoint{20.872906bp}{50.049114bp}}{\pgfpoint{20.872906bp}{46.793966bp}}
\pgfclosepath
\color[rgb]{0.0,0.0,0.0}
\pgfusepath{stroke}
\end{pgfscope}
\begin{pgfscope}
\pgfsetlinewidth{1.0bp}
\pgfsetrectcap 
\pgfsetmiterjoin \pgfsetmiterlimit{10.0}
\pgfpathmoveto{\pgfpoint{20.272906bp}{7.193966bp}}
\pgfpathcurveto{\pgfpoint{20.272906bp}{3.93882bp}}{\pgfpoint{17.741799bp}{1.300001bp}}{\pgfpoint{14.619514bp}{1.300001bp}}
\pgfpathcurveto{\pgfpoint{11.49723bp}{1.300001bp}}{\pgfpoint{8.966117bp}{3.93882bp}}{\pgfpoint{8.966117bp}{7.193966bp}}
\pgfpathcurveto{\pgfpoint{8.966117bp}{10.449114bp}}{\pgfpoint{11.49723bp}{13.087931bp}}{\pgfpoint{14.619514bp}{13.087931bp}}
\pgfpathcurveto{\pgfpoint{17.741799bp}{13.087931bp}}{\pgfpoint{20.272906bp}{10.449114bp}}{\pgfpoint{20.272906bp}{7.193966bp}}
\pgfclosepath
\color[rgb]{0.0,0.0,0.0}\pgfseteorule\pgfusepath{fill}
\pgfpathmoveto{\pgfpoint{20.272906bp}{7.193966bp}}
\pgfpathcurveto{\pgfpoint{20.272906bp}{3.93882bp}}{\pgfpoint{17.741799bp}{1.300001bp}}{\pgfpoint{14.619514bp}{1.300001bp}}
\pgfpathcurveto{\pgfpoint{11.49723bp}{1.300001bp}}{\pgfpoint{8.966117bp}{3.93882bp}}{\pgfpoint{8.966117bp}{7.193966bp}}
\pgfpathcurveto{\pgfpoint{8.966117bp}{10.449114bp}}{\pgfpoint{11.49723bp}{13.087931bp}}{\pgfpoint{14.619514bp}{13.087931bp}}
\pgfpathcurveto{\pgfpoint{17.741799bp}{13.087931bp}}{\pgfpoint{20.272906bp}{10.449114bp}}{\pgfpoint{20.272906bp}{7.193966bp}}
\pgfclosepath
\color[rgb]{0.0,0.0,0.0}
\pgfusepath{stroke}
\end{pgfscope}
\begin{pgfscope}
\pgfsetlinewidth{1.0bp}
\pgfsetrectcap 
\pgfsetmiterjoin \pgfsetmiterlimit{10.0}
\pgfpathmoveto{\pgfpoint{95.472906bp}{45.993966bp}}
\pgfpathcurveto{\pgfpoint{95.472906bp}{42.73882bp}}{\pgfpoint{92.941799bp}{40.100001bp}}{\pgfpoint{89.819514bp}{40.100001bp}}
\pgfpathcurveto{\pgfpoint{86.69723bp}{40.100001bp}}{\pgfpoint{84.166117bp}{42.73882bp}}{\pgfpoint{84.166117bp}{45.993966bp}}
\pgfpathcurveto{\pgfpoint{84.166117bp}{49.249114bp}}{\pgfpoint{86.69723bp}{51.887931bp}}{\pgfpoint{89.819514bp}{51.887931bp}}
\pgfpathcurveto{\pgfpoint{92.941799bp}{51.887931bp}}{\pgfpoint{95.472906bp}{49.249114bp}}{\pgfpoint{95.472906bp}{45.993966bp}}
\pgfclosepath
\color[rgb]{0.0,0.0,0.0}\pgfseteorule\pgfusepath{fill}
\pgfpathmoveto{\pgfpoint{95.472906bp}{45.993966bp}}
\pgfpathcurveto{\pgfpoint{95.472906bp}{42.73882bp}}{\pgfpoint{92.941799bp}{40.100001bp}}{\pgfpoint{89.819514bp}{40.100001bp}}
\pgfpathcurveto{\pgfpoint{86.69723bp}{40.100001bp}}{\pgfpoint{84.166117bp}{42.73882bp}}{\pgfpoint{84.166117bp}{45.993966bp}}
\pgfpathcurveto{\pgfpoint{84.166117bp}{49.249114bp}}{\pgfpoint{86.69723bp}{51.887931bp}}{\pgfpoint{89.819514bp}{51.887931bp}}
\pgfpathcurveto{\pgfpoint{92.941799bp}{51.887931bp}}{\pgfpoint{95.472906bp}{49.249114bp}}{\pgfpoint{95.472906bp}{45.993966bp}}
\pgfclosepath
\color[rgb]{0.0,0.0,0.0}
\pgfusepath{stroke}
\end{pgfscope}
\begin{pgfscope}
\pgfsetlinewidth{1.0bp}
\pgfsetrectcap 
\pgfsetmiterjoin \pgfsetmiterlimit{10.0}
\pgfpathmoveto{\pgfpoint{94.872906bp}{6.393966bp}}
\pgfpathcurveto{\pgfpoint{94.872906bp}{3.13882bp}}{\pgfpoint{92.341799bp}{0.500001bp}}{\pgfpoint{89.219514bp}{0.500001bp}}
\pgfpathcurveto{\pgfpoint{86.09723bp}{0.500001bp}}{\pgfpoint{83.566117bp}{3.13882bp}}{\pgfpoint{83.566117bp}{6.393966bp}}
\pgfpathcurveto{\pgfpoint{83.566117bp}{9.649114bp}}{\pgfpoint{86.09723bp}{12.287931bp}}{\pgfpoint{89.219514bp}{12.287931bp}}
\pgfpathcurveto{\pgfpoint{92.341799bp}{12.287931bp}}{\pgfpoint{94.872906bp}{9.649114bp}}{\pgfpoint{94.872906bp}{6.393966bp}}
\pgfclosepath
\color[rgb]{0.0,0.0,0.0}\pgfseteorule\pgfusepath{fill}
\pgfpathmoveto{\pgfpoint{94.872906bp}{6.393966bp}}
\pgfpathcurveto{\pgfpoint{94.872906bp}{3.13882bp}}{\pgfpoint{92.341799bp}{0.500001bp}}{\pgfpoint{89.219514bp}{0.500001bp}}
\pgfpathcurveto{\pgfpoint{86.09723bp}{0.500001bp}}{\pgfpoint{83.566117bp}{3.13882bp}}{\pgfpoint{83.566117bp}{6.393966bp}}
\pgfpathcurveto{\pgfpoint{83.566117bp}{9.649114bp}}{\pgfpoint{86.09723bp}{12.287931bp}}{\pgfpoint{89.219514bp}{12.287931bp}}
\pgfpathcurveto{\pgfpoint{92.341799bp}{12.287931bp}}{\pgfpoint{94.872906bp}{9.649114bp}}{\pgfpoint{94.872906bp}{6.393966bp}}
\pgfclosepath
\color[rgb]{0.0,0.0,0.0}
\pgfusepath{stroke}
\end{pgfscope}
\begin{pgfscope}
\pgfsetlinewidth{1.0bp}
\pgfsetrectcap 
\pgfsetmiterjoin \pgfsetmiterlimit{10.0}
\pgfpathmoveto{\pgfpoint{91.65bp}{126.495892bp}}
\pgfpathlineto{\pgfpoint{17.25bp}{126.495892bp}}
\color[rgb]{0.0,0.0,0.0}
\pgfusepath{stroke}
\end{pgfscope}
\begin{pgfscope}
\pgfsetlinewidth{1.0bp}
\pgfsetrectcap 
\pgfsetmiterjoin \pgfsetmiterlimit{10.0}
\pgfpathmoveto{\pgfpoint{17.85bp}{126.495892bp}}
\pgfpathlineto{\pgfpoint{88.65bp}{86.895892bp}}
\color[rgb]{0.0,0.0,0.0}
\pgfusepath{stroke}
\end{pgfscope}
\begin{pgfscope}
\pgfsetlinewidth{1.0bp}
\pgfsetrectcap 
\pgfsetmiterjoin \pgfsetmiterlimit{10.0}
\pgfpathmoveto{\pgfpoint{17.85bp}{128.295892bp}}
\pgfpathlineto{\pgfpoint{89.25bp}{46.095892bp}}
\color[rgb]{0.0,0.0,0.0}
\pgfusepath{stroke}
\end{pgfscope}
\begin{pgfscope}
\pgfsetlinewidth{1.0bp}
\pgfsetrectcap 
\pgfsetmiterjoin \pgfsetmiterlimit{10.0}
\pgfpathmoveto{\pgfpoint{17.85bp}{126.495892bp}}
\pgfpathlineto{\pgfpoint{88.65bp}{6.495892bp}}
\color[rgb]{0.0,0.0,0.0}
\pgfusepath{stroke}
\end{pgfscope}
\begin{pgfscope}
\pgfsetlinewidth{1.0bp}
\pgfsetrectcap 
\pgfsetmiterjoin \pgfsetmiterlimit{10.0}
\pgfpathmoveto{\pgfpoint{89.85bp}{127.095892bp}}
\pgfpathlineto{\pgfpoint{15.45bp}{86.295892bp}}
\color[rgb]{0.0,0.0,0.0}
\pgfusepath{stroke}
\end{pgfscope}
\begin{pgfscope}
\pgfsetlinewidth{1.0bp}
\pgfsetrectcap 
\pgfsetmiterjoin \pgfsetmiterlimit{10.0}
\pgfpathmoveto{\pgfpoint{90.45bp}{127.695892bp}}
\pgfpathlineto{\pgfpoint{17.25bp}{49.095892bp}}
\color[rgb]{0.0,0.0,0.0}
\pgfusepath{stroke}
\end{pgfscope}
\begin{pgfscope}
\pgfsetlinewidth{1.0bp}
\pgfsetrectcap 
\pgfsetmiterjoin \pgfsetmiterlimit{10.0}
\pgfpathmoveto{\pgfpoint{91.05bp}{126.495892bp}}
\pgfpathlineto{\pgfpoint{17.85bp}{8.895892bp}}
\color[rgb]{0.0,0.0,0.0}
\pgfusepath{stroke}
\end{pgfscope}
\begin{pgfscope}
\pgfsetlinewidth{1.0bp}
\pgfsetrectcap 
\pgfsetmiterjoin \pgfsetmiterlimit{10.0}
\pgfpathmoveto{\pgfpoint{90.45bp}{87.495892bp}}
\pgfpathlineto{\pgfpoint{16.05bp}{87.495892bp}}
\color[rgb]{0.0,0.0,0.0}
\pgfusepath{stroke}
\end{pgfscope}
\begin{pgfscope}
\pgftransformcm{1.0}{0.0}{0.0}{1.0}{\pgfpoint{0.05bp}{121.495892bp}}
\pgftext[left,base]{\sffamily\mdseries\upshape\large
\color[rgb]{0.0,0.0,0.0}$1$}
\end{pgfscope}
\begin{pgfscope}
\pgftransformcm{1.0}{0.0}{0.0}{1.0}{\pgfpoint{0.25bp}{82.895892bp}}
\pgftext[left,base]{\sffamily\mdseries\upshape\large
\color[rgb]{0.0,0.0,0.0}$2$}
\end{pgfscope}
\begin{pgfscope}
\pgftransformcm{1.0}{0.0}{0.0}{1.0}{\pgfpoint{-0.15bp}{43.695892bp}}
\pgftext[left,base]{\sffamily\mdseries\upshape\large
\color[rgb]{0.0,0.0,0.0}$3$}
\end{pgfscope}
\begin{pgfscope}
\pgftransformcm{1.0}{0.0}{0.0}{1.0}{\pgfpoint{-0.15bp}{1.695892bp}}
\pgftext[left,base]{\sffamily\mdseries\upshape\large
\color[rgb]{0.0,0.0,0.0}$4$}
\end{pgfscope}
\begin{pgfscope}
\pgftransformcm{1.0}{0.0}{0.0}{1.0}{\pgfpoint{99.85bp}{121.295892bp}}
\pgftext[left,base]{\sffamily\mdseries\upshape\large
\color[rgb]{0.0,0.0,0.0}$1'$}
\end{pgfscope}
\begin{pgfscope}
\pgftransformcm{1.0}{0.0}{0.0}{1.0}{\pgfpoint{99.85bp}{82.295892bp}}
\pgftext[left,base]{\sffamily\mdseries\upshape\large
\color[rgb]{0.0,0.0,0.0}$2'$}
\end{pgfscope}
\begin{pgfscope}
\pgftransformcm{1.0}{0.0}{0.0}{1.0}{\pgfpoint{99.25bp}{42.695892bp}}
\pgftext[left,base]{\sffamily\mdseries\upshape\large
\color[rgb]{0.0,0.0,0.0}$3'$}
\end{pgfscope}
\begin{pgfscope}
\pgftransformcm{1.0}{0.0}{0.0}{1.0}{\pgfpoint{98.85bp}{2.295892bp}}
\pgftext[left,base]{\sffamily\mdseries\upshape\large
\color[rgb]{0.0,0.0,0.0}$4'$}
\end{pgfscope}
\begin{pgfscope}
\pgfsetlinewidth{1.0bp}
\pgfsetrectcap 
\pgfsetmiterjoin \pgfsetmiterlimit{10.0}
\pgfpathmoveto{\pgfpoint{89.85bp}{46.095892bp}}
\pgfpathlineto{\pgfpoint{14.85bp}{46.095892bp}}
\color[rgb]{0.0,0.0,0.0}
\pgfusepath{stroke}
\end{pgfscope}
\end{pgfpicture}

%% file: Hexample6.tex
% C:\Users\agalanis\Dropbox\H-colorings\Hexample6.jdr
% Created by Jpgfdraw version 0.5.6b
% 31-Jan-2015 07:32:15
\begin{pgfpicture}{0bp}{0bp}{499.199997bp}{262.062507bp}
\begin{pgfscope}
\pgftransformcm{1.0}{0.0}{0.0}{1.0}{\pgfpoint{230.899997bp}{4.162501bp}}
\pgftext[left,base]{\sffamily\mdseries\upshape\huge
\color[rgb]{0.0,0.0,0.0} $K_{a,b}(G')$}
\end{pgfscope}
\begin{pgfscope}
\pgftransformcm{1.0}{0.0}{0.0}{1.0}{\pgfpoint{22.099997bp}{152.962501bp}}
\pgftext[left,base]{\sffamily\mdseries\upshape\huge
\color[rgb]{0.0,0.0,0.0}$L(K_{a,b})$}
\end{pgfscope}
\begin{pgfscope}
\pgfsetlinewidth{1.0bp}
\pgfsetrectcap 
\pgfsetmiterjoin \pgfsetmiterlimit{10.0}
\pgfpathmoveto{\pgfpoint{375.099997bp}{148.562501bp}}
\pgfpathcurveto{\pgfpoint{375.099997bp}{113.768564bp}}{\pgfpoint{361.131277bp}{85.562501bp}}{\pgfpoint{343.9bp}{85.562501bp}}
\pgfpathcurveto{\pgfpoint{326.668715bp}{85.562501bp}}{\pgfpoint{312.699995bp}{113.768564bp}}{\pgfpoint{312.699995bp}{148.562501bp}}
\pgfpathcurveto{\pgfpoint{312.699995bp}{183.35644bp}}{\pgfpoint{326.668715bp}{211.562501bp}}{\pgfpoint{343.9bp}{211.562501bp}}
\pgfpathcurveto{\pgfpoint{361.131277bp}{211.562501bp}}{\pgfpoint{375.099997bp}{183.35644bp}}{\pgfpoint{375.099997bp}{148.562501bp}}
\pgfclosepath
\color[rgb]{0.0,0.0,0.0}
\pgfusepath{stroke}
\end{pgfscope}
\begin{pgfscope}
\pgfsetlinewidth{1.0bp}
\pgfsetrectcap 
\pgfsetmiterjoin \pgfsetmiterlimit{10.0}
\pgfpathmoveto{\pgfpoint{498.699991bp}{146.162501bp}}
\pgfpathcurveto{\pgfpoint{498.699991bp}{112.694044bp}}{\pgfpoint{485.268533bp}{85.562498bp}}{\pgfpoint{468.699991bp}{85.562498bp}}
\pgfpathcurveto{\pgfpoint{452.131448bp}{85.562498bp}}{\pgfpoint{438.699991bp}{112.694044bp}}{\pgfpoint{438.699991bp}{146.162501bp}}
\pgfpathcurveto{\pgfpoint{438.699991bp}{179.630957bp}}{\pgfpoint{452.131448bp}{206.762501bp}}{\pgfpoint{468.699991bp}{206.762501bp}}
\pgfpathcurveto{\pgfpoint{485.268533bp}{206.762501bp}}{\pgfpoint{498.699991bp}{179.630957bp}}{\pgfpoint{498.699991bp}{146.162501bp}}
\pgfclosepath
\color[rgb]{0.0,0.0,0.0}
\pgfusepath{stroke}
\end{pgfscope}
\begin{pgfscope}
\pgftransformcm{1.0}{0.0}{0.0}{1.0}{\pgfpoint{180.099997bp}{150.962501bp}}
\pgftext[left,base]{\sffamily\mdseries\upshape\huge
\color[rgb]{0.0,0.0,0.0}$R(K_{a,b})$}
\end{pgfscope}
\begin{pgfscope}
\pgfsetlinewidth{1.0bp}
\pgfsetrectcap 
\pgfsetmiterjoin \pgfsetmiterlimit{10.0}
\pgfsetdash{{3.0bp}{5.0bp}}{0.0bp}
\pgfpathmoveto{\pgfpoint{375.499997bp}{145.162501bp}}
\pgfpathcurveto{\pgfpoint{405.299997bp}{212.962501bp}}{\pgfpoint{410.499997bp}{82.762501bp}}{\pgfpoint{437.899997bp}{145.162501bp}}
\color[rgb]{0.0,0.0,0.0}
\pgfusepath{stroke}
\end{pgfscope}
\begin{pgfscope}
\pgftransformcm{1.0}{0.0}{0.0}{1.0}{\pgfpoint{324.099997bp}{143.762501bp}}
\pgftext[left,base]{\sffamily\mdseries\upshape\huge
\color[rgb]{0.0,0.0,0.0}$L(G')$}
\end{pgfscope}
\begin{pgfscope}
\pgftransformcm{1.0}{0.0}{0.0}{1.0}{\pgfpoint{447.299997bp}{142.362501bp}}
\pgftext[left,base]{\sffamily\mdseries\upshape\huge
\color[rgb]{0.0,0.0,0.0}$R(G')$}
\end{pgfscope}
\begin{pgfscope}
\pgfsetlinewidth{1.0bp}
\pgfsetrectcap 
\pgfsetmiterjoin \pgfsetmiterlimit{10.0}
\pgfpathmoveto{\pgfpoint{92.899991bp}{152.962501bp}}
\pgfpathcurveto{\pgfpoint{92.899991bp}{92.984382bp}}{\pgfpoint{72.215558bp}{44.362495bp}}{\pgfpoint{46.699994bp}{44.362495bp}}
\pgfpathcurveto{\pgfpoint{21.184445bp}{44.362495bp}}{\pgfpoint{0.499997bp}{92.984382bp}}{\pgfpoint{0.499997bp}{152.962501bp}}
\pgfpathcurveto{\pgfpoint{0.499997bp}{212.940626bp}}{\pgfpoint{21.184445bp}{261.562502bp}}{\pgfpoint{46.699994bp}{261.562502bp}}
\pgfpathcurveto{\pgfpoint{72.215558bp}{261.562502bp}}{\pgfpoint{92.899991bp}{212.940626bp}}{\pgfpoint{92.899991bp}{152.962501bp}}
\pgfclosepath
\color[rgb]{0.0,0.0,0.0}
\pgfusepath{stroke}
\end{pgfscope}
\begin{pgfscope}
\pgfsetlinewidth{1.0bp}
\pgfsetrectcap 
\pgfsetmiterjoin \pgfsetmiterlimit{10.0}
\pgfpathmoveto{\pgfpoint{253.299991bp}{150.362501bp}}
\pgfpathcurveto{\pgfpoint{253.299991bp}{90.384382bp}}{\pgfpoint{232.615558bp}{41.762495bp}}{\pgfpoint{207.099994bp}{41.762495bp}}
\pgfpathcurveto{\pgfpoint{181.584445bp}{41.762495bp}}{\pgfpoint{160.899997bp}{90.384382bp}}{\pgfpoint{160.899997bp}{150.362501bp}}
\pgfpathcurveto{\pgfpoint{160.899997bp}{210.340626bp}}{\pgfpoint{181.584445bp}{258.962502bp}}{\pgfpoint{207.099994bp}{258.962502bp}}
\pgfpathcurveto{\pgfpoint{232.615558bp}{258.962502bp}}{\pgfpoint{253.299991bp}{210.340626bp}}{\pgfpoint{253.299991bp}{150.362501bp}}
\pgfclosepath
\color[rgb]{0.0,0.0,0.0}
\pgfusepath{stroke}
\end{pgfscope}
\begin{pgfscope}
\pgfsetlinewidth{1.0bp}
\pgfsetrectcap 
\pgfsetmiterjoin \pgfsetmiterlimit{10.0}
\pgfpathmoveto{\pgfpoint{206.299997bp}{259.762501bp}}
\pgfpathlineto{\pgfpoint{44.899997bp}{260.962501bp}}
\color[rgb]{0.0,0.0,0.0}
\pgfusepath{stroke}
\end{pgfscope}
\begin{pgfscope}
\pgfsetlinewidth{1.0bp}
\pgfsetrectcap 
\pgfsetmiterjoin \pgfsetmiterlimit{10.0}
\pgfpathmoveto{\pgfpoint{205.699997bp}{41.362501bp}}
\pgfpathlineto{\pgfpoint{43.699997bp}{44.962501bp}}
\color[rgb]{0.0,0.0,0.0}
\pgfusepath{stroke}
\end{pgfscope}
\begin{pgfscope}
\pgfsetlinewidth{1.0bp}
\pgfsetrectcap 
\pgfsetmiterjoin \pgfsetmiterlimit{10.0}
\pgfpathmoveto{\pgfpoint{217.099997bp}{257.362501bp}}
\pgfpathlineto{\pgfpoint{347.899997bp}{211.762501bp}}
\color[rgb]{0.0,0.0,0.0}
\pgfusepath{stroke}
\end{pgfscope}
\begin{pgfscope}
\pgfsetlinewidth{1.0bp}
\pgfsetrectcap 
\pgfsetmiterjoin \pgfsetmiterlimit{10.0}
\pgfpathmoveto{\pgfpoint{212.299997bp}{42.562501bp}}
\pgfpathlineto{\pgfpoint{349.699997bp}{86.362501bp}}
\color[rgb]{0.0,0.0,0.0}
\pgfusepath{stroke}
\end{pgfscope}
\end{pgfpicture}

%% file: Hexample9.tex
% C:\Users\agalanis\Dropbox\H-colorings\Hexample9.jdr
% Created by Jpgfdraw version 0.5.6b
% 01-Feb-2015 12:54:16
\begin{pgfpicture}{0bp}{0bp}{303.83847bp}{226.5375bp}
\begin{pgfscope}
\pgfsetlinewidth{1.0bp}
\pgfsetrectcap 
\pgfsetmiterjoin \pgfsetmiterlimit{10.0}
\pgfpathmoveto{\pgfpoint{297.444513bp}{146.844675bp}}
\pgfpathcurveto{\pgfpoint{294.189367bp}{146.844675bp}}{\pgfpoint{291.550548bp}{149.375782bp}}{\pgfpoint{291.550548bp}{152.498066bp}}
\pgfpathcurveto{\pgfpoint{291.550548bp}{155.62035bp}}{\pgfpoint{294.189367bp}{158.151463bp}}{\pgfpoint{297.444513bp}{158.151463bp}}
\pgfpathcurveto{\pgfpoint{300.699661bp}{158.151463bp}}{\pgfpoint{303.338478bp}{155.62035bp}}{\pgfpoint{303.338478bp}{152.498066bp}}
\pgfpathcurveto{\pgfpoint{303.338478bp}{149.375782bp}}{\pgfpoint{300.699661bp}{146.844675bp}}{\pgfpoint{297.444513bp}{146.844675bp}}
\pgfclosepath
\color[rgb]{0.0,0.0,0.0}\pgfseteorule\pgfusepath{fill}
\pgfpathmoveto{\pgfpoint{297.444513bp}{146.844675bp}}
\pgfpathcurveto{\pgfpoint{294.189367bp}{146.844675bp}}{\pgfpoint{291.550548bp}{149.375782bp}}{\pgfpoint{291.550548bp}{152.498066bp}}
\pgfpathcurveto{\pgfpoint{291.550548bp}{155.62035bp}}{\pgfpoint{294.189367bp}{158.151463bp}}{\pgfpoint{297.444513bp}{158.151463bp}}
\pgfpathcurveto{\pgfpoint{300.699661bp}{158.151463bp}}{\pgfpoint{303.338478bp}{155.62035bp}}{\pgfpoint{303.338478bp}{152.498066bp}}
\pgfpathcurveto{\pgfpoint{303.338478bp}{149.375782bp}}{\pgfpoint{300.699661bp}{146.844675bp}}{\pgfpoint{297.444513bp}{146.844675bp}}
\pgfclosepath
\color[rgb]{0.0,0.0,0.0}
\pgfusepath{stroke}
\end{pgfscope}
\begin{pgfscope}
\pgfsetlinewidth{1.0bp}
\pgfsetrectcap 
\pgfsetmiterjoin \pgfsetmiterlimit{10.0}
\pgfpathmoveto{\pgfpoint{297.444487bp}{85.544676bp}}
\pgfpathcurveto{\pgfpoint{294.189341bp}{85.544676bp}}{\pgfpoint{291.550522bp}{88.075783bp}}{\pgfpoint{291.550522bp}{91.198067bp}}
\pgfpathcurveto{\pgfpoint{291.550522bp}{94.320352bp}}{\pgfpoint{294.189341bp}{96.851465bp}}{\pgfpoint{297.444487bp}{96.851465bp}}
\pgfpathcurveto{\pgfpoint{300.699635bp}{96.851465bp}}{\pgfpoint{303.338452bp}{94.320352bp}}{\pgfpoint{303.338452bp}{91.198067bp}}
\pgfpathcurveto{\pgfpoint{303.338452bp}{88.075783bp}}{\pgfpoint{300.699635bp}{85.544676bp}}{\pgfpoint{297.444487bp}{85.544676bp}}
\pgfclosepath
\color[rgb]{0.0,0.0,0.0}\pgfseteorule\pgfusepath{fill}
\pgfpathmoveto{\pgfpoint{297.444487bp}{85.544676bp}}
\pgfpathcurveto{\pgfpoint{294.189341bp}{85.544676bp}}{\pgfpoint{291.550522bp}{88.075783bp}}{\pgfpoint{291.550522bp}{91.198067bp}}
\pgfpathcurveto{\pgfpoint{291.550522bp}{94.320352bp}}{\pgfpoint{294.189341bp}{96.851465bp}}{\pgfpoint{297.444487bp}{96.851465bp}}
\pgfpathcurveto{\pgfpoint{300.699635bp}{96.851465bp}}{\pgfpoint{303.338452bp}{94.320352bp}}{\pgfpoint{303.338452bp}{91.198067bp}}
\pgfpathcurveto{\pgfpoint{303.338452bp}{88.075783bp}}{\pgfpoint{300.699635bp}{85.544676bp}}{\pgfpoint{297.444487bp}{85.544676bp}}
\pgfclosepath
\color[rgb]{0.0,0.0,0.0}
\pgfusepath{stroke}
\end{pgfscope}
\begin{pgfscope}
\pgfsetlinewidth{1.0bp}
\pgfsetrectcap 
\pgfsetmiterjoin \pgfsetmiterlimit{10.0}
\pgfpathmoveto{\pgfpoint{257.844513bp}{146.844669bp}}
\pgfpathcurveto{\pgfpoint{254.589367bp}{146.844669bp}}{\pgfpoint{251.950548bp}{149.375776bp}}{\pgfpoint{251.950548bp}{152.49806bp}}
\pgfpathcurveto{\pgfpoint{251.950548bp}{155.620344bp}}{\pgfpoint{254.589367bp}{158.151457bp}}{\pgfpoint{257.844513bp}{158.151457bp}}
\pgfpathcurveto{\pgfpoint{261.099661bp}{158.151457bp}}{\pgfpoint{263.738478bp}{155.620344bp}}{\pgfpoint{263.738478bp}{152.49806bp}}
\pgfpathcurveto{\pgfpoint{263.738478bp}{149.375776bp}}{\pgfpoint{261.099661bp}{146.844669bp}}{\pgfpoint{257.844513bp}{146.844669bp}}
\pgfclosepath
\color[rgb]{0.0,0.0,0.0}\pgfseteorule\pgfusepath{fill}
\pgfpathmoveto{\pgfpoint{257.844513bp}{146.844669bp}}
\pgfpathcurveto{\pgfpoint{254.589367bp}{146.844669bp}}{\pgfpoint{251.950548bp}{149.375776bp}}{\pgfpoint{251.950548bp}{152.49806bp}}
\pgfpathcurveto{\pgfpoint{251.950548bp}{155.620344bp}}{\pgfpoint{254.589367bp}{158.151457bp}}{\pgfpoint{257.844513bp}{158.151457bp}}
\pgfpathcurveto{\pgfpoint{261.099661bp}{158.151457bp}}{\pgfpoint{263.738478bp}{155.620344bp}}{\pgfpoint{263.738478bp}{152.49806bp}}
\pgfpathcurveto{\pgfpoint{263.738478bp}{149.375776bp}}{\pgfpoint{261.099661bp}{146.844669bp}}{\pgfpoint{257.844513bp}{146.844669bp}}
\pgfclosepath
\color[rgb]{0.0,0.0,0.0}
\pgfusepath{stroke}
\end{pgfscope}
\begin{pgfscope}
\pgfsetlinewidth{1.0bp}
\pgfsetrectcap 
\pgfsetmiterjoin \pgfsetmiterlimit{10.0}
\pgfpathmoveto{\pgfpoint{257.844548bp}{85.54467bp}}
\pgfpathcurveto{\pgfpoint{254.589402bp}{85.54467bp}}{\pgfpoint{251.950583bp}{88.075777bp}}{\pgfpoint{251.950583bp}{91.198061bp}}
\pgfpathcurveto{\pgfpoint{251.950583bp}{94.320346bp}}{\pgfpoint{254.589402bp}{96.851459bp}}{\pgfpoint{257.844548bp}{96.851459bp}}
\pgfpathcurveto{\pgfpoint{261.099696bp}{96.851459bp}}{\pgfpoint{263.738513bp}{94.320346bp}}{\pgfpoint{263.738513bp}{91.198061bp}}
\pgfpathcurveto{\pgfpoint{263.738513bp}{88.075777bp}}{\pgfpoint{261.099696bp}{85.54467bp}}{\pgfpoint{257.844548bp}{85.54467bp}}
\pgfclosepath
\color[rgb]{0.0,0.0,0.0}\pgfseteorule\pgfusepath{fill}
\pgfpathmoveto{\pgfpoint{257.844548bp}{85.54467bp}}
\pgfpathcurveto{\pgfpoint{254.589402bp}{85.54467bp}}{\pgfpoint{251.950583bp}{88.075777bp}}{\pgfpoint{251.950583bp}{91.198061bp}}
\pgfpathcurveto{\pgfpoint{251.950583bp}{94.320346bp}}{\pgfpoint{254.589402bp}{96.851459bp}}{\pgfpoint{257.844548bp}{96.851459bp}}
\pgfpathcurveto{\pgfpoint{261.099696bp}{96.851459bp}}{\pgfpoint{263.738513bp}{94.320346bp}}{\pgfpoint{263.738513bp}{91.198061bp}}
\pgfpathcurveto{\pgfpoint{263.738513bp}{88.075777bp}}{\pgfpoint{261.099696bp}{85.54467bp}}{\pgfpoint{257.844548bp}{85.54467bp}}
\pgfclosepath
\color[rgb]{0.0,0.0,0.0}
\pgfusepath{stroke}
\end{pgfscope}
\begin{pgfscope}
\pgfsetlinewidth{1.0bp}
\pgfsetrectcap 
\pgfsetmiterjoin \pgfsetmiterlimit{10.0}
\pgfpathmoveto{\pgfpoint{217.444513bp}{146.844675bp}}
\pgfpathcurveto{\pgfpoint{214.189367bp}{146.844675bp}}{\pgfpoint{211.550548bp}{149.375782bp}}{\pgfpoint{211.550548bp}{152.498066bp}}
\pgfpathcurveto{\pgfpoint{211.550548bp}{155.62035bp}}{\pgfpoint{214.189367bp}{158.151463bp}}{\pgfpoint{217.444513bp}{158.151463bp}}
\pgfpathcurveto{\pgfpoint{220.699661bp}{158.151463bp}}{\pgfpoint{223.338478bp}{155.62035bp}}{\pgfpoint{223.338478bp}{152.498066bp}}
\pgfpathcurveto{\pgfpoint{223.338478bp}{149.375782bp}}{\pgfpoint{220.699661bp}{146.844675bp}}{\pgfpoint{217.444513bp}{146.844675bp}}
\pgfclosepath
\color[rgb]{0.0,0.0,0.0}\pgfseteorule\pgfusepath{fill}
\pgfpathmoveto{\pgfpoint{217.444513bp}{146.844675bp}}
\pgfpathcurveto{\pgfpoint{214.189367bp}{146.844675bp}}{\pgfpoint{211.550548bp}{149.375782bp}}{\pgfpoint{211.550548bp}{152.498066bp}}
\pgfpathcurveto{\pgfpoint{211.550548bp}{155.62035bp}}{\pgfpoint{214.189367bp}{158.151463bp}}{\pgfpoint{217.444513bp}{158.151463bp}}
\pgfpathcurveto{\pgfpoint{220.699661bp}{158.151463bp}}{\pgfpoint{223.338478bp}{155.62035bp}}{\pgfpoint{223.338478bp}{152.498066bp}}
\pgfpathcurveto{\pgfpoint{223.338478bp}{149.375782bp}}{\pgfpoint{220.699661bp}{146.844675bp}}{\pgfpoint{217.444513bp}{146.844675bp}}
\pgfclosepath
\color[rgb]{0.0,0.0,0.0}
\pgfusepath{stroke}
\end{pgfscope}
\begin{pgfscope}
\pgfsetlinewidth{1.0bp}
\pgfsetrectcap 
\pgfsetmiterjoin \pgfsetmiterlimit{10.0}
\pgfpathmoveto{\pgfpoint{217.444518bp}{85.544676bp}}
\pgfpathcurveto{\pgfpoint{214.189371bp}{85.544676bp}}{\pgfpoint{211.550552bp}{88.075783bp}}{\pgfpoint{211.550552bp}{91.198067bp}}
\pgfpathcurveto{\pgfpoint{211.550552bp}{94.320352bp}}{\pgfpoint{214.189371bp}{96.851465bp}}{\pgfpoint{217.444518bp}{96.851465bp}}
\pgfpathcurveto{\pgfpoint{220.699665bp}{96.851465bp}}{\pgfpoint{223.338483bp}{94.320352bp}}{\pgfpoint{223.338483bp}{91.198067bp}}
\pgfpathcurveto{\pgfpoint{223.338483bp}{88.075783bp}}{\pgfpoint{220.699665bp}{85.544676bp}}{\pgfpoint{217.444518bp}{85.544676bp}}
\pgfclosepath
\color[rgb]{0.0,0.0,0.0}\pgfseteorule\pgfusepath{fill}
\pgfpathmoveto{\pgfpoint{217.444518bp}{85.544676bp}}
\pgfpathcurveto{\pgfpoint{214.189371bp}{85.544676bp}}{\pgfpoint{211.550552bp}{88.075783bp}}{\pgfpoint{211.550552bp}{91.198067bp}}
\pgfpathcurveto{\pgfpoint{211.550552bp}{94.320352bp}}{\pgfpoint{214.189371bp}{96.851465bp}}{\pgfpoint{217.444518bp}{96.851465bp}}
\pgfpathcurveto{\pgfpoint{220.699665bp}{96.851465bp}}{\pgfpoint{223.338483bp}{94.320352bp}}{\pgfpoint{223.338483bp}{91.198067bp}}
\pgfpathcurveto{\pgfpoint{223.338483bp}{88.075783bp}}{\pgfpoint{220.699665bp}{85.544676bp}}{\pgfpoint{217.444518bp}{85.544676bp}}
\pgfclosepath
\color[rgb]{0.0,0.0,0.0}
\pgfusepath{stroke}
\end{pgfscope}
\begin{pgfscope}
\pgfsetlinewidth{1.0bp}
\pgfsetrectcap 
\pgfsetmiterjoin \pgfsetmiterlimit{10.0}
\pgfpathmoveto{\pgfpoint{177.844513bp}{146.844675bp}}
\pgfpathcurveto{\pgfpoint{174.589367bp}{146.844675bp}}{\pgfpoint{171.950548bp}{149.375782bp}}{\pgfpoint{171.950548bp}{152.498066bp}}
\pgfpathcurveto{\pgfpoint{171.950548bp}{155.62035bp}}{\pgfpoint{174.589367bp}{158.151463bp}}{\pgfpoint{177.844513bp}{158.151463bp}}
\pgfpathcurveto{\pgfpoint{181.099661bp}{158.151463bp}}{\pgfpoint{183.738478bp}{155.62035bp}}{\pgfpoint{183.738478bp}{152.498066bp}}
\pgfpathcurveto{\pgfpoint{183.738478bp}{149.375782bp}}{\pgfpoint{181.099661bp}{146.844675bp}}{\pgfpoint{177.844513bp}{146.844675bp}}
\pgfclosepath
\color[rgb]{0.0,0.0,0.0}\pgfseteorule\pgfusepath{fill}
\pgfpathmoveto{\pgfpoint{177.844513bp}{146.844675bp}}
\pgfpathcurveto{\pgfpoint{174.589367bp}{146.844675bp}}{\pgfpoint{171.950548bp}{149.375782bp}}{\pgfpoint{171.950548bp}{152.498066bp}}
\pgfpathcurveto{\pgfpoint{171.950548bp}{155.62035bp}}{\pgfpoint{174.589367bp}{158.151463bp}}{\pgfpoint{177.844513bp}{158.151463bp}}
\pgfpathcurveto{\pgfpoint{181.099661bp}{158.151463bp}}{\pgfpoint{183.738478bp}{155.62035bp}}{\pgfpoint{183.738478bp}{152.498066bp}}
\pgfpathcurveto{\pgfpoint{183.738478bp}{149.375782bp}}{\pgfpoint{181.099661bp}{146.844675bp}}{\pgfpoint{177.844513bp}{146.844675bp}}
\pgfclosepath
\color[rgb]{0.0,0.0,0.0}
\pgfusepath{stroke}
\end{pgfscope}
\begin{pgfscope}
\pgfsetlinewidth{1.0bp}
\pgfsetrectcap 
\pgfsetmiterjoin \pgfsetmiterlimit{10.0}
\pgfpathmoveto{\pgfpoint{177.844518bp}{85.544676bp}}
\pgfpathcurveto{\pgfpoint{174.589371bp}{85.544676bp}}{\pgfpoint{171.950552bp}{88.075783bp}}{\pgfpoint{171.950552bp}{91.198067bp}}
\pgfpathcurveto{\pgfpoint{171.950552bp}{94.320352bp}}{\pgfpoint{174.589371bp}{96.851465bp}}{\pgfpoint{177.844518bp}{96.851465bp}}
\pgfpathcurveto{\pgfpoint{181.099665bp}{96.851465bp}}{\pgfpoint{183.738483bp}{94.320352bp}}{\pgfpoint{183.738483bp}{91.198067bp}}
\pgfpathcurveto{\pgfpoint{183.738483bp}{88.075783bp}}{\pgfpoint{181.099665bp}{85.544676bp}}{\pgfpoint{177.844518bp}{85.544676bp}}
\pgfclosepath
\color[rgb]{0.0,0.0,0.0}\pgfseteorule\pgfusepath{fill}
\pgfpathmoveto{\pgfpoint{177.844518bp}{85.544676bp}}
\pgfpathcurveto{\pgfpoint{174.589371bp}{85.544676bp}}{\pgfpoint{171.950552bp}{88.075783bp}}{\pgfpoint{171.950552bp}{91.198067bp}}
\pgfpathcurveto{\pgfpoint{171.950552bp}{94.320352bp}}{\pgfpoint{174.589371bp}{96.851465bp}}{\pgfpoint{177.844518bp}{96.851465bp}}
\pgfpathcurveto{\pgfpoint{181.099665bp}{96.851465bp}}{\pgfpoint{183.738483bp}{94.320352bp}}{\pgfpoint{183.738483bp}{91.198067bp}}
\pgfpathcurveto{\pgfpoint{183.738483bp}{88.075783bp}}{\pgfpoint{181.099665bp}{85.544676bp}}{\pgfpoint{177.844518bp}{85.544676bp}}
\pgfclosepath
\color[rgb]{0.0,0.0,0.0}
\pgfusepath{stroke}
\end{pgfscope}
\begin{pgfscope}
\pgfsetlinewidth{1.0bp}
\pgfsetrectcap 
\pgfsetmiterjoin \pgfsetmiterlimit{10.0}
\pgfpathmoveto{\pgfpoint{138.044513bp}{146.844687bp}}
\pgfpathcurveto{\pgfpoint{134.789367bp}{146.844687bp}}{\pgfpoint{132.150548bp}{149.375794bp}}{\pgfpoint{132.150548bp}{152.498078bp}}
\pgfpathcurveto{\pgfpoint{132.150548bp}{155.620362bp}}{\pgfpoint{134.789367bp}{158.151476bp}}{\pgfpoint{138.044513bp}{158.151476bp}}
\pgfpathcurveto{\pgfpoint{141.299661bp}{158.151476bp}}{\pgfpoint{143.938478bp}{155.620362bp}}{\pgfpoint{143.938478bp}{152.498078bp}}
\pgfpathcurveto{\pgfpoint{143.938478bp}{149.375794bp}}{\pgfpoint{141.299661bp}{146.844687bp}}{\pgfpoint{138.044513bp}{146.844687bp}}
\pgfclosepath
\color[rgb]{0.0,0.0,0.0}\pgfseteorule\pgfusepath{fill}
\pgfpathmoveto{\pgfpoint{138.044513bp}{146.844687bp}}
\pgfpathcurveto{\pgfpoint{134.789367bp}{146.844687bp}}{\pgfpoint{132.150548bp}{149.375794bp}}{\pgfpoint{132.150548bp}{152.498078bp}}
\pgfpathcurveto{\pgfpoint{132.150548bp}{155.620362bp}}{\pgfpoint{134.789367bp}{158.151476bp}}{\pgfpoint{138.044513bp}{158.151476bp}}
\pgfpathcurveto{\pgfpoint{141.299661bp}{158.151476bp}}{\pgfpoint{143.938478bp}{155.620362bp}}{\pgfpoint{143.938478bp}{152.498078bp}}
\pgfpathcurveto{\pgfpoint{143.938478bp}{149.375794bp}}{\pgfpoint{141.299661bp}{146.844687bp}}{\pgfpoint{138.044513bp}{146.844687bp}}
\pgfclosepath
\color[rgb]{0.0,0.0,0.0}
\pgfusepath{stroke}
\end{pgfscope}
\begin{pgfscope}
\pgfsetlinewidth{1.0bp}
\pgfsetrectcap 
\pgfsetmiterjoin \pgfsetmiterlimit{10.0}
\pgfpathmoveto{\pgfpoint{138.044536bp}{85.544676bp}}
\pgfpathcurveto{\pgfpoint{134.78939bp}{85.544676bp}}{\pgfpoint{132.150571bp}{88.075783bp}}{\pgfpoint{132.150571bp}{91.198067bp}}
\pgfpathcurveto{\pgfpoint{132.150571bp}{94.320352bp}}{\pgfpoint{134.78939bp}{96.851465bp}}{\pgfpoint{138.044536bp}{96.851465bp}}
\pgfpathcurveto{\pgfpoint{141.299684bp}{96.851465bp}}{\pgfpoint{143.938501bp}{94.320352bp}}{\pgfpoint{143.938501bp}{91.198067bp}}
\pgfpathcurveto{\pgfpoint{143.938501bp}{88.075783bp}}{\pgfpoint{141.299684bp}{85.544676bp}}{\pgfpoint{138.044536bp}{85.544676bp}}
\pgfclosepath
\color[rgb]{0.0,0.0,0.0}\pgfseteorule\pgfusepath{fill}
\pgfpathmoveto{\pgfpoint{138.044536bp}{85.544676bp}}
\pgfpathcurveto{\pgfpoint{134.78939bp}{85.544676bp}}{\pgfpoint{132.150571bp}{88.075783bp}}{\pgfpoint{132.150571bp}{91.198067bp}}
\pgfpathcurveto{\pgfpoint{132.150571bp}{94.320352bp}}{\pgfpoint{134.78939bp}{96.851465bp}}{\pgfpoint{138.044536bp}{96.851465bp}}
\pgfpathcurveto{\pgfpoint{141.299684bp}{96.851465bp}}{\pgfpoint{143.938501bp}{94.320352bp}}{\pgfpoint{143.938501bp}{91.198067bp}}
\pgfpathcurveto{\pgfpoint{143.938501bp}{88.075783bp}}{\pgfpoint{141.299684bp}{85.544676bp}}{\pgfpoint{138.044536bp}{85.544676bp}}
\pgfclosepath
\color[rgb]{0.0,0.0,0.0}
\pgfusepath{stroke}
\end{pgfscope}
\begin{pgfscope}
\pgfsetlinewidth{1.0bp}
\pgfsetrectcap 
\pgfsetmiterjoin \pgfsetmiterlimit{10.0}
\pgfpathmoveto{\pgfpoint{98.044513bp}{146.844687bp}}
\pgfpathcurveto{\pgfpoint{94.789367bp}{146.844687bp}}{\pgfpoint{92.150548bp}{149.375794bp}}{\pgfpoint{92.150548bp}{152.498078bp}}
\pgfpathcurveto{\pgfpoint{92.150548bp}{155.620362bp}}{\pgfpoint{94.789367bp}{158.151476bp}}{\pgfpoint{98.044513bp}{158.151476bp}}
\pgfpathcurveto{\pgfpoint{101.299661bp}{158.151476bp}}{\pgfpoint{103.938478bp}{155.620362bp}}{\pgfpoint{103.938478bp}{152.498078bp}}
\pgfpathcurveto{\pgfpoint{103.938478bp}{149.375794bp}}{\pgfpoint{101.299661bp}{146.844687bp}}{\pgfpoint{98.044513bp}{146.844687bp}}
\pgfclosepath
\color[rgb]{0.0,0.0,0.0}\pgfseteorule\pgfusepath{fill}
\pgfpathmoveto{\pgfpoint{98.044513bp}{146.844687bp}}
\pgfpathcurveto{\pgfpoint{94.789367bp}{146.844687bp}}{\pgfpoint{92.150548bp}{149.375794bp}}{\pgfpoint{92.150548bp}{152.498078bp}}
\pgfpathcurveto{\pgfpoint{92.150548bp}{155.620362bp}}{\pgfpoint{94.789367bp}{158.151476bp}}{\pgfpoint{98.044513bp}{158.151476bp}}
\pgfpathcurveto{\pgfpoint{101.299661bp}{158.151476bp}}{\pgfpoint{103.938478bp}{155.620362bp}}{\pgfpoint{103.938478bp}{152.498078bp}}
\pgfpathcurveto{\pgfpoint{103.938478bp}{149.375794bp}}{\pgfpoint{101.299661bp}{146.844687bp}}{\pgfpoint{98.044513bp}{146.844687bp}}
\pgfclosepath
\color[rgb]{0.0,0.0,0.0}
\pgfusepath{stroke}
\end{pgfscope}
\begin{pgfscope}
\pgfsetlinewidth{1.0bp}
\pgfsetrectcap 
\pgfsetmiterjoin \pgfsetmiterlimit{10.0}
\pgfpathmoveto{\pgfpoint{98.044518bp}{85.544688bp}}
\pgfpathcurveto{\pgfpoint{94.789371bp}{85.544688bp}}{\pgfpoint{92.150552bp}{88.075795bp}}{\pgfpoint{92.150552bp}{91.19808bp}}
\pgfpathcurveto{\pgfpoint{92.150552bp}{94.320364bp}}{\pgfpoint{94.789371bp}{96.851477bp}}{\pgfpoint{98.044518bp}{96.851477bp}}
\pgfpathcurveto{\pgfpoint{101.299665bp}{96.851477bp}}{\pgfpoint{103.938483bp}{94.320364bp}}{\pgfpoint{103.938483bp}{91.19808bp}}
\pgfpathcurveto{\pgfpoint{103.938483bp}{88.075795bp}}{\pgfpoint{101.299665bp}{85.544688bp}}{\pgfpoint{98.044518bp}{85.544688bp}}
\pgfclosepath
\color[rgb]{0.0,0.0,0.0}\pgfseteorule\pgfusepath{fill}
\pgfpathmoveto{\pgfpoint{98.044518bp}{85.544688bp}}
\pgfpathcurveto{\pgfpoint{94.789371bp}{85.544688bp}}{\pgfpoint{92.150552bp}{88.075795bp}}{\pgfpoint{92.150552bp}{91.19808bp}}
\pgfpathcurveto{\pgfpoint{92.150552bp}{94.320364bp}}{\pgfpoint{94.789371bp}{96.851477bp}}{\pgfpoint{98.044518bp}{96.851477bp}}
\pgfpathcurveto{\pgfpoint{101.299665bp}{96.851477bp}}{\pgfpoint{103.938483bp}{94.320364bp}}{\pgfpoint{103.938483bp}{91.19808bp}}
\pgfpathcurveto{\pgfpoint{103.938483bp}{88.075795bp}}{\pgfpoint{101.299665bp}{85.544688bp}}{\pgfpoint{98.044518bp}{85.544688bp}}
\pgfclosepath
\color[rgb]{0.0,0.0,0.0}
\pgfusepath{stroke}
\end{pgfscope}
\begin{pgfscope}
\pgfsetlinewidth{1.0bp}
\pgfsetrectcap 
\pgfsetmiterjoin \pgfsetmiterlimit{10.0}
\pgfpathmoveto{\pgfpoint{57.644513bp}{146.844693bp}}
\pgfpathcurveto{\pgfpoint{54.389367bp}{146.844693bp}}{\pgfpoint{51.750548bp}{149.3758bp}}{\pgfpoint{51.750548bp}{152.498084bp}}
\pgfpathcurveto{\pgfpoint{51.750548bp}{155.620369bp}}{\pgfpoint{54.389367bp}{158.151482bp}}{\pgfpoint{57.644513bp}{158.151482bp}}
\pgfpathcurveto{\pgfpoint{60.899661bp}{158.151482bp}}{\pgfpoint{63.538478bp}{155.620369bp}}{\pgfpoint{63.538478bp}{152.498084bp}}
\pgfpathcurveto{\pgfpoint{63.538478bp}{149.3758bp}}{\pgfpoint{60.899661bp}{146.844693bp}}{\pgfpoint{57.644513bp}{146.844693bp}}
\pgfclosepath
\color[rgb]{0.0,0.0,0.0}\pgfseteorule\pgfusepath{fill}
\pgfpathmoveto{\pgfpoint{57.644513bp}{146.844693bp}}
\pgfpathcurveto{\pgfpoint{54.389367bp}{146.844693bp}}{\pgfpoint{51.750548bp}{149.3758bp}}{\pgfpoint{51.750548bp}{152.498084bp}}
\pgfpathcurveto{\pgfpoint{51.750548bp}{155.620369bp}}{\pgfpoint{54.389367bp}{158.151482bp}}{\pgfpoint{57.644513bp}{158.151482bp}}
\pgfpathcurveto{\pgfpoint{60.899661bp}{158.151482bp}}{\pgfpoint{63.538478bp}{155.620369bp}}{\pgfpoint{63.538478bp}{152.498084bp}}
\pgfpathcurveto{\pgfpoint{63.538478bp}{149.3758bp}}{\pgfpoint{60.899661bp}{146.844693bp}}{\pgfpoint{57.644513bp}{146.844693bp}}
\pgfclosepath
\color[rgb]{0.0,0.0,0.0}
\pgfusepath{stroke}
\end{pgfscope}
\begin{pgfscope}
\pgfsetlinewidth{1.0bp}
\pgfsetrectcap 
\pgfsetmiterjoin \pgfsetmiterlimit{10.0}
\pgfpathmoveto{\pgfpoint{57.644518bp}{85.544664bp}}
\pgfpathcurveto{\pgfpoint{54.389371bp}{85.544664bp}}{\pgfpoint{51.750552bp}{88.075771bp}}{\pgfpoint{51.750552bp}{91.198055bp}}
\pgfpathcurveto{\pgfpoint{51.750552bp}{94.32034bp}}{\pgfpoint{54.389371bp}{96.851453bp}}{\pgfpoint{57.644518bp}{96.851453bp}}
\pgfpathcurveto{\pgfpoint{60.899665bp}{96.851453bp}}{\pgfpoint{63.538483bp}{94.32034bp}}{\pgfpoint{63.538483bp}{91.198055bp}}
\pgfpathcurveto{\pgfpoint{63.538483bp}{88.075771bp}}{\pgfpoint{60.899665bp}{85.544664bp}}{\pgfpoint{57.644518bp}{85.544664bp}}
\pgfclosepath
\color[rgb]{0.0,0.0,0.0}\pgfseteorule\pgfusepath{fill}
\pgfpathmoveto{\pgfpoint{57.644518bp}{85.544664bp}}
\pgfpathcurveto{\pgfpoint{54.389371bp}{85.544664bp}}{\pgfpoint{51.750552bp}{88.075771bp}}{\pgfpoint{51.750552bp}{91.198055bp}}
\pgfpathcurveto{\pgfpoint{51.750552bp}{94.32034bp}}{\pgfpoint{54.389371bp}{96.851453bp}}{\pgfpoint{57.644518bp}{96.851453bp}}
\pgfpathcurveto{\pgfpoint{60.899665bp}{96.851453bp}}{\pgfpoint{63.538483bp}{94.32034bp}}{\pgfpoint{63.538483bp}{91.198055bp}}
\pgfpathcurveto{\pgfpoint{63.538483bp}{88.075771bp}}{\pgfpoint{60.899665bp}{85.544664bp}}{\pgfpoint{57.644518bp}{85.544664bp}}
\pgfclosepath
\color[rgb]{0.0,0.0,0.0}
\pgfusepath{stroke}
\end{pgfscope}
\begin{pgfscope}
\pgfsetlinewidth{1.0bp}
\pgfsetrectcap 
\pgfsetmiterjoin \pgfsetmiterlimit{10.0}
\pgfpathmoveto{\pgfpoint{18.044513bp}{146.844687bp}}
\pgfpathcurveto{\pgfpoint{14.789367bp}{146.844687bp}}{\pgfpoint{12.150548bp}{149.375794bp}}{\pgfpoint{12.150548bp}{152.498078bp}}
\pgfpathcurveto{\pgfpoint{12.150548bp}{155.620362bp}}{\pgfpoint{14.789367bp}{158.151476bp}}{\pgfpoint{18.044513bp}{158.151476bp}}
\pgfpathcurveto{\pgfpoint{21.299661bp}{158.151476bp}}{\pgfpoint{23.938478bp}{155.620362bp}}{\pgfpoint{23.938478bp}{152.498078bp}}
\pgfpathcurveto{\pgfpoint{23.938478bp}{149.375794bp}}{\pgfpoint{21.299661bp}{146.844687bp}}{\pgfpoint{18.044513bp}{146.844687bp}}
\pgfclosepath
\color[rgb]{0.0,0.0,0.0}\pgfseteorule\pgfusepath{fill}
\pgfpathmoveto{\pgfpoint{18.044513bp}{146.844687bp}}
\pgfpathcurveto{\pgfpoint{14.789367bp}{146.844687bp}}{\pgfpoint{12.150548bp}{149.375794bp}}{\pgfpoint{12.150548bp}{152.498078bp}}
\pgfpathcurveto{\pgfpoint{12.150548bp}{155.620362bp}}{\pgfpoint{14.789367bp}{158.151476bp}}{\pgfpoint{18.044513bp}{158.151476bp}}
\pgfpathcurveto{\pgfpoint{21.299661bp}{158.151476bp}}{\pgfpoint{23.938478bp}{155.620362bp}}{\pgfpoint{23.938478bp}{152.498078bp}}
\pgfpathcurveto{\pgfpoint{23.938478bp}{149.375794bp}}{\pgfpoint{21.299661bp}{146.844687bp}}{\pgfpoint{18.044513bp}{146.844687bp}}
\pgfclosepath
\color[rgb]{0.0,0.0,0.0}
\pgfusepath{stroke}
\end{pgfscope}
\begin{pgfscope}
\pgfsetlinewidth{1.0bp}
\pgfsetrectcap 
\pgfsetmiterjoin \pgfsetmiterlimit{10.0}
\pgfpathmoveto{\pgfpoint{18.044518bp}{85.544688bp}}
\pgfpathcurveto{\pgfpoint{14.789371bp}{85.544688bp}}{\pgfpoint{12.150552bp}{88.075795bp}}{\pgfpoint{12.150552bp}{91.19808bp}}
\pgfpathcurveto{\pgfpoint{12.150552bp}{94.320364bp}}{\pgfpoint{14.789371bp}{96.851477bp}}{\pgfpoint{18.044518bp}{96.851477bp}}
\pgfpathcurveto{\pgfpoint{21.299665bp}{96.851477bp}}{\pgfpoint{23.938483bp}{94.320364bp}}{\pgfpoint{23.938483bp}{91.19808bp}}
\pgfpathcurveto{\pgfpoint{23.938483bp}{88.075795bp}}{\pgfpoint{21.299665bp}{85.544688bp}}{\pgfpoint{18.044518bp}{85.544688bp}}
\pgfclosepath
\color[rgb]{0.0,0.0,0.0}\pgfseteorule\pgfusepath{fill}
\pgfpathmoveto{\pgfpoint{18.044518bp}{85.544688bp}}
\pgfpathcurveto{\pgfpoint{14.789371bp}{85.544688bp}}{\pgfpoint{12.150552bp}{88.075795bp}}{\pgfpoint{12.150552bp}{91.19808bp}}
\pgfpathcurveto{\pgfpoint{12.150552bp}{94.320364bp}}{\pgfpoint{14.789371bp}{96.851477bp}}{\pgfpoint{18.044518bp}{96.851477bp}}
\pgfpathcurveto{\pgfpoint{21.299665bp}{96.851477bp}}{\pgfpoint{23.938483bp}{94.320364bp}}{\pgfpoint{23.938483bp}{91.19808bp}}
\pgfpathcurveto{\pgfpoint{23.938483bp}{88.075795bp}}{\pgfpoint{21.299665bp}{85.544688bp}}{\pgfpoint{18.044518bp}{85.544688bp}}
\pgfclosepath
\color[rgb]{0.0,0.0,0.0}
\pgfusepath{stroke}
\end{pgfscope}
\begin{pgfscope}
\pgftransformcm{1.0}{0.0}{0.0}{1.0}{\pgfpoint{278.484375bp}{144.0bp}}
\pgftext[left,base]{\sffamily\mdseries\upshape\LARGE
\color[rgb]{0.0,0.0,0.0}$9$}
\end{pgfscope}
\begin{pgfscope}
\pgftransformcm{1.0}{0.0}{0.0}{1.0}{\pgfpoint{277.884375bp}{87.4bp}}
\pgftext[left,base]{\sffamily\mdseries\upshape\LARGE
\color[rgb]{0.0,0.0,0.0}$9'$}
\end{pgfscope}
\begin{pgfscope}
\pgfsetlinewidth{1.0bp}
\pgfsetrectcap 
\pgfsetmiterjoin \pgfsetmiterlimit{10.0}
\pgfpathmoveto{\pgfpoint{154.944525bp}{207.444669bp}}
\pgfpathcurveto{\pgfpoint{151.689379bp}{207.444669bp}}{\pgfpoint{149.05056bp}{209.975776bp}}{\pgfpoint{149.05056bp}{213.09806bp}}
\pgfpathcurveto{\pgfpoint{149.05056bp}{216.220344bp}}{\pgfpoint{151.689379bp}{218.751457bp}}{\pgfpoint{154.944525bp}{218.751457bp}}
\pgfpathcurveto{\pgfpoint{158.199673bp}{218.751457bp}}{\pgfpoint{160.838491bp}{216.220344bp}}{\pgfpoint{160.838491bp}{213.09806bp}}
\pgfpathcurveto{\pgfpoint{160.838491bp}{209.975776bp}}{\pgfpoint{158.199673bp}{207.444669bp}}{\pgfpoint{154.944525bp}{207.444669bp}}
\pgfclosepath
\color[rgb]{0.0,0.0,0.0}\pgfseteorule\pgfusepath{fill}
\pgfpathmoveto{\pgfpoint{154.944525bp}{207.444669bp}}
\pgfpathcurveto{\pgfpoint{151.689379bp}{207.444669bp}}{\pgfpoint{149.05056bp}{209.975776bp}}{\pgfpoint{149.05056bp}{213.09806bp}}
\pgfpathcurveto{\pgfpoint{149.05056bp}{216.220344bp}}{\pgfpoint{151.689379bp}{218.751457bp}}{\pgfpoint{154.944525bp}{218.751457bp}}
\pgfpathcurveto{\pgfpoint{158.199673bp}{218.751457bp}}{\pgfpoint{160.838491bp}{216.220344bp}}{\pgfpoint{160.838491bp}{213.09806bp}}
\pgfpathcurveto{\pgfpoint{160.838491bp}{209.975776bp}}{\pgfpoint{158.199673bp}{207.444669bp}}{\pgfpoint{154.944525bp}{207.444669bp}}
\pgfclosepath
\color[rgb]{0.0,0.0,0.0}
\pgfusepath{stroke}
\end{pgfscope}
\begin{pgfscope}
\pgfsetlinewidth{1.0bp}
\pgfsetrectcap 
\pgfsetmiterjoin \pgfsetmiterlimit{10.0}
\pgfpathmoveto{\pgfpoint{154.944524bp}{23.744688bp}}
\pgfpathcurveto{\pgfpoint{151.689378bp}{23.744688bp}}{\pgfpoint{149.050558bp}{26.275795bp}}{\pgfpoint{149.050558bp}{29.39808bp}}
\pgfpathcurveto{\pgfpoint{149.050558bp}{32.520364bp}}{\pgfpoint{151.689378bp}{35.051477bp}}{\pgfpoint{154.944524bp}{35.051477bp}}
\pgfpathcurveto{\pgfpoint{158.199671bp}{35.051477bp}}{\pgfpoint{160.838489bp}{32.520364bp}}{\pgfpoint{160.838489bp}{29.39808bp}}
\pgfpathcurveto{\pgfpoint{160.838489bp}{26.275795bp}}{\pgfpoint{158.199671bp}{23.744688bp}}{\pgfpoint{154.944524bp}{23.744688bp}}
\pgfclosepath
\color[rgb]{0.0,0.0,0.0}\pgfseteorule\pgfusepath{fill}
\pgfpathmoveto{\pgfpoint{154.944524bp}{23.744688bp}}
\pgfpathcurveto{\pgfpoint{151.689378bp}{23.744688bp}}{\pgfpoint{149.050558bp}{26.275795bp}}{\pgfpoint{149.050558bp}{29.39808bp}}
\pgfpathcurveto{\pgfpoint{149.050558bp}{32.520364bp}}{\pgfpoint{151.689378bp}{35.051477bp}}{\pgfpoint{154.944524bp}{35.051477bp}}
\pgfpathcurveto{\pgfpoint{158.199671bp}{35.051477bp}}{\pgfpoint{160.838489bp}{32.520364bp}}{\pgfpoint{160.838489bp}{29.39808bp}}
\pgfpathcurveto{\pgfpoint{160.838489bp}{26.275795bp}}{\pgfpoint{158.199671bp}{23.744688bp}}{\pgfpoint{154.944524bp}{23.744688bp}}
\pgfclosepath
\color[rgb]{0.0,0.0,0.0}
\pgfusepath{stroke}
\end{pgfscope}
\begin{pgfscope}
\pgftransformcm{1.0}{0.0}{0.0}{1.0}{\pgfpoint{164.878125bp}{18.8bp}}
\pgftext[left,base]{\sffamily\mdseries\upshape\LARGE
\color[rgb]{0.0,0.0,0.0}$1$}
\end{pgfscope}
\begin{pgfscope}
\pgftransformcm{1.0}{0.0}{0.0}{1.0}{\pgfpoint{164.878125bp}{213.6bp}}
\pgftext[left,base]{\sffamily\mdseries\upshape\LARGE
\color[rgb]{0.0,0.0,0.0}$1'$}
\end{pgfscope}
\begin{pgfscope}
\pgftransformcm{1.0}{0.0}{0.0}{1.0}{\pgfpoint{-0.525bp}{148.2bp}}
\pgftext[left,base]{\sffamily\mdseries\upshape\LARGE
\color[rgb]{0.0,0.0,0.0}$2$}
\end{pgfscope}
\begin{pgfscope}
\pgftransformcm{1.0}{0.0}{0.0}{1.0}{\pgfpoint{-0.525bp}{87.0bp}}
\pgftext[left,base]{\sffamily\mdseries\upshape\LARGE
\color[rgb]{0.0,0.0,0.0}$2'$}
\end{pgfscope}
\begin{pgfscope}
\pgftransformcm{1.0}{0.0}{0.0}{1.0}{\pgfpoint{37.875bp}{146.8bp}}
\pgftext[left,base]{\sffamily\mdseries\upshape\LARGE
\color[rgb]{0.0,0.0,0.0}$3$}
\end{pgfscope}
\begin{pgfscope}
\pgftransformcm{1.0}{0.0}{0.0}{1.0}{\pgfpoint{37.875bp}{85.2bp}}
\pgftext[left,base]{\sffamily\mdseries\upshape\LARGE
\color[rgb]{0.0,0.0,0.0}$3'$}
\end{pgfscope}
\begin{pgfscope}
\pgftransformcm{1.0}{0.0}{0.0}{1.0}{\pgfpoint{79.03125bp}{147.4bp}}
\pgftext[left,base]{\sffamily\mdseries\upshape\LARGE
\color[rgb]{0.0,0.0,0.0}$4$}
\end{pgfscope}
\begin{pgfscope}
\pgftransformcm{1.0}{0.0}{0.0}{1.0}{\pgfpoint{79.03125bp}{85.6bp}}
\pgftext[left,base]{\sffamily\mdseries\upshape\LARGE
\color[rgb]{0.0,0.0,0.0}$4'$}
\end{pgfscope}
\begin{pgfscope}
\pgftransformcm{1.0}{0.0}{0.0}{1.0}{\pgfpoint{118.284375bp}{147.2bp}}
\pgftext[left,base]{\sffamily\mdseries\upshape\LARGE
\color[rgb]{0.0,0.0,0.0}$5$}
\end{pgfscope}
\begin{pgfscope}
\pgftransformcm{1.0}{0.0}{0.0}{1.0}{\pgfpoint{118.284375bp}{85.2bp}}
\pgftext[left,base]{\sffamily\mdseries\upshape\LARGE
\color[rgb]{0.0,0.0,0.0}$5'$}
\end{pgfscope}
\begin{pgfscope}
\pgftransformcm{1.0}{0.0}{0.0}{1.0}{\pgfpoint{157.93125bp}{148.0bp}}
\pgftext[left,base]{\sffamily\mdseries\upshape\LARGE
\color[rgb]{0.0,0.0,0.0}$6$}
\end{pgfscope}
\begin{pgfscope}
\pgftransformcm{1.0}{0.0}{0.0}{1.0}{\pgfpoint{157.93125bp}{85.6bp}}
\pgftext[left,base]{\sffamily\mdseries\upshape\LARGE
\color[rgb]{0.0,0.0,0.0}$6'$}
\end{pgfscope}
\begin{pgfscope}
\pgftransformcm{1.0}{0.0}{0.0}{1.0}{\pgfpoint{198.075bp}{146.6bp}}
\pgftext[left,base]{\sffamily\mdseries\upshape\LARGE
\color[rgb]{0.0,0.0,0.0}$7$}
\end{pgfscope}
\begin{pgfscope}
\pgftransformcm{1.0}{0.0}{0.0}{1.0}{\pgfpoint{198.075bp}{85.6bp}}
\pgftext[left,base]{\sffamily\mdseries\upshape\LARGE
\color[rgb]{0.0,0.0,0.0}$7'$}
\end{pgfscope}
\begin{pgfscope}
\pgftransformcm{1.0}{0.0}{0.0}{1.0}{\pgfpoint{237.675bp}{146.6bp}}
\pgftext[left,base]{\sffamily\mdseries\upshape\LARGE
\color[rgb]{0.0,0.0,0.0}$8$}
\end{pgfscope}
\begin{pgfscope}
\pgftransformcm{1.0}{0.0}{0.0}{1.0}{\pgfpoint{237.675bp}{86.0bp}}
\pgftext[left,base]{\sffamily\mdseries\upshape\LARGE
\color[rgb]{0.0,0.0,0.0}$8'$}
\end{pgfscope}
\begin{pgfscope}
\pgfsetlinewidth{1.0bp}
\pgfsetrectcap 
\pgfsetmiterjoin \pgfsetmiterlimit{10.0}
\pgfpathmoveto{\pgfpoint{155.475bp}{212.8bp}}
\pgfpathlineto{\pgfpoint{18.075bp}{152.8bp}}
\color[rgb]{0.0,0.0,0.0}
\pgfusepath{stroke}
\end{pgfscope}
\begin{pgfscope}
\pgfsetlinewidth{1.0bp}
\pgfsetrectcap 
\pgfsetmiterjoin \pgfsetmiterlimit{10.0}
\pgfpathmoveto{\pgfpoint{154.875bp}{212.8bp}}
\pgfpathlineto{\pgfpoint{57.675bp}{152.2bp}}
\color[rgb]{0.0,0.0,0.0}
\pgfusepath{stroke}
\end{pgfscope}
\begin{pgfscope}
\pgfsetlinewidth{1.0bp}
\pgfsetrectcap 
\pgfsetmiterjoin \pgfsetmiterlimit{10.0}
\pgfpathmoveto{\pgfpoint{154.875bp}{212.2bp}}
\pgfpathlineto{\pgfpoint{97.875bp}{152.8bp}}
\color[rgb]{0.0,0.0,0.0}
\pgfusepath{stroke}
\end{pgfscope}
\begin{pgfscope}
\pgfsetlinewidth{1.0bp}
\pgfsetrectcap 
\pgfsetmiterjoin \pgfsetmiterlimit{10.0}
\pgfpathmoveto{\pgfpoint{154.875bp}{212.2bp}}
\pgfpathlineto{\pgfpoint{138.075bp}{152.8bp}}
\color[rgb]{0.0,0.0,0.0}
\pgfusepath{stroke}
\end{pgfscope}
\begin{pgfscope}
\pgfsetlinewidth{1.0bp}
\pgfsetrectcap 
\pgfsetmiterjoin \pgfsetmiterlimit{10.0}
\pgfpathmoveto{\pgfpoint{154.875bp}{212.8bp}}
\pgfpathlineto{\pgfpoint{178.275bp}{152.2bp}}
\color[rgb]{0.0,0.0,0.0}
\pgfusepath{stroke}
\end{pgfscope}
\begin{pgfscope}
\pgfsetlinewidth{1.0bp}
\pgfsetrectcap 
\pgfsetmiterjoin \pgfsetmiterlimit{10.0}
\pgfpathmoveto{\pgfpoint{154.275bp}{212.2bp}}
\pgfpathlineto{\pgfpoint{217.275bp}{152.8bp}}
\color[rgb]{0.0,0.0,0.0}
\pgfusepath{stroke}
\end{pgfscope}
\begin{pgfscope}
\pgfsetlinewidth{1.0bp}
\pgfsetrectcap 
\pgfsetmiterjoin \pgfsetmiterlimit{10.0}
\pgfpathmoveto{\pgfpoint{154.875bp}{212.2bp}}
\pgfpathlineto{\pgfpoint{258.075bp}{152.2bp}}
\color[rgb]{0.0,0.0,0.0}
\pgfusepath{stroke}
\end{pgfscope}
\begin{pgfscope}
\pgfsetlinewidth{1.0bp}
\pgfsetrectcap 
\pgfsetmiterjoin \pgfsetmiterlimit{10.0}
\pgfpathmoveto{\pgfpoint{155.475bp}{212.2bp}}
\pgfpathlineto{\pgfpoint{297.675bp}{152.2bp}}
\color[rgb]{0.0,0.0,0.0}
\pgfusepath{stroke}
\end{pgfscope}
\begin{pgfscope}
\pgfsetlinewidth{1.0bp}
\pgfsetrectcap 
\pgfsetmiterjoin \pgfsetmiterlimit{10.0}
\pgfpathmoveto{\pgfpoint{155.475bp}{29.2bp}}
\pgfpathlineto{\pgfpoint{18.075bp}{91.6bp}}
\color[rgb]{0.0,0.0,0.0}
\pgfusepath{stroke}
\end{pgfscope}
\begin{pgfscope}
\pgfsetlinewidth{1.0bp}
\pgfsetrectcap 
\pgfsetmiterjoin \pgfsetmiterlimit{10.0}
\pgfpathmoveto{\pgfpoint{155.475bp}{29.2bp}}
\pgfpathlineto{\pgfpoint{57.675bp}{91.6bp}}
\color[rgb]{0.0,0.0,0.0}
\pgfusepath{stroke}
\end{pgfscope}
\begin{pgfscope}
\pgfsetlinewidth{1.0bp}
\pgfsetrectcap 
\pgfsetmiterjoin \pgfsetmiterlimit{10.0}
\pgfpathmoveto{\pgfpoint{154.875bp}{30.4bp}}
\pgfpathlineto{\pgfpoint{97.875bp}{91.0bp}}
\color[rgb]{0.0,0.0,0.0}
\pgfusepath{stroke}
\end{pgfscope}
\begin{pgfscope}
\pgfsetlinewidth{1.0bp}
\pgfsetrectcap 
\pgfsetmiterjoin \pgfsetmiterlimit{10.0}
\pgfpathmoveto{\pgfpoint{155.475bp}{30.4bp}}
\pgfpathlineto{\pgfpoint{137.475bp}{91.0bp}}
\color[rgb]{0.0,0.0,0.0}
\pgfusepath{stroke}
\end{pgfscope}
\begin{pgfscope}
\pgfsetlinewidth{1.0bp}
\pgfsetrectcap 
\pgfsetmiterjoin \pgfsetmiterlimit{10.0}
\pgfpathmoveto{\pgfpoint{155.475bp}{29.8bp}}
\pgfpathlineto{\pgfpoint{177.675bp}{91.0bp}}
\color[rgb]{0.0,0.0,0.0}
\pgfusepath{stroke}
\end{pgfscope}
\begin{pgfscope}
\pgfsetlinewidth{1.0bp}
\pgfsetrectcap 
\pgfsetmiterjoin \pgfsetmiterlimit{10.0}
\pgfpathmoveto{\pgfpoint{155.475bp}{29.2bp}}
\pgfpathlineto{\pgfpoint{217.275bp}{91.0bp}}
\color[rgb]{0.0,0.0,0.0}
\pgfusepath{stroke}
\end{pgfscope}
\begin{pgfscope}
\pgfsetlinewidth{1.0bp}
\pgfsetrectcap 
\pgfsetmiterjoin \pgfsetmiterlimit{10.0}
\pgfpathmoveto{\pgfpoint{155.475bp}{29.8bp}}
\pgfpathlineto{\pgfpoint{258.075bp}{91.0bp}}
\color[rgb]{0.0,0.0,0.0}
\pgfusepath{stroke}
\end{pgfscope}
\begin{pgfscope}
\pgfsetlinewidth{1.0bp}
\pgfsetrectcap 
\pgfsetmiterjoin \pgfsetmiterlimit{10.0}
\pgfpathmoveto{\pgfpoint{154.875bp}{29.2bp}}
\pgfpathlineto{\pgfpoint{297.675bp}{91.0bp}}
\color[rgb]{0.0,0.0,0.0}
\pgfusepath{stroke}
\end{pgfscope}
\begin{pgfscope}
\pgfsetlinewidth{1.0bp}
\pgfsetrectcap 
\pgfsetmiterjoin \pgfsetmiterlimit{10.0}
\pgfpathmoveto{\pgfpoint{155.475bp}{212.8bp}}
\pgfpathlineto{\pgfpoint{155.475bp}{30.4bp}}
\pgfpathlineto{\pgfpoint{155.475bp}{30.4bp}}
\color[rgb]{0.0,0.0,0.0}
\pgfusepath{stroke}
\end{pgfscope}
\begin{pgfscope}
\pgfsetlinewidth{1.0bp}
\pgfsetrectcap 
\pgfsetmiterjoin \pgfsetmiterlimit{10.0}
\pgfpathmoveto{\pgfpoint{57.675bp}{91.0bp}}
\pgfpathlineto{\pgfpoint{18.075bp}{153.4bp}}
\color[rgb]{0.0,0.0,0.0}
\pgfusepath{stroke}
\end{pgfscope}
\begin{pgfscope}
\pgfsetlinewidth{1.0bp}
\pgfsetrectcap 
\pgfsetmiterjoin \pgfsetmiterlimit{10.0}
\pgfpathmoveto{\pgfpoint{57.675bp}{152.8bp}}
\pgfpathlineto{\pgfpoint{18.075bp}{91.6bp}}
\color[rgb]{0.0,0.0,0.0}
\pgfusepath{stroke}
\end{pgfscope}
\begin{pgfscope}
\pgfsetlinewidth{1.0bp}
\pgfsetrectcap 
\pgfsetmiterjoin \pgfsetmiterlimit{10.0}
\pgfpathmoveto{\pgfpoint{17.475bp}{152.8bp}}
\pgfpathlineto{\pgfpoint{17.475bp}{91.0bp}}
\color[rgb]{0.0,0.0,0.0}
\pgfusepath{stroke}
\end{pgfscope}
\begin{pgfscope}
\pgfsetlinewidth{1.0bp}
\pgfsetrectcap 
\pgfsetmiterjoin \pgfsetmiterlimit{10.0}
\pgfpathmoveto{\pgfpoint{57.675bp}{152.8bp}}
\pgfpathlineto{\pgfpoint{57.675bp}{91.6bp}}
\color[rgb]{0.0,0.0,0.0}
\pgfusepath{stroke}
\end{pgfscope}
\begin{pgfscope}
\pgfsetlinewidth{1.0bp}
\pgfsetrectcap 
\pgfsetmiterjoin \pgfsetmiterlimit{10.0}
\pgfpathmoveto{\pgfpoint{258.075bp}{152.8bp}}
\pgfpathlineto{\pgfpoint{258.075bp}{90.4bp}}
\pgfpathlineto{\pgfpoint{258.075bp}{90.4bp}}
\color[rgb]{0.0,0.0,0.0}
\pgfusepath{stroke}
\end{pgfscope}
\begin{pgfscope}
\pgfsetlinewidth{1.0bp}
\pgfsetrectcap 
\pgfsetmiterjoin \pgfsetmiterlimit{10.0}
\pgfpathmoveto{\pgfpoint{297.675bp}{152.2bp}}
\pgfpathlineto{\pgfpoint{297.675bp}{91.6bp}}
\color[rgb]{0.0,0.0,0.0}
\pgfusepath{stroke}
\end{pgfscope}
\begin{pgfscope}
\pgfsetlinewidth{1.0bp}
\pgfsetrectcap 
\pgfsetmiterjoin \pgfsetmiterlimit{10.0}
\pgfpathmoveto{\pgfpoint{298.275bp}{90.4bp}}
\pgfpathlineto{\pgfpoint{258.075bp}{152.2bp}}
\color[rgb]{0.0,0.0,0.0}
\pgfusepath{stroke}
\end{pgfscope}
\begin{pgfscope}
\pgfsetlinewidth{1.0bp}
\pgfsetrectcap 
\pgfsetmiterjoin \pgfsetmiterlimit{10.0}
\pgfpathmoveto{\pgfpoint{297.675bp}{152.8bp}}
\pgfpathlineto{\pgfpoint{258.075bp}{92.2bp}}
\color[rgb]{0.0,0.0,0.0}
\pgfusepath{stroke}
\end{pgfscope}
\begin{pgfscope}
\pgfsetlinewidth{1.0bp}
\pgfsetrectcap 
\pgfsetmiterjoin \pgfsetmiterlimit{10.0}
\pgfpathmoveto{\pgfpoint{57.675bp}{153.4bp}}
\pgfpathlineto{\pgfpoint{97.875bp}{91.0bp}}
\color[rgb]{0.0,0.0,0.0}
\pgfusepath{stroke}
\end{pgfscope}
\begin{pgfscope}
\pgfsetlinewidth{1.0bp}
\pgfsetrectcap 
\pgfsetmiterjoin \pgfsetmiterlimit{10.0}
\pgfpathmoveto{\pgfpoint{258.075bp}{91.0bp}}
\pgfpathlineto{\pgfpoint{217.875bp}{152.8bp}}
\color[rgb]{0.0,0.0,0.0}
\pgfusepath{stroke}
\end{pgfscope}
\begin{pgfscope}
\pgftransformcm{1.0}{0.0}{0.0}{1.0}{\pgfpoint{145.675bp}{0.0bp}}
\pgftext[left,base]{\sffamily\mdseries\upshape\huge
\color[rgb]{0.0,0.0,0.0}$H$}
\end{pgfscope}
\end{pgfpicture}

%% file: Hexample9b.tex
% C:\Users\agalanis\Dropbox\H-colorings\Hexample9b.jdr
% Created by Jpgfdraw version 0.5.6b
% 01-Feb-2015 13:14:34
\begin{pgfpicture}{0bp}{0bp}{13.115624bp}{113.295206bp}
\begin{pgfscope}
\pgfsetlinewidth{1.0bp}
\pgfsetrectcap 
\pgfsetmiterjoin \pgfsetmiterlimit{10.0}
\pgfpathmoveto{\pgfpoint{6.419521bp}{22.788426bp}}
\pgfpathcurveto{\pgfpoint{3.164375bp}{22.788426bp}}{\pgfpoint{0.525555bp}{25.319533bp}}{\pgfpoint{0.525555bp}{28.441817bp}}
\pgfpathcurveto{\pgfpoint{0.525555bp}{31.564101bp}}{\pgfpoint{3.164375bp}{34.095215bp}}{\pgfpoint{6.419521bp}{34.095215bp}}
\pgfpathcurveto{\pgfpoint{9.674668bp}{34.095215bp}}{\pgfpoint{12.313486bp}{31.564101bp}}{\pgfpoint{12.313486bp}{28.441817bp}}
\pgfpathcurveto{\pgfpoint{12.313486bp}{25.319533bp}}{\pgfpoint{9.674668bp}{22.788426bp}}{\pgfpoint{6.419521bp}{22.788426bp}}
\pgfclosepath
\color[rgb]{0.0,0.0,0.0}\pgfseteorule\pgfusepath{fill}
\pgfpathmoveto{\pgfpoint{6.419521bp}{22.788426bp}}
\pgfpathcurveto{\pgfpoint{3.164375bp}{22.788426bp}}{\pgfpoint{0.525555bp}{25.319533bp}}{\pgfpoint{0.525555bp}{28.441817bp}}
\pgfpathcurveto{\pgfpoint{0.525555bp}{31.564101bp}}{\pgfpoint{3.164375bp}{34.095215bp}}{\pgfpoint{6.419521bp}{34.095215bp}}
\pgfpathcurveto{\pgfpoint{9.674668bp}{34.095215bp}}{\pgfpoint{12.313486bp}{31.564101bp}}{\pgfpoint{12.313486bp}{28.441817bp}}
\pgfpathcurveto{\pgfpoint{12.313486bp}{25.319533bp}}{\pgfpoint{9.674668bp}{22.788426bp}}{\pgfpoint{6.419521bp}{22.788426bp}}
\pgfclosepath
\color[rgb]{0.0,0.0,0.0}
\pgfusepath{stroke}
\end{pgfscope}
\begin{pgfscope}
\pgfsetlinewidth{1.0bp}
\pgfsetrectcap 
\pgfsetmiterjoin \pgfsetmiterlimit{10.0}
\pgfpathmoveto{\pgfpoint{6.419525bp}{101.488418bp}}
\pgfpathcurveto{\pgfpoint{3.164379bp}{101.488418bp}}{\pgfpoint{0.52556bp}{104.019525bp}}{\pgfpoint{0.52556bp}{107.141809bp}}
\pgfpathcurveto{\pgfpoint{0.52556bp}{110.264094bp}}{\pgfpoint{3.164379bp}{112.795207bp}}{\pgfpoint{6.419525bp}{112.795207bp}}
\pgfpathcurveto{\pgfpoint{9.674673bp}{112.795207bp}}{\pgfpoint{12.313491bp}{110.264094bp}}{\pgfpoint{12.313491bp}{107.141809bp}}
\pgfpathcurveto{\pgfpoint{12.313491bp}{104.019525bp}}{\pgfpoint{9.674673bp}{101.488418bp}}{\pgfpoint{6.419525bp}{101.488418bp}}
\pgfclosepath
\color[rgb]{0.0,0.0,0.0}\pgfseteorule\pgfusepath{fill}
\pgfpathmoveto{\pgfpoint{6.419525bp}{101.488418bp}}
\pgfpathcurveto{\pgfpoint{3.164379bp}{101.488418bp}}{\pgfpoint{0.52556bp}{104.019525bp}}{\pgfpoint{0.52556bp}{107.141809bp}}
\pgfpathcurveto{\pgfpoint{0.52556bp}{110.264094bp}}{\pgfpoint{3.164379bp}{112.795207bp}}{\pgfpoint{6.419525bp}{112.795207bp}}
\pgfpathcurveto{\pgfpoint{9.674673bp}{112.795207bp}}{\pgfpoint{12.313491bp}{110.264094bp}}{\pgfpoint{12.313491bp}{107.141809bp}}
\pgfpathcurveto{\pgfpoint{12.313491bp}{104.019525bp}}{\pgfpoint{9.674673bp}{101.488418bp}}{\pgfpoint{6.419525bp}{101.488418bp}}
\pgfclosepath
\color[rgb]{0.0,0.0,0.0}
\pgfusepath{stroke}
\end{pgfscope}
\begin{pgfscope}
\pgfsetlinewidth{1.0bp}
\pgfsetrectcap 
\pgfsetmiterjoin \pgfsetmiterlimit{10.0}
\pgfpathmoveto{\pgfpoint{6.75bp}{107.04375bp}}
\pgfpathlineto{\pgfpoint{6.75bp}{29.64375bp}}
\color[rgb]{0.0,0.0,0.0}
\pgfusepath{stroke}
\end{pgfscope}
\begin{pgfscope}
\pgftransformcm{1.0}{0.0}{0.0}{1.0}{\pgfpoint{-1.05bp}{0.24375bp}}
\pgftext[left,base]{\sffamily\mdseries\upshape\huge
\color[rgb]{0.0,0.0,0.0}$\Gamma$}
\end{pgfscope}
\end{pgfpicture}

%% file: Hexample7a.tex
% C:\Users\agalanis\Dropbox\H-colorings\Hexample7a.jdr
% Created by Jpgfdraw version 0.5.6b
% 03-Feb-2015 19:21:28
\begin{pgfpicture}{0bp}{0bp}{514.621875bp}{167.273286bp}
\begin{pgfscope}
\pgfsetlinewidth{1.0bp}
\pgfsetrectcap 
\pgfsetmiterjoin \pgfsetmiterlimit{10.0}
\pgfpathmoveto{\pgfpoint{157.719781bp}{160.479323bp}}
\pgfpathcurveto{\pgfpoint{157.719781bp}{157.224177bp}}{\pgfpoint{155.188674bp}{154.585358bp}}{\pgfpoint{152.066389bp}{154.585358bp}}
\pgfpathcurveto{\pgfpoint{148.944105bp}{154.585358bp}}{\pgfpoint{146.412992bp}{157.224177bp}}{\pgfpoint{146.412992bp}{160.479323bp}}
\pgfpathcurveto{\pgfpoint{146.412992bp}{163.734471bp}}{\pgfpoint{148.944105bp}{166.373288bp}}{\pgfpoint{152.066389bp}{166.373288bp}}
\pgfpathcurveto{\pgfpoint{155.188674bp}{166.373288bp}}{\pgfpoint{157.719781bp}{163.734471bp}}{\pgfpoint{157.719781bp}{160.479323bp}}
\pgfclosepath
\color[rgb]{0.0,0.0,0.0}\pgfseteorule\pgfusepath{fill}
\pgfpathmoveto{\pgfpoint{157.719781bp}{160.479323bp}}
\pgfpathcurveto{\pgfpoint{157.719781bp}{157.224177bp}}{\pgfpoint{155.188674bp}{154.585358bp}}{\pgfpoint{152.066389bp}{154.585358bp}}
\pgfpathcurveto{\pgfpoint{148.944105bp}{154.585358bp}}{\pgfpoint{146.412992bp}{157.224177bp}}{\pgfpoint{146.412992bp}{160.479323bp}}
\pgfpathcurveto{\pgfpoint{146.412992bp}{163.734471bp}}{\pgfpoint{148.944105bp}{166.373288bp}}{\pgfpoint{152.066389bp}{166.373288bp}}
\pgfpathcurveto{\pgfpoint{155.188674bp}{166.373288bp}}{\pgfpoint{157.719781bp}{163.734471bp}}{\pgfpoint{157.719781bp}{160.479323bp}}
\pgfclosepath
\color[rgb]{0.0,0.0,0.0}
\pgfusepath{stroke}
\end{pgfscope}
\begin{pgfscope}
\pgfsetlinewidth{1.0bp}
\pgfsetrectcap 
\pgfsetmiterjoin \pgfsetmiterlimit{10.0}
\pgfpathmoveto{\pgfpoint{157.119781bp}{120.879323bp}}
\pgfpathcurveto{\pgfpoint{157.119781bp}{117.624177bp}}{\pgfpoint{154.588674bp}{114.985358bp}}{\pgfpoint{151.466389bp}{114.985358bp}}
\pgfpathcurveto{\pgfpoint{148.344105bp}{114.985358bp}}{\pgfpoint{145.812992bp}{117.624177bp}}{\pgfpoint{145.812992bp}{120.879323bp}}
\pgfpathcurveto{\pgfpoint{145.812992bp}{124.134471bp}}{\pgfpoint{148.344105bp}{126.773288bp}}{\pgfpoint{151.466389bp}{126.773288bp}}
\pgfpathcurveto{\pgfpoint{154.588674bp}{126.773288bp}}{\pgfpoint{157.119781bp}{124.134471bp}}{\pgfpoint{157.119781bp}{120.879323bp}}
\pgfclosepath
\color[rgb]{0.0,0.0,0.0}\pgfseteorule\pgfusepath{fill}
\pgfpathmoveto{\pgfpoint{157.119781bp}{120.879323bp}}
\pgfpathcurveto{\pgfpoint{157.119781bp}{117.624177bp}}{\pgfpoint{154.588674bp}{114.985358bp}}{\pgfpoint{151.466389bp}{114.985358bp}}
\pgfpathcurveto{\pgfpoint{148.344105bp}{114.985358bp}}{\pgfpoint{145.812992bp}{117.624177bp}}{\pgfpoint{145.812992bp}{120.879323bp}}
\pgfpathcurveto{\pgfpoint{145.812992bp}{124.134471bp}}{\pgfpoint{148.344105bp}{126.773288bp}}{\pgfpoint{151.466389bp}{126.773288bp}}
\pgfpathcurveto{\pgfpoint{154.588674bp}{126.773288bp}}{\pgfpoint{157.119781bp}{124.134471bp}}{\pgfpoint{157.119781bp}{120.879323bp}}
\pgfclosepath
\color[rgb]{0.0,0.0,0.0}
\pgfusepath{stroke}
\end{pgfscope}
\begin{pgfscope}
\pgfsetlinewidth{1.0bp}
\pgfsetrectcap 
\pgfsetmiterjoin \pgfsetmiterlimit{10.0}
\pgfpathmoveto{\pgfpoint{232.319781bp}{159.679323bp}}
\pgfpathcurveto{\pgfpoint{232.319781bp}{156.424177bp}}{\pgfpoint{229.788674bp}{153.785358bp}}{\pgfpoint{226.666389bp}{153.785358bp}}
\pgfpathcurveto{\pgfpoint{223.544105bp}{153.785358bp}}{\pgfpoint{221.012992bp}{156.424177bp}}{\pgfpoint{221.012992bp}{159.679323bp}}
\pgfpathcurveto{\pgfpoint{221.012992bp}{162.934471bp}}{\pgfpoint{223.544105bp}{165.573288bp}}{\pgfpoint{226.666389bp}{165.573288bp}}
\pgfpathcurveto{\pgfpoint{229.788674bp}{165.573288bp}}{\pgfpoint{232.319781bp}{162.934471bp}}{\pgfpoint{232.319781bp}{159.679323bp}}
\pgfclosepath
\color[rgb]{0.0,0.0,0.0}\pgfseteorule\pgfusepath{fill}
\pgfpathmoveto{\pgfpoint{232.319781bp}{159.679323bp}}
\pgfpathcurveto{\pgfpoint{232.319781bp}{156.424177bp}}{\pgfpoint{229.788674bp}{153.785358bp}}{\pgfpoint{226.666389bp}{153.785358bp}}
\pgfpathcurveto{\pgfpoint{223.544105bp}{153.785358bp}}{\pgfpoint{221.012992bp}{156.424177bp}}{\pgfpoint{221.012992bp}{159.679323bp}}
\pgfpathcurveto{\pgfpoint{221.012992bp}{162.934471bp}}{\pgfpoint{223.544105bp}{165.573288bp}}{\pgfpoint{226.666389bp}{165.573288bp}}
\pgfpathcurveto{\pgfpoint{229.788674bp}{165.573288bp}}{\pgfpoint{232.319781bp}{162.934471bp}}{\pgfpoint{232.319781bp}{159.679323bp}}
\pgfclosepath
\color[rgb]{0.0,0.0,0.0}
\pgfusepath{stroke}
\end{pgfscope}
\begin{pgfscope}
\pgfsetlinewidth{1.0bp}
\pgfsetrectcap 
\pgfsetmiterjoin \pgfsetmiterlimit{10.0}
\pgfpathmoveto{\pgfpoint{231.719781bp}{120.079323bp}}
\pgfpathcurveto{\pgfpoint{231.719781bp}{116.824177bp}}{\pgfpoint{229.188674bp}{114.185358bp}}{\pgfpoint{226.066389bp}{114.185358bp}}
\pgfpathcurveto{\pgfpoint{222.944105bp}{114.185358bp}}{\pgfpoint{220.412992bp}{116.824177bp}}{\pgfpoint{220.412992bp}{120.079323bp}}
\pgfpathcurveto{\pgfpoint{220.412992bp}{123.334471bp}}{\pgfpoint{222.944105bp}{125.973288bp}}{\pgfpoint{226.066389bp}{125.973288bp}}
\pgfpathcurveto{\pgfpoint{229.188674bp}{125.973288bp}}{\pgfpoint{231.719781bp}{123.334471bp}}{\pgfpoint{231.719781bp}{120.079323bp}}
\pgfclosepath
\color[rgb]{0.0,0.0,0.0}\pgfseteorule\pgfusepath{fill}
\pgfpathmoveto{\pgfpoint{231.719781bp}{120.079323bp}}
\pgfpathcurveto{\pgfpoint{231.719781bp}{116.824177bp}}{\pgfpoint{229.188674bp}{114.185358bp}}{\pgfpoint{226.066389bp}{114.185358bp}}
\pgfpathcurveto{\pgfpoint{222.944105bp}{114.185358bp}}{\pgfpoint{220.412992bp}{116.824177bp}}{\pgfpoint{220.412992bp}{120.079323bp}}
\pgfpathcurveto{\pgfpoint{220.412992bp}{123.334471bp}}{\pgfpoint{222.944105bp}{125.973288bp}}{\pgfpoint{226.066389bp}{125.973288bp}}
\pgfpathcurveto{\pgfpoint{229.188674bp}{125.973288bp}}{\pgfpoint{231.719781bp}{123.334471bp}}{\pgfpoint{231.719781bp}{120.079323bp}}
\pgfclosepath
\color[rgb]{0.0,0.0,0.0}
\pgfusepath{stroke}
\end{pgfscope}
\begin{pgfscope}
\pgfsetlinewidth{1.0bp}
\pgfsetrectcap 
\pgfsetmiterjoin \pgfsetmiterlimit{10.0}
\pgfpathmoveto{\pgfpoint{156.719781bp}{80.479323bp}}
\pgfpathcurveto{\pgfpoint{156.719781bp}{77.224177bp}}{\pgfpoint{154.188674bp}{74.585358bp}}{\pgfpoint{151.066389bp}{74.585358bp}}
\pgfpathcurveto{\pgfpoint{147.944105bp}{74.585358bp}}{\pgfpoint{145.412992bp}{77.224177bp}}{\pgfpoint{145.412992bp}{80.479323bp}}
\pgfpathcurveto{\pgfpoint{145.412992bp}{83.734471bp}}{\pgfpoint{147.944105bp}{86.373288bp}}{\pgfpoint{151.066389bp}{86.373288bp}}
\pgfpathcurveto{\pgfpoint{154.188674bp}{86.373288bp}}{\pgfpoint{156.719781bp}{83.734471bp}}{\pgfpoint{156.719781bp}{80.479323bp}}
\pgfclosepath
\color[rgb]{0.0,0.0,0.0}\pgfseteorule\pgfusepath{fill}
\pgfpathmoveto{\pgfpoint{156.719781bp}{80.479323bp}}
\pgfpathcurveto{\pgfpoint{156.719781bp}{77.224177bp}}{\pgfpoint{154.188674bp}{74.585358bp}}{\pgfpoint{151.066389bp}{74.585358bp}}
\pgfpathcurveto{\pgfpoint{147.944105bp}{74.585358bp}}{\pgfpoint{145.412992bp}{77.224177bp}}{\pgfpoint{145.412992bp}{80.479323bp}}
\pgfpathcurveto{\pgfpoint{145.412992bp}{83.734471bp}}{\pgfpoint{147.944105bp}{86.373288bp}}{\pgfpoint{151.066389bp}{86.373288bp}}
\pgfpathcurveto{\pgfpoint{154.188674bp}{86.373288bp}}{\pgfpoint{156.719781bp}{83.734471bp}}{\pgfpoint{156.719781bp}{80.479323bp}}
\pgfclosepath
\color[rgb]{0.0,0.0,0.0}
\pgfusepath{stroke}
\end{pgfscope}
\begin{pgfscope}
\pgfsetlinewidth{1.0bp}
\pgfsetrectcap 
\pgfsetmiterjoin \pgfsetmiterlimit{10.0}
\pgfpathmoveto{\pgfpoint{156.119781bp}{40.879323bp}}
\pgfpathcurveto{\pgfpoint{156.119781bp}{37.624177bp}}{\pgfpoint{153.588674bp}{34.985358bp}}{\pgfpoint{150.466389bp}{34.985358bp}}
\pgfpathcurveto{\pgfpoint{147.344105bp}{34.985358bp}}{\pgfpoint{144.812992bp}{37.624177bp}}{\pgfpoint{144.812992bp}{40.879323bp}}
\pgfpathcurveto{\pgfpoint{144.812992bp}{44.134471bp}}{\pgfpoint{147.344105bp}{46.773288bp}}{\pgfpoint{150.466389bp}{46.773288bp}}
\pgfpathcurveto{\pgfpoint{153.588674bp}{46.773288bp}}{\pgfpoint{156.119781bp}{44.134471bp}}{\pgfpoint{156.119781bp}{40.879323bp}}
\pgfclosepath
\color[rgb]{0.0,0.0,0.0}\pgfseteorule\pgfusepath{fill}
\pgfpathmoveto{\pgfpoint{156.119781bp}{40.879323bp}}
\pgfpathcurveto{\pgfpoint{156.119781bp}{37.624177bp}}{\pgfpoint{153.588674bp}{34.985358bp}}{\pgfpoint{150.466389bp}{34.985358bp}}
\pgfpathcurveto{\pgfpoint{147.344105bp}{34.985358bp}}{\pgfpoint{144.812992bp}{37.624177bp}}{\pgfpoint{144.812992bp}{40.879323bp}}
\pgfpathcurveto{\pgfpoint{144.812992bp}{44.134471bp}}{\pgfpoint{147.344105bp}{46.773288bp}}{\pgfpoint{150.466389bp}{46.773288bp}}
\pgfpathcurveto{\pgfpoint{153.588674bp}{46.773288bp}}{\pgfpoint{156.119781bp}{44.134471bp}}{\pgfpoint{156.119781bp}{40.879323bp}}
\pgfclosepath
\color[rgb]{0.0,0.0,0.0}
\pgfusepath{stroke}
\end{pgfscope}
\begin{pgfscope}
\pgfsetlinewidth{1.0bp}
\pgfsetrectcap 
\pgfsetmiterjoin \pgfsetmiterlimit{10.0}
\pgfpathmoveto{\pgfpoint{231.319781bp}{79.679323bp}}
\pgfpathcurveto{\pgfpoint{231.319781bp}{76.424177bp}}{\pgfpoint{228.788674bp}{73.785358bp}}{\pgfpoint{225.666389bp}{73.785358bp}}
\pgfpathcurveto{\pgfpoint{222.544105bp}{73.785358bp}}{\pgfpoint{220.012992bp}{76.424177bp}}{\pgfpoint{220.012992bp}{79.679323bp}}
\pgfpathcurveto{\pgfpoint{220.012992bp}{82.934471bp}}{\pgfpoint{222.544105bp}{85.573288bp}}{\pgfpoint{225.666389bp}{85.573288bp}}
\pgfpathcurveto{\pgfpoint{228.788674bp}{85.573288bp}}{\pgfpoint{231.319781bp}{82.934471bp}}{\pgfpoint{231.319781bp}{79.679323bp}}
\pgfclosepath
\color[rgb]{0.0,0.0,0.0}\pgfseteorule\pgfusepath{fill}
\pgfpathmoveto{\pgfpoint{231.319781bp}{79.679323bp}}
\pgfpathcurveto{\pgfpoint{231.319781bp}{76.424177bp}}{\pgfpoint{228.788674bp}{73.785358bp}}{\pgfpoint{225.666389bp}{73.785358bp}}
\pgfpathcurveto{\pgfpoint{222.544105bp}{73.785358bp}}{\pgfpoint{220.012992bp}{76.424177bp}}{\pgfpoint{220.012992bp}{79.679323bp}}
\pgfpathcurveto{\pgfpoint{220.012992bp}{82.934471bp}}{\pgfpoint{222.544105bp}{85.573288bp}}{\pgfpoint{225.666389bp}{85.573288bp}}
\pgfpathcurveto{\pgfpoint{228.788674bp}{85.573288bp}}{\pgfpoint{231.319781bp}{82.934471bp}}{\pgfpoint{231.319781bp}{79.679323bp}}
\pgfclosepath
\color[rgb]{0.0,0.0,0.0}
\pgfusepath{stroke}
\end{pgfscope}
\begin{pgfscope}
\pgfsetlinewidth{1.0bp}
\pgfsetrectcap 
\pgfsetmiterjoin \pgfsetmiterlimit{10.0}
\pgfpathmoveto{\pgfpoint{230.719781bp}{40.079323bp}}
\pgfpathcurveto{\pgfpoint{230.719781bp}{36.824177bp}}{\pgfpoint{228.188674bp}{34.185358bp}}{\pgfpoint{225.066389bp}{34.185358bp}}
\pgfpathcurveto{\pgfpoint{221.944105bp}{34.185358bp}}{\pgfpoint{219.412992bp}{36.824177bp}}{\pgfpoint{219.412992bp}{40.079323bp}}
\pgfpathcurveto{\pgfpoint{219.412992bp}{43.334471bp}}{\pgfpoint{221.944105bp}{45.973288bp}}{\pgfpoint{225.066389bp}{45.973288bp}}
\pgfpathcurveto{\pgfpoint{228.188674bp}{45.973288bp}}{\pgfpoint{230.719781bp}{43.334471bp}}{\pgfpoint{230.719781bp}{40.079323bp}}
\pgfclosepath
\color[rgb]{0.0,0.0,0.0}\pgfseteorule\pgfusepath{fill}
\pgfpathmoveto{\pgfpoint{230.719781bp}{40.079323bp}}
\pgfpathcurveto{\pgfpoint{230.719781bp}{36.824177bp}}{\pgfpoint{228.188674bp}{34.185358bp}}{\pgfpoint{225.066389bp}{34.185358bp}}
\pgfpathcurveto{\pgfpoint{221.944105bp}{34.185358bp}}{\pgfpoint{219.412992bp}{36.824177bp}}{\pgfpoint{219.412992bp}{40.079323bp}}
\pgfpathcurveto{\pgfpoint{219.412992bp}{43.334471bp}}{\pgfpoint{221.944105bp}{45.973288bp}}{\pgfpoint{225.066389bp}{45.973288bp}}
\pgfpathcurveto{\pgfpoint{228.188674bp}{45.973288bp}}{\pgfpoint{230.719781bp}{43.334471bp}}{\pgfpoint{230.719781bp}{40.079323bp}}
\pgfclosepath
\color[rgb]{0.0,0.0,0.0}
\pgfusepath{stroke}
\end{pgfscope}
\begin{pgfscope}
\pgfsetlinewidth{1.0bp}
\pgfsetrectcap 
\pgfsetmiterjoin \pgfsetmiterlimit{10.0}
\pgfpathmoveto{\pgfpoint{227.496875bp}{160.181249bp}}
\pgfpathlineto{\pgfpoint{153.096875bp}{160.181249bp}}
\color[rgb]{0.0,0.0,0.0}
\pgfusepath{stroke}
\end{pgfscope}
\begin{pgfscope}
\pgfsetlinewidth{1.0bp}
\pgfsetrectcap 
\pgfsetmiterjoin \pgfsetmiterlimit{10.0}
\pgfpathmoveto{\pgfpoint{153.696875bp}{160.181249bp}}
\pgfpathlineto{\pgfpoint{224.496875bp}{120.581249bp}}
\color[rgb]{0.0,0.0,0.0}
\pgfusepath{stroke}
\end{pgfscope}
\begin{pgfscope}
\pgfsetlinewidth{1.0bp}
\pgfsetrectcap 
\pgfsetmiterjoin \pgfsetmiterlimit{10.0}
\pgfpathmoveto{\pgfpoint{153.696875bp}{161.981249bp}}
\pgfpathlineto{\pgfpoint{225.096875bp}{79.781249bp}}
\color[rgb]{0.0,0.0,0.0}
\pgfusepath{stroke}
\end{pgfscope}
\begin{pgfscope}
\pgfsetlinewidth{1.0bp}
\pgfsetrectcap 
\pgfsetmiterjoin \pgfsetmiterlimit{10.0}
\pgfpathmoveto{\pgfpoint{153.696875bp}{160.181249bp}}
\pgfpathlineto{\pgfpoint{224.496875bp}{40.181249bp}}
\color[rgb]{0.0,0.0,0.0}
\pgfusepath{stroke}
\end{pgfscope}
\begin{pgfscope}
\pgfsetlinewidth{1.0bp}
\pgfsetrectcap 
\pgfsetmiterjoin \pgfsetmiterlimit{10.0}
\pgfpathmoveto{\pgfpoint{225.696875bp}{160.781249bp}}
\pgfpathlineto{\pgfpoint{151.296875bp}{119.981249bp}}
\color[rgb]{0.0,0.0,0.0}
\pgfusepath{stroke}
\end{pgfscope}
\begin{pgfscope}
\pgfsetlinewidth{1.0bp}
\pgfsetrectcap 
\pgfsetmiterjoin \pgfsetmiterlimit{10.0}
\pgfpathmoveto{\pgfpoint{226.296875bp}{161.381249bp}}
\pgfpathlineto{\pgfpoint{153.096875bp}{82.781249bp}}
\color[rgb]{0.0,0.0,0.0}
\pgfusepath{stroke}
\end{pgfscope}
\begin{pgfscope}
\pgfsetlinewidth{1.0bp}
\pgfsetrectcap 
\pgfsetmiterjoin \pgfsetmiterlimit{10.0}
\pgfpathmoveto{\pgfpoint{226.896875bp}{160.181249bp}}
\pgfpathlineto{\pgfpoint{153.696875bp}{42.581249bp}}
\color[rgb]{0.0,0.0,0.0}
\pgfusepath{stroke}
\end{pgfscope}
\begin{pgfscope}
\pgfsetlinewidth{1.0bp}
\pgfsetrectcap 
\pgfsetmiterjoin \pgfsetmiterlimit{10.0}
\pgfpathmoveto{\pgfpoint{226.296875bp}{121.181249bp}}
\pgfpathlineto{\pgfpoint{151.896875bp}{121.181249bp}}
\color[rgb]{0.0,0.0,0.0}
\pgfusepath{stroke}
\end{pgfscope}
\begin{pgfscope}
\pgftransformcm{1.0}{0.0}{0.0}{1.0}{\pgfpoint{135.896875bp}{155.181249bp}}
\pgftext[left,base]{\sffamily\mdseries\upshape\large
\color[rgb]{0.0,0.0,0.0}$1$}
\end{pgfscope}
\begin{pgfscope}
\pgftransformcm{1.0}{0.0}{0.0}{1.0}{\pgfpoint{136.096875bp}{116.581249bp}}
\pgftext[left,base]{\sffamily\mdseries\upshape\large
\color[rgb]{0.0,0.0,0.0}$2$}
\end{pgfscope}
\begin{pgfscope}
\pgftransformcm{1.0}{0.0}{0.0}{1.0}{\pgfpoint{135.696875bp}{77.381249bp}}
\pgftext[left,base]{\sffamily\mdseries\upshape\large
\color[rgb]{0.0,0.0,0.0}$3$}
\end{pgfscope}
\begin{pgfscope}
\pgftransformcm{1.0}{0.0}{0.0}{1.0}{\pgfpoint{135.696875bp}{35.381249bp}}
\pgftext[left,base]{\sffamily\mdseries\upshape\large
\color[rgb]{0.0,0.0,0.0}$4$}
\end{pgfscope}
\begin{pgfscope}
\pgftransformcm{1.0}{0.0}{0.0}{1.0}{\pgfpoint{235.696875bp}{154.981249bp}}
\pgftext[left,base]{\sffamily\mdseries\upshape\large
\color[rgb]{0.0,0.0,0.0}$1'$}
\end{pgfscope}
\begin{pgfscope}
\pgftransformcm{1.0}{0.0}{0.0}{1.0}{\pgfpoint{235.696875bp}{115.981249bp}}
\pgftext[left,base]{\sffamily\mdseries\upshape\large
\color[rgb]{0.0,0.0,0.0}$2'$}
\end{pgfscope}
\begin{pgfscope}
\pgftransformcm{1.0}{0.0}{0.0}{1.0}{\pgfpoint{235.096875bp}{76.381249bp}}
\pgftext[left,base]{\sffamily\mdseries\upshape\large
\color[rgb]{0.0,0.0,0.0}$3'$}
\end{pgfscope}
\begin{pgfscope}
\pgftransformcm{1.0}{0.0}{0.0}{1.0}{\pgfpoint{234.696875bp}{35.981249bp}}
\pgftext[left,base]{\sffamily\mdseries\upshape\large
\color[rgb]{0.0,0.0,0.0}$4'$}
\end{pgfscope}
\begin{pgfscope}
\pgfsetlinewidth{1.0bp}
\pgfsetrectcap 
\pgfsetmiterjoin \pgfsetmiterlimit{10.0}
\pgfpathmoveto{\pgfpoint{20.519781bp}{160.879323bp}}
\pgfpathcurveto{\pgfpoint{20.519781bp}{157.624177bp}}{\pgfpoint{17.988674bp}{154.985358bp}}{\pgfpoint{14.866389bp}{154.985358bp}}
\pgfpathcurveto{\pgfpoint{11.744105bp}{154.985358bp}}{\pgfpoint{9.212992bp}{157.624177bp}}{\pgfpoint{9.212992bp}{160.879323bp}}
\pgfpathcurveto{\pgfpoint{9.212992bp}{164.134471bp}}{\pgfpoint{11.744105bp}{166.773288bp}}{\pgfpoint{14.866389bp}{166.773288bp}}
\pgfpathcurveto{\pgfpoint{17.988674bp}{166.773288bp}}{\pgfpoint{20.519781bp}{164.134471bp}}{\pgfpoint{20.519781bp}{160.879323bp}}
\pgfclosepath
\color[rgb]{0.0,0.0,0.0}\pgfseteorule\pgfusepath{fill}
\pgfpathmoveto{\pgfpoint{20.519781bp}{160.879323bp}}
\pgfpathcurveto{\pgfpoint{20.519781bp}{157.624177bp}}{\pgfpoint{17.988674bp}{154.985358bp}}{\pgfpoint{14.866389bp}{154.985358bp}}
\pgfpathcurveto{\pgfpoint{11.744105bp}{154.985358bp}}{\pgfpoint{9.212992bp}{157.624177bp}}{\pgfpoint{9.212992bp}{160.879323bp}}
\pgfpathcurveto{\pgfpoint{9.212992bp}{164.134471bp}}{\pgfpoint{11.744105bp}{166.773288bp}}{\pgfpoint{14.866389bp}{166.773288bp}}
\pgfpathcurveto{\pgfpoint{17.988674bp}{166.773288bp}}{\pgfpoint{20.519781bp}{164.134471bp}}{\pgfpoint{20.519781bp}{160.879323bp}}
\pgfclosepath
\color[rgb]{0.0,0.0,0.0}
\pgfusepath{stroke}
\end{pgfscope}
\begin{pgfscope}
\pgfsetlinewidth{1.0bp}
\pgfsetrectcap 
\pgfsetmiterjoin \pgfsetmiterlimit{10.0}
\pgfpathmoveto{\pgfpoint{95.119781bp}{160.079323bp}}
\pgfpathcurveto{\pgfpoint{95.119781bp}{156.824177bp}}{\pgfpoint{92.588674bp}{154.185358bp}}{\pgfpoint{89.466389bp}{154.185358bp}}
\pgfpathcurveto{\pgfpoint{86.344105bp}{154.185358bp}}{\pgfpoint{83.812992bp}{156.824177bp}}{\pgfpoint{83.812992bp}{160.079323bp}}
\pgfpathcurveto{\pgfpoint{83.812992bp}{163.334471bp}}{\pgfpoint{86.344105bp}{165.973288bp}}{\pgfpoint{89.466389bp}{165.973288bp}}
\pgfpathcurveto{\pgfpoint{92.588674bp}{165.973288bp}}{\pgfpoint{95.119781bp}{163.334471bp}}{\pgfpoint{95.119781bp}{160.079323bp}}
\pgfclosepath
\color[rgb]{0.0,0.0,0.0}\pgfseteorule\pgfusepath{fill}
\pgfpathmoveto{\pgfpoint{95.119781bp}{160.079323bp}}
\pgfpathcurveto{\pgfpoint{95.119781bp}{156.824177bp}}{\pgfpoint{92.588674bp}{154.185358bp}}{\pgfpoint{89.466389bp}{154.185358bp}}
\pgfpathcurveto{\pgfpoint{86.344105bp}{154.185358bp}}{\pgfpoint{83.812992bp}{156.824177bp}}{\pgfpoint{83.812992bp}{160.079323bp}}
\pgfpathcurveto{\pgfpoint{83.812992bp}{163.334471bp}}{\pgfpoint{86.344105bp}{165.973288bp}}{\pgfpoint{89.466389bp}{165.973288bp}}
\pgfpathcurveto{\pgfpoint{92.588674bp}{165.973288bp}}{\pgfpoint{95.119781bp}{163.334471bp}}{\pgfpoint{95.119781bp}{160.079323bp}}
\pgfclosepath
\color[rgb]{0.0,0.0,0.0}
\pgfusepath{stroke}
\end{pgfscope}
\begin{pgfscope}
\pgfsetlinewidth{1.0bp}
\pgfsetrectcap 
\pgfsetmiterjoin \pgfsetmiterlimit{10.0}
\pgfpathmoveto{\pgfpoint{94.519781bp}{120.479323bp}}
\pgfpathcurveto{\pgfpoint{94.519781bp}{117.224177bp}}{\pgfpoint{91.988674bp}{114.585358bp}}{\pgfpoint{88.866389bp}{114.585358bp}}
\pgfpathcurveto{\pgfpoint{85.744105bp}{114.585358bp}}{\pgfpoint{83.212992bp}{117.224177bp}}{\pgfpoint{83.212992bp}{120.479323bp}}
\pgfpathcurveto{\pgfpoint{83.212992bp}{123.734471bp}}{\pgfpoint{85.744105bp}{126.373288bp}}{\pgfpoint{88.866389bp}{126.373288bp}}
\pgfpathcurveto{\pgfpoint{91.988674bp}{126.373288bp}}{\pgfpoint{94.519781bp}{123.734471bp}}{\pgfpoint{94.519781bp}{120.479323bp}}
\pgfclosepath
\color[rgb]{0.0,0.0,0.0}\pgfseteorule\pgfusepath{fill}
\pgfpathmoveto{\pgfpoint{94.519781bp}{120.479323bp}}
\pgfpathcurveto{\pgfpoint{94.519781bp}{117.224177bp}}{\pgfpoint{91.988674bp}{114.585358bp}}{\pgfpoint{88.866389bp}{114.585358bp}}
\pgfpathcurveto{\pgfpoint{85.744105bp}{114.585358bp}}{\pgfpoint{83.212992bp}{117.224177bp}}{\pgfpoint{83.212992bp}{120.479323bp}}
\pgfpathcurveto{\pgfpoint{83.212992bp}{123.734471bp}}{\pgfpoint{85.744105bp}{126.373288bp}}{\pgfpoint{88.866389bp}{126.373288bp}}
\pgfpathcurveto{\pgfpoint{91.988674bp}{126.373288bp}}{\pgfpoint{94.519781bp}{123.734471bp}}{\pgfpoint{94.519781bp}{120.479323bp}}
\pgfclosepath
\color[rgb]{0.0,0.0,0.0}
\pgfusepath{stroke}
\end{pgfscope}
\begin{pgfscope}
\pgfsetlinewidth{1.0bp}
\pgfsetrectcap 
\pgfsetmiterjoin \pgfsetmiterlimit{10.0}
\pgfpathmoveto{\pgfpoint{94.119781bp}{80.079323bp}}
\pgfpathcurveto{\pgfpoint{94.119781bp}{76.824177bp}}{\pgfpoint{91.588674bp}{74.185358bp}}{\pgfpoint{88.466389bp}{74.185358bp}}
\pgfpathcurveto{\pgfpoint{85.344105bp}{74.185358bp}}{\pgfpoint{82.812992bp}{76.824177bp}}{\pgfpoint{82.812992bp}{80.079323bp}}
\pgfpathcurveto{\pgfpoint{82.812992bp}{83.334471bp}}{\pgfpoint{85.344105bp}{85.973288bp}}{\pgfpoint{88.466389bp}{85.973288bp}}
\pgfpathcurveto{\pgfpoint{91.588674bp}{85.973288bp}}{\pgfpoint{94.119781bp}{83.334471bp}}{\pgfpoint{94.119781bp}{80.079323bp}}
\pgfclosepath
\color[rgb]{0.0,0.0,0.0}\pgfseteorule\pgfusepath{fill}
\pgfpathmoveto{\pgfpoint{94.119781bp}{80.079323bp}}
\pgfpathcurveto{\pgfpoint{94.119781bp}{76.824177bp}}{\pgfpoint{91.588674bp}{74.185358bp}}{\pgfpoint{88.466389bp}{74.185358bp}}
\pgfpathcurveto{\pgfpoint{85.344105bp}{74.185358bp}}{\pgfpoint{82.812992bp}{76.824177bp}}{\pgfpoint{82.812992bp}{80.079323bp}}
\pgfpathcurveto{\pgfpoint{82.812992bp}{83.334471bp}}{\pgfpoint{85.344105bp}{85.973288bp}}{\pgfpoint{88.466389bp}{85.973288bp}}
\pgfpathcurveto{\pgfpoint{91.588674bp}{85.973288bp}}{\pgfpoint{94.119781bp}{83.334471bp}}{\pgfpoint{94.119781bp}{80.079323bp}}
\pgfclosepath
\color[rgb]{0.0,0.0,0.0}
\pgfusepath{stroke}
\end{pgfscope}
\begin{pgfscope}
\pgfsetlinewidth{1.0bp}
\pgfsetrectcap 
\pgfsetmiterjoin \pgfsetmiterlimit{10.0}
\pgfpathmoveto{\pgfpoint{93.519781bp}{40.479323bp}}
\pgfpathcurveto{\pgfpoint{93.519781bp}{37.224177bp}}{\pgfpoint{90.988674bp}{34.585358bp}}{\pgfpoint{87.866389bp}{34.585358bp}}
\pgfpathcurveto{\pgfpoint{84.744105bp}{34.585358bp}}{\pgfpoint{82.212992bp}{37.224177bp}}{\pgfpoint{82.212992bp}{40.479323bp}}
\pgfpathcurveto{\pgfpoint{82.212992bp}{43.734471bp}}{\pgfpoint{84.744105bp}{46.373288bp}}{\pgfpoint{87.866389bp}{46.373288bp}}
\pgfpathcurveto{\pgfpoint{90.988674bp}{46.373288bp}}{\pgfpoint{93.519781bp}{43.734471bp}}{\pgfpoint{93.519781bp}{40.479323bp}}
\pgfclosepath
\color[rgb]{0.0,0.0,0.0}\pgfseteorule\pgfusepath{fill}
\pgfpathmoveto{\pgfpoint{93.519781bp}{40.479323bp}}
\pgfpathcurveto{\pgfpoint{93.519781bp}{37.224177bp}}{\pgfpoint{90.988674bp}{34.585358bp}}{\pgfpoint{87.866389bp}{34.585358bp}}
\pgfpathcurveto{\pgfpoint{84.744105bp}{34.585358bp}}{\pgfpoint{82.212992bp}{37.224177bp}}{\pgfpoint{82.212992bp}{40.479323bp}}
\pgfpathcurveto{\pgfpoint{82.212992bp}{43.734471bp}}{\pgfpoint{84.744105bp}{46.373288bp}}{\pgfpoint{87.866389bp}{46.373288bp}}
\pgfpathcurveto{\pgfpoint{90.988674bp}{46.373288bp}}{\pgfpoint{93.519781bp}{43.734471bp}}{\pgfpoint{93.519781bp}{40.479323bp}}
\pgfclosepath
\color[rgb]{0.0,0.0,0.0}
\pgfusepath{stroke}
\end{pgfscope}
\begin{pgfscope}
\pgfsetlinewidth{1.0bp}
\pgfsetrectcap 
\pgfsetmiterjoin \pgfsetmiterlimit{10.0}
\pgfpathmoveto{\pgfpoint{90.296875bp}{160.581249bp}}
\pgfpathlineto{\pgfpoint{15.896875bp}{160.581249bp}}
\color[rgb]{0.0,0.0,0.0}
\pgfusepath{stroke}
\end{pgfscope}
\begin{pgfscope}
\pgfsetlinewidth{1.0bp}
\pgfsetrectcap 
\pgfsetmiterjoin \pgfsetmiterlimit{10.0}
\pgfpathmoveto{\pgfpoint{16.496875bp}{160.581249bp}}
\pgfpathlineto{\pgfpoint{87.296875bp}{120.981249bp}}
\color[rgb]{0.0,0.0,0.0}
\pgfusepath{stroke}
\end{pgfscope}
\begin{pgfscope}
\pgfsetlinewidth{1.0bp}
\pgfsetrectcap 
\pgfsetmiterjoin \pgfsetmiterlimit{10.0}
\pgfpathmoveto{\pgfpoint{16.496875bp}{162.381249bp}}
\pgfpathlineto{\pgfpoint{87.896875bp}{80.181249bp}}
\color[rgb]{0.0,0.0,0.0}
\pgfusepath{stroke}
\end{pgfscope}
\begin{pgfscope}
\pgfsetlinewidth{1.0bp}
\pgfsetrectcap 
\pgfsetmiterjoin \pgfsetmiterlimit{10.0}
\pgfpathmoveto{\pgfpoint{16.496875bp}{160.581249bp}}
\pgfpathlineto{\pgfpoint{87.296875bp}{40.581249bp}}
\color[rgb]{0.0,0.0,0.0}
\pgfusepath{stroke}
\end{pgfscope}
\begin{pgfscope}
\pgftransformcm{1.0}{0.0}{0.0}{1.0}{\pgfpoint{-1.303125bp}{155.581249bp}}
\pgftext[left,base]{\sffamily\mdseries\upshape\large
\color[rgb]{0.0,0.0,0.0}$1$}
\end{pgfscope}
\begin{pgfscope}
\pgftransformcm{1.0}{0.0}{0.0}{1.0}{\pgfpoint{98.496875bp}{155.381249bp}}
\pgftext[left,base]{\sffamily\mdseries\upshape\large
\color[rgb]{0.0,0.0,0.0}$1'$}
\end{pgfscope}
\begin{pgfscope}
\pgftransformcm{1.0}{0.0}{0.0}{1.0}{\pgfpoint{98.496875bp}{116.381249bp}}
\pgftext[left,base]{\sffamily\mdseries\upshape\large
\color[rgb]{0.0,0.0,0.0}$2'$}
\end{pgfscope}
\begin{pgfscope}
\pgftransformcm{1.0}{0.0}{0.0}{1.0}{\pgfpoint{97.896875bp}{76.781249bp}}
\pgftext[left,base]{\sffamily\mdseries\upshape\large
\color[rgb]{0.0,0.0,0.0}$3'$}
\end{pgfscope}
\begin{pgfscope}
\pgftransformcm{1.0}{0.0}{0.0}{1.0}{\pgfpoint{97.496875bp}{36.381249bp}}
\pgftext[left,base]{\sffamily\mdseries\upshape\large
\color[rgb]{0.0,0.0,0.0}$4'$}
\end{pgfscope}
\begin{pgfscope}
\pgfsetlinewidth{1.0bp}
\pgfsetrectcap 
\pgfsetmiterjoin \pgfsetmiterlimit{10.0}
\pgfpathmoveto{\pgfpoint{298.519781bp}{158.279323bp}}
\pgfpathcurveto{\pgfpoint{298.519781bp}{155.024177bp}}{\pgfpoint{295.988674bp}{152.385358bp}}{\pgfpoint{292.866389bp}{152.385358bp}}
\pgfpathcurveto{\pgfpoint{289.744105bp}{152.385358bp}}{\pgfpoint{287.212992bp}{155.024177bp}}{\pgfpoint{287.212992bp}{158.279323bp}}
\pgfpathcurveto{\pgfpoint{287.212992bp}{161.534471bp}}{\pgfpoint{289.744105bp}{164.173288bp}}{\pgfpoint{292.866389bp}{164.173288bp}}
\pgfpathcurveto{\pgfpoint{295.988674bp}{164.173288bp}}{\pgfpoint{298.519781bp}{161.534471bp}}{\pgfpoint{298.519781bp}{158.279323bp}}
\pgfclosepath
\color[rgb]{0.0,0.0,0.0}\pgfseteorule\pgfusepath{fill}
\pgfpathmoveto{\pgfpoint{298.519781bp}{158.279323bp}}
\pgfpathcurveto{\pgfpoint{298.519781bp}{155.024177bp}}{\pgfpoint{295.988674bp}{152.385358bp}}{\pgfpoint{292.866389bp}{152.385358bp}}
\pgfpathcurveto{\pgfpoint{289.744105bp}{152.385358bp}}{\pgfpoint{287.212992bp}{155.024177bp}}{\pgfpoint{287.212992bp}{158.279323bp}}
\pgfpathcurveto{\pgfpoint{287.212992bp}{161.534471bp}}{\pgfpoint{289.744105bp}{164.173288bp}}{\pgfpoint{292.866389bp}{164.173288bp}}
\pgfpathcurveto{\pgfpoint{295.988674bp}{164.173288bp}}{\pgfpoint{298.519781bp}{161.534471bp}}{\pgfpoint{298.519781bp}{158.279323bp}}
\pgfclosepath
\color[rgb]{0.0,0.0,0.0}
\pgfusepath{stroke}
\end{pgfscope}
\begin{pgfscope}
\pgfsetlinewidth{1.0bp}
\pgfsetrectcap 
\pgfsetmiterjoin \pgfsetmiterlimit{10.0}
\pgfpathmoveto{\pgfpoint{297.919781bp}{118.679323bp}}
\pgfpathcurveto{\pgfpoint{297.919781bp}{115.424177bp}}{\pgfpoint{295.388674bp}{112.785358bp}}{\pgfpoint{292.266389bp}{112.785358bp}}
\pgfpathcurveto{\pgfpoint{289.144105bp}{112.785358bp}}{\pgfpoint{286.612992bp}{115.424177bp}}{\pgfpoint{286.612992bp}{118.679323bp}}
\pgfpathcurveto{\pgfpoint{286.612992bp}{121.934471bp}}{\pgfpoint{289.144105bp}{124.573288bp}}{\pgfpoint{292.266389bp}{124.573288bp}}
\pgfpathcurveto{\pgfpoint{295.388674bp}{124.573288bp}}{\pgfpoint{297.919781bp}{121.934471bp}}{\pgfpoint{297.919781bp}{118.679323bp}}
\pgfclosepath
\color[rgb]{0.0,0.0,0.0}\pgfseteorule\pgfusepath{fill}
\pgfpathmoveto{\pgfpoint{297.919781bp}{118.679323bp}}
\pgfpathcurveto{\pgfpoint{297.919781bp}{115.424177bp}}{\pgfpoint{295.388674bp}{112.785358bp}}{\pgfpoint{292.266389bp}{112.785358bp}}
\pgfpathcurveto{\pgfpoint{289.144105bp}{112.785358bp}}{\pgfpoint{286.612992bp}{115.424177bp}}{\pgfpoint{286.612992bp}{118.679323bp}}
\pgfpathcurveto{\pgfpoint{286.612992bp}{121.934471bp}}{\pgfpoint{289.144105bp}{124.573288bp}}{\pgfpoint{292.266389bp}{124.573288bp}}
\pgfpathcurveto{\pgfpoint{295.388674bp}{124.573288bp}}{\pgfpoint{297.919781bp}{121.934471bp}}{\pgfpoint{297.919781bp}{118.679323bp}}
\pgfclosepath
\color[rgb]{0.0,0.0,0.0}
\pgfusepath{stroke}
\end{pgfscope}
\begin{pgfscope}
\pgfsetlinewidth{1.0bp}
\pgfsetrectcap 
\pgfsetmiterjoin \pgfsetmiterlimit{10.0}
\pgfpathmoveto{\pgfpoint{373.119781bp}{157.479323bp}}
\pgfpathcurveto{\pgfpoint{373.119781bp}{154.224177bp}}{\pgfpoint{370.588674bp}{151.585358bp}}{\pgfpoint{367.466389bp}{151.585358bp}}
\pgfpathcurveto{\pgfpoint{364.344105bp}{151.585358bp}}{\pgfpoint{361.812992bp}{154.224177bp}}{\pgfpoint{361.812992bp}{157.479323bp}}
\pgfpathcurveto{\pgfpoint{361.812992bp}{160.734471bp}}{\pgfpoint{364.344105bp}{163.373288bp}}{\pgfpoint{367.466389bp}{163.373288bp}}
\pgfpathcurveto{\pgfpoint{370.588674bp}{163.373288bp}}{\pgfpoint{373.119781bp}{160.734471bp}}{\pgfpoint{373.119781bp}{157.479323bp}}
\pgfclosepath
\color[rgb]{0.0,0.0,0.0}\pgfseteorule\pgfusepath{fill}
\pgfpathmoveto{\pgfpoint{373.119781bp}{157.479323bp}}
\pgfpathcurveto{\pgfpoint{373.119781bp}{154.224177bp}}{\pgfpoint{370.588674bp}{151.585358bp}}{\pgfpoint{367.466389bp}{151.585358bp}}
\pgfpathcurveto{\pgfpoint{364.344105bp}{151.585358bp}}{\pgfpoint{361.812992bp}{154.224177bp}}{\pgfpoint{361.812992bp}{157.479323bp}}
\pgfpathcurveto{\pgfpoint{361.812992bp}{160.734471bp}}{\pgfpoint{364.344105bp}{163.373288bp}}{\pgfpoint{367.466389bp}{163.373288bp}}
\pgfpathcurveto{\pgfpoint{370.588674bp}{163.373288bp}}{\pgfpoint{373.119781bp}{160.734471bp}}{\pgfpoint{373.119781bp}{157.479323bp}}
\pgfclosepath
\color[rgb]{0.0,0.0,0.0}
\pgfusepath{stroke}
\end{pgfscope}
\begin{pgfscope}
\pgfsetlinewidth{1.0bp}
\pgfsetrectcap 
\pgfsetmiterjoin \pgfsetmiterlimit{10.0}
\pgfpathmoveto{\pgfpoint{372.519781bp}{117.879323bp}}
\pgfpathcurveto{\pgfpoint{372.519781bp}{114.624177bp}}{\pgfpoint{369.988674bp}{111.985358bp}}{\pgfpoint{366.866389bp}{111.985358bp}}
\pgfpathcurveto{\pgfpoint{363.744105bp}{111.985358bp}}{\pgfpoint{361.212992bp}{114.624177bp}}{\pgfpoint{361.212992bp}{117.879323bp}}
\pgfpathcurveto{\pgfpoint{361.212992bp}{121.134471bp}}{\pgfpoint{363.744105bp}{123.773288bp}}{\pgfpoint{366.866389bp}{123.773288bp}}
\pgfpathcurveto{\pgfpoint{369.988674bp}{123.773288bp}}{\pgfpoint{372.519781bp}{121.134471bp}}{\pgfpoint{372.519781bp}{117.879323bp}}
\pgfclosepath
\color[rgb]{0.0,0.0,0.0}\pgfseteorule\pgfusepath{fill}
\pgfpathmoveto{\pgfpoint{372.519781bp}{117.879323bp}}
\pgfpathcurveto{\pgfpoint{372.519781bp}{114.624177bp}}{\pgfpoint{369.988674bp}{111.985358bp}}{\pgfpoint{366.866389bp}{111.985358bp}}
\pgfpathcurveto{\pgfpoint{363.744105bp}{111.985358bp}}{\pgfpoint{361.212992bp}{114.624177bp}}{\pgfpoint{361.212992bp}{117.879323bp}}
\pgfpathcurveto{\pgfpoint{361.212992bp}{121.134471bp}}{\pgfpoint{363.744105bp}{123.773288bp}}{\pgfpoint{366.866389bp}{123.773288bp}}
\pgfpathcurveto{\pgfpoint{369.988674bp}{123.773288bp}}{\pgfpoint{372.519781bp}{121.134471bp}}{\pgfpoint{372.519781bp}{117.879323bp}}
\pgfclosepath
\color[rgb]{0.0,0.0,0.0}
\pgfusepath{stroke}
\end{pgfscope}
\begin{pgfscope}
\pgfsetlinewidth{1.0bp}
\pgfsetrectcap 
\pgfsetmiterjoin \pgfsetmiterlimit{10.0}
\pgfpathmoveto{\pgfpoint{372.119781bp}{77.479323bp}}
\pgfpathcurveto{\pgfpoint{372.119781bp}{74.224177bp}}{\pgfpoint{369.588674bp}{71.585358bp}}{\pgfpoint{366.466389bp}{71.585358bp}}
\pgfpathcurveto{\pgfpoint{363.344105bp}{71.585358bp}}{\pgfpoint{360.812992bp}{74.224177bp}}{\pgfpoint{360.812992bp}{77.479323bp}}
\pgfpathcurveto{\pgfpoint{360.812992bp}{80.734471bp}}{\pgfpoint{363.344105bp}{83.373288bp}}{\pgfpoint{366.466389bp}{83.373288bp}}
\pgfpathcurveto{\pgfpoint{369.588674bp}{83.373288bp}}{\pgfpoint{372.119781bp}{80.734471bp}}{\pgfpoint{372.119781bp}{77.479323bp}}
\pgfclosepath
\color[rgb]{0.0,0.0,0.0}\pgfseteorule\pgfusepath{fill}
\pgfpathmoveto{\pgfpoint{372.119781bp}{77.479323bp}}
\pgfpathcurveto{\pgfpoint{372.119781bp}{74.224177bp}}{\pgfpoint{369.588674bp}{71.585358bp}}{\pgfpoint{366.466389bp}{71.585358bp}}
\pgfpathcurveto{\pgfpoint{363.344105bp}{71.585358bp}}{\pgfpoint{360.812992bp}{74.224177bp}}{\pgfpoint{360.812992bp}{77.479323bp}}
\pgfpathcurveto{\pgfpoint{360.812992bp}{80.734471bp}}{\pgfpoint{363.344105bp}{83.373288bp}}{\pgfpoint{366.466389bp}{83.373288bp}}
\pgfpathcurveto{\pgfpoint{369.588674bp}{83.373288bp}}{\pgfpoint{372.119781bp}{80.734471bp}}{\pgfpoint{372.119781bp}{77.479323bp}}
\pgfclosepath
\color[rgb]{0.0,0.0,0.0}
\pgfusepath{stroke}
\end{pgfscope}
\begin{pgfscope}
\pgfsetlinewidth{1.0bp}
\pgfsetrectcap 
\pgfsetmiterjoin \pgfsetmiterlimit{10.0}
\pgfpathmoveto{\pgfpoint{371.519781bp}{37.879323bp}}
\pgfpathcurveto{\pgfpoint{371.519781bp}{34.624177bp}}{\pgfpoint{368.988674bp}{31.985358bp}}{\pgfpoint{365.866389bp}{31.985358bp}}
\pgfpathcurveto{\pgfpoint{362.744105bp}{31.985358bp}}{\pgfpoint{360.212992bp}{34.624177bp}}{\pgfpoint{360.212992bp}{37.879323bp}}
\pgfpathcurveto{\pgfpoint{360.212992bp}{41.134471bp}}{\pgfpoint{362.744105bp}{43.773288bp}}{\pgfpoint{365.866389bp}{43.773288bp}}
\pgfpathcurveto{\pgfpoint{368.988674bp}{43.773288bp}}{\pgfpoint{371.519781bp}{41.134471bp}}{\pgfpoint{371.519781bp}{37.879323bp}}
\pgfclosepath
\color[rgb]{0.0,0.0,0.0}\pgfseteorule\pgfusepath{fill}
\pgfpathmoveto{\pgfpoint{371.519781bp}{37.879323bp}}
\pgfpathcurveto{\pgfpoint{371.519781bp}{34.624177bp}}{\pgfpoint{368.988674bp}{31.985358bp}}{\pgfpoint{365.866389bp}{31.985358bp}}
\pgfpathcurveto{\pgfpoint{362.744105bp}{31.985358bp}}{\pgfpoint{360.212992bp}{34.624177bp}}{\pgfpoint{360.212992bp}{37.879323bp}}
\pgfpathcurveto{\pgfpoint{360.212992bp}{41.134471bp}}{\pgfpoint{362.744105bp}{43.773288bp}}{\pgfpoint{365.866389bp}{43.773288bp}}
\pgfpathcurveto{\pgfpoint{368.988674bp}{43.773288bp}}{\pgfpoint{371.519781bp}{41.134471bp}}{\pgfpoint{371.519781bp}{37.879323bp}}
\pgfclosepath
\color[rgb]{0.0,0.0,0.0}
\pgfusepath{stroke}
\end{pgfscope}
\begin{pgfscope}
\pgfsetlinewidth{1.0bp}
\pgfsetrectcap 
\pgfsetmiterjoin \pgfsetmiterlimit{10.0}
\pgfpathmoveto{\pgfpoint{368.296875bp}{157.981249bp}}
\pgfpathlineto{\pgfpoint{293.896875bp}{157.981249bp}}
\color[rgb]{0.0,0.0,0.0}
\pgfusepath{stroke}
\end{pgfscope}
\begin{pgfscope}
\pgfsetlinewidth{1.0bp}
\pgfsetrectcap 
\pgfsetmiterjoin \pgfsetmiterlimit{10.0}
\pgfpathmoveto{\pgfpoint{294.496875bp}{157.981249bp}}
\pgfpathlineto{\pgfpoint{365.296875bp}{118.381249bp}}
\color[rgb]{0.0,0.0,0.0}
\pgfusepath{stroke}
\end{pgfscope}
\begin{pgfscope}
\pgfsetlinewidth{1.0bp}
\pgfsetrectcap 
\pgfsetmiterjoin \pgfsetmiterlimit{10.0}
\pgfpathmoveto{\pgfpoint{294.496875bp}{159.781249bp}}
\pgfpathlineto{\pgfpoint{365.896875bp}{77.581249bp}}
\color[rgb]{0.0,0.0,0.0}
\pgfusepath{stroke}
\end{pgfscope}
\begin{pgfscope}
\pgfsetlinewidth{1.0bp}
\pgfsetrectcap 
\pgfsetmiterjoin \pgfsetmiterlimit{10.0}
\pgfpathmoveto{\pgfpoint{294.496875bp}{157.981249bp}}
\pgfpathlineto{\pgfpoint{365.296875bp}{37.981249bp}}
\color[rgb]{0.0,0.0,0.0}
\pgfusepath{stroke}
\end{pgfscope}
\begin{pgfscope}
\pgfsetlinewidth{1.0bp}
\pgfsetrectcap 
\pgfsetmiterjoin \pgfsetmiterlimit{10.0}
\pgfpathmoveto{\pgfpoint{366.496875bp}{158.581249bp}}
\pgfpathlineto{\pgfpoint{292.096875bp}{117.781249bp}}
\color[rgb]{0.0,0.0,0.0}
\pgfusepath{stroke}
\end{pgfscope}
\begin{pgfscope}
\pgfsetlinewidth{1.0bp}
\pgfsetrectcap 
\pgfsetmiterjoin \pgfsetmiterlimit{10.0}
\pgfpathmoveto{\pgfpoint{367.096875bp}{118.981249bp}}
\pgfpathlineto{\pgfpoint{292.696875bp}{118.981249bp}}
\color[rgb]{0.0,0.0,0.0}
\pgfusepath{stroke}
\end{pgfscope}
\begin{pgfscope}
\pgftransformcm{1.0}{0.0}{0.0}{1.0}{\pgfpoint{276.696875bp}{152.981249bp}}
\pgftext[left,base]{\sffamily\mdseries\upshape\large
\color[rgb]{0.0,0.0,0.0}$1$}
\end{pgfscope}
\begin{pgfscope}
\pgftransformcm{1.0}{0.0}{0.0}{1.0}{\pgfpoint{276.896875bp}{114.381249bp}}
\pgftext[left,base]{\sffamily\mdseries\upshape\large
\color[rgb]{0.0,0.0,0.0}$2$}
\end{pgfscope}
\begin{pgfscope}
\pgftransformcm{1.0}{0.0}{0.0}{1.0}{\pgfpoint{376.496875bp}{152.781249bp}}
\pgftext[left,base]{\sffamily\mdseries\upshape\large
\color[rgb]{0.0,0.0,0.0}$1'$}
\end{pgfscope}
\begin{pgfscope}
\pgftransformcm{1.0}{0.0}{0.0}{1.0}{\pgfpoint{376.496875bp}{113.781249bp}}
\pgftext[left,base]{\sffamily\mdseries\upshape\large
\color[rgb]{0.0,0.0,0.0}$2'$}
\end{pgfscope}
\begin{pgfscope}
\pgftransformcm{1.0}{0.0}{0.0}{1.0}{\pgfpoint{375.896875bp}{74.181249bp}}
\pgftext[left,base]{\sffamily\mdseries\upshape\large
\color[rgb]{0.0,0.0,0.0}$3'$}
\end{pgfscope}
\begin{pgfscope}
\pgftransformcm{1.0}{0.0}{0.0}{1.0}{\pgfpoint{375.496875bp}{33.781249bp}}
\pgftext[left,base]{\sffamily\mdseries\upshape\large
\color[rgb]{0.0,0.0,0.0}$4'$}
\end{pgfscope}
\begin{pgfscope}
\pgftransformcm{1.0}{0.0}{0.0}{1.0}{\pgfpoint{42.096875bp}{5.381249bp}}
\pgftext[left,base]{\sffamily\mdseries\upshape\LARGE
\color[rgb]{0.0,0.0,0.0}$H^{\ex}_1$}
\end{pgfscope}
\begin{pgfscope}
\pgftransformcm{1.0}{0.0}{0.0}{1.0}{\pgfpoint{173.696875bp}{3.981249bp}}
\pgftext[left,base]{\sffamily\mdseries\upshape\LARGE
\color[rgb]{0.0,0.0,0.0}$H^{\ex}_2$}
\end{pgfscope}
\begin{pgfscope}
\pgftransformcm{1.0}{0.0}{0.0}{1.0}{\pgfpoint{314.296875bp}{4.381249bp}}
\pgftext[left,base]{\sffamily\mdseries\upshape\LARGE
\color[rgb]{0.0,0.0,0.0}$H_1$}
\end{pgfscope}
\begin{pgfscope}
\pgfsetlinewidth{1.0bp}
\pgfsetrectcap 
\pgfsetmiterjoin \pgfsetmiterlimit{10.0}
\pgfpathmoveto{\pgfpoint{225.696875bp}{79.781249bp}}
\pgfpathlineto{\pgfpoint{151.296875bp}{79.781249bp}}
\color[rgb]{0.0,0.0,0.0}
\pgfusepath{stroke}
\end{pgfscope}
\begin{pgfscope}
\pgfsetlinewidth{1.0bp}
\pgfsetrectcap 
\pgfsetmiterjoin \pgfsetmiterlimit{10.0}
\pgfpathmoveto{\pgfpoint{428.319781bp}{157.479323bp}}
\pgfpathcurveto{\pgfpoint{428.319781bp}{154.224177bp}}{\pgfpoint{425.788674bp}{151.585358bp}}{\pgfpoint{422.666389bp}{151.585358bp}}
\pgfpathcurveto{\pgfpoint{419.544105bp}{151.585358bp}}{\pgfpoint{417.012992bp}{154.224177bp}}{\pgfpoint{417.012992bp}{157.479323bp}}
\pgfpathcurveto{\pgfpoint{417.012992bp}{160.734471bp}}{\pgfpoint{419.544105bp}{163.373288bp}}{\pgfpoint{422.666389bp}{163.373288bp}}
\pgfpathcurveto{\pgfpoint{425.788674bp}{163.373288bp}}{\pgfpoint{428.319781bp}{160.734471bp}}{\pgfpoint{428.319781bp}{157.479323bp}}
\pgfclosepath
\color[rgb]{0.0,0.0,0.0}\pgfseteorule\pgfusepath{fill}
\pgfpathmoveto{\pgfpoint{428.319781bp}{157.479323bp}}
\pgfpathcurveto{\pgfpoint{428.319781bp}{154.224177bp}}{\pgfpoint{425.788674bp}{151.585358bp}}{\pgfpoint{422.666389bp}{151.585358bp}}
\pgfpathcurveto{\pgfpoint{419.544105bp}{151.585358bp}}{\pgfpoint{417.012992bp}{154.224177bp}}{\pgfpoint{417.012992bp}{157.479323bp}}
\pgfpathcurveto{\pgfpoint{417.012992bp}{160.734471bp}}{\pgfpoint{419.544105bp}{163.373288bp}}{\pgfpoint{422.666389bp}{163.373288bp}}
\pgfpathcurveto{\pgfpoint{425.788674bp}{163.373288bp}}{\pgfpoint{428.319781bp}{160.734471bp}}{\pgfpoint{428.319781bp}{157.479323bp}}
\pgfclosepath
\color[rgb]{0.0,0.0,0.0}
\pgfusepath{stroke}
\end{pgfscope}
\begin{pgfscope}
\pgfsetlinewidth{1.0bp}
\pgfsetrectcap 
\pgfsetmiterjoin \pgfsetmiterlimit{10.0}
\pgfpathmoveto{\pgfpoint{502.919781bp}{156.679323bp}}
\pgfpathcurveto{\pgfpoint{502.919781bp}{153.424177bp}}{\pgfpoint{500.388674bp}{150.785358bp}}{\pgfpoint{497.266389bp}{150.785358bp}}
\pgfpathcurveto{\pgfpoint{494.144105bp}{150.785358bp}}{\pgfpoint{491.612992bp}{153.424177bp}}{\pgfpoint{491.612992bp}{156.679323bp}}
\pgfpathcurveto{\pgfpoint{491.612992bp}{159.934471bp}}{\pgfpoint{494.144105bp}{162.573288bp}}{\pgfpoint{497.266389bp}{162.573288bp}}
\pgfpathcurveto{\pgfpoint{500.388674bp}{162.573288bp}}{\pgfpoint{502.919781bp}{159.934471bp}}{\pgfpoint{502.919781bp}{156.679323bp}}
\pgfclosepath
\color[rgb]{0.0,0.0,0.0}\pgfseteorule\pgfusepath{fill}
\pgfpathmoveto{\pgfpoint{502.919781bp}{156.679323bp}}
\pgfpathcurveto{\pgfpoint{502.919781bp}{153.424177bp}}{\pgfpoint{500.388674bp}{150.785358bp}}{\pgfpoint{497.266389bp}{150.785358bp}}
\pgfpathcurveto{\pgfpoint{494.144105bp}{150.785358bp}}{\pgfpoint{491.612992bp}{153.424177bp}}{\pgfpoint{491.612992bp}{156.679323bp}}
\pgfpathcurveto{\pgfpoint{491.612992bp}{159.934471bp}}{\pgfpoint{494.144105bp}{162.573288bp}}{\pgfpoint{497.266389bp}{162.573288bp}}
\pgfpathcurveto{\pgfpoint{500.388674bp}{162.573288bp}}{\pgfpoint{502.919781bp}{159.934471bp}}{\pgfpoint{502.919781bp}{156.679323bp}}
\pgfclosepath
\color[rgb]{0.0,0.0,0.0}
\pgfusepath{stroke}
\end{pgfscope}
\begin{pgfscope}
\pgfsetlinewidth{1.0bp}
\pgfsetrectcap 
\pgfsetmiterjoin \pgfsetmiterlimit{10.0}
\pgfpathmoveto{\pgfpoint{502.319781bp}{117.079323bp}}
\pgfpathcurveto{\pgfpoint{502.319781bp}{113.824177bp}}{\pgfpoint{499.788674bp}{111.185358bp}}{\pgfpoint{496.666389bp}{111.185358bp}}
\pgfpathcurveto{\pgfpoint{493.544105bp}{111.185358bp}}{\pgfpoint{491.012992bp}{113.824177bp}}{\pgfpoint{491.012992bp}{117.079323bp}}
\pgfpathcurveto{\pgfpoint{491.012992bp}{120.334471bp}}{\pgfpoint{493.544105bp}{122.973288bp}}{\pgfpoint{496.666389bp}{122.973288bp}}
\pgfpathcurveto{\pgfpoint{499.788674bp}{122.973288bp}}{\pgfpoint{502.319781bp}{120.334471bp}}{\pgfpoint{502.319781bp}{117.079323bp}}
\pgfclosepath
\color[rgb]{0.0,0.0,0.0}\pgfseteorule\pgfusepath{fill}
\pgfpathmoveto{\pgfpoint{502.319781bp}{117.079323bp}}
\pgfpathcurveto{\pgfpoint{502.319781bp}{113.824177bp}}{\pgfpoint{499.788674bp}{111.185358bp}}{\pgfpoint{496.666389bp}{111.185358bp}}
\pgfpathcurveto{\pgfpoint{493.544105bp}{111.185358bp}}{\pgfpoint{491.012992bp}{113.824177bp}}{\pgfpoint{491.012992bp}{117.079323bp}}
\pgfpathcurveto{\pgfpoint{491.012992bp}{120.334471bp}}{\pgfpoint{493.544105bp}{122.973288bp}}{\pgfpoint{496.666389bp}{122.973288bp}}
\pgfpathcurveto{\pgfpoint{499.788674bp}{122.973288bp}}{\pgfpoint{502.319781bp}{120.334471bp}}{\pgfpoint{502.319781bp}{117.079323bp}}
\pgfclosepath
\color[rgb]{0.0,0.0,0.0}
\pgfusepath{stroke}
\end{pgfscope}
\begin{pgfscope}
\pgfsetlinewidth{1.0bp}
\pgfsetrectcap 
\pgfsetmiterjoin \pgfsetmiterlimit{10.0}
\pgfpathmoveto{\pgfpoint{501.919781bp}{76.679323bp}}
\pgfpathcurveto{\pgfpoint{501.919781bp}{73.424177bp}}{\pgfpoint{499.388674bp}{70.785358bp}}{\pgfpoint{496.266389bp}{70.785358bp}}
\pgfpathcurveto{\pgfpoint{493.144105bp}{70.785358bp}}{\pgfpoint{490.612992bp}{73.424177bp}}{\pgfpoint{490.612992bp}{76.679323bp}}
\pgfpathcurveto{\pgfpoint{490.612992bp}{79.934471bp}}{\pgfpoint{493.144105bp}{82.573288bp}}{\pgfpoint{496.266389bp}{82.573288bp}}
\pgfpathcurveto{\pgfpoint{499.388674bp}{82.573288bp}}{\pgfpoint{501.919781bp}{79.934471bp}}{\pgfpoint{501.919781bp}{76.679323bp}}
\pgfclosepath
\color[rgb]{0.0,0.0,0.0}\pgfseteorule\pgfusepath{fill}
\pgfpathmoveto{\pgfpoint{501.919781bp}{76.679323bp}}
\pgfpathcurveto{\pgfpoint{501.919781bp}{73.424177bp}}{\pgfpoint{499.388674bp}{70.785358bp}}{\pgfpoint{496.266389bp}{70.785358bp}}
\pgfpathcurveto{\pgfpoint{493.144105bp}{70.785358bp}}{\pgfpoint{490.612992bp}{73.424177bp}}{\pgfpoint{490.612992bp}{76.679323bp}}
\pgfpathcurveto{\pgfpoint{490.612992bp}{79.934471bp}}{\pgfpoint{493.144105bp}{82.573288bp}}{\pgfpoint{496.266389bp}{82.573288bp}}
\pgfpathcurveto{\pgfpoint{499.388674bp}{82.573288bp}}{\pgfpoint{501.919781bp}{79.934471bp}}{\pgfpoint{501.919781bp}{76.679323bp}}
\pgfclosepath
\color[rgb]{0.0,0.0,0.0}
\pgfusepath{stroke}
\end{pgfscope}
\begin{pgfscope}
\pgfsetlinewidth{1.0bp}
\pgfsetrectcap 
\pgfsetmiterjoin \pgfsetmiterlimit{10.0}
\pgfpathmoveto{\pgfpoint{501.319781bp}{37.079323bp}}
\pgfpathcurveto{\pgfpoint{501.319781bp}{33.824177bp}}{\pgfpoint{498.788674bp}{31.185358bp}}{\pgfpoint{495.666389bp}{31.185358bp}}
\pgfpathcurveto{\pgfpoint{492.544105bp}{31.185358bp}}{\pgfpoint{490.012992bp}{33.824177bp}}{\pgfpoint{490.012992bp}{37.079323bp}}
\pgfpathcurveto{\pgfpoint{490.012992bp}{40.334471bp}}{\pgfpoint{492.544105bp}{42.973288bp}}{\pgfpoint{495.666389bp}{42.973288bp}}
\pgfpathcurveto{\pgfpoint{498.788674bp}{42.973288bp}}{\pgfpoint{501.319781bp}{40.334471bp}}{\pgfpoint{501.319781bp}{37.079323bp}}
\pgfclosepath
\color[rgb]{0.0,0.0,0.0}\pgfseteorule\pgfusepath{fill}
\pgfpathmoveto{\pgfpoint{501.319781bp}{37.079323bp}}
\pgfpathcurveto{\pgfpoint{501.319781bp}{33.824177bp}}{\pgfpoint{498.788674bp}{31.185358bp}}{\pgfpoint{495.666389bp}{31.185358bp}}
\pgfpathcurveto{\pgfpoint{492.544105bp}{31.185358bp}}{\pgfpoint{490.012992bp}{33.824177bp}}{\pgfpoint{490.012992bp}{37.079323bp}}
\pgfpathcurveto{\pgfpoint{490.012992bp}{40.334471bp}}{\pgfpoint{492.544105bp}{42.973288bp}}{\pgfpoint{495.666389bp}{42.973288bp}}
\pgfpathcurveto{\pgfpoint{498.788674bp}{42.973288bp}}{\pgfpoint{501.319781bp}{40.334471bp}}{\pgfpoint{501.319781bp}{37.079323bp}}
\pgfclosepath
\color[rgb]{0.0,0.0,0.0}
\pgfusepath{stroke}
\end{pgfscope}
\begin{pgfscope}
\pgfsetlinewidth{1.0bp}
\pgfsetrectcap 
\pgfsetmiterjoin \pgfsetmiterlimit{10.0}
\pgfpathmoveto{\pgfpoint{498.096875bp}{157.181249bp}}
\pgfpathlineto{\pgfpoint{423.696875bp}{157.181249bp}}
\color[rgb]{0.0,0.0,0.0}
\pgfusepath{stroke}
\end{pgfscope}
\begin{pgfscope}
\pgfsetlinewidth{1.0bp}
\pgfsetrectcap 
\pgfsetmiterjoin \pgfsetmiterlimit{10.0}
\pgfpathmoveto{\pgfpoint{424.296875bp}{157.181249bp}}
\pgfpathlineto{\pgfpoint{495.096875bp}{117.581249bp}}
\color[rgb]{0.0,0.0,0.0}
\pgfusepath{stroke}
\end{pgfscope}
\begin{pgfscope}
\pgfsetlinewidth{1.0bp}
\pgfsetrectcap 
\pgfsetmiterjoin \pgfsetmiterlimit{10.0}
\pgfpathmoveto{\pgfpoint{424.296875bp}{158.981249bp}}
\pgfpathlineto{\pgfpoint{495.696875bp}{76.781249bp}}
\color[rgb]{0.0,0.0,0.0}
\pgfusepath{stroke}
\end{pgfscope}
\begin{pgfscope}
\pgfsetlinewidth{1.0bp}
\pgfsetrectcap 
\pgfsetmiterjoin \pgfsetmiterlimit{10.0}
\pgfpathmoveto{\pgfpoint{424.296875bp}{157.181249bp}}
\pgfpathlineto{\pgfpoint{495.096875bp}{37.181249bp}}
\color[rgb]{0.0,0.0,0.0}
\pgfusepath{stroke}
\end{pgfscope}
\begin{pgfscope}
\pgftransformcm{1.0}{0.0}{0.0}{1.0}{\pgfpoint{406.496875bp}{152.181249bp}}
\pgftext[left,base]{\sffamily\mdseries\upshape\large
\color[rgb]{0.0,0.0,0.0}$1$}
\end{pgfscope}
\begin{pgfscope}
\pgftransformcm{1.0}{0.0}{0.0}{1.0}{\pgfpoint{506.296875bp}{151.981249bp}}
\pgftext[left,base]{\sffamily\mdseries\upshape\large
\color[rgb]{0.0,0.0,0.0}$1'$}
\end{pgfscope}
\begin{pgfscope}
\pgftransformcm{1.0}{0.0}{0.0}{1.0}{\pgfpoint{506.296875bp}{112.981249bp}}
\pgftext[left,base]{\sffamily\mdseries\upshape\large
\color[rgb]{0.0,0.0,0.0}$2'$}
\end{pgfscope}
\begin{pgfscope}
\pgftransformcm{1.0}{0.0}{0.0}{1.0}{\pgfpoint{505.696875bp}{73.381249bp}}
\pgftext[left,base]{\sffamily\mdseries\upshape\large
\color[rgb]{0.0,0.0,0.0}$3'$}
\end{pgfscope}
\begin{pgfscope}
\pgftransformcm{1.0}{0.0}{0.0}{1.0}{\pgfpoint{505.296875bp}{32.981249bp}}
\pgftext[left,base]{\sffamily\mdseries\upshape\large
\color[rgb]{0.0,0.0,0.0}$4'$}
\end{pgfscope}
\begin{pgfscope}
\pgftransformcm{1.0}{0.0}{0.0}{1.0}{\pgfpoint{444.096875bp}{3.581249bp}}
\pgftext[left,base]{\sffamily\mdseries\upshape\LARGE
\color[rgb]{0.0,0.0,0.0}$H_2$}
\end{pgfscope}
\begin{pgfscope}
\pgfsetlinewidth{1.0bp}
\pgfsetrectcap 
\pgfsetmiterjoin \pgfsetmiterlimit{10.0}
\pgfpathmoveto{\pgfpoint{428.719781bp}{78.479323bp}}
\pgfpathcurveto{\pgfpoint{428.719781bp}{75.224177bp}}{\pgfpoint{426.188674bp}{72.585358bp}}{\pgfpoint{423.066389bp}{72.585358bp}}
\pgfpathcurveto{\pgfpoint{419.944105bp}{72.585358bp}}{\pgfpoint{417.412992bp}{75.224177bp}}{\pgfpoint{417.412992bp}{78.479323bp}}
\pgfpathcurveto{\pgfpoint{417.412992bp}{81.734471bp}}{\pgfpoint{419.944105bp}{84.373288bp}}{\pgfpoint{423.066389bp}{84.373288bp}}
\pgfpathcurveto{\pgfpoint{426.188674bp}{84.373288bp}}{\pgfpoint{428.719781bp}{81.734471bp}}{\pgfpoint{428.719781bp}{78.479323bp}}
\pgfclosepath
\color[rgb]{0.0,0.0,0.0}\pgfseteorule\pgfusepath{fill}
\pgfpathmoveto{\pgfpoint{428.719781bp}{78.479323bp}}
\pgfpathcurveto{\pgfpoint{428.719781bp}{75.224177bp}}{\pgfpoint{426.188674bp}{72.585358bp}}{\pgfpoint{423.066389bp}{72.585358bp}}
\pgfpathcurveto{\pgfpoint{419.944105bp}{72.585358bp}}{\pgfpoint{417.412992bp}{75.224177bp}}{\pgfpoint{417.412992bp}{78.479323bp}}
\pgfpathcurveto{\pgfpoint{417.412992bp}{81.734471bp}}{\pgfpoint{419.944105bp}{84.373288bp}}{\pgfpoint{423.066389bp}{84.373288bp}}
\pgfpathcurveto{\pgfpoint{426.188674bp}{84.373288bp}}{\pgfpoint{428.719781bp}{81.734471bp}}{\pgfpoint{428.719781bp}{78.479323bp}}
\pgfclosepath
\color[rgb]{0.0,0.0,0.0}
\pgfusepath{stroke}
\end{pgfscope}
\begin{pgfscope}
\pgfsetlinewidth{1.0bp}
\pgfsetrectcap 
\pgfsetmiterjoin \pgfsetmiterlimit{10.0}
\pgfpathmoveto{\pgfpoint{498.296875bp}{159.381249bp}}
\pgfpathlineto{\pgfpoint{425.096875bp}{80.781249bp}}
\color[rgb]{0.0,0.0,0.0}
\pgfusepath{stroke}
\end{pgfscope}
\begin{pgfscope}
\pgftransformcm{1.0}{0.0}{0.0}{1.0}{\pgfpoint{407.696875bp}{75.381249bp}}
\pgftext[left,base]{\sffamily\mdseries\upshape\large
\color[rgb]{0.0,0.0,0.0}$3$}
\end{pgfscope}
\begin{pgfscope}
\pgfsetlinewidth{1.0bp}
\pgfsetrectcap 
\pgfsetmiterjoin \pgfsetmiterlimit{10.0}
\pgfpathmoveto{\pgfpoint{497.696875bp}{77.781249bp}}
\pgfpathlineto{\pgfpoint{423.296875bp}{77.781249bp}}
\color[rgb]{0.0,0.0,0.0}
\pgfusepath{stroke}
\end{pgfscope}
\end{pgfpicture}

%% file: Hexample10.tex
% C:\Users\agalanis\Dropbox\H-colorings\Hexample10.jdr
% Created by Jpgfdraw version 0.5.6b
% 01-Feb-2015 14:43:51
\begin{pgfpicture}{0bp}{0bp}{303.83847bp}{226.5375bp}
\begin{pgfscope}
\pgfsetlinewidth{1.0bp}
\pgfsetrectcap 
\pgfsetmiterjoin \pgfsetmiterlimit{10.0}
\pgfpathmoveto{\pgfpoint{297.444513bp}{146.844675bp}}
\pgfpathcurveto{\pgfpoint{294.189367bp}{146.844675bp}}{\pgfpoint{291.550548bp}{149.375782bp}}{\pgfpoint{291.550548bp}{152.498066bp}}
\pgfpathcurveto{\pgfpoint{291.550548bp}{155.62035bp}}{\pgfpoint{294.189367bp}{158.151463bp}}{\pgfpoint{297.444513bp}{158.151463bp}}
\pgfpathcurveto{\pgfpoint{300.699661bp}{158.151463bp}}{\pgfpoint{303.338478bp}{155.62035bp}}{\pgfpoint{303.338478bp}{152.498066bp}}
\pgfpathcurveto{\pgfpoint{303.338478bp}{149.375782bp}}{\pgfpoint{300.699661bp}{146.844675bp}}{\pgfpoint{297.444513bp}{146.844675bp}}
\pgfclosepath
\color[rgb]{0.0,0.0,0.0}\pgfseteorule\pgfusepath{fill}
\pgfpathmoveto{\pgfpoint{297.444513bp}{146.844675bp}}
\pgfpathcurveto{\pgfpoint{294.189367bp}{146.844675bp}}{\pgfpoint{291.550548bp}{149.375782bp}}{\pgfpoint{291.550548bp}{152.498066bp}}
\pgfpathcurveto{\pgfpoint{291.550548bp}{155.62035bp}}{\pgfpoint{294.189367bp}{158.151463bp}}{\pgfpoint{297.444513bp}{158.151463bp}}
\pgfpathcurveto{\pgfpoint{300.699661bp}{158.151463bp}}{\pgfpoint{303.338478bp}{155.62035bp}}{\pgfpoint{303.338478bp}{152.498066bp}}
\pgfpathcurveto{\pgfpoint{303.338478bp}{149.375782bp}}{\pgfpoint{300.699661bp}{146.844675bp}}{\pgfpoint{297.444513bp}{146.844675bp}}
\pgfclosepath
\color[rgb]{0.0,0.0,0.0}
\pgfusepath{stroke}
\end{pgfscope}
\begin{pgfscope}
\pgfsetlinewidth{1.0bp}
\pgfsetrectcap 
\pgfsetmiterjoin \pgfsetmiterlimit{10.0}
\pgfpathmoveto{\pgfpoint{297.444487bp}{85.544676bp}}
\pgfpathcurveto{\pgfpoint{294.189341bp}{85.544676bp}}{\pgfpoint{291.550522bp}{88.075783bp}}{\pgfpoint{291.550522bp}{91.198067bp}}
\pgfpathcurveto{\pgfpoint{291.550522bp}{94.320352bp}}{\pgfpoint{294.189341bp}{96.851465bp}}{\pgfpoint{297.444487bp}{96.851465bp}}
\pgfpathcurveto{\pgfpoint{300.699635bp}{96.851465bp}}{\pgfpoint{303.338452bp}{94.320352bp}}{\pgfpoint{303.338452bp}{91.198067bp}}
\pgfpathcurveto{\pgfpoint{303.338452bp}{88.075783bp}}{\pgfpoint{300.699635bp}{85.544676bp}}{\pgfpoint{297.444487bp}{85.544676bp}}
\pgfclosepath
\color[rgb]{0.0,0.0,0.0}\pgfseteorule\pgfusepath{fill}
\pgfpathmoveto{\pgfpoint{297.444487bp}{85.544676bp}}
\pgfpathcurveto{\pgfpoint{294.189341bp}{85.544676bp}}{\pgfpoint{291.550522bp}{88.075783bp}}{\pgfpoint{291.550522bp}{91.198067bp}}
\pgfpathcurveto{\pgfpoint{291.550522bp}{94.320352bp}}{\pgfpoint{294.189341bp}{96.851465bp}}{\pgfpoint{297.444487bp}{96.851465bp}}
\pgfpathcurveto{\pgfpoint{300.699635bp}{96.851465bp}}{\pgfpoint{303.338452bp}{94.320352bp}}{\pgfpoint{303.338452bp}{91.198067bp}}
\pgfpathcurveto{\pgfpoint{303.338452bp}{88.075783bp}}{\pgfpoint{300.699635bp}{85.544676bp}}{\pgfpoint{297.444487bp}{85.544676bp}}
\pgfclosepath
\color[rgb]{0.0,0.0,0.0}
\pgfusepath{stroke}
\end{pgfscope}
\begin{pgfscope}
\pgfsetlinewidth{1.0bp}
\pgfsetrectcap 
\pgfsetmiterjoin \pgfsetmiterlimit{10.0}
\pgfpathmoveto{\pgfpoint{257.844513bp}{146.844669bp}}
\pgfpathcurveto{\pgfpoint{254.589367bp}{146.844669bp}}{\pgfpoint{251.950548bp}{149.375776bp}}{\pgfpoint{251.950548bp}{152.49806bp}}
\pgfpathcurveto{\pgfpoint{251.950548bp}{155.620344bp}}{\pgfpoint{254.589367bp}{158.151457bp}}{\pgfpoint{257.844513bp}{158.151457bp}}
\pgfpathcurveto{\pgfpoint{261.099661bp}{158.151457bp}}{\pgfpoint{263.738478bp}{155.620344bp}}{\pgfpoint{263.738478bp}{152.49806bp}}
\pgfpathcurveto{\pgfpoint{263.738478bp}{149.375776bp}}{\pgfpoint{261.099661bp}{146.844669bp}}{\pgfpoint{257.844513bp}{146.844669bp}}
\pgfclosepath
\color[rgb]{0.0,0.0,0.0}\pgfseteorule\pgfusepath{fill}
\pgfpathmoveto{\pgfpoint{257.844513bp}{146.844669bp}}
\pgfpathcurveto{\pgfpoint{254.589367bp}{146.844669bp}}{\pgfpoint{251.950548bp}{149.375776bp}}{\pgfpoint{251.950548bp}{152.49806bp}}
\pgfpathcurveto{\pgfpoint{251.950548bp}{155.620344bp}}{\pgfpoint{254.589367bp}{158.151457bp}}{\pgfpoint{257.844513bp}{158.151457bp}}
\pgfpathcurveto{\pgfpoint{261.099661bp}{158.151457bp}}{\pgfpoint{263.738478bp}{155.620344bp}}{\pgfpoint{263.738478bp}{152.49806bp}}
\pgfpathcurveto{\pgfpoint{263.738478bp}{149.375776bp}}{\pgfpoint{261.099661bp}{146.844669bp}}{\pgfpoint{257.844513bp}{146.844669bp}}
\pgfclosepath
\color[rgb]{0.0,0.0,0.0}
\pgfusepath{stroke}
\end{pgfscope}
\begin{pgfscope}
\pgfsetlinewidth{1.0bp}
\pgfsetrectcap 
\pgfsetmiterjoin \pgfsetmiterlimit{10.0}
\pgfpathmoveto{\pgfpoint{257.844548bp}{85.54467bp}}
\pgfpathcurveto{\pgfpoint{254.589402bp}{85.54467bp}}{\pgfpoint{251.950583bp}{88.075777bp}}{\pgfpoint{251.950583bp}{91.198061bp}}
\pgfpathcurveto{\pgfpoint{251.950583bp}{94.320346bp}}{\pgfpoint{254.589402bp}{96.851459bp}}{\pgfpoint{257.844548bp}{96.851459bp}}
\pgfpathcurveto{\pgfpoint{261.099696bp}{96.851459bp}}{\pgfpoint{263.738513bp}{94.320346bp}}{\pgfpoint{263.738513bp}{91.198061bp}}
\pgfpathcurveto{\pgfpoint{263.738513bp}{88.075777bp}}{\pgfpoint{261.099696bp}{85.54467bp}}{\pgfpoint{257.844548bp}{85.54467bp}}
\pgfclosepath
\color[rgb]{0.0,0.0,0.0}\pgfseteorule\pgfusepath{fill}
\pgfpathmoveto{\pgfpoint{257.844548bp}{85.54467bp}}
\pgfpathcurveto{\pgfpoint{254.589402bp}{85.54467bp}}{\pgfpoint{251.950583bp}{88.075777bp}}{\pgfpoint{251.950583bp}{91.198061bp}}
\pgfpathcurveto{\pgfpoint{251.950583bp}{94.320346bp}}{\pgfpoint{254.589402bp}{96.851459bp}}{\pgfpoint{257.844548bp}{96.851459bp}}
\pgfpathcurveto{\pgfpoint{261.099696bp}{96.851459bp}}{\pgfpoint{263.738513bp}{94.320346bp}}{\pgfpoint{263.738513bp}{91.198061bp}}
\pgfpathcurveto{\pgfpoint{263.738513bp}{88.075777bp}}{\pgfpoint{261.099696bp}{85.54467bp}}{\pgfpoint{257.844548bp}{85.54467bp}}
\pgfclosepath
\color[rgb]{0.0,0.0,0.0}
\pgfusepath{stroke}
\end{pgfscope}
\begin{pgfscope}
\pgfsetlinewidth{1.0bp}
\pgfsetrectcap 
\pgfsetmiterjoin \pgfsetmiterlimit{10.0}
\pgfpathmoveto{\pgfpoint{217.444513bp}{146.844675bp}}
\pgfpathcurveto{\pgfpoint{214.189367bp}{146.844675bp}}{\pgfpoint{211.550548bp}{149.375782bp}}{\pgfpoint{211.550548bp}{152.498066bp}}
\pgfpathcurveto{\pgfpoint{211.550548bp}{155.62035bp}}{\pgfpoint{214.189367bp}{158.151463bp}}{\pgfpoint{217.444513bp}{158.151463bp}}
\pgfpathcurveto{\pgfpoint{220.699661bp}{158.151463bp}}{\pgfpoint{223.338478bp}{155.62035bp}}{\pgfpoint{223.338478bp}{152.498066bp}}
\pgfpathcurveto{\pgfpoint{223.338478bp}{149.375782bp}}{\pgfpoint{220.699661bp}{146.844675bp}}{\pgfpoint{217.444513bp}{146.844675bp}}
\pgfclosepath
\color[rgb]{0.0,0.0,0.0}\pgfseteorule\pgfusepath{fill}
\pgfpathmoveto{\pgfpoint{217.444513bp}{146.844675bp}}
\pgfpathcurveto{\pgfpoint{214.189367bp}{146.844675bp}}{\pgfpoint{211.550548bp}{149.375782bp}}{\pgfpoint{211.550548bp}{152.498066bp}}
\pgfpathcurveto{\pgfpoint{211.550548bp}{155.62035bp}}{\pgfpoint{214.189367bp}{158.151463bp}}{\pgfpoint{217.444513bp}{158.151463bp}}
\pgfpathcurveto{\pgfpoint{220.699661bp}{158.151463bp}}{\pgfpoint{223.338478bp}{155.62035bp}}{\pgfpoint{223.338478bp}{152.498066bp}}
\pgfpathcurveto{\pgfpoint{223.338478bp}{149.375782bp}}{\pgfpoint{220.699661bp}{146.844675bp}}{\pgfpoint{217.444513bp}{146.844675bp}}
\pgfclosepath
\color[rgb]{0.0,0.0,0.0}
\pgfusepath{stroke}
\end{pgfscope}
\begin{pgfscope}
\pgfsetlinewidth{1.0bp}
\pgfsetrectcap 
\pgfsetmiterjoin \pgfsetmiterlimit{10.0}
\pgfpathmoveto{\pgfpoint{217.444518bp}{85.544676bp}}
\pgfpathcurveto{\pgfpoint{214.189371bp}{85.544676bp}}{\pgfpoint{211.550552bp}{88.075783bp}}{\pgfpoint{211.550552bp}{91.198067bp}}
\pgfpathcurveto{\pgfpoint{211.550552bp}{94.320352bp}}{\pgfpoint{214.189371bp}{96.851465bp}}{\pgfpoint{217.444518bp}{96.851465bp}}
\pgfpathcurveto{\pgfpoint{220.699665bp}{96.851465bp}}{\pgfpoint{223.338483bp}{94.320352bp}}{\pgfpoint{223.338483bp}{91.198067bp}}
\pgfpathcurveto{\pgfpoint{223.338483bp}{88.075783bp}}{\pgfpoint{220.699665bp}{85.544676bp}}{\pgfpoint{217.444518bp}{85.544676bp}}
\pgfclosepath
\color[rgb]{0.0,0.0,0.0}\pgfseteorule\pgfusepath{fill}
\pgfpathmoveto{\pgfpoint{217.444518bp}{85.544676bp}}
\pgfpathcurveto{\pgfpoint{214.189371bp}{85.544676bp}}{\pgfpoint{211.550552bp}{88.075783bp}}{\pgfpoint{211.550552bp}{91.198067bp}}
\pgfpathcurveto{\pgfpoint{211.550552bp}{94.320352bp}}{\pgfpoint{214.189371bp}{96.851465bp}}{\pgfpoint{217.444518bp}{96.851465bp}}
\pgfpathcurveto{\pgfpoint{220.699665bp}{96.851465bp}}{\pgfpoint{223.338483bp}{94.320352bp}}{\pgfpoint{223.338483bp}{91.198067bp}}
\pgfpathcurveto{\pgfpoint{223.338483bp}{88.075783bp}}{\pgfpoint{220.699665bp}{85.544676bp}}{\pgfpoint{217.444518bp}{85.544676bp}}
\pgfclosepath
\color[rgb]{0.0,0.0,0.0}
\pgfusepath{stroke}
\end{pgfscope}
\begin{pgfscope}
\pgfsetlinewidth{1.0bp}
\pgfsetrectcap 
\pgfsetmiterjoin \pgfsetmiterlimit{10.0}
\pgfpathmoveto{\pgfpoint{177.844513bp}{146.844675bp}}
\pgfpathcurveto{\pgfpoint{174.589367bp}{146.844675bp}}{\pgfpoint{171.950548bp}{149.375782bp}}{\pgfpoint{171.950548bp}{152.498066bp}}
\pgfpathcurveto{\pgfpoint{171.950548bp}{155.62035bp}}{\pgfpoint{174.589367bp}{158.151463bp}}{\pgfpoint{177.844513bp}{158.151463bp}}
\pgfpathcurveto{\pgfpoint{181.099661bp}{158.151463bp}}{\pgfpoint{183.738478bp}{155.62035bp}}{\pgfpoint{183.738478bp}{152.498066bp}}
\pgfpathcurveto{\pgfpoint{183.738478bp}{149.375782bp}}{\pgfpoint{181.099661bp}{146.844675bp}}{\pgfpoint{177.844513bp}{146.844675bp}}
\pgfclosepath
\color[rgb]{0.0,0.0,0.0}\pgfseteorule\pgfusepath{fill}
\pgfpathmoveto{\pgfpoint{177.844513bp}{146.844675bp}}
\pgfpathcurveto{\pgfpoint{174.589367bp}{146.844675bp}}{\pgfpoint{171.950548bp}{149.375782bp}}{\pgfpoint{171.950548bp}{152.498066bp}}
\pgfpathcurveto{\pgfpoint{171.950548bp}{155.62035bp}}{\pgfpoint{174.589367bp}{158.151463bp}}{\pgfpoint{177.844513bp}{158.151463bp}}
\pgfpathcurveto{\pgfpoint{181.099661bp}{158.151463bp}}{\pgfpoint{183.738478bp}{155.62035bp}}{\pgfpoint{183.738478bp}{152.498066bp}}
\pgfpathcurveto{\pgfpoint{183.738478bp}{149.375782bp}}{\pgfpoint{181.099661bp}{146.844675bp}}{\pgfpoint{177.844513bp}{146.844675bp}}
\pgfclosepath
\color[rgb]{0.0,0.0,0.0}
\pgfusepath{stroke}
\end{pgfscope}
\begin{pgfscope}
\pgfsetlinewidth{1.0bp}
\pgfsetrectcap 
\pgfsetmiterjoin \pgfsetmiterlimit{10.0}
\pgfpathmoveto{\pgfpoint{177.844518bp}{85.544676bp}}
\pgfpathcurveto{\pgfpoint{174.589371bp}{85.544676bp}}{\pgfpoint{171.950552bp}{88.075783bp}}{\pgfpoint{171.950552bp}{91.198067bp}}
\pgfpathcurveto{\pgfpoint{171.950552bp}{94.320352bp}}{\pgfpoint{174.589371bp}{96.851465bp}}{\pgfpoint{177.844518bp}{96.851465bp}}
\pgfpathcurveto{\pgfpoint{181.099665bp}{96.851465bp}}{\pgfpoint{183.738483bp}{94.320352bp}}{\pgfpoint{183.738483bp}{91.198067bp}}
\pgfpathcurveto{\pgfpoint{183.738483bp}{88.075783bp}}{\pgfpoint{181.099665bp}{85.544676bp}}{\pgfpoint{177.844518bp}{85.544676bp}}
\pgfclosepath
\color[rgb]{0.0,0.0,0.0}\pgfseteorule\pgfusepath{fill}
\pgfpathmoveto{\pgfpoint{177.844518bp}{85.544676bp}}
\pgfpathcurveto{\pgfpoint{174.589371bp}{85.544676bp}}{\pgfpoint{171.950552bp}{88.075783bp}}{\pgfpoint{171.950552bp}{91.198067bp}}
\pgfpathcurveto{\pgfpoint{171.950552bp}{94.320352bp}}{\pgfpoint{174.589371bp}{96.851465bp}}{\pgfpoint{177.844518bp}{96.851465bp}}
\pgfpathcurveto{\pgfpoint{181.099665bp}{96.851465bp}}{\pgfpoint{183.738483bp}{94.320352bp}}{\pgfpoint{183.738483bp}{91.198067bp}}
\pgfpathcurveto{\pgfpoint{183.738483bp}{88.075783bp}}{\pgfpoint{181.099665bp}{85.544676bp}}{\pgfpoint{177.844518bp}{85.544676bp}}
\pgfclosepath
\color[rgb]{0.0,0.0,0.0}
\pgfusepath{stroke}
\end{pgfscope}
\begin{pgfscope}
\pgfsetlinewidth{1.0bp}
\pgfsetrectcap 
\pgfsetmiterjoin \pgfsetmiterlimit{10.0}
\pgfpathmoveto{\pgfpoint{138.044513bp}{146.844687bp}}
\pgfpathcurveto{\pgfpoint{134.789367bp}{146.844687bp}}{\pgfpoint{132.150548bp}{149.375794bp}}{\pgfpoint{132.150548bp}{152.498078bp}}
\pgfpathcurveto{\pgfpoint{132.150548bp}{155.620362bp}}{\pgfpoint{134.789367bp}{158.151476bp}}{\pgfpoint{138.044513bp}{158.151476bp}}
\pgfpathcurveto{\pgfpoint{141.299661bp}{158.151476bp}}{\pgfpoint{143.938478bp}{155.620362bp}}{\pgfpoint{143.938478bp}{152.498078bp}}
\pgfpathcurveto{\pgfpoint{143.938478bp}{149.375794bp}}{\pgfpoint{141.299661bp}{146.844687bp}}{\pgfpoint{138.044513bp}{146.844687bp}}
\pgfclosepath
\color[rgb]{0.0,0.0,0.0}\pgfseteorule\pgfusepath{fill}
\pgfpathmoveto{\pgfpoint{138.044513bp}{146.844687bp}}
\pgfpathcurveto{\pgfpoint{134.789367bp}{146.844687bp}}{\pgfpoint{132.150548bp}{149.375794bp}}{\pgfpoint{132.150548bp}{152.498078bp}}
\pgfpathcurveto{\pgfpoint{132.150548bp}{155.620362bp}}{\pgfpoint{134.789367bp}{158.151476bp}}{\pgfpoint{138.044513bp}{158.151476bp}}
\pgfpathcurveto{\pgfpoint{141.299661bp}{158.151476bp}}{\pgfpoint{143.938478bp}{155.620362bp}}{\pgfpoint{143.938478bp}{152.498078bp}}
\pgfpathcurveto{\pgfpoint{143.938478bp}{149.375794bp}}{\pgfpoint{141.299661bp}{146.844687bp}}{\pgfpoint{138.044513bp}{146.844687bp}}
\pgfclosepath
\color[rgb]{0.0,0.0,0.0}
\pgfusepath{stroke}
\end{pgfscope}
\begin{pgfscope}
\pgfsetlinewidth{1.0bp}
\pgfsetrectcap 
\pgfsetmiterjoin \pgfsetmiterlimit{10.0}
\pgfpathmoveto{\pgfpoint{138.044536bp}{85.544676bp}}
\pgfpathcurveto{\pgfpoint{134.78939bp}{85.544676bp}}{\pgfpoint{132.150571bp}{88.075783bp}}{\pgfpoint{132.150571bp}{91.198067bp}}
\pgfpathcurveto{\pgfpoint{132.150571bp}{94.320352bp}}{\pgfpoint{134.78939bp}{96.851465bp}}{\pgfpoint{138.044536bp}{96.851465bp}}
\pgfpathcurveto{\pgfpoint{141.299684bp}{96.851465bp}}{\pgfpoint{143.938501bp}{94.320352bp}}{\pgfpoint{143.938501bp}{91.198067bp}}
\pgfpathcurveto{\pgfpoint{143.938501bp}{88.075783bp}}{\pgfpoint{141.299684bp}{85.544676bp}}{\pgfpoint{138.044536bp}{85.544676bp}}
\pgfclosepath
\color[rgb]{0.0,0.0,0.0}\pgfseteorule\pgfusepath{fill}
\pgfpathmoveto{\pgfpoint{138.044536bp}{85.544676bp}}
\pgfpathcurveto{\pgfpoint{134.78939bp}{85.544676bp}}{\pgfpoint{132.150571bp}{88.075783bp}}{\pgfpoint{132.150571bp}{91.198067bp}}
\pgfpathcurveto{\pgfpoint{132.150571bp}{94.320352bp}}{\pgfpoint{134.78939bp}{96.851465bp}}{\pgfpoint{138.044536bp}{96.851465bp}}
\pgfpathcurveto{\pgfpoint{141.299684bp}{96.851465bp}}{\pgfpoint{143.938501bp}{94.320352bp}}{\pgfpoint{143.938501bp}{91.198067bp}}
\pgfpathcurveto{\pgfpoint{143.938501bp}{88.075783bp}}{\pgfpoint{141.299684bp}{85.544676bp}}{\pgfpoint{138.044536bp}{85.544676bp}}
\pgfclosepath
\color[rgb]{0.0,0.0,0.0}
\pgfusepath{stroke}
\end{pgfscope}
\begin{pgfscope}
\pgfsetlinewidth{1.0bp}
\pgfsetrectcap 
\pgfsetmiterjoin \pgfsetmiterlimit{10.0}
\pgfpathmoveto{\pgfpoint{98.044513bp}{146.844687bp}}
\pgfpathcurveto{\pgfpoint{94.789367bp}{146.844687bp}}{\pgfpoint{92.150548bp}{149.375794bp}}{\pgfpoint{92.150548bp}{152.498078bp}}
\pgfpathcurveto{\pgfpoint{92.150548bp}{155.620362bp}}{\pgfpoint{94.789367bp}{158.151476bp}}{\pgfpoint{98.044513bp}{158.151476bp}}
\pgfpathcurveto{\pgfpoint{101.299661bp}{158.151476bp}}{\pgfpoint{103.938478bp}{155.620362bp}}{\pgfpoint{103.938478bp}{152.498078bp}}
\pgfpathcurveto{\pgfpoint{103.938478bp}{149.375794bp}}{\pgfpoint{101.299661bp}{146.844687bp}}{\pgfpoint{98.044513bp}{146.844687bp}}
\pgfclosepath
\color[rgb]{0.0,0.0,0.0}\pgfseteorule\pgfusepath{fill}
\pgfpathmoveto{\pgfpoint{98.044513bp}{146.844687bp}}
\pgfpathcurveto{\pgfpoint{94.789367bp}{146.844687bp}}{\pgfpoint{92.150548bp}{149.375794bp}}{\pgfpoint{92.150548bp}{152.498078bp}}
\pgfpathcurveto{\pgfpoint{92.150548bp}{155.620362bp}}{\pgfpoint{94.789367bp}{158.151476bp}}{\pgfpoint{98.044513bp}{158.151476bp}}
\pgfpathcurveto{\pgfpoint{101.299661bp}{158.151476bp}}{\pgfpoint{103.938478bp}{155.620362bp}}{\pgfpoint{103.938478bp}{152.498078bp}}
\pgfpathcurveto{\pgfpoint{103.938478bp}{149.375794bp}}{\pgfpoint{101.299661bp}{146.844687bp}}{\pgfpoint{98.044513bp}{146.844687bp}}
\pgfclosepath
\color[rgb]{0.0,0.0,0.0}
\pgfusepath{stroke}
\end{pgfscope}
\begin{pgfscope}
\pgfsetlinewidth{1.0bp}
\pgfsetrectcap 
\pgfsetmiterjoin \pgfsetmiterlimit{10.0}
\pgfpathmoveto{\pgfpoint{98.044518bp}{85.544688bp}}
\pgfpathcurveto{\pgfpoint{94.789371bp}{85.544688bp}}{\pgfpoint{92.150552bp}{88.075795bp}}{\pgfpoint{92.150552bp}{91.19808bp}}
\pgfpathcurveto{\pgfpoint{92.150552bp}{94.320364bp}}{\pgfpoint{94.789371bp}{96.851477bp}}{\pgfpoint{98.044518bp}{96.851477bp}}
\pgfpathcurveto{\pgfpoint{101.299665bp}{96.851477bp}}{\pgfpoint{103.938483bp}{94.320364bp}}{\pgfpoint{103.938483bp}{91.19808bp}}
\pgfpathcurveto{\pgfpoint{103.938483bp}{88.075795bp}}{\pgfpoint{101.299665bp}{85.544688bp}}{\pgfpoint{98.044518bp}{85.544688bp}}
\pgfclosepath
\color[rgb]{0.0,0.0,0.0}\pgfseteorule\pgfusepath{fill}
\pgfpathmoveto{\pgfpoint{98.044518bp}{85.544688bp}}
\pgfpathcurveto{\pgfpoint{94.789371bp}{85.544688bp}}{\pgfpoint{92.150552bp}{88.075795bp}}{\pgfpoint{92.150552bp}{91.19808bp}}
\pgfpathcurveto{\pgfpoint{92.150552bp}{94.320364bp}}{\pgfpoint{94.789371bp}{96.851477bp}}{\pgfpoint{98.044518bp}{96.851477bp}}
\pgfpathcurveto{\pgfpoint{101.299665bp}{96.851477bp}}{\pgfpoint{103.938483bp}{94.320364bp}}{\pgfpoint{103.938483bp}{91.19808bp}}
\pgfpathcurveto{\pgfpoint{103.938483bp}{88.075795bp}}{\pgfpoint{101.299665bp}{85.544688bp}}{\pgfpoint{98.044518bp}{85.544688bp}}
\pgfclosepath
\color[rgb]{0.0,0.0,0.0}
\pgfusepath{stroke}
\end{pgfscope}
\begin{pgfscope}
\pgfsetlinewidth{1.0bp}
\pgfsetrectcap 
\pgfsetmiterjoin \pgfsetmiterlimit{10.0}
\pgfpathmoveto{\pgfpoint{57.644513bp}{146.844693bp}}
\pgfpathcurveto{\pgfpoint{54.389367bp}{146.844693bp}}{\pgfpoint{51.750548bp}{149.3758bp}}{\pgfpoint{51.750548bp}{152.498084bp}}
\pgfpathcurveto{\pgfpoint{51.750548bp}{155.620369bp}}{\pgfpoint{54.389367bp}{158.151482bp}}{\pgfpoint{57.644513bp}{158.151482bp}}
\pgfpathcurveto{\pgfpoint{60.899661bp}{158.151482bp}}{\pgfpoint{63.538478bp}{155.620369bp}}{\pgfpoint{63.538478bp}{152.498084bp}}
\pgfpathcurveto{\pgfpoint{63.538478bp}{149.3758bp}}{\pgfpoint{60.899661bp}{146.844693bp}}{\pgfpoint{57.644513bp}{146.844693bp}}
\pgfclosepath
\color[rgb]{0.0,0.0,0.0}\pgfseteorule\pgfusepath{fill}
\pgfpathmoveto{\pgfpoint{57.644513bp}{146.844693bp}}
\pgfpathcurveto{\pgfpoint{54.389367bp}{146.844693bp}}{\pgfpoint{51.750548bp}{149.3758bp}}{\pgfpoint{51.750548bp}{152.498084bp}}
\pgfpathcurveto{\pgfpoint{51.750548bp}{155.620369bp}}{\pgfpoint{54.389367bp}{158.151482bp}}{\pgfpoint{57.644513bp}{158.151482bp}}
\pgfpathcurveto{\pgfpoint{60.899661bp}{158.151482bp}}{\pgfpoint{63.538478bp}{155.620369bp}}{\pgfpoint{63.538478bp}{152.498084bp}}
\pgfpathcurveto{\pgfpoint{63.538478bp}{149.3758bp}}{\pgfpoint{60.899661bp}{146.844693bp}}{\pgfpoint{57.644513bp}{146.844693bp}}
\pgfclosepath
\color[rgb]{0.0,0.0,0.0}
\pgfusepath{stroke}
\end{pgfscope}
\begin{pgfscope}
\pgfsetlinewidth{1.0bp}
\pgfsetrectcap 
\pgfsetmiterjoin \pgfsetmiterlimit{10.0}
\pgfpathmoveto{\pgfpoint{57.644518bp}{85.544664bp}}
\pgfpathcurveto{\pgfpoint{54.389371bp}{85.544664bp}}{\pgfpoint{51.750552bp}{88.075771bp}}{\pgfpoint{51.750552bp}{91.198055bp}}
\pgfpathcurveto{\pgfpoint{51.750552bp}{94.32034bp}}{\pgfpoint{54.389371bp}{96.851453bp}}{\pgfpoint{57.644518bp}{96.851453bp}}
\pgfpathcurveto{\pgfpoint{60.899665bp}{96.851453bp}}{\pgfpoint{63.538483bp}{94.32034bp}}{\pgfpoint{63.538483bp}{91.198055bp}}
\pgfpathcurveto{\pgfpoint{63.538483bp}{88.075771bp}}{\pgfpoint{60.899665bp}{85.544664bp}}{\pgfpoint{57.644518bp}{85.544664bp}}
\pgfclosepath
\color[rgb]{0.0,0.0,0.0}\pgfseteorule\pgfusepath{fill}
\pgfpathmoveto{\pgfpoint{57.644518bp}{85.544664bp}}
\pgfpathcurveto{\pgfpoint{54.389371bp}{85.544664bp}}{\pgfpoint{51.750552bp}{88.075771bp}}{\pgfpoint{51.750552bp}{91.198055bp}}
\pgfpathcurveto{\pgfpoint{51.750552bp}{94.32034bp}}{\pgfpoint{54.389371bp}{96.851453bp}}{\pgfpoint{57.644518bp}{96.851453bp}}
\pgfpathcurveto{\pgfpoint{60.899665bp}{96.851453bp}}{\pgfpoint{63.538483bp}{94.32034bp}}{\pgfpoint{63.538483bp}{91.198055bp}}
\pgfpathcurveto{\pgfpoint{63.538483bp}{88.075771bp}}{\pgfpoint{60.899665bp}{85.544664bp}}{\pgfpoint{57.644518bp}{85.544664bp}}
\pgfclosepath
\color[rgb]{0.0,0.0,0.0}
\pgfusepath{stroke}
\end{pgfscope}
\begin{pgfscope}
\pgfsetlinewidth{1.0bp}
\pgfsetrectcap 
\pgfsetmiterjoin \pgfsetmiterlimit{10.0}
\pgfpathmoveto{\pgfpoint{18.044513bp}{146.844687bp}}
\pgfpathcurveto{\pgfpoint{14.789367bp}{146.844687bp}}{\pgfpoint{12.150548bp}{149.375794bp}}{\pgfpoint{12.150548bp}{152.498078bp}}
\pgfpathcurveto{\pgfpoint{12.150548bp}{155.620362bp}}{\pgfpoint{14.789367bp}{158.151476bp}}{\pgfpoint{18.044513bp}{158.151476bp}}
\pgfpathcurveto{\pgfpoint{21.299661bp}{158.151476bp}}{\pgfpoint{23.938478bp}{155.620362bp}}{\pgfpoint{23.938478bp}{152.498078bp}}
\pgfpathcurveto{\pgfpoint{23.938478bp}{149.375794bp}}{\pgfpoint{21.299661bp}{146.844687bp}}{\pgfpoint{18.044513bp}{146.844687bp}}
\pgfclosepath
\color[rgb]{0.0,0.0,0.0}\pgfseteorule\pgfusepath{fill}
\pgfpathmoveto{\pgfpoint{18.044513bp}{146.844687bp}}
\pgfpathcurveto{\pgfpoint{14.789367bp}{146.844687bp}}{\pgfpoint{12.150548bp}{149.375794bp}}{\pgfpoint{12.150548bp}{152.498078bp}}
\pgfpathcurveto{\pgfpoint{12.150548bp}{155.620362bp}}{\pgfpoint{14.789367bp}{158.151476bp}}{\pgfpoint{18.044513bp}{158.151476bp}}
\pgfpathcurveto{\pgfpoint{21.299661bp}{158.151476bp}}{\pgfpoint{23.938478bp}{155.620362bp}}{\pgfpoint{23.938478bp}{152.498078bp}}
\pgfpathcurveto{\pgfpoint{23.938478bp}{149.375794bp}}{\pgfpoint{21.299661bp}{146.844687bp}}{\pgfpoint{18.044513bp}{146.844687bp}}
\pgfclosepath
\color[rgb]{0.0,0.0,0.0}
\pgfusepath{stroke}
\end{pgfscope}
\begin{pgfscope}
\pgfsetlinewidth{1.0bp}
\pgfsetrectcap 
\pgfsetmiterjoin \pgfsetmiterlimit{10.0}
\pgfpathmoveto{\pgfpoint{18.044518bp}{85.544688bp}}
\pgfpathcurveto{\pgfpoint{14.789371bp}{85.544688bp}}{\pgfpoint{12.150552bp}{88.075795bp}}{\pgfpoint{12.150552bp}{91.19808bp}}
\pgfpathcurveto{\pgfpoint{12.150552bp}{94.320364bp}}{\pgfpoint{14.789371bp}{96.851477bp}}{\pgfpoint{18.044518bp}{96.851477bp}}
\pgfpathcurveto{\pgfpoint{21.299665bp}{96.851477bp}}{\pgfpoint{23.938483bp}{94.320364bp}}{\pgfpoint{23.938483bp}{91.19808bp}}
\pgfpathcurveto{\pgfpoint{23.938483bp}{88.075795bp}}{\pgfpoint{21.299665bp}{85.544688bp}}{\pgfpoint{18.044518bp}{85.544688bp}}
\pgfclosepath
\color[rgb]{0.0,0.0,0.0}\pgfseteorule\pgfusepath{fill}
\pgfpathmoveto{\pgfpoint{18.044518bp}{85.544688bp}}
\pgfpathcurveto{\pgfpoint{14.789371bp}{85.544688bp}}{\pgfpoint{12.150552bp}{88.075795bp}}{\pgfpoint{12.150552bp}{91.19808bp}}
\pgfpathcurveto{\pgfpoint{12.150552bp}{94.320364bp}}{\pgfpoint{14.789371bp}{96.851477bp}}{\pgfpoint{18.044518bp}{96.851477bp}}
\pgfpathcurveto{\pgfpoint{21.299665bp}{96.851477bp}}{\pgfpoint{23.938483bp}{94.320364bp}}{\pgfpoint{23.938483bp}{91.19808bp}}
\pgfpathcurveto{\pgfpoint{23.938483bp}{88.075795bp}}{\pgfpoint{21.299665bp}{85.544688bp}}{\pgfpoint{18.044518bp}{85.544688bp}}
\pgfclosepath
\color[rgb]{0.0,0.0,0.0}
\pgfusepath{stroke}
\end{pgfscope}
\begin{pgfscope}
\pgftransformcm{1.0}{0.0}{0.0}{1.0}{\pgfpoint{278.484375bp}{144.0bp}}
\pgftext[left,base]{\sffamily\mdseries\upshape\LARGE
\color[rgb]{0.0,0.0,0.0}$9$}
\end{pgfscope}
\begin{pgfscope}
\pgftransformcm{1.0}{0.0}{0.0}{1.0}{\pgfpoint{277.884375bp}{87.4bp}}
\pgftext[left,base]{\sffamily\mdseries\upshape\LARGE
\color[rgb]{0.0,0.0,0.0}$9'$}
\end{pgfscope}
\begin{pgfscope}
\pgfsetlinewidth{1.0bp}
\pgfsetrectcap 
\pgfsetmiterjoin \pgfsetmiterlimit{10.0}
\pgfpathmoveto{\pgfpoint{154.944525bp}{207.444669bp}}
\pgfpathcurveto{\pgfpoint{151.689379bp}{207.444669bp}}{\pgfpoint{149.05056bp}{209.975776bp}}{\pgfpoint{149.05056bp}{213.09806bp}}
\pgfpathcurveto{\pgfpoint{149.05056bp}{216.220344bp}}{\pgfpoint{151.689379bp}{218.751457bp}}{\pgfpoint{154.944525bp}{218.751457bp}}
\pgfpathcurveto{\pgfpoint{158.199673bp}{218.751457bp}}{\pgfpoint{160.838491bp}{216.220344bp}}{\pgfpoint{160.838491bp}{213.09806bp}}
\pgfpathcurveto{\pgfpoint{160.838491bp}{209.975776bp}}{\pgfpoint{158.199673bp}{207.444669bp}}{\pgfpoint{154.944525bp}{207.444669bp}}
\pgfclosepath
\color[rgb]{0.0,0.0,0.0}\pgfseteorule\pgfusepath{fill}
\pgfpathmoveto{\pgfpoint{154.944525bp}{207.444669bp}}
\pgfpathcurveto{\pgfpoint{151.689379bp}{207.444669bp}}{\pgfpoint{149.05056bp}{209.975776bp}}{\pgfpoint{149.05056bp}{213.09806bp}}
\pgfpathcurveto{\pgfpoint{149.05056bp}{216.220344bp}}{\pgfpoint{151.689379bp}{218.751457bp}}{\pgfpoint{154.944525bp}{218.751457bp}}
\pgfpathcurveto{\pgfpoint{158.199673bp}{218.751457bp}}{\pgfpoint{160.838491bp}{216.220344bp}}{\pgfpoint{160.838491bp}{213.09806bp}}
\pgfpathcurveto{\pgfpoint{160.838491bp}{209.975776bp}}{\pgfpoint{158.199673bp}{207.444669bp}}{\pgfpoint{154.944525bp}{207.444669bp}}
\pgfclosepath
\color[rgb]{0.0,0.0,0.0}
\pgfusepath{stroke}
\end{pgfscope}
\begin{pgfscope}
\pgfsetlinewidth{1.0bp}
\pgfsetrectcap 
\pgfsetmiterjoin \pgfsetmiterlimit{10.0}
\pgfpathmoveto{\pgfpoint{154.944524bp}{23.744688bp}}
\pgfpathcurveto{\pgfpoint{151.689378bp}{23.744688bp}}{\pgfpoint{149.050558bp}{26.275795bp}}{\pgfpoint{149.050558bp}{29.39808bp}}
\pgfpathcurveto{\pgfpoint{149.050558bp}{32.520364bp}}{\pgfpoint{151.689378bp}{35.051477bp}}{\pgfpoint{154.944524bp}{35.051477bp}}
\pgfpathcurveto{\pgfpoint{158.199671bp}{35.051477bp}}{\pgfpoint{160.838489bp}{32.520364bp}}{\pgfpoint{160.838489bp}{29.39808bp}}
\pgfpathcurveto{\pgfpoint{160.838489bp}{26.275795bp}}{\pgfpoint{158.199671bp}{23.744688bp}}{\pgfpoint{154.944524bp}{23.744688bp}}
\pgfclosepath
\color[rgb]{0.0,0.0,0.0}\pgfseteorule\pgfusepath{fill}
\pgfpathmoveto{\pgfpoint{154.944524bp}{23.744688bp}}
\pgfpathcurveto{\pgfpoint{151.689378bp}{23.744688bp}}{\pgfpoint{149.050558bp}{26.275795bp}}{\pgfpoint{149.050558bp}{29.39808bp}}
\pgfpathcurveto{\pgfpoint{149.050558bp}{32.520364bp}}{\pgfpoint{151.689378bp}{35.051477bp}}{\pgfpoint{154.944524bp}{35.051477bp}}
\pgfpathcurveto{\pgfpoint{158.199671bp}{35.051477bp}}{\pgfpoint{160.838489bp}{32.520364bp}}{\pgfpoint{160.838489bp}{29.39808bp}}
\pgfpathcurveto{\pgfpoint{160.838489bp}{26.275795bp}}{\pgfpoint{158.199671bp}{23.744688bp}}{\pgfpoint{154.944524bp}{23.744688bp}}
\pgfclosepath
\color[rgb]{0.0,0.0,0.0}
\pgfusepath{stroke}
\end{pgfscope}
\begin{pgfscope}
\pgftransformcm{1.0}{0.0}{0.0}{1.0}{\pgfpoint{164.878125bp}{18.8bp}}
\pgftext[left,base]{\sffamily\mdseries\upshape\LARGE
\color[rgb]{0.0,0.0,0.0}$1$}
\end{pgfscope}
\begin{pgfscope}
\pgftransformcm{1.0}{0.0}{0.0}{1.0}{\pgfpoint{164.878125bp}{213.6bp}}
\pgftext[left,base]{\sffamily\mdseries\upshape\LARGE
\color[rgb]{0.0,0.0,0.0}$1'$}
\end{pgfscope}
\begin{pgfscope}
\pgftransformcm{1.0}{0.0}{0.0}{1.0}{\pgfpoint{-0.525bp}{148.2bp}}
\pgftext[left,base]{\sffamily\mdseries\upshape\LARGE
\color[rgb]{0.0,0.0,0.0}$2$}
\end{pgfscope}
\begin{pgfscope}
\pgftransformcm{1.0}{0.0}{0.0}{1.0}{\pgfpoint{-0.525bp}{87.0bp}}
\pgftext[left,base]{\sffamily\mdseries\upshape\LARGE
\color[rgb]{0.0,0.0,0.0}$2'$}
\end{pgfscope}
\begin{pgfscope}
\pgftransformcm{1.0}{0.0}{0.0}{1.0}{\pgfpoint{37.875bp}{146.8bp}}
\pgftext[left,base]{\sffamily\mdseries\upshape\LARGE
\color[rgb]{0.0,0.0,0.0}$3$}
\end{pgfscope}
\begin{pgfscope}
\pgftransformcm{1.0}{0.0}{0.0}{1.0}{\pgfpoint{37.875bp}{85.2bp}}
\pgftext[left,base]{\sffamily\mdseries\upshape\LARGE
\color[rgb]{0.0,0.0,0.0}$3'$}
\end{pgfscope}
\begin{pgfscope}
\pgftransformcm{1.0}{0.0}{0.0}{1.0}{\pgfpoint{79.03125bp}{147.4bp}}
\pgftext[left,base]{\sffamily\mdseries\upshape\LARGE
\color[rgb]{0.0,0.0,0.0}$4$}
\end{pgfscope}
\begin{pgfscope}
\pgftransformcm{1.0}{0.0}{0.0}{1.0}{\pgfpoint{79.03125bp}{85.6bp}}
\pgftext[left,base]{\sffamily\mdseries\upshape\LARGE
\color[rgb]{0.0,0.0,0.0}$4'$}
\end{pgfscope}
\begin{pgfscope}
\pgftransformcm{1.0}{0.0}{0.0}{1.0}{\pgfpoint{118.284375bp}{147.2bp}}
\pgftext[left,base]{\sffamily\mdseries\upshape\LARGE
\color[rgb]{0.0,0.0,0.0}$5$}
\end{pgfscope}
\begin{pgfscope}
\pgftransformcm{1.0}{0.0}{0.0}{1.0}{\pgfpoint{118.284375bp}{85.2bp}}
\pgftext[left,base]{\sffamily\mdseries\upshape\LARGE
\color[rgb]{0.0,0.0,0.0}$5'$}
\end{pgfscope}
\begin{pgfscope}
\pgftransformcm{1.0}{0.0}{0.0}{1.0}{\pgfpoint{157.93125bp}{148.0bp}}
\pgftext[left,base]{\sffamily\mdseries\upshape\LARGE
\color[rgb]{0.0,0.0,0.0}$6$}
\end{pgfscope}
\begin{pgfscope}
\pgftransformcm{1.0}{0.0}{0.0}{1.0}{\pgfpoint{157.93125bp}{85.6bp}}
\pgftext[left,base]{\sffamily\mdseries\upshape\LARGE
\color[rgb]{0.0,0.0,0.0}$6'$}
\end{pgfscope}
\begin{pgfscope}
\pgftransformcm{1.0}{0.0}{0.0}{1.0}{\pgfpoint{198.075bp}{146.6bp}}
\pgftext[left,base]{\sffamily\mdseries\upshape\LARGE
\color[rgb]{0.0,0.0,0.0}$7$}
\end{pgfscope}
\begin{pgfscope}
\pgftransformcm{1.0}{0.0}{0.0}{1.0}{\pgfpoint{198.075bp}{85.6bp}}
\pgftext[left,base]{\sffamily\mdseries\upshape\LARGE
\color[rgb]{0.0,0.0,0.0}$7'$}
\end{pgfscope}
\begin{pgfscope}
\pgftransformcm{1.0}{0.0}{0.0}{1.0}{\pgfpoint{237.675bp}{146.6bp}}
\pgftext[left,base]{\sffamily\mdseries\upshape\LARGE
\color[rgb]{0.0,0.0,0.0}$8$}
\end{pgfscope}
\begin{pgfscope}
\pgftransformcm{1.0}{0.0}{0.0}{1.0}{\pgfpoint{237.675bp}{86.0bp}}
\pgftext[left,base]{\sffamily\mdseries\upshape\LARGE
\color[rgb]{0.0,0.0,0.0}$8'$}
\end{pgfscope}
\begin{pgfscope}
\pgfsetlinewidth{1.0bp}
\pgfsetrectcap 
\pgfsetmiterjoin \pgfsetmiterlimit{10.0}
\pgfpathmoveto{\pgfpoint{155.475bp}{212.8bp}}
\pgfpathlineto{\pgfpoint{18.075bp}{152.8bp}}
\color[rgb]{0.0,0.0,0.0}
\pgfusepath{stroke}
\end{pgfscope}
\begin{pgfscope}
\pgfsetlinewidth{1.0bp}
\pgfsetrectcap 
\pgfsetmiterjoin \pgfsetmiterlimit{10.0}
\pgfpathmoveto{\pgfpoint{154.875bp}{212.8bp}}
\pgfpathlineto{\pgfpoint{57.675bp}{152.2bp}}
\color[rgb]{0.0,0.0,0.0}
\pgfusepath{stroke}
\end{pgfscope}
\begin{pgfscope}
\pgfsetlinewidth{1.0bp}
\pgfsetrectcap 
\pgfsetmiterjoin \pgfsetmiterlimit{10.0}
\pgfpathmoveto{\pgfpoint{154.875bp}{212.2bp}}
\pgfpathlineto{\pgfpoint{97.875bp}{152.8bp}}
\color[rgb]{0.0,0.0,0.0}
\pgfusepath{stroke}
\end{pgfscope}
\begin{pgfscope}
\pgfsetlinewidth{1.0bp}
\pgfsetrectcap 
\pgfsetmiterjoin \pgfsetmiterlimit{10.0}
\pgfpathmoveto{\pgfpoint{154.875bp}{212.2bp}}
\pgfpathlineto{\pgfpoint{138.075bp}{152.8bp}}
\color[rgb]{0.0,0.0,0.0}
\pgfusepath{stroke}
\end{pgfscope}
\begin{pgfscope}
\pgfsetlinewidth{1.0bp}
\pgfsetrectcap 
\pgfsetmiterjoin \pgfsetmiterlimit{10.0}
\pgfpathmoveto{\pgfpoint{154.875bp}{212.8bp}}
\pgfpathlineto{\pgfpoint{178.275bp}{152.2bp}}
\color[rgb]{0.0,0.0,0.0}
\pgfusepath{stroke}
\end{pgfscope}
\begin{pgfscope}
\pgfsetlinewidth{1.0bp}
\pgfsetrectcap 
\pgfsetmiterjoin \pgfsetmiterlimit{10.0}
\pgfpathmoveto{\pgfpoint{154.275bp}{212.2bp}}
\pgfpathlineto{\pgfpoint{217.275bp}{152.8bp}}
\color[rgb]{0.0,0.0,0.0}
\pgfusepath{stroke}
\end{pgfscope}
\begin{pgfscope}
\pgfsetlinewidth{1.0bp}
\pgfsetrectcap 
\pgfsetmiterjoin \pgfsetmiterlimit{10.0}
\pgfpathmoveto{\pgfpoint{154.875bp}{212.2bp}}
\pgfpathlineto{\pgfpoint{258.075bp}{152.2bp}}
\color[rgb]{0.0,0.0,0.0}
\pgfusepath{stroke}
\end{pgfscope}
\begin{pgfscope}
\pgfsetlinewidth{1.0bp}
\pgfsetrectcap 
\pgfsetmiterjoin \pgfsetmiterlimit{10.0}
\pgfpathmoveto{\pgfpoint{155.475bp}{212.2bp}}
\pgfpathlineto{\pgfpoint{297.675bp}{152.2bp}}
\color[rgb]{0.0,0.0,0.0}
\pgfusepath{stroke}
\end{pgfscope}
\begin{pgfscope}
\pgfsetlinewidth{1.0bp}
\pgfsetrectcap 
\pgfsetmiterjoin \pgfsetmiterlimit{10.0}
\pgfpathmoveto{\pgfpoint{155.475bp}{29.2bp}}
\pgfpathlineto{\pgfpoint{18.075bp}{91.6bp}}
\color[rgb]{0.0,0.0,0.0}
\pgfusepath{stroke}
\end{pgfscope}
\begin{pgfscope}
\pgfsetlinewidth{1.0bp}
\pgfsetrectcap 
\pgfsetmiterjoin \pgfsetmiterlimit{10.0}
\pgfpathmoveto{\pgfpoint{155.475bp}{29.2bp}}
\pgfpathlineto{\pgfpoint{57.675bp}{91.6bp}}
\color[rgb]{0.0,0.0,0.0}
\pgfusepath{stroke}
\end{pgfscope}
\begin{pgfscope}
\pgfsetlinewidth{1.0bp}
\pgfsetrectcap 
\pgfsetmiterjoin \pgfsetmiterlimit{10.0}
\pgfpathmoveto{\pgfpoint{154.875bp}{30.4bp}}
\pgfpathlineto{\pgfpoint{97.875bp}{91.0bp}}
\color[rgb]{0.0,0.0,0.0}
\pgfusepath{stroke}
\end{pgfscope}
\begin{pgfscope}
\pgfsetlinewidth{1.0bp}
\pgfsetrectcap 
\pgfsetmiterjoin \pgfsetmiterlimit{10.0}
\pgfpathmoveto{\pgfpoint{155.475bp}{30.4bp}}
\pgfpathlineto{\pgfpoint{137.475bp}{91.0bp}}
\color[rgb]{0.0,0.0,0.0}
\pgfusepath{stroke}
\end{pgfscope}
\begin{pgfscope}
\pgfsetlinewidth{1.0bp}
\pgfsetrectcap 
\pgfsetmiterjoin \pgfsetmiterlimit{10.0}
\pgfpathmoveto{\pgfpoint{155.475bp}{29.8bp}}
\pgfpathlineto{\pgfpoint{177.675bp}{91.0bp}}
\color[rgb]{0.0,0.0,0.0}
\pgfusepath{stroke}
\end{pgfscope}
\begin{pgfscope}
\pgfsetlinewidth{1.0bp}
\pgfsetrectcap 
\pgfsetmiterjoin \pgfsetmiterlimit{10.0}
\pgfpathmoveto{\pgfpoint{155.475bp}{29.2bp}}
\pgfpathlineto{\pgfpoint{217.275bp}{91.0bp}}
\color[rgb]{0.0,0.0,0.0}
\pgfusepath{stroke}
\end{pgfscope}
\begin{pgfscope}
\pgfsetlinewidth{1.0bp}
\pgfsetrectcap 
\pgfsetmiterjoin \pgfsetmiterlimit{10.0}
\pgfpathmoveto{\pgfpoint{155.475bp}{29.8bp}}
\pgfpathlineto{\pgfpoint{258.075bp}{91.0bp}}
\color[rgb]{0.0,0.0,0.0}
\pgfusepath{stroke}
\end{pgfscope}
\begin{pgfscope}
\pgfsetlinewidth{1.0bp}
\pgfsetrectcap 
\pgfsetmiterjoin \pgfsetmiterlimit{10.0}
\pgfpathmoveto{\pgfpoint{154.875bp}{29.2bp}}
\pgfpathlineto{\pgfpoint{297.675bp}{91.0bp}}
\color[rgb]{0.0,0.0,0.0}
\pgfusepath{stroke}
\end{pgfscope}
\begin{pgfscope}
\pgfsetlinewidth{1.0bp}
\pgfsetrectcap 
\pgfsetmiterjoin \pgfsetmiterlimit{10.0}
\pgfpathmoveto{\pgfpoint{155.475bp}{212.8bp}}
\pgfpathlineto{\pgfpoint{155.475bp}{30.4bp}}
\pgfpathlineto{\pgfpoint{155.475bp}{30.4bp}}
\color[rgb]{0.0,0.0,0.0}
\pgfusepath{stroke}
\end{pgfscope}
\begin{pgfscope}
\pgfsetlinewidth{1.0bp}
\pgfsetrectcap 
\pgfsetmiterjoin \pgfsetmiterlimit{10.0}
\pgfpathmoveto{\pgfpoint{57.675bp}{91.0bp}}
\pgfpathlineto{\pgfpoint{18.075bp}{153.4bp}}
\color[rgb]{0.0,0.0,0.0}
\pgfusepath{stroke}
\end{pgfscope}
\begin{pgfscope}
\pgfsetlinewidth{1.0bp}
\pgfsetrectcap 
\pgfsetmiterjoin \pgfsetmiterlimit{10.0}
\pgfpathmoveto{\pgfpoint{57.675bp}{152.8bp}}
\pgfpathlineto{\pgfpoint{18.075bp}{91.6bp}}
\color[rgb]{0.0,0.0,0.0}
\pgfusepath{stroke}
\end{pgfscope}
\begin{pgfscope}
\pgfsetlinewidth{1.0bp}
\pgfsetrectcap 
\pgfsetmiterjoin \pgfsetmiterlimit{10.0}
\pgfpathmoveto{\pgfpoint{17.475bp}{152.8bp}}
\pgfpathlineto{\pgfpoint{17.475bp}{91.0bp}}
\color[rgb]{0.0,0.0,0.0}
\pgfusepath{stroke}
\end{pgfscope}
\begin{pgfscope}
\pgfsetlinewidth{1.0bp}
\pgfsetrectcap 
\pgfsetmiterjoin \pgfsetmiterlimit{10.0}
\pgfpathmoveto{\pgfpoint{57.675bp}{152.8bp}}
\pgfpathlineto{\pgfpoint{57.675bp}{91.6bp}}
\color[rgb]{0.0,0.0,0.0}
\pgfusepath{stroke}
\end{pgfscope}
\begin{pgfscope}
\pgfsetlinewidth{1.0bp}
\pgfsetrectcap 
\pgfsetmiterjoin \pgfsetmiterlimit{10.0}
\pgfpathmoveto{\pgfpoint{258.075bp}{152.8bp}}
\pgfpathlineto{\pgfpoint{258.075bp}{90.4bp}}
\pgfpathlineto{\pgfpoint{258.075bp}{90.4bp}}
\color[rgb]{0.0,0.0,0.0}
\pgfusepath{stroke}
\end{pgfscope}
\begin{pgfscope}
\pgfsetlinewidth{1.0bp}
\pgfsetrectcap 
\pgfsetmiterjoin \pgfsetmiterlimit{10.0}
\pgfpathmoveto{\pgfpoint{297.675bp}{152.2bp}}
\pgfpathlineto{\pgfpoint{297.675bp}{91.6bp}}
\color[rgb]{0.0,0.0,0.0}
\pgfusepath{stroke}
\end{pgfscope}
\begin{pgfscope}
\pgfsetlinewidth{1.0bp}
\pgfsetrectcap 
\pgfsetmiterjoin \pgfsetmiterlimit{10.0}
\pgfpathmoveto{\pgfpoint{298.275bp}{90.4bp}}
\pgfpathlineto{\pgfpoint{258.075bp}{152.2bp}}
\color[rgb]{0.0,0.0,0.0}
\pgfusepath{stroke}
\end{pgfscope}
\begin{pgfscope}
\pgfsetlinewidth{1.0bp}
\pgfsetrectcap 
\pgfsetmiterjoin \pgfsetmiterlimit{10.0}
\pgfpathmoveto{\pgfpoint{297.675bp}{152.8bp}}
\pgfpathlineto{\pgfpoint{258.075bp}{92.2bp}}
\color[rgb]{0.0,0.0,0.0}
\pgfusepath{stroke}
\end{pgfscope}
\begin{pgfscope}
\pgfsetlinewidth{1.0bp}
\pgfsetrectcap 
\pgfsetmiterjoin \pgfsetmiterlimit{10.0}
\pgfpathmoveto{\pgfpoint{57.675bp}{153.4bp}}
\pgfpathlineto{\pgfpoint{97.875bp}{91.0bp}}
\color[rgb]{0.0,0.0,0.0}
\pgfusepath{stroke}
\end{pgfscope}
\begin{pgfscope}
\pgfsetlinewidth{1.0bp}
\pgfsetrectcap 
\pgfsetmiterjoin \pgfsetmiterlimit{10.0}
\pgfpathmoveto{\pgfpoint{258.075bp}{91.0bp}}
\pgfpathlineto{\pgfpoint{217.875bp}{152.8bp}}
\color[rgb]{0.0,0.0,0.0}
\pgfusepath{stroke}
\end{pgfscope}
\begin{pgfscope}
\pgftransformcm{1.0}{0.0}{0.0}{1.0}{\pgfpoint{145.675bp}{0.0bp}}
\pgftext[left,base]{\sffamily\mdseries\upshape\huge
\color[rgb]{0.0,0.0,0.0}$H$}
\end{pgfscope}
\begin{pgfscope}
\pgfsetlinewidth{1.0bp}
\pgfsetrectcap 
\pgfsetmiterjoin \pgfsetmiterlimit{10.0}
\pgfpathmoveto{\pgfpoint{137.875bp}{152.4bp}}
\pgfpathlineto{\pgfpoint{137.875bp}{90.6bp}}
\color[rgb]{0.0,0.0,0.0}
\pgfusepath{stroke}
\end{pgfscope}
\begin{pgfscope}
\pgfsetlinewidth{1.0bp}
\pgfsetrectcap 
\pgfsetmiterjoin \pgfsetmiterlimit{10.0}
\pgfpathmoveto{\pgfpoint{178.075bp}{153.0bp}}
\pgfpathlineto{\pgfpoint{178.075bp}{91.2bp}}
\color[rgb]{0.0,0.0,0.0}
\pgfusepath{stroke}
\end{pgfscope}
\end{pgfpicture}

%% file: Hexample5.tex
% C:\Users\agalanis\Dropbox\H-colorings\Hexample5.jdr
% Created by Jpgfdraw version 0.5.6b
% 31-Jan-2015 07:19:39
\begin{pgfpicture}{490bp}{0bp}{878.400013bp}{370.399994bp}
\begin{pgfscope}
\pgftransformcm{1.0}{0.0}{0.0}{1.0}{\pgfpoint{612.300001bp}{188.499988bp}}
\pgftext[left,base]{\sffamily\mdseries\upshape\huge
\color[rgb]{0.0,0.0,0.0}$w_A$}
\end{pgfscope}
\begin{pgfscope}
\pgftransformcm{1.0}{0.0}{0.0}{1.0}{\pgfpoint{730.700001bp}{187.699988bp}}
\pgftext[left,base]{\sffamily\mdseries\upshape\huge
\color[rgb]{0.0,0.0,0.0}$w_B$}
\end{pgfscope}
\begin{pgfscope}
\pgfsetlinewidth{1.0bp}
\pgfsetrectcap 
\pgfsetmiterjoin \pgfsetmiterlimit{10.0}
\pgfpathmoveto{\pgfpoint{653.500001bp}{306.899988bp}}
\pgfpathcurveto{\pgfpoint{653.500001bp}{272.10605bp}}{\pgfpoint{639.531281bp}{243.899988bp}}{\pgfpoint{622.300004bp}{243.899988bp}}
\pgfpathcurveto{\pgfpoint{605.068719bp}{243.899988bp}}{\pgfpoint{591.099999bp}{272.10605bp}}{\pgfpoint{591.099999bp}{306.899988bp}}
\pgfpathcurveto{\pgfpoint{591.099999bp}{341.693927bp}}{\pgfpoint{605.068719bp}{369.899988bp}}{\pgfpoint{622.300004bp}{369.899988bp}}
\pgfpathcurveto{\pgfpoint{639.531281bp}{369.899988bp}}{\pgfpoint{653.500001bp}{341.693927bp}}{\pgfpoint{653.500001bp}{306.899988bp}}
\pgfclosepath
\color[rgb]{0.0,0.0,0.0}
\pgfusepath{stroke}
\end{pgfscope}
\begin{pgfscope}
\pgfsetlinewidth{1.0bp}
\pgfsetrectcap 
\pgfsetmiterjoin \pgfsetmiterlimit{10.0}
\pgfpathmoveto{\pgfpoint{777.099995bp}{304.499987bp}}
\pgfpathcurveto{\pgfpoint{777.099995bp}{271.031531bp}}{\pgfpoint{763.668537bp}{243.899985bp}}{\pgfpoint{747.099995bp}{243.899985bp}}
\pgfpathcurveto{\pgfpoint{730.531452bp}{243.899985bp}}{\pgfpoint{717.099995bp}{271.031531bp}}{\pgfpoint{717.099995bp}{304.499987bp}}
\pgfpathcurveto{\pgfpoint{717.099995bp}{337.968443bp}}{\pgfpoint{730.531452bp}{365.099988bp}}{\pgfpoint{747.099995bp}{365.099988bp}}
\pgfpathcurveto{\pgfpoint{763.668537bp}{365.099988bp}}{\pgfpoint{777.099995bp}{337.968443bp}}{\pgfpoint{777.099995bp}{304.499987bp}}
\pgfclosepath
\color[rgb]{0.0,0.0,0.0}
\pgfusepath{stroke}
\end{pgfscope}
\begin{pgfscope}
\pgftransformcm{1.0}{0.0}{0.0}{1.0}{\pgfpoint{603.300001bp}{303.099988bp}}
\pgftext[left,base]{\sffamily\mdseries\upshape\huge
\color[rgb]{0.0,0.0,0.0}$L(G')$}
\end{pgfscope}
\begin{pgfscope}
\pgftransformcm{1.0}{0.0}{0.0}{1.0}{\pgfpoint{724.700001bp}{301.099988bp}}
\pgftext[left,base]{\sffamily\mdseries\upshape\huge
\color[rgb]{0.0,0.0,0.0}$R(G')$}
\end{pgfscope}
\begin{pgfscope}
\pgfsetlinewidth{1.0bp}
\pgfsetrectcap 
\pgfsetmiterjoin \pgfsetmiterlimit{10.0}
\pgfsetdash{{10.0bp}{5.0bp}}{0.0bp}
\pgfpathmoveto{\pgfpoint{652.900001bp}{301.499988bp}}
\pgfpathcurveto{\pgfpoint{695.300001bp}{358.499988bp}}{\pgfpoint{670.500001bp}{248.699988bp}}{\pgfpoint{717.100001bp}{301.499988bp}}
\pgfpathcurveto{\pgfpoint{717.100001bp}{301.499988bp}}{\pgfpoint{717.100001bp}{301.499988bp}}{\pgfpoint{717.100001bp}{301.499988bp}}
\color[rgb]{0.0,0.0,0.0}
\pgfusepath{stroke}
\end{pgfscope}
\begin{pgfscope}
\pgfsetlinewidth{1.0bp}
\pgfsetrectcap 
\pgfsetmiterjoin \pgfsetmiterlimit{10.0}
\pgfpathmoveto{\pgfpoint{658.892907bp}{195.563061bp}}
\pgfpathcurveto{\pgfpoint{658.892907bp}{192.307915bp}}{\pgfpoint{656.3618bp}{189.669096bp}}{\pgfpoint{653.239515bp}{189.669096bp}}
\pgfpathcurveto{\pgfpoint{650.117231bp}{189.669096bp}}{\pgfpoint{647.586118bp}{192.307915bp}}{\pgfpoint{647.586118bp}{195.563061bp}}
\pgfpathcurveto{\pgfpoint{647.586118bp}{198.818209bp}}{\pgfpoint{650.117231bp}{201.457027bp}}{\pgfpoint{653.239515bp}{201.457027bp}}
\pgfpathcurveto{\pgfpoint{656.3618bp}{201.457027bp}}{\pgfpoint{658.892907bp}{198.818209bp}}{\pgfpoint{658.892907bp}{195.563061bp}}
\pgfclosepath
\color[rgb]{0.0,0.0,0.0}\pgfseteorule\pgfusepath{fill}
\pgfpathmoveto{\pgfpoint{658.892907bp}{195.563061bp}}
\pgfpathcurveto{\pgfpoint{658.892907bp}{192.307915bp}}{\pgfpoint{656.3618bp}{189.669096bp}}{\pgfpoint{653.239515bp}{189.669096bp}}
\pgfpathcurveto{\pgfpoint{650.117231bp}{189.669096bp}}{\pgfpoint{647.586118bp}{192.307915bp}}{\pgfpoint{647.586118bp}{195.563061bp}}
\pgfpathcurveto{\pgfpoint{647.586118bp}{198.818209bp}}{\pgfpoint{650.117231bp}{201.457027bp}}{\pgfpoint{653.239515bp}{201.457027bp}}
\pgfpathcurveto{\pgfpoint{656.3618bp}{201.457027bp}}{\pgfpoint{658.892907bp}{198.818209bp}}{\pgfpoint{658.892907bp}{195.563061bp}}
\pgfclosepath
\color[rgb]{0.0,0.0,0.0}
\pgfusepath{stroke}
\end{pgfscope}
\begin{pgfscope}
\pgfsetlinewidth{1.0bp}
\pgfsetrectcap 
\pgfsetmiterjoin \pgfsetmiterlimit{10.0}
\pgfpathmoveto{\pgfpoint{722.092907bp}{193.563061bp}}
\pgfpathcurveto{\pgfpoint{722.092907bp}{190.307915bp}}{\pgfpoint{719.5618bp}{187.669096bp}}{\pgfpoint{716.439515bp}{187.669096bp}}
\pgfpathcurveto{\pgfpoint{713.317231bp}{187.669096bp}}{\pgfpoint{710.786118bp}{190.307915bp}}{\pgfpoint{710.786118bp}{193.563061bp}}
\pgfpathcurveto{\pgfpoint{710.786118bp}{196.818209bp}}{\pgfpoint{713.317231bp}{199.457027bp}}{\pgfpoint{716.439515bp}{199.457027bp}}
\pgfpathcurveto{\pgfpoint{719.5618bp}{199.457027bp}}{\pgfpoint{722.092907bp}{196.818209bp}}{\pgfpoint{722.092907bp}{193.563061bp}}
\pgfclosepath
\color[rgb]{0.0,0.0,0.0}\pgfseteorule\pgfusepath{fill}
\pgfpathmoveto{\pgfpoint{722.092907bp}{193.563061bp}}
\pgfpathcurveto{\pgfpoint{722.092907bp}{190.307915bp}}{\pgfpoint{719.5618bp}{187.669096bp}}{\pgfpoint{716.439515bp}{187.669096bp}}
\pgfpathcurveto{\pgfpoint{713.317231bp}{187.669096bp}}{\pgfpoint{710.786118bp}{190.307915bp}}{\pgfpoint{710.786118bp}{193.563061bp}}
\pgfpathcurveto{\pgfpoint{710.786118bp}{196.818209bp}}{\pgfpoint{713.317231bp}{199.457027bp}}{\pgfpoint{716.439515bp}{199.457027bp}}
\pgfpathcurveto{\pgfpoint{719.5618bp}{199.457027bp}}{\pgfpoint{722.092907bp}{196.818209bp}}{\pgfpoint{722.092907bp}{193.563061bp}}
\pgfclosepath
\color[rgb]{0.0,0.0,0.0}
\pgfusepath{stroke}
\end{pgfscope}
\begin{pgfscope}
\pgfsetlinewidth{1.0bp}
\pgfsetrectcap 
\pgfsetmiterjoin \pgfsetmiterlimit{10.0}
\pgfpathmoveto{\pgfpoint{653.500001bp}{195.299988bp}}
\pgfpathlineto{\pgfpoint{749.500001bp}{243.899988bp}}
\color[rgb]{0.0,0.0,0.0}
\pgfusepath{stroke}
\end{pgfscope}
\begin{pgfscope}
\pgfsetlinewidth{1.0bp}
\pgfsetrectcap 
\pgfsetmiterjoin \pgfsetmiterlimit{10.0}
\pgfpathmoveto{\pgfpoint{653.500001bp}{195.899988bp}}
\pgfpathlineto{\pgfpoint{723.700001bp}{341.699988bp}}
\color[rgb]{0.0,0.0,0.0}
\pgfusepath{stroke}
\end{pgfscope}
\begin{pgfscope}
\pgfsetlinewidth{1.0bp}
\pgfsetrectcap 
\pgfsetmiterjoin \pgfsetmiterlimit{10.0}
\pgfpathmoveto{\pgfpoint{717.100001bp}{195.299988bp}}
\pgfpathlineto{\pgfpoint{645.700001bp}{348.299988bp}}
\color[rgb]{0.0,0.0,0.0}
\pgfusepath{stroke}
\end{pgfscope}
\begin{pgfscope}
\pgfsetlinewidth{1.0bp}
\pgfsetrectcap 
\pgfsetmiterjoin \pgfsetmiterlimit{10.0}
\pgfpathmoveto{\pgfpoint{715.900001bp}{195.899988bp}}
\pgfpathlineto{\pgfpoint{617.500001bp}{244.499988bp}}
\color[rgb]{0.0,0.0,0.0}
\pgfusepath{stroke}
\end{pgfscope}
\begin{pgfscope}
\pgfsetlinewidth{1.0bp}
\pgfsetrectcap 
\pgfsetmiterjoin \pgfsetmiterlimit{10.0}
\pgfpathmoveto{\pgfpoint{552.099998bp}{187.499987bp}}
\pgfpathcurveto{\pgfpoint{552.099998bp}{137.794363bp}}{\pgfpoint{534.370475bp}{97.499991bp}}{\pgfpoint{512.500001bp}{97.499991bp}}
\pgfpathcurveto{\pgfpoint{490.629524bp}{97.499991bp}}{\pgfpoint{472.900001bp}{137.794363bp}}{\pgfpoint{472.900001bp}{187.499987bp}}
\pgfpathcurveto{\pgfpoint{472.900001bp}{237.205615bp}}{\pgfpoint{490.629524bp}{277.499987bp}}{\pgfpoint{512.500001bp}{277.499987bp}}
\pgfpathcurveto{\pgfpoint{534.370475bp}{277.499987bp}}{\pgfpoint{552.099998bp}{237.205615bp}}{\pgfpoint{552.099998bp}{187.499987bp}}
\pgfclosepath
\color[rgb]{0.0,0.0,0.0}
\pgfusepath{stroke}
\end{pgfscope}
\begin{pgfscope}
\pgfsetlinewidth{1.0bp}
\pgfsetrectcap 
\pgfsetmiterjoin \pgfsetmiterlimit{10.0}
\pgfpathmoveto{\pgfpoint{877.899998bp}{181.499989bp}}
\pgfpathcurveto{\pgfpoint{877.899998bp}{140.741376bp}}{\pgfpoint{860.707737bp}{107.699986bp}}{\pgfpoint{839.5bp}{107.699986bp}}
\pgfpathcurveto{\pgfpoint{818.292267bp}{107.699986bp}}{\pgfpoint{801.1bp}{140.741376bp}}{\pgfpoint{801.1bp}{181.499989bp}}
\pgfpathcurveto{\pgfpoint{801.1bp}{222.258602bp}}{\pgfpoint{818.292267bp}{255.299989bp}}{\pgfpoint{839.5bp}{255.299989bp}}
\pgfpathcurveto{\pgfpoint{860.707737bp}{255.299989bp}}{\pgfpoint{877.899998bp}{222.258602bp}}{\pgfpoint{877.899998bp}{181.499989bp}}
\pgfclosepath
\color[rgb]{0.0,0.0,0.0}
\pgfusepath{stroke}
\end{pgfscope}
\begin{pgfscope}
\pgfsetlinewidth{1.0bp}
\pgfsetrectcap 
\pgfsetmiterjoin \pgfsetmiterlimit{10.0}
\pgfpathmoveto{\pgfpoint{653.500001bp}{195.899988bp}}
\pgfpathlineto{\pgfpoint{523.300001bp}{273.899988bp}}
\color[rgb]{0.0,0.0,0.0}
\pgfusepath{stroke}
\end{pgfscope}
\begin{pgfscope}
\pgfsetlinewidth{1.0bp}
\pgfsetrectcap 
\pgfsetmiterjoin \pgfsetmiterlimit{10.0}
\pgfpathmoveto{\pgfpoint{652.900001bp}{194.699988bp}}
\pgfpathlineto{\pgfpoint{526.300001bp}{103.499988bp}}
\color[rgb]{0.0,0.0,0.0}
\pgfusepath{stroke}
\end{pgfscope}
\begin{pgfscope}
\pgfsetlinewidth{1.0bp}
\pgfsetrectcap 
\pgfsetmiterjoin \pgfsetmiterlimit{10.0}
\pgfpathmoveto{\pgfpoint{717.100001bp}{193.499988bp}}
\pgfpathlineto{\pgfpoint{829.900001bp}{252.899988bp}}
\color[rgb]{0.0,0.0,0.0}
\pgfusepath{stroke}
\end{pgfscope}
\begin{pgfscope}
\pgfsetlinewidth{1.0bp}
\pgfsetrectcap 
\pgfsetmiterjoin \pgfsetmiterlimit{10.0}
\pgfpathmoveto{\pgfpoint{717.100001bp}{193.499988bp}}
\pgfpathlineto{\pgfpoint{826.300001bp}{112.499988bp}}
\color[rgb]{0.0,0.0,0.0}
\pgfusepath{stroke}
\end{pgfscope}
\begin{pgfscope}
\pgftransformcm{1.0}{0.0}{0.0}{1.0}{\pgfpoint{500.500001bp}{183.899988bp}}
\pgftext[left,base]{\sffamily\mdseries\upshape\huge
\color[rgb]{0.0,0.0,0.0}$A$}
\end{pgfscope}
\begin{pgfscope}
\pgftransformcm{1.0}{0.0}{0.0}{1.0}{\pgfpoint{831.900001bp}{178.299988bp}}
\pgftext[left,base]{\sffamily\mdseries\upshape\huge
\color[rgb]{0.0,0.0,0.0}$B$}
\end{pgfscope}
\begin{pgfscope}
\pgfsetlinewidth{1.0bp}
\pgfsetrectcap 
\pgfsetmiterjoin \pgfsetmiterlimit{10.0}
\pgfpathmoveto{\pgfpoint{716.500001bp}{194.699988bp}}
\pgfpathlineto{\pgfpoint{656.500001bp}{194.699988bp}}
\color[rgb]{0.0,0.0,0.0}
\pgfusepath{stroke}
\end{pgfscope}
\begin{pgfscope}
\pgfsetlinewidth{1.0bp}
\pgfsetrectcap 
\pgfsetmiterjoin \pgfsetmiterlimit{10.0}
\pgfpathmoveto{\pgfpoint{586.099976bp}{34.1bp}}
\pgfpathcurveto{\pgfpoint{586.099976bp}{15.543207bp}}{\pgfpoint{578.57837bp}{0.499994bp}}{\pgfpoint{569.299989bp}{0.499994bp}}
\pgfpathcurveto{\pgfpoint{560.021607bp}{0.499994bp}}{\pgfpoint{552.500001bp}{15.543207bp}}{\pgfpoint{552.500001bp}{34.1bp}}
\pgfpathcurveto{\pgfpoint{552.500001bp}{52.656763bp}}{\pgfpoint{560.021607bp}{67.699976bp}}{\pgfpoint{569.299989bp}{67.699976bp}}
\pgfpathcurveto{\pgfpoint{578.57837bp}{67.699976bp}}{\pgfpoint{586.099976bp}{52.656763bp}}{\pgfpoint{586.099976bp}{34.1bp}}
\pgfclosepath
\color[rgb]{0.0,0.0,0.0}
\pgfusepath{stroke}
\end{pgfscope}
\begin{pgfscope}
\pgfsetlinewidth{1.0bp}
\pgfsetrectcap 
\pgfsetmiterjoin \pgfsetmiterlimit{10.0}
\pgfpathmoveto{\pgfpoint{663.699976bp}{36.3bp}}
\pgfpathcurveto{\pgfpoint{663.699976bp}{17.743207bp}}{\pgfpoint{656.17837bp}{2.699994bp}}{\pgfpoint{646.899989bp}{2.699994bp}}
\pgfpathcurveto{\pgfpoint{637.621607bp}{2.699994bp}}{\pgfpoint{630.100001bp}{17.743207bp}}{\pgfpoint{630.100001bp}{36.3bp}}
\pgfpathcurveto{\pgfpoint{630.100001bp}{54.856763bp}}{\pgfpoint{637.621607bp}{69.899976bp}}{\pgfpoint{646.899989bp}{69.899976bp}}
\pgfpathcurveto{\pgfpoint{656.17837bp}{69.899976bp}}{\pgfpoint{663.699976bp}{54.856763bp}}{\pgfpoint{663.699976bp}{36.3bp}}
\pgfclosepath
\color[rgb]{0.0,0.0,0.0}
\pgfusepath{stroke}
\end{pgfscope}
\begin{pgfscope}
\pgftransformcm{1.0}{0.0}{0.0}{1.0}{\pgfpoint{557.300001bp}{29.899988bp}}
\pgftext[left,base]{\sffamily\mdseries\upshape\Large
\color[rgb]{0.0,0.0,0.0}$L(J)$}
\end{pgfscope}
\begin{pgfscope}
\pgftransformcm{1.0}{0.0}{0.0}{1.0}{\pgfpoint{631.900001bp}{32.699988bp}}
\pgftext[left,base]{\sffamily\mdseries\upshape\Large
\color[rgb]{0.0,0.0,0.0}$R(J)$}
\end{pgfscope}
\begin{pgfscope}
\pgfsetlinewidth{1.0bp}
\pgfsetrectcap 
\pgfsetmiterjoin \pgfsetmiterlimit{10.0}
\pgfsetdash{{3.0bp}{5.0bp}}{0.0bp}
\pgfpathmoveto{\pgfpoint{586.100001bp}{35.899988bp}}
\pgfpathcurveto{\pgfpoint{609.300001bp}{98.899988bp}}{\pgfpoint{606.100001bp}{-12.100012bp}}{\pgfpoint{628.100001bp}{38.299988bp}}
\color[rgb]{0.0,0.0,0.0}
\pgfusepath{stroke}
\end{pgfscope}
\begin{pgfscope}
\pgfsetlinewidth{1.0bp}
\pgfsetrectcap 
\pgfsetmiterjoin \pgfsetmiterlimit{10.0}
\pgfpathmoveto{\pgfpoint{716.300001bp}{194.299988bp}}
\pgfpathlineto{\pgfpoint{580.100001bp}{7.699988bp}}
\color[rgb]{0.0,0.0,0.0}
\pgfusepath{stroke}
\end{pgfscope}
\begin{pgfscope}
\pgfsetlinewidth{1.0bp}
\pgfsetrectcap 
\pgfsetmiterjoin \pgfsetmiterlimit{10.0}
\pgfpathmoveto{\pgfpoint{715.700001bp}{194.899988bp}}
\pgfpathlineto{\pgfpoint{564.500001bp}{66.499988bp}}
\color[rgb]{0.0,0.0,0.0}
\pgfusepath{stroke}
\end{pgfscope}
\begin{pgfscope}
\pgfsetlinewidth{1.0bp}
\pgfsetrectcap 
\pgfsetmiterjoin \pgfsetmiterlimit{10.0}
\pgfpathmoveto{\pgfpoint{764.499976bp}{34.5bp}}
\pgfpathcurveto{\pgfpoint{764.499976bp}{15.943207bp}}{\pgfpoint{756.97837bp}{0.899994bp}}{\pgfpoint{747.699989bp}{0.899994bp}}
\pgfpathcurveto{\pgfpoint{738.421607bp}{0.899994bp}}{\pgfpoint{730.900001bp}{15.943207bp}}{\pgfpoint{730.900001bp}{34.5bp}}
\pgfpathcurveto{\pgfpoint{730.900001bp}{53.056763bp}}{\pgfpoint{738.421607bp}{68.099976bp}}{\pgfpoint{747.699989bp}{68.099976bp}}
\pgfpathcurveto{\pgfpoint{756.97837bp}{68.099976bp}}{\pgfpoint{764.499976bp}{53.056763bp}}{\pgfpoint{764.499976bp}{34.5bp}}
\pgfclosepath
\color[rgb]{0.0,0.0,0.0}
\pgfusepath{stroke}
\end{pgfscope}
\begin{pgfscope}
\pgfsetlinewidth{1.0bp}
\pgfsetrectcap 
\pgfsetmiterjoin \pgfsetmiterlimit{10.0}
\pgfpathmoveto{\pgfpoint{842.099976bp}{36.7bp}}
\pgfpathcurveto{\pgfpoint{842.099976bp}{18.143207bp}}{\pgfpoint{834.57837bp}{3.099994bp}}{\pgfpoint{825.299989bp}{3.099994bp}}
\pgfpathcurveto{\pgfpoint{816.021607bp}{3.099994bp}}{\pgfpoint{808.500001bp}{18.143207bp}}{\pgfpoint{808.500001bp}{36.7bp}}
\pgfpathcurveto{\pgfpoint{808.500001bp}{55.256763bp}}{\pgfpoint{816.021607bp}{70.299976bp}}{\pgfpoint{825.299989bp}{70.299976bp}}
\pgfpathcurveto{\pgfpoint{834.57837bp}{70.299976bp}}{\pgfpoint{842.099976bp}{55.256763bp}}{\pgfpoint{842.099976bp}{36.7bp}}
\pgfclosepath
\color[rgb]{0.0,0.0,0.0}
\pgfusepath{stroke}
\end{pgfscope}
\begin{pgfscope}
\pgftransformcm{1.0}{0.0}{0.0}{1.0}{\pgfpoint{735.700001bp}{30.299988bp}}
\pgftext[left,base]{\sffamily\mdseries\upshape\Large
\color[rgb]{0.0,0.0,0.0}$L(J)$}
\end{pgfscope}
\begin{pgfscope}
\pgftransformcm{1.0}{0.0}{0.0}{1.0}{\pgfpoint{810.300001bp}{33.099988bp}}
\pgftext[left,base]{\sffamily\mdseries\upshape\Large
\color[rgb]{0.0,0.0,0.0}$R(J)$}
\end{pgfscope}
\begin{pgfscope}
\pgfsetlinewidth{1.0bp}
\pgfsetrectcap 
\pgfsetmiterjoin \pgfsetmiterlimit{10.0}
\pgfsetdash{{3.0bp}{5.0bp}}{0.0bp}
\pgfpathmoveto{\pgfpoint{765.100001bp}{32.099988bp}}
\pgfpathcurveto{\pgfpoint{793.700001bp}{111.899988bp}}{\pgfpoint{782.100001bp}{-20.700012bp}}{\pgfpoint{809.500001bp}{39.299988bp}}
\color[rgb]{0.0,0.0,0.0}
\pgfusepath{stroke}
\end{pgfscope}
\begin{pgfscope}
\pgfsetlinewidth{1.0bp}
\pgfsetrectcap 
\pgfsetmiterjoin \pgfsetmiterlimit{10.0}
\pgfpathmoveto{\pgfpoint{655.100001bp}{196.699988bp}}
\pgfpathlineto{\pgfpoint{832.100001bp}{68.899988bp}}
\color[rgb]{0.0,0.0,0.0}
\pgfusepath{stroke}
\end{pgfscope}
\begin{pgfscope}
\pgfsetlinewidth{1.0bp}
\pgfsetrectcap 
\pgfsetmiterjoin \pgfsetmiterlimit{10.0}
\pgfpathmoveto{\pgfpoint{655.700001bp}{196.699988bp}}
\pgfpathlineto{\pgfpoint{818.300001bp}{5.299988bp}}
\color[rgb]{0.0,0.0,0.0}
\pgfusepath{stroke}
\end{pgfscope}
\begin{pgfscope}
\pgfsetlinewidth{1.0bp}
\pgfsetrectcap 
\pgfsetmiterjoin \pgfsetmiterlimit{10.0}
\pgfpathmoveto{\pgfpoint{716.300001bp}{193.699988bp}}
\pgfpathlineto{\pgfpoint{731.300001bp}{33.499988bp}}
\color[rgb]{0.0,0.0,0.0}
\pgfusepath{stroke}
\end{pgfscope}
\begin{pgfscope}
\pgfsetlinewidth{1.0bp}
\pgfsetrectcap 
\pgfsetmiterjoin \pgfsetmiterlimit{10.0}
\pgfpathmoveto{\pgfpoint{718.100001bp}{191.899988bp}}
\pgfpathlineto{\pgfpoint{760.700001bp}{56.899988bp}}
\color[rgb]{0.0,0.0,0.0}
\pgfusepath{stroke}
\end{pgfscope}
\begin{pgfscope}
\pgfsetlinewidth{1.0bp}
\pgfsetrectcap 
\pgfsetmiterjoin \pgfsetmiterlimit{10.0}
\pgfpathmoveto{\pgfpoint{653.900001bp}{195.499988bp}}
\pgfpathlineto{\pgfpoint{630.500001bp}{47.299988bp}}
\color[rgb]{0.0,0.0,0.0}
\pgfusepath{stroke}
\end{pgfscope}
\begin{pgfscope}
\pgfsetlinewidth{1.0bp}
\pgfsetrectcap 
\pgfsetmiterjoin \pgfsetmiterlimit{10.0}
\pgfpathmoveto{\pgfpoint{655.100001bp}{194.899988bp}}
\pgfpathlineto{\pgfpoint{663.500001bp}{40.699988bp}}
\color[rgb]{0.0,0.0,0.0}
\pgfusepath{stroke}
\end{pgfscope}

\end{pgfpicture}